\documentclass[final, 11pt, a4paper, twoside, reqno]{report}
\usepackage[top= 3.0cm, bottom=3.0cm, left=3.2cm, right=3.4cm]{geometry} 
\usepackage[latin1]{inputenc}
\usepackage[T1]{fontenc} 
 \usepackage{amsfonts, amsmath, amssymb, amscd, amsthm, bbm}
\usepackage[english]{babel}
\usepackage{cancel}
\usepackage{epsfig} 
\usepackage{graphicx}
\usepackage{setspace}
\usepackage{color}
\usepackage{framed}

\usepackage{sectsty}
\chapterfont{\huge}
\usepackage[titles]{tocloft}

\makeatletter
\renewcommand{\@pnumwidth}{1.75em}
\renewcommand{\@tocrmarg}{2.75em}
\makeatother

\usepackage{dsfont}
\usepackage{hyperref}
 \usepackage{slashed,enumerate}
\usepackage{epstopdf}
\usepackage[arrow, matrix, curve]{xy}

\usepackage{fancyhdr}

\fancyhfoffset{0.0pt}

\definecolor{grau}{rgb}{0.55,0.55,0.55}% Wir definieren im RGB-Farbraum
\newcommand{\HRule}{\color{grau}{\rule{\linewidth}{0.5mm}}}

\newcommand{\HH}{\mathcal{H}}
\newcommand{\FF}{\mathcal{F}}
\newcommand{\A}{\mathfrak{A}}
\newcommand{\CC}{\mathcal{C}}
\newcommand{\IC}{\mathbb{C}}
\newcommand{\IR}{\mathbb{R}}
\newcommand{\IRinf}{\mathbb{R}\cup\lbrace \pm \infty \rbrace}
\newcommand{\IZ}{\mathbb{Z}}
\newcommand{\IN}{\mathbb{N}}

\newcommand{\cu}{{\mathrm U}}
\newcommand{\Ur}{\cu_{\rm res}}
\newcommand{\Gl}{{\mathrm{GL}}}
\newcommand{\Gres}{\widetilde{\Gl}_{\rm res}}
\newcommand{\GresO}{\widetilde{\Gl}^{0}_{\rm res}}
\newcommand{\Greso}{\Gl^{0}_{\rm res}}
\newcommand{\Ures}{\widetilde{\cu}_{\rm res}}
\newcommand{\scpro}{\langle \cdot , \cdot \rangle}
\newcommand{\Gr}{\mathrm{Gr}(\mathcal{H})}
\newcommand{\Grass}{\mathrm{Gr}}
\newcommand{\Hplusn}{\mathcal{H}_{\lbrace \geq -n \rbrace}}
\newcommand{\GL}{\mathrm{GL}^1(\mathcal{H}_+)}
\newcommand{\Hol}{{\mathrm{Hol}}}

\newcommand{\TC}{I_1(\HH_+)}
\newcommand{\HS}{I_2(\HH)}
\newcommand{\prc}{\mathrm{pr}_{\mathbb{C}}}
\newcommand{\pol}{\mathtt{Pol}}
\newcommand{\sa}{{\sf A}}
\newcommand{\car}{{\mathcal{A(H)}}}
\newcommand{\cc}{{\mathcal{C}}}
\newcommand{\cs}{{\mathcal{S}}}
\newcommand{\csq}{{\mathbb{S}}}
\newcommand{\sk}[1]{\left\langle #1\right\rangle}
\newcommand{\fingbox}{\addtolength{\fboxsep}{5pt}\boxed}

\newcommand{\UU}{{\mathfrak{U}}}
\newcommand{\UUt}{{\widetilde{U}}}
\newcommand{\T}{\operatorname{T}}
\newcommand{\ind}{\operatorname{ind}}
\newcommand{\im}{\operatorname{im}}
\newcommand{\coker}{\operatorname{coker}}
\newcommand{\charge}{\operatorname{charge}}
\newcommand{\trace}{\operatorname{tr}}
\newcommand{\spn}{\mathrm{span}}
\newcommand{\seas}{\mathtt{Seas}(\mathcal{H})}
\newcommand{\iseas}{\mathtt{Seas}^{\perp}(\mathcal{H})}
\newcommand{\ocean}{\mathtt{Ocean}}
\newcommand{\St}{\mathrm{St}(\HH)}
\newcommand{\lop}[1]{{\cal L}_{#1}}
\newcommand{\rop}[1]{{\cal R}_{#1}}
\newcommand{\bwdge}{{\textstyle\bigwedge}}

\newcommand{\wedgeu}{{\sideset{}{_\cu}\bigwedge}}

\numberwithin{equation}{section}

\theoremstyle{plain}
\newtheorem{Theorem}{Theorem}[section]
\newtheorem{Lemma}[Theorem]{Lemma}
\newtheorem{Corollary}[Theorem]{Corollary}
\newtheorem{Proposition}[Theorem]{Proposition}
\newtheorem*{nnprop}{Proposition}
\theoremstyle{definition}
\newtheorem{Definition}[Theorem]{Definition}

\newtheorem*{nndef}{Definition}
\newtheorem{Construction}[Theorem]{Construction}
\newtheorem{Remark}[Theorem]{Remark}
\newtheorem{Example}[Theorem]{Examples}

\begin{document}

%\newpage 
%\thispagestyle{empty}
%\quad 
%\newpage

%DECKBLATT
\begin{titlepage}
\begin{center}
\begin{bfseries}
\Large Ludwig-Maximilians-Universit\"at M\"unchen\\[0.9cm]

\Large Diplomarbeit
\underline{\parbox{\linewidth}{\huge }}

\HRule \\[0.5cm]
%{ \color{black}\LARGE \bfseries Time Evolution in\\[0.4cm] Quantum Electrodynamics\\[0.4cm] with external fields}\\[0.4cm]
{ \color{black}\LARGE \bfseries Time Evolution in the external field problem\\[0.5cm] of Quantum Electrodynamics}\\[0.4cm]

\HRule \\[1cm]
\color{black}
\Large Dustin Lazarovici\\[0.75cm]
\end{bfseries}

\Large betreut von Prof. Dr. Peter Pickl,\\
 Mathematisches Institut der  Ludwig-Maximilians-Universit\"at M\"unchen \\[1cm]

- Zur Erlangung des Grades eines Diplom-Mathematikers - \\[2cm]
\begin{center}
   \includegraphics[width=2in]{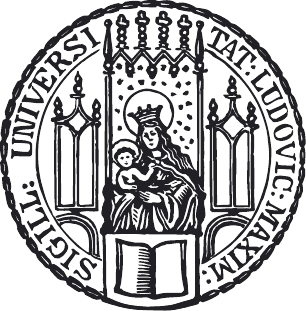}
   %\LARGE{\colorbox{black}{\textcolor{white}{?}\textcolor{red}{?}\textcolor{blue}{?}}}
\end{center}
\vspace{2cm}

\begin{bfseries} \Large M\"unchen, den 30. September 2011\end{bfseries}
\center{\"Uberarbeitete Version vom 19. Oktober 2011.}

%\title{The Geometric Phase in QED}
%\author{Dustin Lazarovici}
%\date{}
%\maketitle
\end{center}
\end{titlepage}

\newpage 
\thispagestyle{empty}
\quad 
\newpage

%INNERE TITELSEITE (ENGL)
%\pagestyle{plain}
%\cfoot{}
\begin{titlepage}
\begin{center}
\begin{bfseries}
\Large Ludwig-Maximilians-Universit\"at M\"unchen\\[0.9cm]

 \Large Diploma Thesis

\underline{\parbox{\linewidth}{\huge }}
\HRule \\[0.5cm]
%{ \color{black}\huge \bfseries Time Evolution in external field\\[0.4cm]Quantum Electrodynamics}\\[0.4cm]
{ \color{black}\LARGE \bfseries Time Evolution in the external field problem\\[0.5cm] of Quantum Electrodynamics}\\[0.4cm]
\HRule \\[1.2cm]
\color{black}
\LARGE Dustin Lazarovici\\[0.75cm]
\end{bfseries}

\Large under supervision of Prof. Dr. Peter Pickl,\\
 Mathematisches Institut der  Ludwig-Maximilians-Universit\"at M\"unchen \\[1cm]
 
 - For the achievement of the degree : Diploma in Mathematics - \\[2cm]

\begin{center}
   \includegraphics[width=2in]{siegel.pdf}
   %\LARGE{\colorbox{black}{\textcolor{white}{?}\textcolor{red}{?}\textcolor{blue}{?}}}
\end{center}
\vspace{2cm}

\begin{bfseries}\Large Submitted: September 30th, 2011 \end{bfseries}

\newpage
  \thispagestyle{empty}
  \vspace*{\stretch{1}}
 % \begin{flushleft}
    %\Large Erstgutachter:  Prof. Dr. D. D\"urr\\[2mm]
    %\Large Zweitgutachter: Prof. Dr. P. Mayr\\[1mm]
  %\end{flushleft}
\cleardoublepage
\newpage 
\thispagestyle{empty}
\quad 
\newpage

%\title{The Geometric Phase in QED}
%\author{Dustin Lazarovici}
%\date{}
%\maketitle
\end{center}
\end{titlepage}
\newpage
\thispagestyle{empty}
\quad 
\newpage
\pagenumbering{roman}
{
\begin{center}\begin{bfseries} \Large Time Evolution in the external field problem\\ of Quantum Electrodynamics \end{bfseries}\end{center}

\renewcommand\thefootnote{\fnsymbol{footnote}}
%\vspace*{1mm}
\begin{center}
\large Dustin Lazarovici \footnote{Mathematisches Institut, LMU M\"unchen. Dustin.Lazarovici@mathematik.uni-muenchen.de}\\ \vspace{0.1cm}
%January 4, 2011
\end{center}
}
\renewcommand{\thefootnote}{\arabic{footnote}}

\vspace*{0.4cm}
\begin{center}{\bf Abstract}\end{center}
\small \noindent

A general problem of quantum field theories is the fact that the free vacuum and the vacuum for an interacting theory belong to different, non-equivalent representations of the canonical (anti-)commutation relations. In the external field problem of QED, we encounter this problem in the form that the Dirac time evolution for an external field with non-vanishing magnetic components will not satisfy the Shale-Stinespring condition, known to be necessary and sufficient for the existence of an implementation on the fermionic Fock space. Therefore, a second quantization of the time evolution in the usual way is impossible.

In this work, we will present several rigorous approaches to QED in time-dependent, external fields and analyze in what sense a time evolution can exist in the second quantized theory. It is shown that the construction of the fermionic Fock space as an ``infinite wedge space'', introduced in \cite{DeDueMeScho}, is equivalent to the more established construction of the Fock space as the space of holomorphic sections in the dual of the determinant bundle over the infinite-dimensional Grassmann manifold (see e.g. \cite{PreSe}). 

Concerning the problem of the time evolution, we first give a comprehensive presentation of the solutions proposed by Deckert et. al. in \cite{DeDueMeScho}, where the time evolution is realized as unitary transformations between time-varying Fock spaces, and by Langmann and Mickelsson in \cite{LaMi}, where the authors construct a ``renormalization'' for the time evolution and present a method to fix the phase of the second quantized scattering operator by parallel transport in the principle fibre bundle $\Gres(\HH) \rightarrow \Gl_{\rm res}(\HH)$.

We provide a systematic treatment of the ``renormalizations'' introduced in $\cite{LaMi}$ and show how they can be used to translate between the second quantization procedure on time-varying Fock spaces and the second quantization of the renormalized time evolution. We argue that non-uniqueness of the renormalizations corresponds to an additional freedom in the construction of the second quantized S-operator in \cite{LaMi} which is expressed by the holonomy of the of the principle bundle $\Ures(\HH)$ and might no have been fully appreciated in the original paper.

We provide rigorous proof for the fact that the second quantization by parallel transport preserves causality. These findings seem to refute claims made in \cite{Scha} that the phase of the second quantized S-matrix is essentially determined by the requirement of causality. The result can also be applied in the context of time-varying, showing that the implementations can be chosen such that the semi-group structure of the time evolution is preserved.

We propose a simple solution to the problem of gauge anomalies in the procedure of Langmann and Mickelsson, showing that the second quantization of the scattering operator can be made gauge-invariant by using a suitable class of renormalizations.

\tableofcontents 

\lhead{}
\onehalfspace
\newpage
\lhead{NOTATION}
\section*{Notation and Mathematical Preliminaries}
\vspace*{0.4cm}
%\noindent \textbf{Notation and Mathematical Preliminaries}\\
\noindent Throughout this work we use ``natural units'' in which $\hbar = c = 1$.\\
We use the Minkowski metric with signature $(+,-,-,-)$.\\

\noindent Furthermore, we introduce the following notations:
\begin{description}
\item[$\HH$] a separable, complex Hilbert space.
\item [$\mathcal{B}(\HH)$] the space of bounded operators on $\HH$.
\item [$\Gl(\HH)$] the space of bounded automorphisms of $\HH$.
\item [$\cu(\HH)$] the group of unitary automorphisms of $\HH$.

\item[$I_p(\HH, \HH')$] the \textit{p-th Schatten class} of linear operators $T: \HH \rightarrow \HH'$ for which \begin{equation*}(\lVert T \rVert_p)^p := \trace[(T^*T)^{p/2}] < \infty. \end{equation*}
$\lVert \cdot\rVert_p$ is a norm that makes $I_p(\HH, \HH')$ a Banach space. 
It satisfies\\ $\lVert AT \rVert_p \leq \lVert A \rVert \, \lVert T \rVert_p$ and $\lVert TB \rVert_p \leq \lVert B \rVert \, \lVert T \rVert_p$, for $A \in \mathcal{B}(\HH'), B \in \mathcal{B}(\HH)$.
Thus, $I_p(\HH, \HH) =: I_p(\HH)$ is a two-sided ideal in the algebra of bounded operators on $\HH$.\\  
If $T = \sum\limits_{k\geq 0} \mu_k \lvert f_k \rangle \langle e_k\rvert$ is a singular-value decomposition, the p-th Schatten norm corresponds to the $\ell^p$ norm on the sequence $(\mu_k)_k$ of singular values.\\
If an operator $T$ is in  $I_p(\HH, \HH')$ for some $p$, then $T$ is compact. 
\end{description}
In particular,
\begin{description}
\item[$I_1(\HH)$] the ideal of \textit{trace-class operators} for which $\trace(T) := \sum\limits_{k\geq0} \sk{e_k, T e_k}$ is well-defined and independent of the Hilbert basis $(e_k)_{k\geq0}$ of $\HH$. 
\item [$I_2(\HH, \HH')$] the class of \textit{Hilbert-Schmidt operators} $\HH \rightarrow \HH'$. The product of two Hilbert-Schmidt operators is in the trace class with $\lVert ST \rVert_1 \leq \lVert S \rVert_2\,\lVert T \rVert _2$. 
\item[$Id + I_1(\HH)$] $= \lbrace T  = Id + A \mid A \in I_1(\HH)\rbrace\;$ the set of operators for which the \textit{Fredholm determinant} is well defined by
\begin{equation*}  \det(1+A) := \sum\limits_{k=0}^\infty \trace\bigl(\sideset{}{^k}\bigwedge (A)\bigr) = \sum\limits_{k=0}^\infty \sum\limits_{i_1 < \dots < i_k} \lambda_{i_1}\cdots\lambda_{i_k}, \end{equation*}
for $(\lambda_n)_n$, the eigenvalues of $A$. The Fredholm determinant has the following properties:
\begin{enumerate}[i)]
\item $\lvert \det(1+A) \rvert \leq \exp(\lVert A \rVert_1)$.
\item $(1 + A)$ is invertible, if and only if $\det(1 +A) \neq 0$.
\item $\det(S) \det(T) = \det(S T)$, for $S, T \in Id + I_1(\HH)$.
\item If $(e_k)_{k\geq0}$ is on ONB of $\HH$, then
$\det(T) = \lim\limits_{N \rightarrow \infty} \det\bigl(\langle e_i, T e_j\rangle\bigr)_{i,j \leq N}$
\end{enumerate}
\item[$\Gl^1(\HH)$] := $\Gl(\HH) \cap \bigl(Id + I_1(\HH)\bigr)$ denotes the set of bdd. isomorphism with well-defined Frdholm determinant.
\item[$\cu^1(\HH)$] := $\cu(\HH)  \cap \bigl(Id + I_1(\HH)\bigr)$ denotes the set of unitary operators with well-defined Frdholm determinant.
\end{description}

\vspace*{0.4cm}
\noindent All further notations will be introduced in the course of the work.

\chapter*{Preface}
\pagestyle{fancy}
\pagenumbering{arabic}
\lhead{PREFACE}
\cfoot{\thepage}
\setcounter{page}{1}
\addcontentsline{toc}{chapter}{Preface}

Quantum Electrodynamics, or short: QED, is widely considered to be the most successful theory in entire physics. Indeed, its predictions have been confirmed time and time again with remarkable precision by various experiments in particle accelerators and laboratories all over the world. Probably the most famous and most spectacular demonstration of the potency of Quantum Electrodynamics is the prediction of the anomalous magnetic moment of the electron, known as ``g - 2'' in the physical literature. This electron g-factor has been measured with an accuracy of 7.6 parts in $10^{13}$, i.e. with a stupendous precision of 12 decimal places and found to be in full agreement with the theoretical prediction (Odom et.al., Phys. Rev. Lett. 97, 030801 (2006)). Actually, we have to be more precise:  since the fine structure constant $\alpha$ enters every QED-calculation as a free parameter, we have to gauge it by other experiments or, equivalently, express every QED-measurement as an independent measurement of $\alpha$. In this sense, theory and experiment are in agreement up to 0.37 parts per billion i.e. to 10 decimal places in the determination of $\alpha$ (Hanneke et.al., Phys. Rev. Lett. 100, 120801 (2008)).
This has often been called the best prediction in physics and whether this is adequate or not, it is certainly very impressive.\\
In one of the standard textbooks on quantum field theory it is even said that \textit{``Quantum Electrodynamics (QED) is perhaps the best fundamental physical theory we have''}. 
(Peskin, Schr\"oder, ``An Introduction to Quantum Field Theory'',  1995 ).\\

\noindent In this light, it might seem surprising that many of the brilliant minds that actually developed the theory were not quite as enthusiastic about it. In a talk given in 1975, P.A.M. Dirac famously expressed:

\begin{quote}\textit{``Most physicists are very satisfied with the situation. They say, Quantum Electrodynamics is a good theory, and we do not have to worry about it any more. I must say that I am very dissatisfied with the situation, because this so-called good theory does involve neglecting infinities which appear in its equations, neglecting them in an arbitrary way. This is just not sensible mathematics. Sensible mathematics involves neglecting a quantity when it turns out to be small - not neglecting it just because it is infinitely great and you do not want it!''} \footnote{What I find most remarkable about this quote is that it takes such a great scientist to state something so obvious and be taken seriously.} 
\begin{flushright}
\footnotesize{Cited after: H. Kragh, Dirac: A scientific biography, CUP 1990}
\end{flushright}\end{quote}

\noindent And even in his Nobel lecture, where the occasion would have excused some enthusiasm, Richard Feynman said:

\begin{quote} \textit{``That is, I believe there is really no satisfactory quantum electrodynamics, but I'm not sure. [...] I don't think we have a completely satisfactory relativistic quantum-mechanical model, even one that doesn't agree with nature, but, at least, agrees with the logic that the sum of probability of all alternatives has to be 100\%. Therefore, I think that the renormalization theory is simply a way to sweep the difficulties of the divergences of electrodynamics under the rug. I am, of course, not sure of that.''}
  \begin{flushright}
    \footnotesize{R.P. Feynman: "The Development of the Space-Time View of Quantum Electrodynamics",
Nobel Lecture (1965)\footnote{http://nobelprize.org/nobel\_prizes/physics/laureates/1965/feynman-lecture.html}}
    \end{flushright}
\end{quote}

\noindent Comparing these two statements with the assessment of Peskin and Schr\"oder and really with the general spirit of the scientific community of today, one might think that great progress has been made on the foundations of QED ever since. Quite frankly, I don't see where. 
Of course, experimental success has steadily strengthened our trust in the usefulness of the framework of quantum field theory, but I don't think this was really the concern expressed by Feynman and Dirac. Over the years, we might have become desensitized to the problems of QED, but it seems to me that we haven't done a very good job at fixing them. So, what then is wrong with QED?\\

For once, more then half a century after its development, Quantum Electrodynamics still lacks a rigorous mathematical formulation. It is well known that QED (just as all realistic quantum field theories) relies on different ``renormalization'' schemes to render its predictions (more or less) finite. And even after renormalization, the S-matrix expansion is widely believed to have \textit{zero} radius of convergence in the coupling constant $\alpha$. No matter how crafty physicists have gotten at manipulating infinities, this fact remains highly unsatisfying from a mathematical point of view. It is also well known that every attempt to axiomatize (3+1 dimensional) QED, i.e. every attempt to base the theory on a consistent set of formal requirements, has failed -- and is bound to fail for any ``serious'' quantum field theory. Interestingly enough, this realization had not so much shaken confidence in the theories themselves but rather ended the program of axiomatizing fundamental physics.

A fact widely ignored by physicist  is that the mathematical deficiencies of QED (and other field theories) do not begin at the computational level --they are really much more basic. The usual formulations of the theory are intrinsically ill-defined and it's not every clear how to write things down in a way that is mathematically meaningful. 
Most physicists seem either not to know or not to care about these kind of problems. This has created a somewhat tragicomical situation. Nowadays, very little work is done on the rotten foundations of the theory. Physicists don't work on it, because they believe the theory is so good that there's nothing left to do. Mathematicians don't work on it, because they find the theory so bad that they don't even know where to start.\\

But apart from all mathematical problems, QED (and really the entire Standard Model of particle physics) is incomplete in a very different sense: it's lacking an ``ontology'', a meaningful interpretation of the mathematical framework, providing a clear and coherent picture of what the theory is really about. I believe that a fundamental theory of nature has to clearly identify the basic elements of its physical description (the \textit{local beables}, in the sense of John Bell), which we can think of as the basic constituents of physical ``reality''. And it has to tell us, how exactly these physical objects correspond to abstract objects or parameters in the mathematical formulation of the theory; indeed, we should insist that this correspondence be as simple and immediate as possible.Thinking about Quantum Electrodynamics, just try in all honesty to answer the question: what is the theory actually about? What are the fundamental physical objects of its description? Is QED essentially a theory about charged particles? This is pretty much the way we use to think and talk about it; but if you inquire a little further, a particle physicist will have to tell you that this is not ``really'' what's going on. And indeed, taking a look into any textbook on quantum field theory or the Standard Model, we're going to find plenty of ``particles'', neatly listed in various tables or swirling around in funny little diagrams  -- yet, there aren't actually any particles in the theory.\footnote{Indeed, in the Standard Model, particles can have mass, charge, spin, even color, flavor or families, but no location in space-time. And therefore no substance as a physical object.}  So, maybe QED is a theory about ``quantum fields''? The expression ``quantum field theory'' might suggest that kind, but the answer is not too convincing, as the role of the ``fields'' in the theory remains rather obscure. Mostly, they seem to appear as a formal device for setting up the perturbation expansions or for deriving the equations of motion from a principle of least action. I might be wrong about that. Even then, however, explaining what the quantum fields are actually supposed to be and how they constitute the physical world that we live in, seems like a formidable task and not many people, who invoke this answer, like to engage in it. Maybe QED is merely about ``transition amplitudes''. To my understanding, this very pragmatic standpoint was advocated by such distinguished scientists as Werner Heisenberg, and it might very well be logically consistent, although I think it requires some mental gymnastics to avoid questions like \textit{what} transitions from \textit{what} into \textit{what}?\\

One might call such concerns ``metaphysical'', but to me, they are as physical as it gets. And with such considerations in mind, the state of modern physics in general and of Quantum Electrodynamics in particular seems pretty bad. I have to repeat, though, that the vast majority of physicists does not share this kind of pessimism. The reasons for this may be historical or sociological or they might, of course, lie in my incomplete understanding of things. There are various possible reasons for this and a detailed discussion would certainly go beyond the scope of this introduction. However, I like the irony in the idea that the very genius of Richard Feynman might have to take some of the blame. His ingenious method of visualizing the formal expansions of quantum field theories by means of the famous diagrams that carry its name, provided us with most of the intuition we have for the physical processes in particle physics and coined the way we use to think and talk about the theory. We are so used to talking about particles scattering from each other, about photons being emitted and absorbed or pairs of virtual particles ``screening" charges and so on, that we tend to forget that \textit{none of this is actually in the theory}. As my first teacher of quantum field theory used to say: \textit{``it's just a nice, cartoonish way to talk about these things.''}\\

Above, I have suggested that physicists don't work on foundations of QED anymore, because they don't see any problems with the theory. This is just partly true. Actually, many physicists \textit{do} acknowledge that the theory is deeply flawed but have -- de facto -- given up on it. Instead, they have adopted the point of view that QED (or rather the entire Standard Model) is indeed not a fundamental theory of nature, but merely a low energy approximation to a fundamental theory (maybe a ``theory of everything'') still waiting to be discovered. All the problems we're facing today, that is the hope, will vanish, once said theory is found. String Theory is usually considered to be the best candidate for such a fundamental theory of nature. Personally, I am sceptical whether the results of String Theory after 20 years of intensive research justify this kind of optimism, but that's a different debate. Anyways, claiming that QED will have to wait for the next big scientific revolution to resolve its various issues might be a valid standpoint. I just have two objections that I have to mention.

First: it has never been like that in the history of physics. Classical mechanics were perfectly well defined and well understood before Special Relativity and Quantum Mechanics came to extend the picture. The Maxwell-Lorentz theory of electromagnetism is a beautiful theory, both physically and mathematically, except for one little detail: the electron self-interaction. This problem was not solved but inherited by the presumingly ``more fundamental'' theory of Quantum Electrodynamics, where it made quite a prominent career under the name of ``ultraviolet divergence''.\footnote{Above all, to this problem, QED added the ``infrared divergence'', so things really just got worse on that front.}

My second objection is this: even if the final answers \textit{do} lie beyond QED, isn't it still important to understand as well as possible what exactly goes wrong and what can and cannot be done? Isn't it possible, even likely, that insights of this kind will lead the way to a new, better behaved, maybe more fundamental description?
As John Bell put it \footnote{Actually, Bell wrote this about non-relativistic Quantum Mechanics, but his appeal seems even more urgent in the context of modern quantum field theories}: 
\begin{quote}\textit{``Suppose that when formulation beyond FAPP} [for all practical purposes] \textit{is attempted, we find an unmovable finger obstinately pointing outside the subject, to the mind of the observer, to the Hindu scriptures, to God, or even only Gravitation? Would not that be very, very interesting?''}
  \begin{flushright}
  \footnotesize{J. Bell, ``Against Measurement'', in: \cite{Bell}}
  \end{flushright}
  \end{quote} 

It is in this spirit that I wrote this thesis and that the research program started by my teachers and colleagues should be understood. The goal is ambitious, yet humble at the same time. We do not expect to ``fix'' QED, or solve all the problems that have troubled so many greater physicists before us. However, we hope to get a better understanding of the difficulties, approach them in a systematic way and see how far one can get with rigorous mathematics. The work I am presenting here is of rather technical nature and mostly concerned with the task of lifting the unitary time evolution to the fermionic Fock space (``second quantization'') in the exterior field problem of QED. I hope that this will provide some insights into the fundamental difficulties of QED in particular and of relativistic quantum theory in general. To me, at least, it was a humbling realization to learn at what basic level the formalism already fails. Ultimately, though, my personal believe is that if we want to make significant progress towards a meaningful, well-defined, fundamental theory we need to think the formal aspects and the conceptual aspects together and return to a way of doing physics that takes both, rigorous mathematics and true physical understanding seriously.

\chapter{Introduction}
\lhead{\small \rightmark}
	
\section{The Dirac Equation}
We begin our study of the external field problem in Quantum Electrodynamics with the one-particle Hilbert space $\HH := L^2(\IR^3, \IC^4)$ of square-integrable $\IC^4$-valued functions. The fundamental equation of motion is the famous Dirac equation
\begin{equation}\label{Diraceq}\addtolength{\fboxsep}{5pt}\boxed{ (i \partial\!\!\!/ - m)\Psi(t) = (i \gamma^{\mu}\partial_{\mu} - m ) \Psi(t) = 0 }\end{equation}
where $\Psi(t) \in L^2(\IR^3, \IC^4)$ for every fixed $t \in \IR$.
The gamma matrices $\lbrace \gamma^0,\gamma^1, \gamma^2,\gamma^3 \rbrace$ form a \mbox{4-dimensional} complex representation of the Clifford algebra $\mathcal{C}l(1,3)$, i.e. they satisfy\\ 
$\lbrace \gamma^{\mu} , \gamma^{\nu} \rbrace = 2 g^{\mu\nu} \cdot \mathds{1}$ where $g^{\mu\nu}$ is the Minkowski metric tensor.

\noindent Using $(\gamma^0)^2 = \mathds{1}$, we can rewrite the Dirac equation in Hamiltonian form as
\begin{equation}\addtolength{\fboxsep}{5pt}\boxed{  i\, \partial_t\, \Psi = D_0\, \Psi :=  \bigl(-i \underline{\alpha} \cdot \underline{\nabla} + m \beta\bigr) \Psi} \end{equation}
Here, the notations $\beta = \gamma^0$ and $\alpha^{\mu} = \gamma^0\gamma^{\mu}$ are common.\\
In the presence of an electromagnetic field described by a vector potential\\ 
$\sa = (\sa_\mu)_{\mu=0,1,2,3} = (\Phi , - \underline{\sa})$, the partial derivative in the Dirac equation \eqref{Diraceq} is replaced by the covariant derivative $ \partial_{\mu} + ie \sa_{\mu}$,  adding to the Hamiltonian the interaction potential 
\begin{equation}\fingbox{ V(t) = e\, \alpha^{\mu} \sa_{\mu} = - e\, \underline{\alpha} \cdot \underline{\sa}  +  e\, \Phi} \end{equation}

\noindent As a side-note we remark that a mathematically more sophisticated description would start with a space-time manifold $M \times \IR$, where $M$ is a (3-dimensional) compact manifold with spin-structure, and realize $\HH$ as the space of $L^2$-sections in a spinor-bundle over $M$ (cf. \cite{LawMi}).

It is well known that the free Dirac Hamiltonian $D_0$ is unstable. It has the continuous spectrum $(-\infty , -m] \cup [+ m , + \infty)$ which gives rise to a splitting of the one-particle Hilbert space $\HH = L^2(\IR^3, \IC^4)$ into two spectral subspaces $\HH = \HH_+ \oplus \HH_-$.
Physical interpretation of the negative energy free states is difficult and has troubled physicists for many years. In particular, as the Hamiltonian is unbounded from below, it would be possible to extract an arbitrary amount of energy from the system, which is unphysical.
To deal with this problems, P.A.M. Dirac proposed the so-called \textit{Dirac sea-} or \textit{hole theory}:
\begin{quote}
  \textit{``Admettons que dans l'Univers tel que nous le connaissons, les \'etats d'energie n\'egative soient presque tous occup\'es par des \'electrons, et que la distribution ainsi obtenue ne soit pas accessible \`a notre observation \`a cause de son uniformit\'e dans toute l'etendue de l'espace. Dans ces conditions, tout \'etat d'energie n\'egative non occup\'e repr\'esentant une rupture de cette uniformit\'e, doit se r\'ev\'evler \`a observation comme une sorte de lacune. Il es possible d'admettre que ces lacunes constituent les positrons.''} \footnotesize{P.A.M. Dirac: \textit{Th\'eorie du Positron}. \cite{Dir}}\\
%  \begin{flushright}
 %   \footnotesize{P.A.M. Dirac. Th\'eorie du Positron. \cite{Dir}}
  % \end{flushright}
\end{quote}

According to Dirac, the negative energy states are almost entirely occupied by an infinite number of electrons -- the Dirac sea -- which due to its homogeneous distribution is hidden from physical observation. The Pauli exclusion principle will then work to prevent transition of positive energy electrons to states of arbitrarily negative energy, so that the system is stable. ``Holes'' in the otherwise homogeneously filled Dirac sea will appear as particles of positive energy but opposite charge and represent what we call positrons. Transition of an electron from the negative energy spectrum to the positive energy spectrum in the presence of an electromagnetic field will look like the simultaneous creation of an electron and a positron to the outside observer. Conversely, when a positive energy electrons drops into an unoccupied state of negative energy, we see the annihilation of an electron/positron pair with energy being emitted in form of radiation.
\vspace*{-2mm}
\begin{figure}[h]
  \begin{center}
    {\includegraphics[scale=0.4]{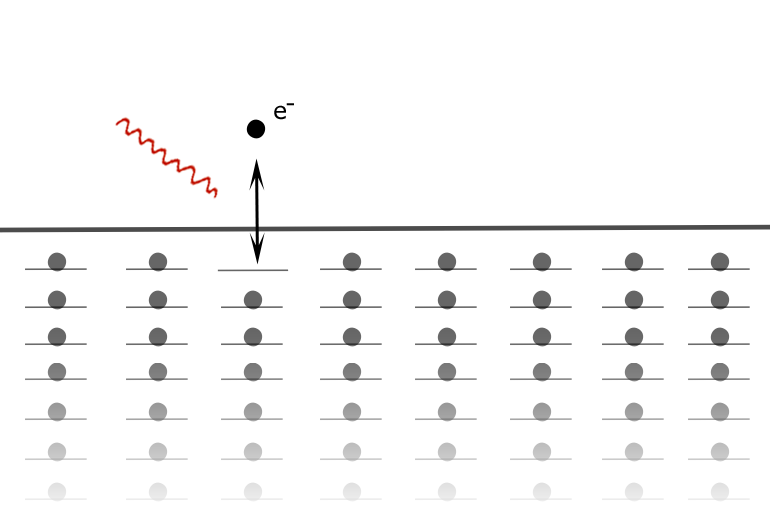}}
    \caption{\label{fig:Dirac Sea}Dirac's hole theory.}
  \end{center}
\end{figure}

Despite the obvious peculiarities of the model -- that is, the existence of an infinite number of particles whose status remains to discuss -- it provides an ingenious picture explaining the most important phenomena of relativistic quantum theory in a clear and elegant way. Admittedly, the Dirac sea is usually absent in ``modern'' presentations of QED, though it seems to me that it hasn't been replaced by an equally compelling physical picture. Anyways, Dirac's theory provides a good intuition for the difficulties of relativistic quantum theory and can be understood as the physical motivation behind most of the mathematically rigorous approaches to the second quantized theory that will be presented in this work.

 \newpage
\singlespace 
\section{An Intuitive Approach}
\noindent From nonrelativistic Quantum Mechanics we are familiar with the fact that the n-fermion Hilbert space is the n-fold exterior product $\bigwedge^n \mathcal{H}$ of the one-particle Hilbert space $\mathcal{H}.$ This space is spanned by decomposable states of the form $v_1 \wedge \dots \wedge v_n$. The Hermitian scalar product is given by: 
\begin{equation}\label{slaterdet} \langle v_1 \wedge \dots \wedge v_n ,
w_1 \wedge \dots \wedge w_n \rangle= \det ( \langle v_i , w_j \rangle)_{i,j} \end{equation}
Under a unitary time evolution $U = U(t_1,t_0 )$, the states evolve in the obvious way:
\begin{equation} v_1 \wedge \dots \wedge v_n \, \xrightarrow{U(t_1,t_0)} \, Uv_1 \wedge \dots \wedge Uv_n\end{equation}

\noindent Projectively (i.e. mod $\IC$), such states are in one-to-one correspondence with n-dimensional subspaces of the Hilbert space $\mathcal{H}$ by 
\begin{equation}v_1 \wedge \dots \wedge v_n \longmapsto \spn(v_1, \dots , v_n) =: V \subset \mathcal{H},\end{equation}
for if we take $w_1, \dots , w_n \in \mathcal{H}$ with $\spn(w_1, \dots , w_n) = \spn(v_1, \dots v_n)$, i.e. a different basis of V, we find
\begin{equation}w_1 \wedge \dots \wedge w_n = \det(R) \, v_1 \wedge \dots \wedge v_n \end{equation}
with $R$, the matrix in $\Gl_n(\mathbb{C})$ transforming the basis $(v_1, \dots , v_n)$ into $(w_1, \dots , w_n)$.\\  
(Note that if $v_1, \dots , v_n$ are not linearly independent, then $v_1 \wedge \dots \wedge v_n = 0$).\\

In the setting of relativistic Quantum Electrodynamics, however, we have a \textit{Dirac sea} containing \textit{infinitely many} particles which makes the situation more complicated. Projective decomposable states are now in correspondence with certain infinite-dimensional subspaces of $\mathcal{H}$, so-called \textit{polarizations} (see Def. \ref{Def:Polarization}). For example, in the unperturbed Dirac sea (the ground state), all negative energy states are occupied and all positive energy states are empty, so this configuration of the Dirac sea corresponds to the subspace $V:= \mathcal{H_-} \subset \mathcal{H}$. Under a unitary transformation $U$, a time evolution $U = U^\sa(t_1,t_0)$, let's say, $V$ evolves into $W := U(V) = UV$.

What does this tell us about the physics? First and foremost we would like to ask: ``How many electrons and how many positrons (holes) were created?''. Thanks to Dirac's ingenious picture, we have a very simple intuition of how to answer this question: we just count! The negative energy states that remain occupied correspond to the subspace $W \cap \mathcal{H_-}$. Consequently, the number of holes is just the codimension of $W \cap \mathcal{H_-}$ in $\mathcal{H_-}$, i.e. the dimension of the factor space $\mathcal{H_-}\slash (W \cap \mathcal{H_-})$.  Similarly, the Dirac sea picture tells us that we can think of the dimensions of $W$ complementary to $W \cap \mathcal{H_-}$ as being populated by electrons that have been lifted from the sea to positive energies. In conclusion: 
\begin{align}\label{partnum}  \# \text{electrons}\; \approx \dim\bigl(W \slash (W \cap \mathcal{H_-})\bigr). \\ \label{holnum} 
\# \text{holes}\; \approx  \dim\bigl(\mathcal{H_-}\slash  (W \cap \mathcal{H_-})\bigr).  \end{align}
The \textit{net-charge} of $W$ is then:
\begin{equation}\label{relcharge} \dim\bigl(W \slash (W \cap \mathcal{H_-})\bigr) - \dim\bigl(\mathcal{H_-}\slash  (W \cap \mathcal{H_-})\bigr). \end{equation}

\noindent In order for all of this to make sense, we have to require that both \eqref{partnum} and \eqref{holnum} are finite. Such polarizations $V$ and $W$ satisfying 
\begin{equation}\label{commensurability} \dim\bigl(W \slash (W \cap V)\bigr) < \infty \;\; \text{and} \; \dim\bigl(V\slash  (W \cap V)\bigr) < \infty \end{equation} are called \textit{commensurable}.
\newpage 
\noindent There is a natural generalization of commensurability (in fact corresponding to a closure in a topological sense):
If $P_V$ and $P_W$ are the orthogonal projections onto $V$ and $W$ respectively, we require that 
\begin{equation*}P_V - P_W\; \text{is a Hilbert-Schmidt operator.}\end{equation*}
In this case, we will say that $V$ and $W$ belong to the same \textit{polarization class}. At first glance, this looks nothing like the condition of commensurability formulated above. But note that $P_V - P_W$ being of Hilbert-Schmidt type just means that $(P_V - P_W)^*(P_V - P_W)$ is in the trace-class, i.e. (since orthogonal projections are self-adjoint) 
\begin{equation}\label{trfinite} \trace(P_V - P_WP_V + P_W -P_VP_W) < \infty. \end{equation}
%Aside we note that \eqref{trfinite} can also be written as
%\begin{equation}\trace(P_{W^{\perp}}P_V + P_{V^{\perp}}P_W) < \infty \end{equation}

This is starting to look more like what we're going for. Orthogonal projections are self-adjoint operators with eigenvalues $1$ (on the respective subspace) and $0$ (on the orthogonal complement). Thus, $\trace(P_V) = \dim(V)$ whenever this is finite. And if $P_V$ and $P_W$ commute (which they usually don't), $P_VP_W = P_WP_V = P_{V \cap W}$ and so $\trace(P_V - P_WP_V)$ really counts the dimension of the orthogonal complement of $W \cap V$ in $V$ and so forth. Therefore, it does indeed make sense to regard \eqref{trfinite} as a generalization of \eqref{commensurability}.\\
The precise relationship between \eqref{trfinite} and \eqref{commensurability} is discussed in Appendix A.1, but we hope that at this point, the reader is convinced that polarization classes are an adequate concept.\\

Similarly, there's an abstract generalization of \eqref{relcharge} counting the net-charges of polarizations : 
If $V$ and $W$ are in the same polarization class, then $\ind(P_W\lvert_{V \rightarrow W})$ is well defined and we call this number the \textit{relative charge} of $V$ and $W$. Again, the motivation becomes more clear if we remember that 
\begin{align*}  \ind(P_W\lvert_{V \rightarrow W}) & = \dim \ker(P_W\lvert_{V \rightarrow W})) - \dim\coker(P_W\lvert_{V \rightarrow W}) \\ &= \dim\ker(P_W\lvert_{V \rightarrow W})) - \dim \, (W\slash P_W(V)) \end{align*}
This coincides with \eqref{relcharge}, whenever the latter is well defined (see Appendix A.1).\\

Returning to our initial setup with $V=\mathcal{H_-}$ and $W = UV$ for the unitary transformation $U \in \cu(\HH)$, we denote by $P_-$ the orthogonal projection onto $\HH_-$ and by $P_+$ the orthogonal projection onto $\HH_+$. In particular, $P_- + P_+ = \mathds{1}$. As $W$ results from $\HH_-$ by the unitary transformation $U$, the orthogonal projection is $P_W = UP_-U^*$. Therefore:

\begin{align*}P_V - P_W \;\text{Hilbert-Schmidt} & \iff P_- - UP_-U^* \hspace{4mm} &\text{Hilbert-Schmidt}\\
& \iff P_-U - UP_- \hspace{6mm} &\text{Hilbert-Schmidt}\\
& \iff P_-U(P_- + P_+) - (P_-+P_+)UP_- \\
& \hspace{4 mm} = \hspace{2mm} P_-UP_+ - P_+UP_-\;&\text{Hilbert-Schmidt.}
\end{align*} 
As $P_-UP_+$ and $P_+UP_-$ map from and into complementary subspaces, this is satisfied if and only if both
\begin{equation}\label{SS}
U_{-+}:= P_-UP_+\vert_{\HH_+\rightarrow \HH_-} \;\text{and}\; U_{+-}:= P_+UP_-\vert_{\HH_-\rightarrow \HH_+}
\end{equation}

\noindent are of Hilbert-Schmidt type. This is known as the \textbf{Shale-Stinespring condition}.\\

In the remainder of this work we will invoke more or less fancy mathematics to construct Fock spaces and implement unitary transformations on them, but in the end, whatever path we take, it always comes down to this: the unitary time evolution would have to fulfill the Shale-Stinespring condition. Otherwise, the Dirac sea becomes too stormy; the electromagnetic field creates infinitely many particles and we have no chance to compare initial and final state in a meaningful way. In other words: we cannot accommodate initial and final state in one and the same Fock space. There is no way around this, at least not in the usual mathematical framework.\\

\newpage
%\begin{figure}[h]
 % \begin{center}
  %  {\includegraphics[scale=0.35]{Commens}}
  %  \caption{\label{fig:Commens} Lifting the time evolution by parallel transport.}
 % \end{center}
%\end{figure}

\noindent It is obvious that we have stumbled upon a serious difficulty, but it might be less obvious how bad the problem really is, because it is not evident how restrictive the Shale-Stinespring criterion turns out to be. We might hope that in most -- if not all -- physically relevant situations, things turn out to be sufficiently nice and everything is going to work out well. But actually, we shouldn't expect that. We should expect something like that, if the difficulties were of a merely technical kind. But this is not the case. The Shale-Stinespring condition is not a petty mathematical prerequisite of the type ``let $f$ twice continuously differentiable''; much rather, it reflects the problem of infinite particle creation which seems to be deeply inherent to relativistic quantum theory. Finally, all remaining hopes are destroyed by the following theorem:

\begin{Theorem}[Ruijsenaars, 1976]\label{Ruij}
\mbox{}
\item Let $\sa = (\sa_0,- \underline{\sa}) \in C_c^{\infty} (\mathbb{R}^4, \mathbb{R}^4)$ be an external vector potential and  $U^\sa(t, t')$ the corresponding Dirac time evolution on $\HH$. Then $U^\sa(t, t')$ satisfies the Shale-Stinespring condition for all $t,t'$ if and only if $\underline{\sa} = 0$, i.e. if and only if the spatial part of the A-field vanishes identically.\footnote{It might seem suspicious that the condition $\sa \equiv 0$ is blatantly not gauge-invariant. This, however, just points to another deficiency of QED. In contrast to common believe, the theory is not gauge-invariant in a naive sense. Gauge transformations as well do generally not satisfy the Shale-Stinespring condition and therefore cannot be implemented as unitary transformations on the Fock space. We'll come back to this problem in Chapters 6 and 7.}
\end{Theorem}

\noindent This and similar results seem not to receive the attention they deserve. Ruijsenaars theorem shows that the unitary time evolution can typically \textit{not} be implemented on the fermionic Fock space. In other words: there is \textit{no well-defined time evolution in Quantum Electrodynamics.} \\ 

\noindent One might wonder why it's only the spatial part of the vector potential ($\approx$ the magnetic field) which causes the problem. The following considerations might provide some intuition:\\
Recall that the free Dirac Hamiltonian, in momentum-space, is
\begin{equation} D_0(p) = \underline{\alpha} \cdot \underline{p} + \beta m \end{equation}
which reduces to $D_0(0)= \beta m$ for a particle at rest.  In the so-called Dirac representation,
\begin{align*}\beta = \begin{pmatrix} \mathds{1} & 0 \\ 0 & -\mathds{1} \end{pmatrix}.\end{align*}
Obviously, the standard basis vectors represent eigenstates with energy $\pm m$.\\
In the presence of an electromagnetic field with vector potential $\sa = (\sa_{\mu})_{\mu=0,1,2,3}$, there is an additional interaction term 
\begin{equation*} V^\sa = e\sum\limits_{\mu=0}^3 \alpha^{\mu}\hat{A}_{\mu},\end{equation*} where $\alpha^0 = \mathds{1}$ and $\hat{A}_{\mu}$ is the Fourier transform of $\sa_{\mu}$ (acting as convolution operators). The electric potential $\Phi= \sa_0$ is harmless, it just shifts the energy of the particles by a finite amount. But the alpha matrices satisfy $\beta \alpha_j = - \alpha_j \beta$ for  $j=1,2,3$, which means that they are mapping negative energy eigenstates of $D_0$ to positive energy eigenstates (and vice-versa). Intuitively speaking, the magnetic field ``rotates'' the Dirac sea into $\HH_+$.\\

%\noindent I think that this result has not received the resonance that it deserves. It shows that typically the unitary time evolution can \textit{not} be lifted to the Fock space.\\
%To physicists it might seem suspicious that the condition $\sa = 0$ is not gauge-invariant, obviously. This, however, just points to another deficiency of QED. Against all common intuition, the theory is not gauge-invariant in a trivial sense. Gauge-transformations (except constant ones) do not satisfy the Shale-Stinespring condition and therefore we have the very same problem as with the time evolution: they cannot be implemented as unitary transformations on the Fock space. We will come back to this, although just briefly, in (X?X).\\

So far, we have only looked at the action of unitary transformations on polarizations which correspond to \textit{projective} states. Given a unitary operator $U$ that does satisfy the Shale-Stinespring condition, we encounter another difficulty when we try to define its action on non-projective states: its lift to the Fock space is well-defined only up to a complex phase. This is a well-known result (for example from representation theory), but our intuitive considerations are good enough to see why this is the case.\\ 
\newpage
\noindent Again we try to generalize the description of an n-fermion system represented in $\bigwedge^n \HH$ to QED-states with an infinite number of particles. We may think of such a state as an infinite wedge product, i.e. a formal expression
\begin{equation*}\Psi\; \widehat{=}\; v_0 \wedge v_1 \wedge v_2 \wedge v_3 \wedge \ldots ,\end{equation*}
where the subspace spanned by the $v_j$ is in the polarization class of $\HH_-$. Now let's consider the simplest case in which the $v_j$ are all eigenstates of the operator $U$ with eigenvalues $e^{i\varphi_j}$.
On the n-fermion state $ v_0\wedge v_1\wedge \ldots \wedge v_n$, the operator $U$ acts like
\begin{equation*}  v_0\wedge v_1\wedge \ldots \wedge v_n \xrightarrow{\;\;U\;}\; e^{i\varphi_0}v_0\wedge e^{i\varphi_1}v_1\wedge\ldots\wedge e^{i\varphi_n}v_n \; = \; e^{i\bigl(\sum\limits_{j=1}^n \varphi_j\bigr)}\; \bigl( v_0\wedge v_1\wedge \ldots \wedge v_n \bigr) \end{equation*}
But again, things become more complicated in the infinite-dimensional case. On $\Psi = v_0 \wedge v_1 \wedge v_2  \wedge \ldots\, $, the unitary transformation $U$ would have to act like
\begin{equation*}  v_0 \wedge v_1 \wedge v_2  \wedge \ldots  \xrightarrow{\;\;U\;} \; Uv_0\wedge Uv_1\wedge Uv_2 \wedge\ldots \;=\; e^{i\varphi_0}v_0\wedge e^{i\varphi_1}v_1\wedge e^{i\varphi_2}v_2 \wedge\ldots\end{equation*}
\hspace{3.5cm}$``=``\hspace{0.6cm} e^{ \displaystyle{i(?)}}\; \bigl(v_0\wedge v_1\wedge v_2 \wedge\ldots \bigr)$.\\

\noindent The product of infinitely many phases does not converge in general and thus the phase of the ``second quantization'' of $U$, acting on the Fock space, is not well-defined.
This $\cu(1)$-freedom is known as the \textit{geometric phase} in QED. ``Geometric'', because we will identify it as the structure group of a principle fibre bundle over the Lie group of implementable unitary operators.\\

\section{Dirac Sea versus Electron-Positron Picture}
Considering the problems described so far, one might observe that all difficulties originate in the fact that we are dealing with an infinite number of particles. Hence one might go on to suggest that the problems might be solved by abandoning the Dirac sea with its infinitely many electrons and switching from the ``electron-hole picture'' to the ``particle-antiparticle'' picture, as is usually done in modern quantum field theory (indeed, most textbooks do suggest, implicitly or explicitly, that this constitutes a significant improvement.) The first statement is, of course, correct. The second one is not. To see why this is so, let's discuss what all this actually means. 

What I'm referring to as the ``particle-antiparticle-picture'' or ``electron-positron-picture'', as opposed to the Dirac sea -  or electron-hole - picture, is a formulation in which the problematic negative spectrum of the Dirac Hamiltonian is fixed in a rather ad hoc way by mapping negative energy solutions of the Dirac equation to positive energy solutions with opposite charge, which are then interpreted as  antiparticles (positrons). The Dirac sea is then omitted altogether as part of the physical description. This reinterpretation is also reflected in the mathematical formalism. In standard physics textbooks, this is usually done in a rather naive way :\\

\noindent The formal quantization of the Dirac field yields the expression
\begin{equation}\label{quantizedfield} \Psi(x) = \int \frac{d^3p}{(2\pi)^3} \frac{1}{\sqrt{2E_p}} \sum\limits_s (a^s_{\underline{p}}  u^s(\underline{p}) e^{-ipx} +  b^{s}_{\underline{p}}\  v^s(\underline{p}) e^{ipx}) \end{equation}    
leading to the second quantized Hamiltonian
\begin{equation} H = \int \frac{d^3p}{(2\pi)^3} \sum\limits_s (E(\underline{p})\; a^{s*}_{\underline{p}}a^s_{\underline{p}} - E(\underline{p})\; b^{s*}_{\underline{p}}b^s_{\underline{p}}). \end{equation}
This Hamiltonian is unstable, i.e. unbounded from below, as every particle of the type created by $b^*$ decreases the total energy by at least its rest-mass $m$.
\newpage 
\noindent Thus, one uses the canonical commutation relations 
\begin{equation} \lbrace b^r_{\underline{p}} , b_{\underline{q}}^{s*} \rbrace = b^r_{\underline{p}} b_{\underline{q}}^{s*} + b^{s*}_{\underline{q}} b_{\underline{p}}^{r} = (2\pi)^3 \delta^3(\underline{p}-\underline{q}) \delta^{r,s} \end{equation}
to write $- b^{s*}_{\underline{q}} b_{\underline{p}}^{s} = + b^s_{\underline{p}} b_{\underline{q}}^{s*} - (2\pi)^3 \; \delta(0)$ and simply interchanges the roles of $b$ and $b^*$.\\
After renaming the operators accordingly, one gets the stable Hamiltonian
\begin{equation}\label{quantizedhamiltonian} H = \int \frac{d^3p}{(2\pi)^3} \sum\limits_s (E(\underline{p})\; a^{s*}_{\underline{p}}a^s_{\underline{p}} + E(\underline{p})\; b^{s*}_{\underline{p}}b^s_{\underline{p}}) \end{equation}
plus an infinite ``vacuum energy'' that physicists boldly get right of by ``shifting the energy'' or, in other words, ignoring it.

\noindent A less playful approach would use the \textit{charge conjugation operator} $\mathcal{C}$, an anti-unitary operator mapping negative energy solutions to positive energy solutions of the Dirac equation with opposite charge. Then construct the Fock space
\begin{equation}\mathcal{F} = \bigwedge \mathcal{H}_+ \otimes  \bigwedge \mathcal{C (H_-)}\end{equation}
and define the \textit{field operator} $\Psi$ which is a complex anti-linear map into the space $\mathcal{B (F)}$ of bounded operators on $\mathcal{F}$. Explicitely,
\begin{equation}\label{CAR1} \Psi: \mathcal{H} \rightarrow \mathcal{B (F)} , \; \Psi (f) = a(P_+ f) + b^*(P_- f) \end{equation}
where $a$ is  the annihilation operator on $\bigwedge \mathcal{H}_+$ and $b^*$ the creation operator on $\bigwedge \mathcal{C (H_-)}$\\
(i.e. for $g \in \mathcal{H_-}$, $b^*(g)$ creates the state $\mathcal{C}g$ in $\bigwedge \mathcal{C (H_-)}$ ).The idea is that $\bigwedge \mathcal{H}_+$ contains the electron-states and  $\bigwedge \mathcal{C (H_-)}$ the positron-states, all of positive energy. The vacuum state $\Omega$ is just $1 = 1 \otimes 1 \in \mathcal{F}$. If the reader excuses some physics jargon, we can say that the field operator $\Psi$ ``creates'' antiparticles and ``annihilates'' particles. Correspondingly, the adjoint operator $\Psi^*$ ``creates'' particles and ``annihilates'' antiparticles.

This construction is carried out in more detail in the next chapter but clearly, as it involves only states of positive energy, it leads to a positive definite second quantized Hamiltonian on the Fock space. Now, we can relate this construction to the Dirac sea description in the following way: 
In the Dirac sea description, the vacuum state corresponds to the ``unperturbed sea'' in which all free, negative energy states are occupied by electrons and all positive energy states are empty. Formally, we may think of this state as an ``infinite-wedge product''. We pick a basis $(e_k)_{k \in \mathbb{Z}}$ of $\mathcal{H}$, s.t. $(e_k)_{k \leq 0}$ is a basis of $\mathcal{H}_-$ and $(e_k)_{k > 0}$ a basis of $\mathcal{H}_+$, and represent the vacuum by the formal expression
\begin{equation*} \Omega = e_0 \wedge e_{-1} \wedge e_{-2}\wedge e_{-3}\wedge \, \dots \end{equation*}
Now, we can define ``field operators'' $\Psi$ and $\Psi^*$ acting on $\Omega$ as follows:
\begin{equation}\begin{split}\label{CAR2} & \mathbf{\Psi}^*(e_k)\Omega = e_0 \wedge e_{-1} \wedge e_{-2}\wedge \dots \wedge e_{k+1} \wedge \cancel{e_k} \wedge e_{k-1} \wedge \, \dots \hspace{5mm}\text{for}\; k<0 \\
& \mathbf{\Psi}(e_{k'})\, \Omega = e_{k'}\wedge e_0 \wedge e_{-1} \wedge e_{-2}\wedge e_{-3}\wedge \, \dots  \hspace{25mm} \text{for}\; k \in \IZ \end{split}\end{equation}
We see that $\Psi^*$ acting on $\Omega$ creates positive energy states and $\Psi$ annihilates negative energy states, i.e creates holes. In particular:
\begin{align*}\mathbf{\Psi}^*(g) \Omega = 0,\; \text{for}\; g \in \mathcal{H_-}\\
\mathbf{\Psi}\,(f)\, \Omega = 0,\; \text{for}\; f \in \mathcal{H_+}\end{align*}
Acting with these operators on $\Omega$ successively, we reach configurations of the Dirac sea i.e. infinite particle states, where finitely many positive energy states and almost all negative energy states are occupied.
Formally, those are (linear combinations of) states of the form 
\begin{equation}\label{infinitewedge}
\Phi = e_{i_0} \wedge e_{i_1} \wedge e_{i_2} \wedge \dots 
\end{equation}
where $(i_0, i_1, i_2,...)$ is a strictly decreasing sequence in $\mathbb{Z}$ with $i_{(k+1)} = i_k -1$ for all sufficiently large indices k.

On the formal level, the transition from the hole theory to what we dubbed electron-positron picture now simply results in mapping (linear combinations of) states of the form \eqref{infinitewedge} into the Fock-space $\mathcal{F} = \bigwedge \mathcal{H}_+ \otimes  \bigwedge \mathcal{C (H_-)}$, such that the holes in the sea ($\sim$ the missing negative indices) are mapped to antiparticle-states and the positive energy states ($\sim$ positive indices) are mapped to the corresponding particle states. The idea is quite simple, although it's somewhat tedious to write things down properly. We will spare the reader the formal details until Section 5.4. Anyways, after introducing the appropriate Fock-space structure on Dirac seas, it's fairly easy to see that these assignments do indeed define an isomorphism and that under this isomorphism, the operator $\Psi$ as defined in \eqref{CAR2} acts just as the field operator in \eqref{CAR1}.\\ 
%\noindent If  $S=(i_0, i_1, i_2,...)$ is a sequence as above, with $i_0 > i_1 > \dots i_N > 0 \geq i_{N+1} > i_{N+2}$ and $\lbrace j_1, \dots j_M \rbrace = \mathbb{Z}^-_0 \backslash S \; , j_1 > j_2 \dots > j_M$ (i.e. the ``holes" in the negative spectrum), then  
%\begin{equation}
%\Phi =  \, e_{i_0} \wedge e_{i_1} \wedge e_{i_2} \wedge \dots  \longmapsto e_{i_N}\wedge e_{i_{N-1}} \dots  \wedge e_{i_0} \otimes \mathcal{C} e_{j_1}\wedge \mathcal{C} e_{j_2}\wedge \dots \mathcal{C} e_{j_M}\end{equation}
%which lies in $\bigwedge^N \mathcal{H}_+ \otimes \bigwedge ^M \mathcal{C (H_-)} \subset \mathcal{F}$.

So we note that the two descriptions are actually mathematically equivalent. If we carry out the whole process of second quantization in the electron-positron picture without any reference to the Dirac sea whatsoever (and we will do that, in the next chapter), we'll arrive at the very same obstacles, in particular at the Shale-Stinespring condition for second quantization of unitary operators. Personally, I would say that the meaning of these results remains more obscure without reference to infinite particle states. Indeed, certain features of the theory seem to indicate that the Dirac sea picture is a more honest description. For instance, in the formal quantization of the Dirac field, the vacuum-charge, just as the vacuum-energy above, appears as an infinite constant even after renaming creators and annihilators. Nevertheless, most physicists nowadays seem to strongly favor the electron-positron description, often calling the idea of a Dirac sea an outdated concept that became obsolete once the theory was properly understood. On the other hand, my personal teachers strongly advocate Dirac's picture. What they ultimately have in mind is that the Dirac sea is just an approximate description of a universe consisting of a very large, yet finite number of charged particles and that electrons or holes can be understood as deviations from an equilibrium state in which the particles are so homogeneously distributed that the net-interaction is zero and therefore the sea ``invisible''. The Dirac sea formalism with infinitely many particles is then best understood as a thermodynamic limit of a fully interacting N-particle theory, whose consistent formulation has not yet been achieved.
Until we have a complete and well-defined fully interacting theory of Quantum Electrodynamics, the choice between the two descriptions ultimately remains a matter of personal taste. However, it should be clear that at this level of description no serious problems can be solved by choosing one over the other.

\chapter{Polarization Classes and the Restricted Transformation Groups}
In this chapter, we give precise definitions for the concepts introduced in Chapter 1 and study the \textit{restricted} unitary- and general linear group of automorphism that do satisfy the Shale-Stinespring condition \eqref{SS} and can be implemented on the fermionic Fock space.\\ 

\noindent By $\HH$ we will always denote an infinite-dimensional, complex, separable Hilbert space.\\

\noindent\textbf{Note:} With the physical context in mind, we would think of $\HH= \HH_+\oplus\HH_-$ as the spectral decomposition w.r.to the free Dirac Hamiltonian and of the subspace $\HH_-$ as the projective ``vacuum state'', corresponding to the Dirac sea. Unfortunately, the mathematical convention doesn't follow the physical motivation and mathematicians for some reason prefer to work with the subspace denoted by $\HH_+$ instead of $\HH_-$. We will follow this convention, in order to match the mathematical literature. Therefore, $\HH_+$ and not $\HH_-$ will play the role of the ``Dirac sea'' in the following and most  of the ``mathematical'' chapters.  

\section{Polarization Classes}

\begin{Definition}[Polarizations]\label{Def:Polarization}
\item A \emph{polarization} of $\HH$ is an infinite-dimensional, closed subspace $V \subset \HH$ with infinite-dimensional orthogonal complement $V^{\perp}$.
\item Accordingly, we will call $\HH$ a \emph{polarized Hilbert space} if we have a distinct splitting\\  
$\HH = V \oplus V^{\perp}$ for a polarization $V$.
\item The set of all polarizations of $\HH$ is denoted by $\pol(\HH)$.
\item By $P_V: \HH \rightarrow \HH$ we denote the orthogonal projection of $\HH$ onto $V$, so $P_V + P_{V^{\perp}} = \mathds{1}_{\HH}$.\\
\end{Definition}

\begin{Definition}[Polarization Classes]
\mbox{}\\
On $\pol(\HH)$, we introduce the equivalence relation
\begin{equation}\label{equivrelationpol} V \approx W \iff  P_V - P_W \in I_2(\HH) \end{equation}
The equivalence classes $C \in \pol(\HH)$ are called \emph{polarization classes}.
\end{Definition}

\begin{Lemma}[Characterization of $\approx$]\label{approx}
\mbox{}\\
For $V,W\in \pol(\HH)$, the following are equivalent:
\begin{enumerate}[i)]
\item $V\approx W$
\item $P_{W^\perp} P_V,\, P_W P_{V^\perp} \in I_2(\HH)$ 
\item$P_W|_{V\to W}$ is a Fredholm operator and $P_{W^\perp}|_{V\to W^\perp}\in I_2(V)$.
\end{enumerate}
\end{Lemma}

\begin{proof}
\mbox{}
\begin{list}{}{}
 \item[ i)$\Rightarrow$ii):]  If $V \approx W$, i.e. $P_V - P_W \in I_2(\HH)$ then
\begin{align*}  P_{W^\perp} P_V = & (Id_{\HH} - P_W) P_V = (P_V - P_W)P_V \in I_2(\HH)\\
P_W P_{V^\perp} = & P_W(Id_{\HH} - P_V) = - P_W (P_V - P_W) \in I_2(\HH)
\end{align*}
\item[ ii)$\Rightarrow$iii):]  We write the identity on $\HH$ in matrix form as
    \begin{equation*}
      Id_{\HH}:V\oplus V^\perp\to W\oplus W^\perp = \begin{pmatrix}
        P_W|_{V\to W} & P_W|_{V^\perp\to W}\\
        P_{W^\perp}|_{V\to W^\perp} & P_{W^\perp}|_{V^\perp\to W^\perp}
      \end{pmatrix}
    \end{equation*}
    By ii), the off-diagonal terms are of Hilbert-Schmidt type, so the operator 
    \begin{equation*}
      \begin{pmatrix}
        P_W|_{V\to W} & 0 \\
        0 & P_{W^\perp}|_{V^\perp\to W^\perp}
      \end{pmatrix}
    \end{equation*}
 is a compact perturbation of the identity. Consequently, $P_W|_{V\to W}$ and $P_{W^\perp}|_{V^\perp\to W^\perp}$ are Fredholm-operators.  
 \item[ ii)$\Rightarrow$i):]  If \begin{align*} & P_{W^\perp} P_V = (P_V - P_W)P_V \in I_2(\HH) ,\\
& P_W P_{V^\perp} = - P_W (P_V - P_W) \in I_2(\HH)\end{align*}
 then \begin{equation*}(P_V - P_W) = (P_V - P_W)P_V +  P_W (P_V - P_W)  \in I_2(\HH)\end{equation*}
 \item[ iii)$\Rightarrow$ii):]$P_W P_{V^\perp} \in I_2(\HH)$ by assumption and as $P_W|_{V\to W}$ is a Fredholm operator, the cokernel of $P_WP_V$ in $W$ is finite-dimensional. Hence $P_{W^\perp} P_V$ is a finite-rank operator, in particular Hilbert-Schmidt.
\end{list}\end{proof}
\noindent By this Lemma, the following is well-defined:
\begin{Definition}[Relative Charge]
\mbox{}\\
For $V,W\in \pol(\HH)$ with $V\approx W$, we define the \emph{relative charge} of $V,W$ to be the Fredholm index of $P_{W}|_{V\to W}$. I.e.:
\begin{equation}\label{relativecharge} \begin{split} \charge(V, W) := & \; \ind(P_W\lvert_{V \rightarrow W}) \\
= & \dim\, \ker(P_W\lvert_{V \rightarrow W}) - \dim\, \coker(P_W\lvert_{V \rightarrow W}) \\ = & \dim\, \ker(P_W\lvert_{V \rightarrow W}) - \dim\, \ker\bigl((P_W\lvert_{V \rightarrow W})^*\bigr). \end{split} \end{equation}
\end{Definition}

\begin{Lemma}[Properties of Relative Charge]
\mbox{}\\
The relative charge has the following intuitive properties:
\begin{enumerate}[i)] 
\item $\charge(V , W) = - \charge(W , V)$
\item $\charge(V , W) + \charge(W , X) = \charge(V , X)$
for $V \approx W \approx X \in \pol(\HH)$. 
\end{enumerate}
\end{Lemma}

\begin{proof} 
For $V,W,X$ from the same polarization class in $\pol(\HH)\slash \approx$ we find 
\begin{align*}
& \charge(V , W) + \charge(W , X)\\
= \; & \ind(P_W\lvert_{V \rightarrow W}) +  \ind(P_X\lvert_{W \rightarrow X}) =  \ind(P_XP_W\lvert_{V\rightarrow X})\\
=\; & \ind\bigl(P_X P_X + P_X (P_W - P_X)\, \bigl\lvert_{V\rightarrow X}\bigr) =  \ind(P_X\lvert_{V\rightarrow X})\\ 
=\; & \charge(V ,X) \end{align*}
The equality in the third line holds because $P_W - P_X$ is of Hilbert-Schmidt type, in particular a compact perturbation, and therefore doesn't change the index.\\ 
As a special case we get 
\begin{equation*}\charge(V , W) + \charge(W , V) = \charge(V ,V) = 0, \end{equation*} which proves the first identity.\\
\end{proof}
\noindent It follows from the Lemma that we get a finer equivalence relation $"\approx_0"$, by setting
\begin{equation} V \approx_0 W :\iff V \approx W\; \text{and}\; \charge(V , W) = 0.\end{equation}
For our purposes, it is often more convenient to work with these \textit{equal charge classes}. Also, the physical principle of charge conservation tells us that the Dirac time evolution should preserve the relative charge.\\ 

In general, unitary transformations do NOT preserve the polarization class - they will map one polarization class into another. This is the source of all evil when we try to implement the unitary time evolution on Fock spaces. However, a unitary transformation induces at a well-defined map between polarization classes by $U[V] = [UV] \in \pol(\HH) \slash \approx$.
This is true, because $P_{UV} - P_{UW} = UP_VU^* - UP_WU^* = U (P_V - P_W) U^*$ is of Hilbert-Schmidt type if and only $P_V - P_W$ is. The analogous map is also well-defined between equal charge classes as follows, for example, from \eqref{chargetransform}.

\begin{Definition}[Restricted Unitary Operators]
\item For polarization classes $C, C' \in \pol(\HH) \slash \approx$ we define
\begin{align*} \Ur(\HH; C, C') =\, & \lbrace U: \HH \rightarrow \HH \;\text{unitary} \mid \forall \, V \in C:\, UV \in C' \rbrace \\
= \,& \lbrace U: \HH \rightarrow \HH\;\text{unitary} \mid \exists\, V \in C \; \text{with}\; UV \in C' \rbrace
\end{align*}
as the set of unitary operators mapping the polarization class $C$ into $C'$.
If we want to restrict to equal charge classes, we write $\Ur^{0}(\HH; C_0, C'_0)$ for $C_0, C'_0 \in \pol(\HH)\slash \approx_0$ etc.. 

\item The definition can be immediately generalized to unitary maps between different Hilbert spaces $\HH$ and $\HH'$.
\end{Definition}
\noindent Note that $\Ur(\HH; C, C')$ is \emph{not} a group, unless $C = C'$.\\ 
However, they satisfy the composition property
\begin{equation}\label{eq:Urescomposition}\Ur(\HH; C', C'')\Ur(\HH; C, C') = \Ur(\HH; C, C'').\end{equation}

\noindent The group $\cu_{\rm res}(\HH; C, C)$ of unitary operators preserving a fixed polarization class $C$ will be of crucial importance for our further discussion. In particular, as argued in the introductory sections, it will turn out that a unitary transformation can be implemented on the (standard) Fock space if and only if it's compatible with the natural polarization $\HH = \HH_+ \oplus \HH_-$ in the sense that it preserves the polarization class $C = [\HH_+]$.
\newpage
\section{The restricted Unitary and General Linear Group}
We want to study the \textit{restricted unitary group} consisting of those operators which do preserve the polarization class $[\HH_+]$ and will turn out to be implementable on the Fock space.

\begin{Definition}[Restricted Unitary Group]\label{UresDef}
\mbox{}\\
The group
\begin{equation*} \Ur(\HH) := \Ur(\HH; [\HH_+], [\HH_+]) \end{equation*}
is called the \emph{restricted unitary group} on $\HH$.
\end{Definition}

\noindent We introduce the notation $\epsilon = P_+ - P_-$ for the sign of the free Dirac-Hamiltonian.

\begin{Lemma}[Characterization of $\Ur(\HH)$]
\mbox{}\\
For a unitary operator $U \in \cu(\HH)$ the following statements are equivalent:
\begin{enumerate}[i)]
\item $U \in \Ur(\HH)$ 
\item $U_{+-} \in I_2(\HH_{-},\HH_{+})$ and $U_{-+} \in I_2(\HH_{+},\HH_{-})$
\item $ [\epsilon , U] \in I_2(\HH)$ 
\end{enumerate}
\end{Lemma}

\noindent Note that $ii)$ is the infamous Shale-Stinespring condition,  which appears very naturally in this setting. 

\begin{proof} Let $U \in \cu(\HH)$. Then:
\begin{align*} U \in \Ur(\HH) & \iff  U\HH_+ \approx \HH_+ \in \pol(\HH) \iff  P_+ -UP_ +U^* \in I_2(\HH)\\
& \iff P_+U - UP_+ \in I_2(\HH)\\
& \iff P_+U(P_+ + P_-) - (P_+ + P_-)UP_- = P_+UP_- - P_-UP_+ \in I_2(\HH)\\
& \iff U_{+-} \in I_2(\HH_{-},\HH_{+})\;\text{and}\; U_{-+} \in I_2(\HH_{+},\HH_{-})
\end{align*}
where we have used that $P_{UV} = UP_ VU^*$ for any $U \in \cu(\HH)$ and any subspace $V\subset\HH$.\\
\noindent Furthermore we compute 
\begin{align*}[\epsilon , U] =\; & (P_+ - P_-)U - U(P_+ - P_-)\\
=\; & (P_+ - P_-)U(P_+ + P_-) - (P_+ + P_-)U(P_+ - P_-)\\
=\; & P_+UP_- - P_-UP_+ + P_+UP_- - P_-UP_+\\
=\; & 2\, \bigl(P_+UP_- - P_-UP_+ \bigr).
\end{align*}
\begin{equation}\begin{split}\label{HSnormwithepsilon}\hspace{-1cm}\text{And therefore\footnote{Note that we use the notation $U^*_{+-}$ for $(U^*)_{+-}$ and not for $(U_{+-})^* = (U^*)_{-+}$}:} \hspace{.5cm} \frac{1}{4}\, \lVert [\epsilon , U] \rVert_{2}^2 = & \lVert (P_+UP_- - P_-UP_+)^{*} (P_+UP_- - P_-UP_+ ) \rVert_{1}\\
= & \lVert (P_-U^{*}P_+ - P_+U^{*}P_-)(P_+UP_- - P_-UP_+) \rVert_{1}\\
= & \lVert P_-U^{*}P_+UP_- + P_+U^{*}P_-UP_+\rVert_{1}\\
= & \lVert U^{*}_{-+}U_{+-}\rVert_1 +  \lVert U^{*}_{+-}U_{-+}\rVert_{1}\\
=  & \lVert U_{+-} \rVert^2_{2} + \lVert U_{-+} \rVert^2_{2}\end{split}\end{equation} 
as the traces of the odd parts vanish. This finishes the proof.\\
\end{proof}
\noindent Thus, we can describe the restricted unitary group as 
\begin{equation}\begin{split}\label{UresCharacterization}
\Ur(\HH) =&\; \lbrace U \in \cu(\HH) \mid [\epsilon, U] \in I_2(\HH) \rbrace\\
= & \; \lbrace U \in \cu(\HH) \mid  U_{+-}\;\text{and}\;U_{-+} \;\text{are Hilbert-Schmidt operators} \rbrace.\\
 \end{split} \end{equation}
\noindent This alternative definition extends immediately to arbitrary isomorphisms. The unitary case is the most relevant one for physical applications, but for the development of the mathematical framework it is very natural and convenient to study the general case as well.

\begin{Definition}[Restricted General Linear Group]
\mbox{}\\
%Consider:
%\begin{equation*} \Gl(\HH) = \lbrace A: \HH \rightarrow \HH \; \text{bounded isomorphism with bounded inverse} \rbrace \footnote{In fact, by bounded inverse theorem, the inverse $A^{-1}$ is automatically bounded if A is bounded and bijective.} \end{equation*}
In analogy to \eqref{UresCharacterization} we define the group
%\footnote{That $\Gl_{\rm res}(\HH)$ is indeed a subgroup of $\Gl(\HH)$, i.e. closed under multiplication, can be easily checked using the fact that the Hilbert-Schmidt operators are a two-sided ideal in the space of bounded operators.}
\begin{equation*}\begin{split}\Gl_{\rm res}(\HH) := &\; \lbrace A \in \Gl(\HH) \mid [\epsilon, A] \in I_2(\HH) \rbrace\\
  = &\; \lbrace A\in \Gl(\HH)  \mid  A_{+-}\;\text{and}\;A_{-+} \;\text{are Hilbert-Schmidt operators}\rbrace. \end{split} \end{equation*} 
$\Gl_{\rm res}(\HH)$ is called the \emph{restricted general linear group} of $\HH$.
\end{Definition}

\noindent With respect to the decomposition $\HH = \HH_+ \oplus \HH_-$ we can write any $A \in \Gl(\HH)$ in matrix form as
\begin{align}
  A = \begin{pmatrix}
    A_{++}\lvert_{\HH_+\to\HH_+} & A_{+-}\lvert_{\HH_-\to\HH_+}\\
    A_{-+}\lvert_{\HH_+\to\HH_-} & A_{--}\lvert_{\HH_-\to\HH_-}
  \end{pmatrix}
  = \begin{pmatrix}
    a & b\\
    c & d
  \end{pmatrix}.
\end{align}
\vspace*{1mm}
\noindent We call $A_{even} = A_{++} + A_{--}$ the even part and $A_{odd} = A_{+-} + A_{-+}$ the odd part of the operator. 
Then, $A$ is in $\Gl_{\rm res}(\HH)$ if and only if $A_{odd} \in I_2(\HH)$, if and only if $b$ and $c$ are Hilbert-Schmidt operators. In this case, as the odd part is just a compact perturbation, it follows immediately that $a$ and $d$ are \textit{Fredholm operators} with $\ind(a) = - \ind(d)$, because
\begin{align}
  0 = \ind(A) = \ind \begin{pmatrix}
    a & b\\
    c & d
  \end{pmatrix} = \ind \begin{pmatrix}
    a & 0\\
    0 & d
  \end{pmatrix} = \ind(a) + \ind(d).
\end{align}
\noindent N.B. that for $A \in \Gl(\HH)$, $A \in \Gl_{\rm res}(\HH)$ is \textit{not equivalent} to $[A\HH_{+}] \approx [\HH_{+}]$ in $\pol(\HH)$, but the first implies the latter. Using Lemma, \ref{approx} iii), we see that for $[A\HH_{+}] \approx [\HH_{+}]$ it suffices that $c$ is Hilbert-Schmidt and $a$ a Fredholm operator which for general isomorphisms is less restrictive than $A \in \Gl_{\rm res}(\HH)$.

\subsection{Lie Group structure of $\Gl_{\rm res}(\HH)$ and $\Ur(\HH)$}
We recall a few basic definitions:

\begin{Definition}(Linear Banach Lie group)
\item Let $(\mathcal{A}, \lVert\cdot\rVert)$ a unital Banach algebra (cf. Def. \ref{DefBanachalgebra}) and 
$\mathcal{A}^{\times}$  its multiplicative group of units. Recall that the exponential map $\exp: \mathcal{A} \to \mathcal{A}^{\times},\;  x \mapsto \sum\limits_{k=0}^{\infty} \frac{x^k}{k!}$ is well-defined and smooth with $D_0\exp x = x, \; \forall x \in \mathcal{A}$. 
\item A \emph{linear Banach Lie group} is a closed subgroup $G$ of $\mathcal{A}^{\times}$ which is locally exponential. This means that if we identify the Lie algebra of $G$ with $\mathrm{L}(G):= \lbrace x \in \mathcal{A} \mid \exp(x) \in G \rbrace$, the exponential map restricted to $L(G)$ is a local diffeomorphism around $0$. 
\item Every (continuous) Lie group homomorphism $\varphi: G_1 \to G_2$ between two linear Banach Lie groups is already smooth. There exists a unique Lie algebra homomorphism $\mathrm{Lie(\varphi)}$ with 
\begin{equation*} \exp_{G_2} \circ\, \mathrm{Lie}(\varphi) = \varphi \circ \exp_{G_1}. \end{equation*}
\end{Definition}
\noindent The last statement is proven in Appendix A.2. 

\begin{Example}(Linear Lie groups)
\begin{enumerate} 
\item For every unital Banach algebra $\mathcal{A}$, $\mathcal{A}^{\times}$ itself is a Lie group. It is always locally exponential due to the Inverse Mapping Theorem for Banach spaces (see for example \cite{Lang}).
\item Let $\HH$ a separable Hilbert-space. The set $\mathcal{B}(\HH)$ of bounded operators, together with the operator norm, is a unital Banach algebra. The group $\cu(\HH)$ of unitary operators is then a real linear Banach Lie group. Its Lie algebra can be easily identified as $\mathrm{L}(\cu(\HH)) = \lbrace iX \mid X\; \text{Hermitian operator on}\;  \HH \rbrace$, which is an algebra over $\IR$. The exponential map is surjective onto $\cu(\HH)$, as follows from the spectral theorem. 
\end{enumerate}
\end{Example}

\noindent For our purposes, we consider the algebra $\mathcal{B}_\epsilon(\HH)$ of all bounded operators $A: \HH \rightarrow \HH$ with $[\epsilon , A] \in I_2(\HH)$. We introduce a norm $\lVert\cdot\rVert_{\epsilon}$, defined by
\begin{equation}\label{EpsilonNorm} \lVert A \rVert_{\epsilon} := \lVert A \rVert + \lVert A_{+-} \rVert_{2} + \lVert A_{-+} \rVert_{2} \end{equation}
(this is equivalent to the norm $\lVert \cdot \rVert + \lVert [\epsilon, \cdot] \rVert_2$, also found in the mathematical literature). The space $\mathcal{B}(\HH)$ of bounded operators is complete in the operator-norm, and $I_2(\HH_+,\HH_-)$ as well as $I_2(\HH_-,\HH_+)$ are complete in the Hilbert-Schmidt norm. Thus, it is easily verified  that $\mathcal{B}_\epsilon(\HH)$ equipped with the norm $\lVert\cdot\rVert_{\epsilon}$ is a unital Banach algebra.
%With this norm, $\mathcal{B}_\epsilon(\HH)$ becomes a \textit{unital Banach algebra}.
%\footnote{Note that$\lVert A_{+-}\rVert_2 + \lVert A_{-+}\rVert_2 \leq \lVert [\epsilon , A] \rVert_{2} = 2 \sqrt{ \lVert A_{+-}\rVert_2^2 + \lVert A_{-+}\rVert_2^2} \leq 2(\lVert A_{+-}\rVert_2 + \lVert A_{-+}\rVert_2)$. The space $\mathcal{B}(\HH)$ of bounded operators is complete in the operator-norm, and $I_2(\HH_+,\HH_-)$ as well as $I_2(\HH_-,\HH_+)$ are complete in the Hilbert-Schmidt norm, therefore $(\mathcal{B}_\epsilon(\HH), \lVert\cdot\rVert_{\epsilon}$ is a Banach space.}
Now, $\Gl_{\rm res}(\HH) = \mathcal{B}^{\times}_\epsilon$ is just the multiplicative group of units in $\mathcal{B}_\epsilon$ and $\Ur(\HH)$ is the closed subgroup of unitary operators in $\mathcal{B}_\epsilon(\HH)$. It follows that $\Gl_{\rm res}(\HH)$ is a complex linear Lie group and $\Ur(\HH)$ a real linear Lie group. The corresponding Lie algebras are
\begin{align}
\mathfrak{gl}_1
= & \lbrace X \;\text{bounded operator on}\; \HH \mid [\epsilon, X] \in I_2(\HH) \rbrace = \mathcal{B}^{\times}_\epsilon\\\notag
\\
\mathfrak{u}_{\rm res} = &  i \cdot \,\lbrace X\; \text{Hermitian operator on}\; \HH \mid [\epsilon, X] \in I_2(\HH) \rbrace.
\end{align}
%\footnotetext{For $\mathfrak{u}_{\rm res}$ we use the physicist convention with an additional factor of $i^{-1}$ in front of the commutator and $i$ in the exponential map.}

\noindent Geometrically, $\Gl_{\rm res}(\HH)$ can also be understood as the \textit{complexification} of $\Ur(\HH)$. In particular, it inhibits $\Ur(\HH)$ as a real Lie subgroup and $\mathfrak{gl}_1\cong \mathfrak{u}_{\rm res} \otimes_{\IR} \IC$.\\
\noindent Finally, two more results about the topology of $\Gl_{\rm res}(\HH)$.

\begin{Lemma}[Homotopy type of $\Gl_{\rm res}$]
\mbox{}\\
The map 
\vspace{-2mm} \begin{align*} A = \begin{pmatrix}
    a & b\\
    c & d
  \end{pmatrix} \longmapsto a\end{align*}
  from $\Gl_{\rm res}(\HH)$ to the space $\mathrm{Fred}(\HH_+)$ of Fredholm operators on $\HH_+$ is a homotopy equivalence.
In particular, $\Gl_{\rm res}(\HH)$ has infinitely many connected components 
indexed by $\mathbb{Z}$, corresponding to the index of the $(++)$-component $a$ of $A$. 
\end{Lemma}

\noindent The last conclusion is true, because $\ind: \mathrm{Fred}\rightarrow \IZ$ is continuous. We denote the n-th connected component by $\Gl^{(n)}_{\rm res}(\HH)$. In particular, $\Gl^0_{\rm res}(\HH)$ denotes the identity component of $\Gl_{\rm res}(\HH)$. For $\Ur(\HH)$ we use the analogous notations.\\ 

\noindent The Fredholm-index of $a$ has a very important physical interpretation.\\ 
For $V \approx \HH_+ \in \pol(\HH)$ and $A \in \Gl_{\rm res}(\HH)$, we find that 
\begin{align*} \charge(AV, \HH_+) = & \ind(P_+ \lvert_{AV \rightarrow \HH_+}) = \ind(P_+A\lvert_{V \rightarrow \HH_+})\\
= & \ind\bigl(P_+A(P_+ + P_-)\lvert_{V \rightarrow \HH_+}\bigr) = \ind\bigl(( P_+AP_+ + P_+ AP_- )\lvert_{AV \rightarrow \HH_+}\bigr)\\
= & \ind(P_+AP_+ \lvert_{V \rightarrow \HH_+})= \ind(P_+AP_+ \circ P_+ \lvert_{V \rightarrow \HH_+})\\
= & \ind(P_+\lvert_{V \rightarrow \HH_+}) +  \ind(P_+AP_+\lvert \HH_+ \rightarrow \HH_+)
\end{align*}
Thus: \begin{equation}\label{chargetransform}\addtolength{\fboxsep}{5pt} \boxed{\charge(AV, \HH_+) = \charge(V, \HH_+ ) + \ind(a)} \end{equation}

\vspace*{3mm}\noindent So, physically, the index of the (++)-component corresponds to the net-charge that the transformation $A$ ``creates'' from the vacuum ($\sim \HH_+$,  in this convention).

\begin{Lemma}[Homotopy Groups of $\Greso$]
\mbox{}\\
The homotopy groups of the connected Lie group $\Greso(\HH)$ are for $k \geq 0$
\begin{equation*}\pi_{2k+1}(\Greso) = \lbrace 0 \rbrace \; \text{and} \;\; \pi_{2k+2}(\Greso) \cong \IZ \end{equation*}
In particular, $\Greso(\HH)$ and $\cu^{0}_{\rm res}(\HH)$ are simply-connected.\\
\end{Lemma}

\noindent For the proofs of the Lemmatas see \cite{PreSe} \S6 or \cite{Wurz06}, for a more complete version.

\newpage
\section{The restricted Grassmannian}
\label{Sec:Grassmannian}
To the spectral decomposition $\HH = \HH_+ \oplus \HH_-$ corresponds the polarization class \\
$[\HH_+] \in \pol(\HH) \slash_\approx$. This set of subspaces of $\HH$ is known in the literature as the \textit{(restricted) Grassmannian} of the Hilbert space and denoted by $\Gr$ (the notations $\mathrm{Gr}_{\rm res}(\HH)$ or $\mathrm{Gr}_1(\HH)$ are also common). From the mathematical point of view, it's a remarkably nice object. It can actually be given the structure of an infinite-dimensional complex (K\"ahler) manifold.\\

\noindent By Lemma \ref{approx} we can describe $\Gr$ as the set of all closed subspaces $W$ of $\HH$ such that
\begin{enumerate}[i)]
\item the orthogonal projection $P_+: W \rightarrow \HH_+$ is a Fredholm operator
\item the orthogonal projection $P_-: W \rightarrow \HH_-$ is a Hilbert-Schmidt operator.\\
\end{enumerate}

\noindent This is the definition most commonly found in the mathematical literature (e.g. \cite{PreSe}, Def. 7.1.1).
Another way of saying the same thing is the following:\\
A subspace $W$ belongs to $\Gr$ if and only if it's the image of an operator
$w: \HH_+ \rightarrow \HH$ with $P_+ \circ w \in \mathrm{Fred}(\HH_+)$ and $P_- \circ w \in I_2(\HH_+ , \HH_-)$. \\

\noindent We can cover $\Gr$ by the sets $\bigl\lbrace U_W \bigr\rbrace_{W \in\Gr}$\, , where 
\begin{equation}\label{U_W} U_W := \lbrace W' \in\Gr \mid P_W\vert_{W'\rightarrow W}\; \text{is an isomorphism}\,\rbrace.\end{equation} 
The elements of $U_W$ are ``close'' to $W$ in the sense that they differ from $W$ only by a Hilbert-Schmidt operator $T : W \rightarrow W^{\perp}$. We make this more precise:

\begin{Lemma}[Characterization of $U_W$]
\mbox{}\\
The set $U_W$ consists of all graphs of Hilbert-Schmidt operators $W \rightarrow W^{\perp}$. There is a bijection between $U_W$ and $I_2(W , W^{\perp})$.
\end{Lemma}

\begin{proof} Let $W \in \Gr$ a polarisation and $T: W \rightarrow W^{\perp}$ be a Hilbert-Schmidt operator.
Let $w:\HH_+ \rightarrow W$ an isomorphism with image $W\subset \HH$. Then:
\begin{equation*}\mathrm{Graph(T)} = \lbrace (\mathrm{w} , T\,\mathrm{w}) \mid \mathrm{w} \in W \rbrace = \mathrm{im}(w \oplus T\,w\bigl\lvert_{W \to W \oplus W^{\perp}}) = \mathrm{im}(w + T\,w\bigl\lvert_{W \to \HH}).\end{equation*}
(under the canonical identification $W\oplus W^{\perp} \cong \HH$). Now, since $W \in \Gr$, it follows that $P_+(w + T\circ w) = P_+w \,+\, P_+Tw$ is a Fredholm-operator and $P_-(w + T\circ w) = P_-w + P_-Tw$ is Hilbert-Schmidt, hence $\mathrm{Graph}(T) \in \Gr$.
Furthermore, the orthogonal projection $P_W \bigl\lvert_{\mathrm{Graph(T)} \to W}$ is obviously an isomorphism, whose inverse is $T$, so that $\mathrm{Graph}(T) \in U_W$.

\item Conversely, if $W' \in U_W$, we can set 
\begin{equation*} (P_W\vert_{W'\rightarrow W})^{-1} = (\mathrm{Id}_W + T)\lvert_{W \to W'}\;  \text{with} \;T: W \rightarrow W^{\perp}.\end{equation*} 
Since $W' \approx W$, it follows from Lemma \ref{approx} that $P_{W^\perp}\lvert_{W'\rightarrow W^{\perp}}$ is Hilbert-Schmidt, hence 
\begin{equation*} P_{W^\perp}\lvert_{W'\rightarrow W^{\perp}}\,(\mathrm{Id}_W + T)\lvert_{W \to W'} = T \end{equation*} is Hilbert-Schmidt. Finally, it is easy to see that the assignments $T \longleftrightarrow W' = \mathrm{Graph}(T)$ are inverses of each other, so the statement of the Lemma is proven.\\
\end{proof}

\begin{Proposition}[Manifold Structure of $\Gr$]
\mbox{}\\
$\Gr$ is a complex Hilbert manifold modelled on $I_2(\HH_+ , \HH_-)$.\\
The sets $U_W, W \in \Gr$ form an open covering of $\Gr$.
\end{Proposition}
\begin{proof} \cite{PreSe}, Prop. 7.1.2\\ \end{proof}

\noindent By definition, a unitary transformation maps $\Gr$ into $\Gr$ if and only if it is in $\Ur(\HH)$. Conversely, for any two polarizations $W, W'$ of $\HH$ there exists $U \in \cu(\HH)$ with $W' = UW$. Then, if $W, W'$ belong both to 
$\Gr$, this transformation $U$ is necessarily in $\Ur(\HH)$.\\
 
\noindent In other words: $\Ur(\HH)$, and consequently also $\Gl_{\rm res}(\HH)$, \textit{act transitively} on $\Gr$.  
This leads us to a different description of the Grassmann manifold as a \textit{homogeneous space} under $\Ur(\HH)$ or $\Gl_{\rm res}(\HH)$. 
The corresponding isotropy groups of $\HH_+ \in \Gr$ are
\begin{align*} & P := \Biggl\lbrace A = \begin{pmatrix}
    a & b\\
    c & d
  \end{pmatrix} \in \Gl_{\rm res}(\HH) \; \Biggl\lvert \; c = 0 \Biggr\rbrace\end{align*}
  respectively
  \begin{align*}
  &Q := \Biggl\lbrace U = \begin{pmatrix}
    a & b\\
    c & d
  \end{pmatrix} \in \Ur(\HH)\;  \Biggl\lvert \; b=c = 0 \Biggr\rbrace
\end{align*}
This means: 
\begin{equation}\label{homogspace} \addtolength{\fboxsep}{5pt} \boxed{\Gl_{\rm res}(\HH) \slash P \; \cong \Ur(\HH) \slash Q \; \cong \Gr}\end{equation}

\noindent In particular, it follows that, just as $\Gl_{\rm res}(\HH)$ and $\Ur(\HH)$, the Grassmannian $\Gr$ has $\mathbb{Z}$ connected components corresponding to the relative charges 
\begin{equation*} \charge(W , \HH_+) = \ind(P_+\vert_{W \rightarrow \HH_+}).\end{equation*}
Alternatively, we could have noted that \\
\begin{equation*} \charge(W' , \HH_+) = \charge(W, \HH_+),\; \forall\, W' \in U_W\, \forall\, W \in \Gr \end{equation*}
and hence the sets $\mathrm{Gr}^{(c)}(\HH):=\lbrace W \in \Gr \mid \charge(W, \HH_+) = c \in \mathbb{Z} \rbrace$ are open and disjoint. They are connected, because $\Ur^0(\HH)$ acts continuously and transitively on each of them. Consequently, they correspond to different connected components.\\  

It is a nice feature of the mathematical structure that different charges are separated
topologically. This also reflects the physical intuition that a continuous time evolution preserves the total charge i.e. that particles and anti-particles are always created in pairs.\\
%We hope that the reader excuses that we won't always be too careful in distinguishing between $\Gl_{\rm res}(\HH)/\Gr$ and their identity components $\Gl^{0}_{\rm res}(\HH)/\mathrm{Gr}^{0}(\HH)$. Often, it is just the latter that we care for. Extension to arbitrary charges is usually unproblematic, but can become rather tedious. We will discuss it for example on \S 4.1.2.

Finally, we can also use the transitive action of $\Ur(\HH)$ on $\Gr$ to define a $\Ur$-invariant Hermitian form on the restricted Grassmannian. Since $\Gr$ is a complex Hilbert-manifold modelled on $I_2(\HH_+, \HH_-)$, the tangent space at $\HH_+ \in \Gr$ is naturally isomorphic to $I_2(\HH_+, \HH_-)$. This space carries an  inner product
\begin{equation}\label{I2innerproduct} (X,Y) \mapsto 2 \trace(X^*Y), \end{equation}  
which is obviously invariant under the action of the isotropy group $\cu(\HH_+)\times \cu(\HH_-)$. Therefore, we can extend $\eqref{I2innerproduct}$ to a $\Ur$-invariant Hermitian form on $\Gr$ (by pull-backs with appropriate unitary transformations). In fact, this Hermitian form defines a K\"ahler structure on $\Gr$, turning it into an infinite-dimensional K\"ahler-manifold (\cite{PreSe} \S 7.8).

\chapter{Projective Representations and Central Extensions}
In this section we introduce the concept of \textit{central extensions} of (Lie-)groups which is essential for the treatment of representations of Lie groups in a quantum mechanical setting. 
%Some of the results will go a little beyond what we actually need for our further discussion, but as we develop the formalism anyways, it would be a pity not to mention them. 
Again, we will start with a rather intuitive approach to motivate the concept.

\section{Motivation}
We have argued -- and will state rigorous results -- that a unitary operator can be lifted to the fermionic Fock space $\FF$ if and only it satisfies the Shale-Stinespring condition i.e. if and only if it's in $\Ur(\HH)$. In this case, we will see that the implementation, i.e. the ``second quantization'' of the operator is well-defined only up to a complex phase of modulus one. In this chapter, we are concerned with this phase-freedom. Suppose we are trying to tackle this latter problem. Suppose we choose \textit{any} prescription for fixing the phase and denote by $\Gamma(U)$ the corresponding lift of $U \in \Ur$ to an operator on the Fock space $\FF$. This yields a map
\begin{equation}\Gamma: \Ur(\HH) \rightarrow \cu(\mathcal{F})\end{equation}
from $\Ur(\HH)$ into the group of unitary automorphisms of the Fock space. Such a map is what's usually - and somewhat mysteriously - called a ``second quantization'' of the unitary operators. However, fixing the phases of the lifts in some arbitrary way, we have no reason to expect that these lifts will preserve the group structure: for any $U,V \in \Ur(\HH)$, $\Gamma(U) \Gamma(V)$ and $\Gamma(UV)$  are both implementations of the same unitary operator but will, in general, differ by a complex phase. In other words, $\Gamma$ will fail to be a \textit{representation} of the restricted unitary group $\Ur$ on $\mathcal{F}$. Instead, we merely get a \textit{projective representation} of $\Ur$ on the projective Fock space $\mathbb{P}(\mathcal{F}) = \mathcal{F}\setminus\lbrace0\rbrace \; mod \; \mathbb{C}$.\\

\noindent Now, we would like to ask: is it possible to fix the phases of the lifts in such a way as to preserve the group structure? In other words:
\begin{center}\textit{Does a proper representation of $\Ur(\HH)$ on the Fock space exist?} \end{center}
The answer to this question will require a more sophisticated analysis. First, we observe that we can always write
\begin{equation} \Gamma(U) \Gamma(V) = \chi(U,V)\, \Gamma(UV)\end{equation}
for $U,V \in \Ur$ with $\chi(U,V) \in \cu(1)$. This defines a map $\chi: \Ur \times \Ur \rightarrow \cu(1)$.\\ 
Reasonably, we demand $\Gamma(\mathds{1})=\mathds{1}_{\FF}$, which implies $\chi(\mathds{1},\mathds{1}) = 1$. Such a map is called a \textit{2-cocycle} or a \textit{factor set}. Using this factor set, we can define a group 
\begin{equation*}\cu(1) \times_{\chi} \Ur \end{equation*} as the direct product $\cu(1) \times \Ur$ with the multiplication given by
\begin{equation} (a, U ) \cdot_{\chi} (b, V) := \bigl( \chi(U,V) \, ab , UV \bigr) \end{equation} 
and set $\widehat{\Gamma}\bigl( (a,U) \bigr) := a \, \Gamma(U)$. But the multiplication on $\cu(1) \times_{\chi} \Ur$  was defined precisely in such a way as to compensate the cocycle coming from $\Gamma$ and make $\widehat{\Gamma}$ a homomorphism of groups:
\begin{align*}\widehat{\Gamma}\bigl( (a,U) (b,V) \bigr) & = \widehat{\Gamma} \bigl((\chi(U,V) \, ab , UV)\bigr) = ab \, \chi(U,V) \Gamma\bigl(UV\bigr)\\&= ab\, \Gamma(U)\,\Gamma(V) = \widehat{\Gamma}\bigl( (a,U) \bigr)\; \widehat{\Gamma}\bigl( (b,V) \bigr) \end{align*}
The group $\cu(1) \times_{\chi} \Ur$ is called a \textit{central extension} of $\Ur$ by $\cu(1)$.\\

\noindent This is quite nice. Indeed, there is an intimate connection between central extensions and the problem of lifting projective representations to proper representations that we are going to exploit in the following chapter. Of course, things aren't quite as simple as we have presented them so far and thus a more general approach is required.\\
%In particular, representations of  $\Ur$ itself, would just correspond to a homomorphisms 
%$\sigma: \Ur \rightarrow \cu(1) \times_{\chi} \Ur$ with $\mathrm{pr}_2 \circ \sigma \equiv Id$, picking for every $U \in \Ur$ one of the corresponding elements in the bigger group $\cu(1) \times_{\chi} \Ur$ which does have a representation. This enables us to study algebraic/topological/geometrical obstructions to constructing such ``sections''.\\

\section{Central Extensions of Groups}

Throughout this section, let $G$ be an arbitrary group and $A$ an abelian group. The trivial group, consisting of the neutral element only, is denoted by 1. All definitions and results apply directly to topological groups and Lie groups, if the appropriate structures on the groups and continuity, respectively smoothness, of the maps are implied.
  
\begin{Definition}[Central Extension]
\mbox{}\\
A central extension of G by A is a short exact sequence of group homomorphisms
\begin{equation*} 1 \longrightarrow A \xrightarrow{\;\;\imath\;} E \xrightarrow{\;\;\pi\;} G \longrightarrow 1\end{equation*}
such that $\imath(A)$ is in the center of $E$ (i.e commutes with all the elements in $E$). \end{Definition}

Let's dissect this abstract definition: 
The sequence being \textit{exact} means that the kernel of each map equals the image of the previous map.
So $\imath$ must be injective, $\pi$ surjective and $\ker(\pi) = \im(\imath) \cong A$. Thus, $E$ is a covering of $G$ and the preimage of every $g \in G$ is isomorphic to $A$. The requirement that $\imath(A)$ is central in $E$ is crucial in the context of studying (projective) representations of $G$, since then, \textit{Schurr's Lemma} ensures that in any irreducible representation, the images of $A$ in $E$ will be constant multiples of the identity. In a certain sense, $A$ will carry the information about the phases that lead to ambiguities when trying to lift a projective representation of $G$ to a proper one.\\

\begin{Example}[Trivial extension and universal covering group]\label{CEexamples}
\mbox{}
\begin{enumerate}
\item A \textit{trivial extension}  has the form
\begin{equation*}1 \longrightarrow A \xrightarrow{\;\;\imath\;} A \times G \xrightarrow{\;\;pr_2\;} G \longrightarrow 1\end{equation*}
where $E$ is just the direct product of $A$ and $G$, $i(a) = (a, 1), \; \forall a \in A$ and $pr_2$ is the projection onto the second entry.  
%\item The exact sequence 
%\begin{equation*}1 \longrightarrow \mathbb{Z}/<n> \rightarrow \cu(1) \xrightarrow{\;\pi\;} \cu(1) \longrightarrow 1\end{equation*}
%with $\pi(z):= z^n$ is a central extension which is non-trivial for $n\geq 2$, because $\cu(1)$ is not isomorphic to $ \mathbb{Z}\slash<n> \times \cu(1)$

\item Let G be a connected topological group, $\hat{G}$ the universal covering group and\\ 
$A = \pi_1(G)  \cong \mathrm{Cov}(\hat{G} , G)$ the fundamental group of G. Then
\begin{equation*}1 \longrightarrow \pi_1(G) \xrightarrow{\;\;\imath\;} \hat{G} \xrightarrow{\;\;\pi\;} G \longrightarrow 1\end{equation*} 
is a central extension of $G$ by $\pi_1(G)$, where $\pi$ is the covering homomorphism and $\imath$ is given by the action of the covering group $\mathrm{Cov}(\hat{G} , G) \cong \pi_1(G)$ on $1 \in \hat{G}$.

\begin{proof} It is easily checked that the sequence is exact. Furthermore, the orbit of $1_{\hat{G}}$ under the action of $\mathrm{Cov}(\hat{G} , G)$ is by definition the preimage of $1_G$ under $\pi$. Thus: $i(\pi_1(G)) = \pi^{-1}(1_G) = \ker(\pi)$. 
The interesting part is to show that $\imath(\pi_1(G))$ is central in $\hat{G}$. To this end, note that $\imath(\pi_1(G)) = \pi^{-1}(1_G)$ is a discrete set by definition of covering spaces. As the kernel of $\pi$, it is also a normal subgroup of $\hat{G}$. Now, for any fixed $a \in \imath(\pi_1(G))$ we can consider the map 
\begin{equation*} \hat{G} \ni g \mapsto g^{-1} a g.\end{equation*}
This is a continuous map from $\hat{G}$ into $\imath(\pi_1(G))$, sending $1$ to $a$. As $\hat{G}$ is connected and $\imath(\pi_1(G))$ is discrete, the map must be constant. Thus, $g^{-1} a g = a \; \forall g \in \hat{G}$, i.e. $a$ belongs to the center of $\hat{G}$.   
\end{proof}

\item As a special case of 2. , we can consider the well known covering
\begin{equation*} 1 \longrightarrow \lbrace \pm 1 \rbrace \rightarrow \mathrm{SU}(2) \rightarrow \mathrm{SO}(3) \longrightarrow 1\end{equation*}

It is a well-known fact (in physics prominently drom the discussion of spin in non-relativistic quantum mechanics) that there is no irreducible, 2-dimensional unitary representation of $\mathrm{SO}(3)$. However, there exists one of its universal covering group $\mathrm{SU}(2)$ (generated by the Pauli matrices).\\ 
\end{enumerate}
\end{Example}

\begin{Definition}[Equivalence of Central Extensions]\label{equivce}
\mbox{}\\
\vspace*{-2mm}

\noindent Two central extension $E$ and $E'$ of a group $G$ by $A$ are \emph{equivalent}, if there exists an isomorphism $\varphi: E\to E'$ which is compatible with the extensions i.e. such that the following diagram commutes:

\begin{equation*}
\begin{xy}
  \xymatrix{
      1 \ar[r] & A \ar[r] \ar[d]_{Id}    &   E \ar[d]^{\varphi} \ar[r]  & G \ar[d]_{Id} \ar[r]          & 1  \\
      1 \ar[r] & A \ar[r]            &   E' \ar[r]          & G \ar[r] & 1 
  }
\end{xy}
 \end{equation*}  
\end{Definition}
\vspace*{2mm}

\begin{Lemma}[Trivial Extensions]\label{trivext}
\item A central extension $1 \longrightarrow A \xrightarrow{\;\;\imath\;} E \xrightarrow{\;\;\pi\;} G \longrightarrow 1$ 
is equivalent to the trivial extension if and only if there is a homomorphism $\sigma : G \rightarrow E$ with $\pi \circ \sigma = Id_G$.
\item In other words: $\sigma$ is a section of $G$ in $E$ which is also a homomorphism of groups.
In this case, $\sigma$ is called a \emph{splitting map} and the extension is said to \emph{split}.
\item If $G$ and $E$ are Lie groups, a section $\sigma$ is called a splitting map if it is also a Lie group homomorphism. 
\end{Lemma}
\begin{proof} If the extension is trivial i.e. $E \cong A \times G$, set $\sigma(g) := (1,g) \in A \times G$.\\
Conversely, suppose there exists $\sigma: G \rightarrow E$ as above.
Set $\varphi: A \times G \rightarrow E, \; (a,g) \mapsto \imath(a)\,\sigma(g)$.\\
It is easily checked that this is a homomorphism compatible with the extension in the sense of Def. \ref{equivce}. Furthermore, $\varphi$ is bijective, since for every $g \in G$ and $\xi \in \pi^{-1}(g) \subset E$ there is one and only one $a \in A$ with $\xi = \imath(a) \, \sigma(g)$.\\
\end{proof}

\newpage 
\section{Projective Representations}

\begin{Definition}[Projective Hilbert space]
\mbox{}\\
Let $\HH$ be a complex Hilbert space of finite or infinite dimension.\\
The projective Hilbert space $\mathbb{P}(\HH)$ is the space of rays in $\HH$, i.e.
\begin{equation*} \mathbb{P}(\HH) := (\HH \backslash \lbrace 0 \rbrace) \slash \IC^{\times} \end{equation*} 
where the equivalence relation is given by $\varphi \sim \lambda \, \varphi\,, \text{for}\; \lambda \in \IC \backslash \lbrace 0 \rbrace$.
The topology on $\mathbb{P}(\HH)$ is the quotient topology, induced by the quotient map $\gamma:\,\HH \rightarrow \mathbb{P}(\HH)$.\\ 
We will also write $\hat{\varphi}$ instead of $\gamma(\varphi)$.
\end{Definition}
 
\noindent Obviously, the projective Hilbert space is not a linear space. However, \textit{transition probabilities} are still well-defined. By this, we mean the map $\delta: \mathbb{P}(\HH) \times \mathbb{P}(\HH) \rightarrow [0,1]$ defined by
\begin{equation} \delta(\hat{\varphi} , \hat{\psi}) := \frac{\vert \langle \varphi , \psi \rangle \vert ^2}{\lVert \varphi \rVert^2 \lVert \psi \rVert^2} \end{equation} 
where $\scpro$ is the Hermitian scalar product on $\HH$. 
%Now, just as we are interested in unitary transformations preserving the scalar product  on the Hilbert space, on $\mathbb{P}(\HH)$ we can study projective automorphisms preserving the transition probability $\delta$.

\begin{Definition}[Projective Automorphisms]
\mbox{}\\
A projective automorphism is a bijective map $T:  \mathbb{P}(\HH) \rightarrow  \mathbb{P}(\HH)$ preserving the transition probability i.e. satisfying
\begin{equation*}  \delta(T \hat{\varphi} ,T \hat{\psi}) = \delta(\hat{\varphi} , \hat{\psi})\,,\, \forall \hat{\varphi}, \hat{\psi} \in \mathbb{P}(\HH). \end{equation*}
We denote the set of all projective automorphisms by $\mathrm{Aut}(\mathbb{P}(\HH))$.
\end{Definition}
\noindent If $U$ is a unitary map on $\HH$, we can define $\hat{U}$ on $\mathbb{P}(\HH)$ by 
\begin{equation} \hat{U} (\gamma({\varphi})) =  \gamma\bigl({U(\varphi)}\bigr). \end{equation}
Clearly, this is a projective automorphism. Thus, we get a group-homomorphism
 \begin{align*}\hat{\gamma}:\,  \cu( & \HH) \rightarrow \mathrm{Aut}(\mathbb{P}(\HH))\,,\\
 & U \mapsto \hat{\gamma}(U) =  \hat{U}. \end{align*}
 As the projective Hilbert space doesn't care for multiplicative constants, the same definition also works for \textit{anti-unitary} maps $U$ on $\HH$, i.e. anti-linear maps satisfying $\langle U x, U y \rangle = \overline{ \langle x,  y \rangle}$ for $x,y \in \HH$. As a matter of fact, every projective automorphism comes from a unitary or anti-unitary map on $\HH$:
 
 \begin{Theorem}[Wigner, 31]
 \mbox{}\\
For every projective automorphism $T \in \mathrm{Aut}(\mathbb{P}(\HH))$ there exists a unitary or an anti-unitary operator U on $\HH$ with T = $\hat{\gamma}(U)$.
 \end{Theorem}
 
\noindent The projective automorphisms coming from unitary transformations on $\HH$ form a subgroup of $\mathrm{Aut}(\mathbb{P}(\HH))$, denoted by $\cu(\mathbb{P}(\HH))$. We can express this reasoning in the language of central extensions:
 
 \begin{Lemma}[Unitary Projective Automorphisms]
 \mbox{}\\
 The sequence
 \begin{equation}\label{Uprojective} 1 \longrightarrow \cu(1) \longrightarrow \cu(\HH) \xrightarrow{\;\hat{\gamma}\;} \cu(\mathbb{P}(\HH)) \longrightarrow 1\end{equation}
 defines a central extension of $\cu(\mathbb{P}(\HH))$ by $\cu(1)$ which is non-trivial. 
 \end{Lemma}
 \begin{proof} The only part in proving that the sequence is exact and central that might be non-trivial is to identify $\ker(\hat{\gamma})$ with $\cu(1)\cdot Id \subset \cu(\HH)$.  "$\supseteq$" is clear. For "$\subseteq$" pick $U \in \ker(\hat{\gamma})$. Then for any $\varphi \in \HH:\, \hat{\gamma}(U)(\gamma(\varphi)) = \gamma(U \varphi) = \gamma(\varphi)  \Rightarrow  \exists \lambda \in \IC : U \varphi = \lambda \varphi.$\\ 
Since $U$ is unitary, $\vert \lambda \vert = 1$ i.e. $\lambda \in \cu(1)$. If we take any other $\psi \in \HH$ which is not a multiple of $\varphi$ it is by the previous consideration also an eigenvector to some eigenvalue $\lambda' \in \cu(1)$; but so is the sum $\varphi + \psi$ (with eigenvalue $\mu$).
It follows that $U (\varphi + \psi) = \mu (\varphi + \psi) = \lambda \varphi + \lambda' \psi$ and thus, by linear independence, $\lambda' = \lambda = \mu$. Hence, $U = \lambda Id$.\\

To prove that the central extension is not trivial, we embed $\IC^2$ in $\HH$ be fixing any 2-dimensional subspace $V \subset \HH$. Now consider the subgroup of all unitary operators on $\HH$ leaving $V$ invariant, i.e. $\lbrace U \in \cu(\HH) \mid \cu(V) = V \rbrace =: \mathcal{V}$. On $\mathcal{V}$, we introduce an equivalence relation and identify two operators if they agree on $V$. Then $\mathcal{V}\slash\sim \; \cong U(V) \cong \cu(2)$.\\ 
Now all the homomorphisms descend to a central extension
 \begin{equation*} 1 \longrightarrow \cu(1) \longrightarrow \cu(2) \longrightarrow \cu(\mathbb{P}(\IC^2)) \longrightarrow 1. \end{equation*}
Note that  $\mathbb{P}(\IC^2) = \mathrm{CP}^1 \cong \mathrm{S}^1$, so its unitary group is just the isometry group of the sphere, i.e. $\mathrm{SO}(3)$. Also, $\cu(1) \times \mathrm{SU}(2) \cong \cu(2)$ by $(e^{i\phi}, U) \rightarrow e^{i\phi \slash 2} \,  U$, for $\phi \in [0,2\pi)$. Thus, if the central extension \eqref{Uprojective} was trivial with splitting map $\sigma$, this $\sigma$ would descend to a spitting map for
\begin{equation*} 1 \longrightarrow \cu(1) \longrightarrow \cu(1) \times \mathrm{SU}(2) \longrightarrow \mathrm{SO}(3) \longrightarrow 1. \end{equation*}
 But then, the second component of the homomorphism $\sigma: \mathrm{SO}(3) \rightarrow \cu(1) \times \mathrm{SU}(2)$ is a splitting map for the universal covering
  \begin{equation*} 1 \longrightarrow \lbrace \pm 1 \rbrace \longrightarrow \mathrm{SU}(2) \longrightarrow \mathrm{SO}(3) \longrightarrow 1 \end{equation*}
We know, however, that such a splitting map cannot exist, because the covering of $ \mathrm{SO}(3)$ by  $\mathrm{SU}(2)$ is not trivial. Thus, we get a contradiction.\\
 \end{proof}

 \subsubsection{Lifting projective representations}
 \noindent Now we are able to formulate the problem of lifting projective representations in a more precise way:
 \textit{Given a projective representation $\Gamma:\, G \rightarrow \cu(\mathbb{P}(\HH))$, is there a representation $\rho:\, G \rightarrow \cu(\HH)$ such that $\hat{\gamma}\circ \rho  = \Gamma$, i.e. such that the following diagram commutes?}
 
\begin{equation*} 
\begin{xy}
  \xymatrix{
  & & &  G   \ar[d]^ {\Gamma} \ar@{-->}[dl]_{\rho} \\
      1 \ar[r] & \cu(1) \ar[r]            &   \cu(\HH) \ar[r]^ {\hat{\gamma}}          & \cu(\mathbb{P}({\HH})) \ar[r] & 1 
  }
\end{xy}
 \end{equation*}\\
 
\noindent In general, the answer is NO. However, as we have suggested in the introducing remarks,  there always exists a central extension $\widetilde{G}$ of $G$ such that the projective representation of $G$ lifts to a proper representation of $\widetilde{G}$.\\

\begin{Lemma}[Lifting Projective Representations]\label{liftingrep}
\mbox{}\\
Let $G$ be a group and $\Gamma:\, G \rightarrow \cu(\mathbb{P}(\HH))$ a homomorphism. 
There exists a central extension $\widetilde{G}$ of $G$ by $\cu(1)$ and a homomorphism $\widetilde{\Gamma}:\, \widetilde{G} \rightarrow \cu(\HH)$, such that the following diagram commutes:\\
\begin{equation}\label{liftdiagram}
\begin{xy}
  \xymatrix{
      1 \ar[r] & \cu(1) \ar[r]^{\imath} \ar[d]_{Id}    &   \widetilde{G} \ar[d]^{\widetilde{\Gamma}} \ar[r]^ \pi & G   \ar[d]^ {\Gamma} \ar[r]          & 1  \\
      1 \ar[r] & \cu(1) \ar[r]            &   \cu(\HH) \ar[r]^ {\hat{\gamma}}          & \cu(\mathbb{P}({\HH})) \ar[r] & 1 
  }
\end{xy}
 \end{equation}
 
 \vspace*{2mm}
\noindent $\widetilde{\Gamma}$ is a unitary representation of $\widetilde{G}$ on $\HH$. 
The group $\widetilde{G}$ is called the \emph{deprojectivization} of $G$.
\end{Lemma}
\newpage
 \begin{proof} We define 
 \begin{equation*} \widetilde{G} := \lbrace (U , g ) \in \cu(\HH) \times G \mid \hat{\gamma}(U) = \Gamma (g) \rbrace. \end{equation*}
 This is a subgroup of $\cu(\HH) \times G$. The inclusion $\cu(1) \ni \lambda \stackrel{\imath}\mapsto (\lambda \cdot Id , 1)$ and the projection onto the second component $\pi = \mathrm{pr}_2 : \widetilde{G} \rightarrow G$ are homomorphisms that make the upper row of the diagram \eqref{liftingrep} into a central extension. The projection onto the first component defines a representation $\widetilde{\Gamma}:= \mathrm{pr}_1 : \, \widetilde{G} \rightarrow \cu(\HH)$ which, by construction, satisfies $\hat{\gamma} \circ \widetilde{\Gamma} = \Gamma \circ \pi$.\\
 \end{proof}
 
So far, this is a pure algebraic statement. In the more interesting cases, where $G$ is a topological group or even a Lie group, we will have to check if and how the results are compatible with the topological structure, respectively the Lie group structure, of $G$ and its central extension. Here, we summarize the main results. A more complete treatment can be found in \cite{Scho}, for example.
\begin{enumerate}[1)]
\item If $G$ is a topological group, $\widetilde{G}$ can be given the structure of a topological group as a subgroup of $\cu(\HH)\times G$. Then, if $\Gamma$ is continuous, so is $\widetilde{\Gamma}$.
\item  If $G$ is a \textit{finite-dimensional} Lie group, then $\widetilde{G}$ can be given the structure of a Lie group so that the upper sequence in \eqref{liftdiagram} becomes a sequence of differentiable homomorphisms. If $\Gamma$ is smooth (in a strong sense), so is $\widetilde{\Gamma}$.\footnote{This statement seems rather harmless, but it's quite the opposite. Indeed, it requires the solution of one of the famous ``Hilbert problems'': every topological group which is also a finite-dimensional topological manifold is already a Lie group. This theorem was proven by Montgomory and Zippin in 1955.} 
\item If we have to deal with infinite-dimensional Lie groups -- and we do -- the Lie group structure of the central extension is not generally for free. Fortunately, things work out nicely in the cases relevant to our discussion.
\end{enumerate}

\noindent In conclusion, if the central extension $\widetilde{G}$ carries the structure of a Lie group with a representation $\widetilde{\Gamma}$ on $\HH$, any prescription for lifting the projective action of $G$ to the Hilbert-space corresponds to a section $\sigma:\, G \rightarrow \widetilde{G}$, via 
\begin{equation}\rho:= \widetilde{\Gamma} \circ \sigma : G \rightarrow \cu(\HH). \end{equation} 
Conversely, every attempt to lift the action to the Hilbert space, i.e. every (continuous) map $\rho: G \to \cu(\HH)$ defines a section in $\widetilde{G}$, by assigning to $g \in G$ the unique element $\tilde{g}$ in $\pi^{-1}(g)$ with $\widetilde{\Gamma}(\tilde{g}) = \rho(g)$. Now, $\rho:= \widetilde{\Gamma} \circ \sigma$ defines a (continuous/smooth) representation of $G$ on $\HH$ if and only if the section $\sigma$ is a (continuous/smooth) homomorphism of groups i.e. if and only if the central extension is trivial as an extension of groups/topological groups/ Lie groups.
We summarize this insight in the following proposition.\\
 
%\begin{Remark}[$\cu(\HH)$ is a topological group]
%\mbox{}\\
 %$\cu(\HH)$ is a topological group with respect to the strong topology \cite{Scho}. So, if in the context of the previous Lemma $G$ is topological group, $\widetilde{G}$ can be given the structure of a topological group as a subgroup of $\cu(\HH) \times G$. Then, if $\Gamma$ is continuous, so is $\widetilde{\Gamma}$.
%\end{Remark}
%\noindent For finite-dimensional Lie groups we have a much stronger result:
 %\begin{Theorem}[Montgomory,Zippin, 1955]\label{MoZi}
 %\mbox{}\\
%If $G$ in the previous Lemma is an n-dimensional Lie group ($n < \infty$), then $\widetilde{G}$ can be given the structure of a Lie group of dimension n+1. The upper sequence in \eqref{liftdiagram} is then a sequence of differentiable homomorphisms and if $\Gamma$ is smooth (in a strong sense) so is $\widetilde{\Gamma}$.
%\end{Theorem}

%\noindent Note that this particular result holds for finite-dimensional Lie groups only. If we have to deal with infinite-dimensional Lie groups - and we do- the Lie group structure of the central extension is not for free. Anyways, with there results we are able to answer our initial question about the existence of proper unitary representations. 

\begin{Proposition}[Lifting Projective Representations]\label{Corliftingrep}
\item Let $G$ be a topological group. A projective representation $\Gamma:\, G \rightarrow \cu(\mathbb{P}(\HH))$ can be lifted to a (continuous) unitary representation $\rho:\, G \rightarrow \cu(\HH)$ with $\hat{\gamma} \circ \rho = \Gamma$ if and only if the central extension $1 \longrightarrow \cu(1) \longrightarrow \widetilde{G} \longrightarrow G \longrightarrow 1$ splits by a (continuous) section i.e. is trivial.\\ 
%If $G$ is a finite-dimensional Lie group, smoothness of $\Gamma$ implies smoothness of $\rho$.
\end{Proposition}
\begin{proof} By Lemma \eqref{trivext}, the central extension is (algebraically) trivial if and only if there is a section $\sigma:\, G \rightarrow \widetilde{G}$ which is also a homomorphism of groups. Then $\rho := \widetilde{\Gamma} \circ \sigma$ is a homomorphism with
\begin{equation*}\hat{\gamma} \circ \rho =  \hat{\gamma} \circ \widetilde{\Gamma} \circ \sigma = \Gamma \circ \pi \circ \sigma = \Gamma.\end{equation*}  
If the section  $\sigma$ is continuous, so is $\rho$.\\ 
Conversely, if $\rho$ is a unitary representation of $G$, then $\sigma(g) := (\rho(g), g) \in \widetilde{G}$ is a section in $\widetilde{G}$ and a homomorphism of groups and thus the desired splitting map.\\
\end{proof}
\newpage

\section{Cocycles and Cohomology Group}

In the previous section, we have seen that lifts of a projective representation of a (Lie) group $G$ to an action on the Hilbert-space $\HH$ (not necessarily a representation) correspond to sections in a central extension of $G$. We want to study this more thoroughly.\\

\noindent Let
\begin{equation*}1 \longrightarrow A \xrightarrow{\;\;\imath\;} E \xrightarrow{\;\;\pi\;} G \longrightarrow 1\end{equation*}
be a central extension of G by A. Let $\tau : G \rightarrow E$ be a map with 
 \begin{equation}\label{section} \pi \circ \tau = Id_G\; \text{and}\; \tau(1) = 1. \end{equation}
The map might be defined in a neighborhood of the identity, only.\\
$\tau$ will in general fail to be a homomorphism. Nevertheless:
\begin{equation*} \pi(\tau(g)\tau(h)) = \pi(\tau(gh)) = gh, \; \forall g, h \in G. \end{equation*} 
Therefore, there exists $\chi(g,h) \in A$ with 
\begin{equation} \tau(g) \tau(h) = \chi(g , h) \tau(gh). \end{equation}
This defines a map $\chi: G \times G \rightarrow A$. Obviously, this $\chi$ satisfies 
\begin{equation}\label{chi1} \chi(1,1) = 1.\end{equation}
\onehalfspace
Furthermore:  $\tau(x)\tau(y)\tau(z) = \chi(x,y) \tau(xyb) \tau(z) = \chi(x,y) \chi(xy, z) \tau(xyz)$ and 
similarly:    $\tau(x)\tau(y)\tau(z) = \tau(x) \chi(y,z) \tau(yz) = \chi(x,yz)\chi(y,z) \tau(xyz)$, from which we deduce: 
\begin{equation}\label{chi2}\chi(x,y) \chi(xy,z) = \chi(x,yz) \chi(y,z), \; \forall \; x,y,z \in G. \end{equation}

\noindent Now It's possible that we've just made a ``bad'' choice for $\tau$ and that there really does exist a (local) section $\tau'$ which is also a homomorphism, i.e. for which the corresponding cocycle vanishes.
Let's fix this by writing
$\tau'(x) = \tau(x) \lambda(x)$ with a suitable function $\lambda: G \rightarrow A$.\\
Then: $\tau'(x) \tau'(y) = \tau(x)\tau(y) \lambda(x)\lambda(y) = \tau(xy)\chi(x,y)\lambda(x)\lambda(y)$.
But also, since $\tau'$ is a\\ homomorphism: $\tau'(x) \tau'(y) = \tau ' (xy) = \tau(xy) \lambda(xy)$. Together, this implies for $\lambda$: \begin{equation} \lambda(xy) = \chi(x,y) \lambda(x) \lambda(y), \; \forall x,y,z \in G. \end{equation}
This motivates the following definition:

\begin{Definition}[Cocycles and Second Cohomology Group]
\mbox{}
\begin{enumerate}[i)]
\item A map $\chi: G \times G \rightarrow A$ satisfying \eqref{chi1} and \eqref{chi2} is called a \emph{factor set} or a \emph{2-cocycle} on $G$ with values in $A$.
\item A 2-cocycle  $\chi: G \times G \rightarrow$ is called \emph{trivial} if there exists a map $\lambda: G \rightarrow A$ such that  $\lambda(xy) = \chi(x,y) \lambda(x) \lambda(y) \; , \forall x,y,z \in G$.
\item Furthermore, we define the \emph{second cohomology group} of G with coefficients in $A$ as the set of all 2-cocycles on $G$ with values in $A$ modulo trivial cocycles, i.e.
\begin{equation}\mathrm{H}^2(G,A) := \lbrace \chi: G \times G \rightarrow A \mid \chi \;  \text{is a cocycle} \rbrace \slash \sim\end{equation}
where $\chi_1 \sim \chi_2  :\iff \chi_1\chi_2^{-1}$ is trivial.
\end{enumerate}
\end{Definition}
\newpage
\noindent We have seen that given a central extension $1 \longrightarrow A \xrightarrow{\;\;\imath\;} E \xrightarrow{\;\;\pi\;} G \longrightarrow 1$, every section $\tau : G \rightarrow E \; \text{with} \;  \pi \circ \tau = Id_G\; \text{and}\; \tau(1) = 1$ defines a cocycle $\chi:G \times G \rightarrow A$. Different choices for $\tau$ lead to equivalent cocycles. Conversely, given a cocycle $\chi$, we can define a group $A \times_{\chi} G$ as the direct product $A \times G$ with the multiplication
\begin{equation}
(a,x)(b,y) := \bigl(\chi(x,y) ab , xy\bigr).
\end{equation}
This is a central extension of $G$ by $A$ with the cocycle $\chi$ coming from the obvious section $\tau(x) := (1, x)$ (c.f. the construction in the introduction of this chapter).

\noindent If the cocycle $\chi$ comes from a (global) section $\tau : G \rightarrow E$ as above, the central extensions
\begin{equation*} 1 \longrightarrow A \xrightarrow{\;\;\imath\;} E \xrightarrow{\;\;\pi\;} G \longrightarrow 1\end{equation*}
\vspace*{-2mm} and 
\begin{equation*} 1 \longrightarrow A \longrightarrow A \times_{\chi} G  \xrightarrow{\;\; pr_2\;} G \longrightarrow 1\end{equation*}
are equivalent by the isomorphism $\varphi: A \times_{\chi} G \rightarrow E, \; \varphi\bigl((a,c)\bigr) := \imath(a)\cdot\tau(x)$. The group-multiplication in $A \times_{\chi} G$ is defined just in such a way as to cancel the cocycle coming from $\tau$ on the right-hand-side. We reserve for the reader the little joy of checking for him- or herself how everything is designed to work out nicely. Now it's an easy exercise to show that for cocycles $\chi_1$ and $\chi_2$, the corresponding groups $A \times_{\chi_1} G$ and $A \times_{\chi_2} G$ are equivalent (isomorphic) as central extensions if and only if $\chi_1$ and $\chi_2$ are equivalent as cocycles. Altogether, this yields the following theorem:

\begin{Theorem}[Central Extensions correspond to Cohomology Classes]
\mbox{}\\
There is a one-to-one correspondence between equivalence classes of central extension of G by A and the second cohomology classes of G with values in A.
\end{Theorem}

In particular, a central extension is trivial, if and only if it allows a (global) section whose corresponding cocycle is trivial. In our context, this means that a 2-cocycle or, more generally, the corresponding cohomology group represents the \textit{algebraic} or, if questions of continuity are involved, \textit{topological obstructions} to lifting projective representations to proper representations on the Hilbert space.

Note that the theorem is so far a pure algebraic statement! If we are dealing with Lie groups and have to take topological aspects into account, things aren't quite as easy and the arguments above will, in general, work only locally. In fact, there might not be any \textit{continuous} section $\tau : G \rightarrow E \; \text{with} \;  \pi \circ \tau = Id_G\; \text{and}\; \tau(1) = 1$. Therefore, it is usually more convenient to discuss cocycles of the corresponding \textit{Lie algebras}, which may be thought of as the infinitesimal version of the Lie group cocylces and are already determined by a \textit{local} section $\tau$ defined in some arbitrarily small neighborhood of the identity of the Lie group $G$. We will do this in the next chapter. However, the following is true:

\begin{Proposition}[Lie Group Extensions and local Cocycles]
\mbox{}\\ 
Let $1\rightarrow A \rightarrow E \rightarrow G \rightarrow 1$ be a central extension of a connected (finite- oder infinite-dimensional) Lie group $G$ by the abelian Lie group $A$. Then, $E$ carries the structure of a Lie group such that the central extension is smooth if and only if the central extension can be described by a cocycle $\chi: G \times G \rightarrow A$ which is smooth in a neighborhood of $(e,e)\in G \times G$.\\
For Banach Lie groups, the statement applies also to non-connected $G$. 
\end{Proposition}
\noindent For the proof we refer to \cite{Ne} Prop. 4.2 and \cite{TW87} Prop. 3.11 .
\newpage

\section{Central Extensions of Lie Algebras}

In the following, we assume that  $G$ is a \textit{locally exponential} Lie group (finite- oder infinite-dimensional) and that the assignment  $G \to \mathrm{Lie}(G)$ of a Lie group to its Lie algebra is \textit{functorial} in the following sense:
If $G_1, G_2$ are Lie groups and $\varphi: G_1 \to G_2$ a Lie group homomorphism, then their exists a unique Lie algebra homomorphism $\mathrm{Lie}(\varphi) = \dot{\varphi}$ such that the following diagram commutes:

\begin{equation}
\begin{xy}
  \xymatrix{
G_1  \ar[r]^ \varphi & G_2    \\
   \mathrm{Lie}(G_1)\ar[u]^{\exp} \ar[r]^ {\mathrm{Lie}(\varphi)}         &  \mathrm{Lie}(G_2) \ar[u]_{\exp} 
  }
\end{xy}
 \end{equation}

\noindent Usually, we'd like to define a Lie group homomorphism as a continuous homomorphism between Lie groups. However, for some ``exotic'' examples, it can be necessary to explicitly demand differentiability of the homomorphism at the identity, in order to get the functorial property defined above.

%This assumption is not \textit{always} satisfied, however, it is satisfied for all ``reasonable'' examples. In particular it is true for:
%\begin{itemize}
%\item Finite-dimensional Lie groups. If we identify the Lie algebras with the tangent space of the corresponding Lie groups at the identity, then $ \dot{\varphi} = D_{e}\varphi: T_eG_1 \rightarrow T_e G_2$ is the desired Lie algebra homomorphism.
%\item In the infinite-dimensional case, the functorial property can be shown to hold, if $G_1$ is simply-connected and $G_2$ is regular. 
%\item 
%\end{itemize}

\begin{Definition}[Central Extension of Lie Algebras]\label{DefinitionAlgebraextension}
\mbox{}\\
Let $\mathfrak{a}$ be an abelian Lie algebra and $\mathfrak{g}$ a  Lie algebra over $\mathbb{R}$ or $\mathbb{C}$ (the dimensions may be infinite). A \emph{central extension} of $\mathfrak{g}$
by $\mathfrak{a}$ is an exact sequence of  Lie algebra homomorphisms 
\[
0\longrightarrow \mathfrak{a}\stackrel
{\mathfrak{i} } \longrightarrow \mathfrak{h}\stackrel
{\mathfrak{\pi} }{\longrightarrow }\mathfrak{g}\longrightarrow 0
\]
s.t. $\mathfrak{i(a)} \subset  \mathfrak{h}$ is central in $\mathfrak{h}$, i.e. if $\left[
\mathfrak{i}(X),Y \right] = 0 $ for all $X \in \mathfrak{a}$ and $Y \in
\mathfrak{h}$.
\end{Definition}

\begin{Proposition}[From Lie Groups to Lie Algebras]\label{LAfromLG}
\mbox{}\\
Let $A$, $E$ and $G$ finite-dimensional Lie groups and
\[
 1 \longrightarrow   A  \stackrel i \longrightarrow   E  \stackrel{\pi}
{\longrightarrow }  G  \longrightarrow   1 
\] 
a central extension of Lie groups. Then 
\begin{equation}\label{LAfromLG2}
0 \longrightarrow  \mathrm{Lie}(A)  \stackrel {\dot \imath} \longrightarrow   
\mathrm{Lie}(E)  \stackrel{\dot \pi}{\longrightarrow }  \mathrm{Lie}(G) 
\longrightarrow   0
\end{equation}
is a central extension of the corresponding Lie algebras.
\end{Proposition}
\noindent Now, for a central extension of a Lie algebras \[
0\longrightarrow \mathfrak{a}\stackrel
{\mathfrak{i} } \longrightarrow \mathfrak{h}\stackrel
{\pi }{\longrightarrow }\mathfrak{g}\longrightarrow 0
\]
we can always define a linear map $\beta:\; \mathfrak{g} \rightarrow \mathfrak{h}$ satisfying $\pi \circ \beta = Id$.\\
Analogously to the central extensions of groups, the central extension of Lie algebras is equivalent to the trivial extension $\mathfrak{h} = \mathfrak{g} \oplus \mathfrak{a}$ if and only if $\beta$ can be chosen to be a Lie algebra homomorphism, i.e. such that $[\beta(X), \beta(Y)] = \beta([X,Y]), \; \forall X, Y \in \mathfrak{g}$. In general, how $\beta$ \textit{fails} to be a Lie algebra homomorphism is expressed by the skew-symmetric map \begin{equation}\label{Theta} c(X,Y) := [\beta(X), \beta(Y)] - \beta([X,Y]). \end{equation}
$c: \, \mathfrak{g} \times \mathfrak{g} \rightarrow \mathfrak{a}$ has the following properties: 
\begin{equation} \begin{split} \label{liealgebracocycle}  & i)\;\; c \; \text{is bilinear and skew-symmetric} \\ &ii)\;  c \left( X,\left[ Y,Z\right] \right) + c \left( Y,\left[ Z,X\right] \right) + c \left( Z,\left[ X,Y\right] \right) =0. \end{split} \end{equation} 

\noindent In particular, if the central extension of Lie algebras comes from a central extension of Lie groups, any differentiable local section $\tau : \; G \rightarrow E$ as in \eqref{section} defines a corresponding Lie algebra section by $\beta := D_e\tau$ (under identification of the Lie algebras with the tangent space of the corresponding Lie groups at the identity $e$). But as $\tau$ fails to be a Lie group homomorphism, $D{\tau}$ will fail to be a homomorphism of the Lie algebras.

\begin{Definition}[Lie Algebra Cocycles]
\mbox{}\\
Let $\mathfrak{a}, \mathfrak{g}$ be two Lie algebras, $\mathfrak{a}$ abelian.  
A map $c:\; \mathfrak{g} \times \mathfrak{g} \rightarrow \mathfrak{a}$ satisfying the conditions \eqref{liealgebracocycle} is called a \emph{Lie algebra 2-cocycle} or simply a cocycle. 
\end{Definition}

\begin{Proposition}[Computation of Lie Algebra Cocycles]
\mbox{}\\
If $\chi$ is a (local) Lie group cocycle coming from the (local) section $\tau$, then the corresponding Lie algebra cocycle coming from $\dot{\tau}$ can be computed from $\chi$ as
\begin{equation}\label{cocycleformula} c(X , Y)\, = \, \frac{\partial}{\partial t} \frac{\partial}{\partial s} \Bigl\vert_{t=s=0}\, \chi(e^{sX} , e^{tY})\, - \, \frac{\partial}{\partial t} \frac{\partial}{\partial s} \Bigl\vert_{t=s=0}\, \chi(e^{tY} , e^{sX}) \end{equation}
\end{Proposition}
\begin{proof}: Here, we write $\dot{\tau}$ for the Lie algebra map, corresponding to $D_e\tau$. We compute:
\begin{align*} & \frac{\partial}{\partial t} \frac{\partial}{\partial s} \Bigl\vert_{t=s=0}\, \chi(e^{sX} , e^{tY})\, - \, \frac{\partial}{\partial t} \frac{\partial}{\partial s} \Bigl\vert_{t=s=0}\, \chi(e^{tY} , e^{sX})\\[1.3ex]
= \; & \frac{\partial}{\partial t} \frac{\partial}{\partial s} \Bigl\vert_{t=s=0}\, \tau(e^{sX})\tau(e^{tY})\tau(e^{sX}e^{tY})^{-1}\, - \,\Bigl\vert_{t=s=0}\, \tau(e^{tY})\tau(e^{sX})\tau(e^{tY}e^{sX})^{-1}\\[1.3ex]
= \;& \dot{\tau}(X)\dot{\tau}(Y)\,-\, \dot{\tau}(XY)\, - \, \dot{\tau}(Y)\dot{\tau}(X)\,+\,\dot{\tau}(YX)\\[1.3ex]
= \; & \dot{\tau}(X)\dot{\tau}(Y) - \dot{\tau}(Y)\dot{\tau}(X) - (\,\dot{\tau}(XY) - \dot{\tau}(YX)\, )\\[1.3ex]
= \; & [\dot{\tau}(X) , \dot{\tau}(Y)] - \dot{\tau}([X , Y])\, = c(X, Y). \end{align*}
\end{proof}
\noindent Just as for groups, there is a correspondence between central extensions of Lie algebras and cocycles which is 1-to-1 modulo trivial extensions/cocycles.
%However, whereas the result was of pure algebraic nature for groups and the situation becomes unclear as soon as topology is involved, everything works out nicely for Lie algebras:\\
 Given bilinear form b $c:\; \mathfrak{g} \times \mathfrak{g} \rightarrow \mathfrak{a}$, consider the vector space $\mathfrak{h} := \mathfrak{g} \oplus \mathfrak{a}$ and define a bilinear map $[\cdot , \cdot]_{c}$ on $\mathfrak{h}$ by
\begin{equation}\label{Lieklammer} [X_1 \oplus Y_1 , X_2 \oplus Y_2]_{c} := [X_1 , X_2]_{\mathfrak{g}} + c(X_1,X_2) \end{equation}
for  $X_1, X_2 \in \mathfrak{g}, Y_1,Y_2 \in  \mathfrak{a}$.
It is straight forward to check that $[\cdot , \cdot]_{c}$ is a Lie bracket on $\mathfrak{h}$, if and only if $c$ is a cocycle.
In this case, $\mathfrak{h}$ becomes a Lie algebra and projection onto the first component makes 
 \[
0\longrightarrow \mathfrak{a}
 \longrightarrow \mathfrak{h}\stackrel
{pr_1 }{\longrightarrow }\mathfrak{g}\longrightarrow 0
\]
a central extension of $\mathfrak{g}$ by $\mathfrak{a}$. Conversely, if $c$ comes from a central extension and a linear map $\beta:\; \mathfrak{g} \rightarrow \mathfrak{h}$, the algebra $\mathfrak{h}$ is isomorphic to $\mathfrak{g\oplus a}$ as a linear space, by the isomorphism
\[
\mathfrak{F}\; : \mathfrak{g} \times \mathfrak{a} \to \mathfrak{h}, \quad
X \oplus Y = (X,Y) \mapsto \beta(X) + Y\,.
\]
And equipped with the Lie bracket \eqref{Lieklammer}, $\mathfrak{g\oplus a}$ becomes a Lie algebra and $\mathfrak{F}$ an isomorphism of Lie algebras:
\begin{align*} [\mathfrak{F}(X_1,Y_1), \mathfrak{F}(X_2,Y_2)] &= [\beta(X_1) + Y_1, \beta(X_2) + Y_2] = [\beta(X_1), \beta(X_2)] \\ &= \beta([X_1, X_2]) + c(X_1, X_2) = \mathfrak{F} ([X_1, X_2]) + c(X_1, X_2)) \\ & = \mathfrak{F}\; ([(X_1,Y_1),(X_2,Y_2)]_{c}). \end{align*}
We summarize the results in the following lemma:
\vspace*{-2mm}
\begin{Lemma}[Central Extensions and Lie Algebra Cocycles]
\mbox{}\\
Every central extension of Lie algebras comes from a cocycle.\\
Conversely, every cocycle $c:\; \mathfrak{g} \times \mathfrak{g} \rightarrow \mathfrak{a}$ induces a central extension of $\mathfrak{g}$ by $\mathfrak{a}$ as above.
\end{Lemma}

\noindent Note that in all of these construction a particular choice of $\beta$ was involved. 
What if we choose a different linear map $\beta ':\; \mathfrak{g} \rightarrow \mathfrak{h}$ with $\dot{\pi} \circ \beta '= Id$? Then, since $\pi \circ (\beta' - \beta) = 0$,  the difference of $\beta$ and $\beta'$ is a linear map with values in $\mathfrak{a}$ (identified with the corresponding subalgebra in $\mathfrak{h}$), i.e.
\begin{equation*} \beta' - \beta = \mu :\mathfrak{g} \rightarrow \mathfrak{a} \cong \mathfrak{i} (\mathfrak{a}) \subset \mathfrak{h} \end{equation*}
For the corresponding cocycles, this means
\begin{align*} c'(X,Y)  :&= [\beta'(X), \beta'(Y)] - \beta'([X,Y])\\ 
& = [\beta(X) + \mu(Y), \beta(Y) + \mu(Y)] - \beta([X,Y]) - \mu([X,Y])\\
& = [\beta(X), \beta(Y)] - \beta([X,Y])  - \mu([X,Y])\\
& = c(X,Y) - \mu([X,Y])
\end{align*} 
since the image of $\mu$ is central in $\mathfrak{h}$. We deduce:

%\noindent \textit{There exists a section $\beta'$ that is also a Lie algebra homomorphism if and only if the corresponding cocycle $c'$ vanishes if and only if there exists a linear map $\mu :\mathfrak{g} \rightarrow \mathfrak{a}$ with $c(X,Y) \equiv \mu([X,Y])$. By the previous Lemma, this is the case if and only if the corresponding central extension splits, i.e. is trivial.}
\begin{Theorem}[Triviality of Lie Algebra Extensions]\label{CEalgebraTriviality}
\mbox{}\\
Let $c$ be a Lie algebra cocycle for the central extension $\mathfrak{h}$ of $\mathfrak{g}$.
The following are equivalent:
\begin{enumerate}[i)]
\item There exists a section $\beta: \mathfrak{g} \rightarrow \mathfrak{h} $ that is also a Lie algebra homomorphism, i.e. a splitting-map for the central extension.
\item The cocycle $c$ defined by $\beta$ vanishes. 
\item There exists a linear map $\mu :\mathfrak{g} \rightarrow \mathfrak{a}$ with $c(X,Y) \equiv \mu([X,Y])$.
\item The central extension is trivial, i.e. $\mathfrak{h} \cong \mathfrak{g} \oplus \mathfrak{a}$ as Lie algebras.
\end{enumerate}
In particular, if $\tau: G \rightarrow H$ is a splitting map for a central extension of Lie groups, $\beta := \dot{\tau}$ is a splitting map for the corresponding central extension of Lie algebras.
\end{Theorem}

\noindent These considerations lead us straight to the mathematical concept of cohomology groups.

\begin{Definition}[Second Cohomology Group]
\mbox{}\\
Let $\mathfrak{g}$ be a Lie algebra and $\mathfrak{a}$ an abelian Lie algebra.
\item Let $Z^2(\mathfrak{g}, \mathfrak{a})$ be the set of all 2-cocycles on $\mathfrak{g}$ with values in $\mathfrak{a}$. 
\item Let $B^2(\mathfrak{g}, \mathfrak{a}) = \lbrace c \in Z^2(\mathfrak{g}, \mathfrak{a})  \mid \exists \mu \in \mathrm{Hom}(\mathfrak{g}, \mathfrak{a}): c(X,Y) = \mu([X,Y]) \rbrace  $
\item Let $\mathrm{H}^2(\mathfrak{g}, \mathfrak{a}) := Z^2(\mathfrak{g}, \mathfrak{a}) \slash B^2(\mathfrak{g}, \mathfrak{a})$.
\item$\mathrm{H}^2(\mathfrak{g}, \mathfrak{a}) $ is the \emph{second cohomology group} of $\mathfrak{g}$ with values in $ \mathfrak{a}$. 
\end{Definition}

\begin{Theorem}[Central Extensions correspond to Cohomology Classes]\label{CohomologyExt}
\mbox{}\\
The cohomology group $\mathrm{H}^2(\mathfrak{g}, \mathfrak{a}) $ is in one-to-one correspondence with the set of equivalence classes of central extensions of $\mathfrak{g}$ by $\mathfrak{a}$.
\end{Theorem}

Let's summarize the insights we've got so far. We started with a Lie group $G$ and a projective unitary representation $\Gamma$ of $G$ on a projective Hilbert space $\mathbb{P}(\HH)$. In general, we cannot expect to be able to lift the projective representation to a proper representation on $\HH$. However, there exists a central extension $\widetilde{G}$ of $G$ by $\cu(1)$ that does have a unitary representation $\widetilde{\Gamma}$ on the Hilbert space, with $\hat{\gamma} \circ \widetilde{\Gamma} = \Gamma$ \ref{liftingrep}. Conversely, any prescription for lifting the projective representation $\Gamma$ to $\cu(\HH)$, i.e. any choice of phases for the lift, corresponds to a section $\tau$ of $G$ in $\widetilde{G}$. In particular, for any such section, $\widetilde{\Gamma} \circ \tau$ defines a lift as desired. This will \textit{not} be a representation, though, unless $\tau$ is a continuous homomorphism. The way in which this fails to be a representation, or equivalently, in which the section fails to be a homomorphism, is encoded in the corresponding \textit{cocycle}. We have deduced that the following statements are equivalent:

\begin{enumerate}
\item There exists unitary representation $\rho:\, G \rightarrow \cu(\HH)$ with $\hat{\gamma} \circ \rho = \Gamma$.
\item There exists a continuous section $\sigma:\, G \rightarrow \widetilde{G}$ which is also a homomorphism [Cor.\ref{Corliftingrep}].
\item The central extension $1 \longrightarrow \cu(1) \rightarrow \widetilde{G} \rightarrow G \rightarrow 1$ is trivial  i.e. splits by a continuous section [Lem. \ref{trivext} ].
\item Any Lie group cocycle corresponding to the central extension above is trivial i.e. corresponds to $0$ in $\mathrm{H}^2(G,\cu(1))$ [Thm. \ref{CohomologyExt} ].
\end{enumerate}

\noindent The ``infinitesimal version'' of this, so to speak, is the corresponding central extension of Lie algebras \eqref{LAfromLG} and its Lie algebra cocycle \eqref{Theta}. We have explained why it is more convenient to work with Lie algebras in general, but we still have to clarify how exactly this infinitesimal version is related to our original problem.
So far it is clear that if $\sigma: G \rightarrow E$ is a section of $G$ in $E$ which is also a Lie group homomorphism (i.e. a splitting map for the central extension of Lie groups), then $ \dot{\sigma}= \mathrm{Lie}(\sigma)$ is a continuous Lie algebra homomorphism with $\dot{\pi} \circ \beta = Id$ and thus a splitting map, defining a trivialization for the corresponding central extension of Lie algebras. Therefore, triviality of the Lie algebra extensions is a necessary condition for the triviality of the Lie group extension. The converse is not true in general, but at least under rather mild assumptions:
\begin{Proposition}[Lie Group Homomorphisms from Lie Algebra Homomorphisms]
\mbox{}\\
Let $G_1 , G_2$ two Lie groups with Lie algebra $\mathfrak{g}_1$ and  $\mathfrak{g}_2$, respectively. Assume that $G_1$ is connected and simply connected, and $G_2$ is regular. Then, for every continuous Lie algebra homomorphism $\mathfrak{F}:\, \mathfrak{g}_1 \rightarrow \mathfrak{g}_2$ there exists a unique Lie algebra homomorphism $f:\, G_1 \rightarrow G_2$ with $\mathfrak{F} = \dot{f}$.
\end{Proposition}
\begin{proof} This is theorem 8.1 in \cite{Miln} \end{proof}

\noindent From this, we conclude:

\begin{Proposition}[Triviality of Extensions]\label{PropTrivialityofExtensions}
\mbox{}\\
Let $1 \longrightarrow A \xrightarrow{\;\;\imath\;} E \xrightarrow{\;\;\pi\;} G \longrightarrow 1$ be a central extension of $G$ by $A$. Let  $G$ and $E$ satisfy the requirements of the previous proposition. Then, the central extension splits by a continuous section if and only if the corresponding central extension of Lie algebras \eqref{LAfromLG2} splits by a continuous Lie algebra homomorphism.
\end{Proposition}

\begin{proof} If the central extension of Lie groups splits by a section $\sigma: G \rightarrow E$, then $\beta = \dot{\sigma}$ is a Lie algebra homomorphism with $\dot{\pi} \circ \beta = Id$ and thus a splitting map for the LIe algebra extension which trivializes the central extension. Conversely, let $\beta: \mathrm{Lie}(G) \rightarrow  \mathrm{Lie}(E)$ be a Lie algebra homomorphism with $\dot{\pi} \circ \beta = Id$. Using the previous proposition, we can deduce that there exists a Lie group homomorphism $\sigma: G \rightarrow E$ with $\beta = \dot{\sigma}$. Since $\dot{(\pi \circ \sigma)} = \dot{\pi}\, \beta = Id_{\mathfrak{g}} = \dot{Id_G}$, the uniqueness part of the proposition yields $\pi \circ \sigma = Id_G$ and thus $\sigma$ is a splitting map for the central extension of Lie groups.
\end{proof}

%\noindent Applying all these results to our discussion of projective representations, we have found several necessary and/or sufficient conditions for being able to lift a projective representation of a Lie group to a proper representation. In particular, we see that the \textit{cocycles represent the obstruction} for lifting a projective representation to a proper representation.\\ 
\noindent For completeness, we state Bargmanns theorem which is probably the main result of the theory of central extensions, and give a brief sketch of the proof. We will not make any further use of this result, however,  it's certainly of great interest on its own. For the complete, original proof we refer the reader to \cite{Bar54}.

\begin{Theorem}[Bargmann, 1954]
\mbox{}\\
Let G be a connected and simply connected finite-dimensional Lie group with
\begin{equation*}\mathrm{H}^2(\mathrm{Lie}(G), \mathbb{R}) = 0\end{equation*}
Then, every projective unitary representation of $G$ on $\mathbb{P}(\mathcal{H})$ can be lifted to a unitary representation  on $\mathcal{H}$.
\end{Theorem}
\begin{proof}By Lemma\eqref{liftingrep} there exists a central extension $\widetilde{G}$ of $G$ by $\cu(1)$, such that the following diagram commutes:
\begin{equation*}
\begin{xy}
  \xymatrix{
      1 \ar[r] & \cu(1) \ar[r]^i \ar[d]_{Id}    &   \widetilde{G} \ar[d]^{\widetilde{\Gamma}} \ar[r]^ \pi & G  \ar@{-->}@/_0.6cm/[l]_\sigma \ar[d]^ {\Gamma} \ar[r]          & 1  \\
      1 \ar[r] & \cu(1) \ar[r]            &   \cu(\HH) \ar[r]^ {\hat{\gamma}}          & \cu(\mathbb{P}({\HH})) \ar[r] & 1 
  }
\end{xy}
 \end{equation*}
Now, one has to show that $\widetilde{G}$ can be given the structure of a (dim($G$) + 1)-dimensional Lie group. Thus, to this central extension of Lie groups corresponds a central extension of Lie algebras. Since $\mathrm{H}^2(\mathrm{Lie}(G), \mathbb{R}) = 0$, this central extension splits. By  the previous Lemma, this implies that the central extension of Lie groups split, i.e. there is a differentiable homomorphism $\sigma:\, G \rightarrow \widetilde{G}$ with $\pi \circ \sigma = Id_G$. Then, $\widetilde{\Gamma} \circ \sigma$ is the postulated lift.\\
\end{proof}

\subsubsection{A final remark} 
Our discussion shows why in quantum theory it makes sense to study representations not of the generic symmetry groups but of their corresponding \textit{universal covering group}, which is (the unique) simply connected Lie group $\widehat{G} \xrightarrow{\pi} G$ covering $G$. Of course, every representation $\rho$ of the original Lie group G can be lifted to a representation of the universal covering, simply by setting $\hat{\rho} := \rho \circ \pi$. The converse is not true in general, but every representation of the universal covering group defines a \textit{projective representation} of the original symmetry group (cf. the example in 3.2.2.).
For example, one might study the group $\mathrm{SL}(2,\IC)$, which is the universal covering group of the (proper orthochronous) Lorentz-group $L_{\uparrow}^+ = \mathrm{SO}(1,3)$, or, more generally, the \textit{spin groups} $\mathrm{Spin}_n \xrightarrow{\pi} \mathrm{SO}(n)$ which can be realized as subgroups of the units in a \textit{Clifford algebra} and lead to the so-called \textit{spin representations.} 

\newpage 
\thispagestyle{empty}
\quad 
\newpage

\chapter{The Central Extensions $\Gres(\HH)$ \& $\Ures(\HH)$}

We construct a central extension $\Gres(\HH)$ of the restricted general linear group $\Gl_{\rm res}(\HH)$ by $\IC^ {\times} = \IC \backslash \lbrace 0 \rbrace$ which restricts to a central extension of $\Ur(\HH)$ by $\cu(1)$. These groups act naturally on the fermionic Fock space which will be constructed in the next chapter in the spirit of the Dirac sea theory. This action lifts the projective representation of $\Ur$ (respectively, $\Gl_{\rm res}$) on $\Gr$, given by
\begin{equation}\begin{split}\label{rhoGr} &\Ur(\HH) \times  \Gr \to \Gr,\\ &\bigl(U,W\bigr) \longmapsto \rho(U)\,W = UW \end{split}\end{equation}
to a proper representation on the Fock space. In this sense, the central extension can be understood as the ``second quantization'' of the groups $\Gres(\HH)$ and $\Ures(\HH)$. In more mathematical terms, they represent an explicit construction of the ``deprojectivization'', discussed in a more abstract context in \S 3.3.1. (c.f.  \cite{Wurz}, Prop. II.27). One may think of them as ``bigger'' groups containing the information about the transformation itself, plus all possible choices for the phase. A little more to the point, we can think of it in the following way: $\Ur(\HH)$ and $\Gl_{\rm res}(\HH)$ act on Dirac seas only projectively.  $\Gres(\HH)$ and $\Ures(\HH)$ contain the additional information about how to rotate the basis of the Hilbert space appropriately, in order to turn this into a proper action on infinite-wedge-vectors (cf. the discussion of the left- and right-action in Section 5.2). This provides a generalization of the product-wise lift
\begin{equation*} \varphi_1 \wedge \varphi_2 \wedge \dots \wedge \varphi_n \xrightarrow{U\; }  U\varphi_1 \wedge U\varphi_2 \wedge \dots \wedge U\varphi_n \end{equation*}
to states of infinitely many fermions.

\section{The Central Extension of $\Gl_{\rm res}$}

\noindent We start with the construction of the central extension of the identity component $\Gl^0_{\rm res}(\HH)$ of  $\Gl_{\rm res}(\HH)$. The construction is based on the standard polarization $\HH = \HH_{+} \oplus \HH_{-}$, but application to a different (fixed) polarization is immediate.
With respect to this splitting, we can write every linear operator as
\begin{align*}
  A = \begin{pmatrix}
    A_{++} & A_{+-}\\
    A_{-+} & A_{--}
  \end{pmatrix}
  = \begin{pmatrix}
    a & b\\
    c & d
  \end{pmatrix}.
\end{align*}
\noindent We define
\begin{equation*} 
\mathcal{E} = \bigl\lbrace (A , q) \in \Gl^{0}_{\rm res}(\HH) \times \Gl(\HH_+) \mid a - q \in I_1({\HH_+}) \bigr\rbrace .
\end{equation*}
%\vspace*{1mm}
\begin{Lemma}[Group structure of $\mathcal{E}$]
\item $\mathcal{E}$ is a subgroup of $\Gl^{0}_{\rm res}(\HH) \times \Gl(\HH_+)$.
\end{Lemma}

\begin{proof}: Clearly, $(Id_{\HH} , Id_{\HH_+} ) \in \mathcal{E}$.
Furthermore:
\begin{enumerate}[i)]
\item  If $(A_1 , q_1) , (A_2, q_2) \in \mathcal{E}$ with 
 $A_1 = \begin{pmatrix}
    a_1 & b_1\\
    c_1 & d_1
  \end{pmatrix}
  , \; A_2 = \begin{pmatrix}
    a_2 & b_2\\
    c_2 & d_2
  \end{pmatrix}$, we have\\
   $(A_1A_2)_{++} = a_1a_2 + b_1c_2$ and
  \begin{equation*}(A_1A_2)_{++} - \,q_1q_2 =  a_1a_2 + b_1c_2 - q_1q_2 = \underbrace{(a_1-q_1)} \limits_{\in \TC}a_2 \,+ \, q_1\underbrace{(a_2-q_2)}\limits_{\in \TC}+\underbrace{b_1}\limits_{\in \HS} \underbrace{c_2}\limits_{\in \HS} \in \TC. \end{equation*}
Thus: $(A_1 A_2 , q_1 q_2) \in \mathcal{E}$.
 
\item  Let $(A , q) \in \mathcal{E}$.

\begin{equation*} A = \begin{pmatrix}
    a & b \\
    c & d
  \end{pmatrix} \in \Gl^{0}_{\rm res}(\HH) \; \Rightarrow \; A^{-1} = \begin{pmatrix}
    \alpha & \beta \\
    \gamma & \delta
  \end{pmatrix} \in \Gl^{0}_{\rm res}(\HH). \end{equation*}
  We deduce: $a \alpha + b \beta = Id_{\HH_+} \Rightarrow a \alpha \in Id + \TC$ and thus: \begin{equation*}  a\, (\alpha - q^{-1}) = \underbrace{a \alpha}\limits_{\in Id + \TC} - \underbrace{a q^{-1}} \limits_{\in Id + \TC} \in \TC \end{equation*} 
 Since $a$ is a Fredholm operator with index $0$, it has a pseudoinverse. It follows that $(\alpha - q^{-1})\in\TC$ and so $(A^{-1} , q^{-1} ) \in \mathcal{E}$.
 \end{enumerate} 
 \end{proof}

\begin{Lemma}[Some properties of $\mathcal{E}$]
\mbox{}
\begin{enumerate}[i)]
\item $a - q \in  I_1({\HH_+})$ is equivalent to $q^{-1}a $ having a determinant. 
\item For all $A \in \Gl_{\rm res} \; \exists\;q \in \Gl(\HH_+)$ such that $(A, q) \in \mathcal{E}$.
\item If $aq^{-1}$ has a determinant then $aq'^{-1}$ has a determinant if and only if $q^{-1}q'$ does.
\end{enumerate}
\end{Lemma}
\noindent It follows that we can think of $\mathcal{E}$ as an extension (not a central one!) of $\Gl_{\rm res}$ by $\GL$:
\begin{equation*}
1 \longrightarrow \GL \longrightarrow \mathcal{E} \xrightarrow{\; pr_2 \;} \Gl^{0}_{\rm res}(\HH) \longrightarrow 1
\end{equation*}
$ii)$ implies that the projection $pr_2$ is surjective onto $\Gl^{0}_{\rm res}(\HH)$ and $iii)$ implies that two preimages differ (multiplicatively) by an element in $\GL$.
\begin{proof}
i) $a - q \in  I_1({\HH_+}) \iff q^{-1}(a-q)  = q^{-1}a - Id \in  I_1({\HH_+}) \iff q^{-1}a  \in  Id + I_1(\HH_+)$.\\
ii) $A \in \Greso \Rightarrow U_{++}\lvert_{\HH_+ \to \HH_+} =: a$ is a Fredholm operator of index 0. Thus, there exists a finite rank operator $t :\HH_+ \rightarrow \HH_+$ such that $q:= a - t$ is invertible and $a - q = t \in  I_1({\HH_+})$.\\
iii) If $q^{-1}a$ and $q'^{-1}a$ both have determinants, so does $ (q^{-1}a)(q'^{-1}a)^{-1} = q^{-1}q'$.\\ Conversely, if $q^{-1}a , q^{-1}q' $ have determinants, so does $ (q^{-1}q' )^{-1} (q^{-1}a) = q'^{-1}a $.\\
%\end{enumerate}
\end{proof}

\begin{Definition}[Central Extension of $\Gl_{\rm res}$]
\mbox{}\\
The group $\mathrm{SL}(\HH_+) = \lbrace q \in \mathrm{GL}^1(\HH_+) \mid \det(q)=1 \rbrace$ is embedded into $\mathcal{E}$ as 
\begin{equation*} \mathrm{SL}(\HH_+) \hookrightarrow \mathcal{E}, \; q \mapsto (\mathds{1} , q).\end{equation*} 
\noindent We define 
\begin{equation*}\GresO(\HH) \; := \; \mathcal{E} /\, \mathrm{SL}(\HH_+) \end{equation*}
and will show that this is a central extension of $\Greso$ by $\IC^{\times} := \IC\setminus\lbrace 0 \rbrace$.
\end{Definition} 
\noindent More explicitely, $\GresO(\HH)$ is the quotient of $\mathcal{E}$, modulo the equivalence relations
\begin{equation}\label{equivrelationE} (A_1 , q_1) \sim (A_2 , q_2) : \iff \, A_1 = A_2 \; \text{and}\; \det(q_2^{-1}\, q_1) = 1. \end{equation}

\begin{Lemma}[Central Extension of $\Gl_{\rm res}$]
\mbox{}\\
The extension 
\begin{equation*}
1 \longrightarrow \GL \longrightarrow \mathcal{E} \xrightarrow{\; pr_2 \;} \Gl^{0}_{\rm res}(\HH) \longrightarrow 1
\end{equation*}
descends to a central extension 
\begin{equation}\label{GLresExtension}
\fingbox{1 \longrightarrow \IC^{\times} \xrightarrow{\; \imath \;} \GresO(\HH) \xrightarrow{\; \pi \;} \Gl^{0}_{\rm res}(\HH) \longrightarrow 1}
\end{equation}
with $\pi: [(A , q)] \mapsto A$ and $\imath: c \mapsto [(Id , t)]$ for any $t \in \GL$ with $\det(t) = c$.
\end{Lemma}

\begin{proof}:
\begin{itemize}
\item Multiplication is well-defined on $\mathcal{E} \slash \sim$:\\
If $(A_1, q_1) \sim (A_1, r_1)$ and $(A_2, q_2) \sim (A_2, r_2)$ i.e. $\det(q^{-1}_1 r_1) = 1 = \det(q^{-1}_2 r_2)$ then $\det(q_2^{-1} q_1^{-1}r_1r_2 ) = \det(q_1^{-1}r_1r_2q_2^{-1}) = \det(q^{-1}_1 r_1)\, \det(r_2 q^{-1}_2) = 1$,\\ i.e. $(A_1A_2 , q_1q_2) \sim (A_1A_2, r_1 r_2)$.

\item Now it's easy to see that $\pi$ and $\mathfrak{i}$ are well-defined homomorphisms of groups and that  $\pi$ surjective and $\mathfrak{i}$ is injective.

\item We show $\ker(\pi) = \im(\mathfrak{i})$: Note that $\ker(\pi) = \lbrace (\mathds{1},q) \rbrace \subset \mathcal{E}\,$, with the $q$'s satisfying $\mathds{1}_{++} - q = \mathds{1}_{\HH_+} - q \in  \TC$ and thus $q \in \Gl(\HH_+)\cap (Id + \TC) = \GL$.
Hence, $\ker(\pi) = \im(\mathfrak{i})$.

\item The extension is central:  For all $(A , q ) \in \mathcal{E}$ and $t \in \GL$, we find
$(A , q t) \sim (A , t q)$, since $\det((qt)^{-1}tq \, ) = \det(t^{-1}\,q^{-1}\,t \, q  ) = \det(t)^{-1} \det( q^{-1}\,  t q) = \det(t)^{-1}\,\det(t) = 1$.\\
Thus,  $[(A , q)] \cdot [(\mathds{1} , t)] = [(\mathds{1} , t)] \cdot [(A , q)]\, , \; \forall\, [(A , q)] \in \GresO \; \text{and} \; [(\mathds{1} , t)] \in \im(\mathfrak{i})$. 
\end{itemize}
\end{proof}
\newpage
\subsection{The complete $\widetilde{\Gl}_{\rm res}$}
So far we have constructed the identity component $\GresO$ of the central extension, corresponding to transformations preserving the net-charge. Generalization to other connected components indexed by the relative charge $q = \charge(W, AW)$ for $A \in \Gl_{\rm res}, W \in \mathrm{Gr}(\HH)$, is obtained by concatenating charge-preserving transformations with a shift of a suitably chosen Hilbert-basis.\\

\noindent Let $(e_k)_{k \in \IZ}$ be a basis of $\HH$ such that $(e_k)_{k \geq 0}$ is a basis of $\mathcal{H}_+$ and $(e_k)_{k < 0}$ a basis of $\mathcal{H}_-$. We need a $\sigma \in \Gl_{\rm res}$ with $\ind(\sigma_{++}) = \pm 1$. Conveniently, we choose $\sigma$ defined by $e_k \longmapsto e_{k+1}$. This $\sigma$ is unitary with $\ind(\sigma_{++}) = -1$. It acts on $\Gl^{0}_{\rm res}$ by $A \mapsto \sigma A \sigma^{-1}$. Now, we can define a semi-direct product on $\IZ \ltimes \Gl^{0}_{\rm res}$ by 
\begin{equation} (n, A ) \cdot_{\ltimes} (m, A') = (n+m , A \sigma^n A' \sigma ^{-n}). \end{equation}
With this group-structure, the obvious map
\begin{equation} \IZ \ltimes \Gl^{0}_{\rm res} \rightarrow \Gl_{\rm res},\;(A , n) \mapsto A\,\sigma^n \end{equation}
becomes an isomorphism of groups: 
\begin{equation*}(n , A)\cdot (m , A') = (n + m , A \sigma^n A' \sigma ^{-n} ) \longmapsto A \sigma^n A' \sigma ^{-n} \sigma^{n+m} = A \sigma^n A' \sigma ^m. \end{equation*}  
Thus, we can describe the restricted general group as a semi-direct product of the identity component $\Gl^{0}_{\rm res}$ with $\IZ$. The $\IZ$-component corresponds to $(-1) \times$ the index of the $++$ component, i.e. to the relative charge ``created'' by the transformation.\\

\noindent The action of $\sigma$ on $\Greso$ (by conjugation) is covered by an endomorphism
 $\tilde{\sigma}: \mathcal{E} \rightarrow \mathcal{E}$, $\tilde{\sigma}( (A , q) ) = (\sigma A \sigma^{-1} , q_{\sigma} )$, where
\[
q_{\sigma} \; = \; \left\{
\begin{array}{ll}
\sigma q \sigma^{-1}  
& \mbox{; on } \sigma(\HH_+) = \HH_{\geq1} \\ \\
\mathrm{Id}
& \mbox{; on } \sigma(\HH_+)^{\perp} = \spn(e_0)
\end{array}
\right. 
\]
This is not an automorphism on $\mathcal{E}$, but it descends to an automorphism on $\GresO$. With a little abuse of notation, we write $\tilde{\sigma}$ for this map as well.\\
 
\begin{Definition}[The complete $\Gres$]
\mbox{}\\
We define
\begin{equation*}\Gres(\HH) := \IZ \ltimes \GresO(\HH) \end{equation*}
with the action of $\IZ$ on $\GresO(\HH)$ generated by  $\tilde\sigma$.
Then 
 \begin{equation*}1 \longrightarrow \IC^{\times} \longrightarrow \Gres(\HH) \xrightarrow{\; \pi \;} \Gl_{\rm res}(\HH) \longrightarrow 1 \end{equation*}
 is a central extension of $\Gl_{\rm res}$ by $\IC^{\times}$ and restricts to \eqref{GLresExtension} over the identity components of the Lie groups.
 \end{Definition}
\newpage

\subsection{Lie Group- and Bundle- structure}

We will show that $\GresO$ carries the structure of an infinite-dimensional Banach Lie group. To this end, we first introduce a Lie group structure on $\mathcal{E}$ by giving it not the subgroup topology, but the topology induced by the embedding 
 \begin{equation*} \iota:\; \mathcal{E} \to \mathcal{B}_{\epsilon} \times \TC; \; (A , q) \mapsto (A , a - q).\end{equation*}
 
 \begin{Lemma}[Lie Group structure of $\mathcal{E}$]
 \mbox{}\\
 $\mathcal{E}$ with the topology defined by the embedding $\iota$ is a Banach Lie group.
 \end{Lemma}
 
 \begin{proof} Let $\mathcal{A}:=\mathcal{B}_{\epsilon} \times \TC$ be equipped with the product-norm
 \begin{equation*} \lVert(A, t)\rVert_{\mathcal{A}} := \lVert A \rVert _{\epsilon} + \lVert t \rVert_1. \end{equation*}
 We define a multiplication $\star$ on $\mathcal{A}$ in such a way, as to make $\iota$ an algebra-homomorphism:
 \begin{equation*}Ê(A,t) \star (A',t') := (AA', at' + ta' + bc' -tt'). \end{equation*}  
It is easily checked that $\bigl((\mathcal{A}, + , \star), \lVert \cdot \rVert_{\mathcal{A}} \bigr)$ is a Banach algebra with unit element $(\mathds{1}, 0 )$.   
\noindent We claim that $\iota(\mathcal{\epsilon})$ is precisely the group  $\mathcal{A}^{\times}$ of unit elements in $\mathcal{B}_{\epsilon} \times \TC$ and therefore a linear Banach Lie group.
 Note that $(A,t) \in \mathcal{A}$ is invertible if and only if $A \in \GL_{\rm res}(\HH)$ and there exists $t' \in \TC$ s.t.
\vspace*{-4mm}\begin{align*}  
(A,t) \star (A^{-1},t') :=& (\mathds{1}, at' + ta' + bc' -tt' ) = (\mathds{1},0)
\end{align*}
for $aa' + bc' = \mathds{1}_{\HH_+}$. And this is true, if and only if 
\begin{align*}  &(a-t) t' + ta' + bc' = (a-t) t' + ta' + ( \mathds{1}_{\HH_+} - aa') = 0\\[2ex]
& \iff (t-a) t' - (t-a) a' = (t - a) (t' - a') = \mathds{1}_{\HH_+}\\
& \iff  (A, t) = \iota\bigl( (A, t-a) \bigr) ; \; \text{with} \; A \in  \GL_{\rm res}(\HH), \, (t - a) \in \GL.
\end{align*} 
This finishes the proof.
\end{proof} 

\noindent Now, we can understand $\GresO(\HH)$ as homogeneous space for the Lie subgroup $\mathrm{SL}(\HH_+).$ It follows that there exist a unique Lie group structure on $\GresO(\HH) = \mathcal{E} / \, \mathrm{SL}(\HH_+)$ such that the canonical projection becomes a submersion (\cite{Bourbaki}, \S 1.6, Prop. 11). We can conclude that $\GresO(\HH)$ is an infinite-dimensional Banach Lie group.

\subsubsection{Bundle structure}
\noindent In fact, $\pi:\,\GresO \rightarrow \Gl^{0}_{\rm res}$ is not only a central extension, it also carries the structure of a \textit{principle fibre bundle}.
On principle-bundles, local trivializations are given by \textit{local sections}. As we're dealing we Lie groups, the bundle structure is already defined by a smooth section in a neighborhood of the identity (because any such section can be translated to arbitrary points on the manifold by the multiplicative action of group on itself).  For $\GresO$, there exists a very natural local section about the identity in $\Greso(\HH)$ that will be of great importance for the further discussion.\\

\noindent Consider 
\begin{equation*} W := \lbrace A \in \Gl_{\rm res} \mid  a = A_{++} \in \Gl({\HH_+}) \rbrace  \subset \Greso \end{equation*}

\noindent\textbf{Claim:} W is a dense, open subset of $\Greso(\HH)$.
\vspace*{-2mm}
\begin{proof} Recall that the topology in $\Gres$ is given by the norm $\lVert\cdot\rVert_{\epsilon} = \lVert\cdot\rVert_{\infty} + \lVert [\epsilon, \cdot] \rVert_{2}$.\\
Clearly, $A \mapsto a = P_+ A P_+$ is continuous w.r.to this norm, therefore $W$ is open in $\Gres$ because $\Gl(\HH_+)$ is open in the space of bounded operators on $\HH_+$: For $a \in \Gl(\HH_+)$ and $\lVert k \rVert$ small enough, $a - k$ is also invertible with $(a - k)^{-1}=\bigl[ \sum\limits_{\nu = o}^{\infty} (a^{-1}k)^{\nu} \bigr] a^{-1}$.\\
Furthermore, $W$ is dense in $\Gl_{\rm res}$ because for any Fredholm operator $a$ with index 0 -- such as $a = A_{++}$ for $A \in \Gl_{\rm res}$ -- there exists a compact operator $k$ of arbitrary small norm so that $a+k$ is invertible.\\
\end{proof}

\noindent On $W$ we can define a section
\begin{equation}\label{tau}\addtolength{\fboxsep}{5pt} \boxed{\tau:\; W \longrightarrow \Gres(\HH);\hspace{0.2cm}A \longmapsto [(A , A_{++})].}\end{equation}
Obviously, $\tau$ is smooth because $\iota\circ\tau(A) = (A, 0) \in \Gl_{\rm res} \times \TC$. Also, $\tau(\mathds{1}) = \mathds{1}$. This section induces the local trivialization
\begin{equation}\label{localtriv}\addtolength{\fboxsep}{5pt} \boxed{\begin{split} \phi: \; & \Gres(\HH)  \supseteq \,  W \longrightarrow \Gl_{\rm res}(\HH) \times \IC^{\times}; \\
& [(A , q)] \longmapsto \bigl( A \, , \det(a^{-1}q) \bigr)\end{split}} \end{equation} 
\noindent by $\phi^{-1}(A, \lambda) := \tau(A) \cdot \lambda \in \Gres(\HH)$.

\begin{Proposition}[Lie Group Cocycle]
\mbox{}\\
On $W$ we compute the continuous 2-cocycle $ \tau(A) \tau(B) = \chi(A, B)\, \tau(AB) $
to be
\begin{equation}\label{Liegroupcocycle} \chi(A , B) = \det[A_{++} B_{++} (AB_{++})^{-1}]. \end{equation}
\end{Proposition}
\begin{proof} Let $A, B \in W$ s.t. $AB \in W$. Then $\tau(AB)=  [(AB , (AB)_{\,++})]$. And thus:
\begin{align*} \tau(A)\cdot \tau(B) = & [(A , A_{++})]\,[(B , B_{++})] = [(AB , A_{++} B_{++})] \\ = & [(\mathds{1}, A_{++} B_{++} (AB_{++})^{-1})] \,\tau(AB)\\
= & \det[A_{++} B_{++} (AB_{++})^{-1}]\cdot \tau(AB)\end{align*}
From which we can immediately read off the cocycle.\\
\end{proof}
\noindent The central extension of Lie groups
\begin{equation}
1 \longrightarrow \IC^{\times} \xrightarrow{\; \imath \;} \GresO(\HH) \xrightarrow{\; \pi \;} \Gl^{0}_{\rm res}(\HH) \longrightarrow 1
\end{equation}
induces a central extension 
\begin{equation}\label{g1Extension}
0 \longrightarrow \IC \xrightarrow{\;\dot{\imath}\;} \tilde{\mathfrak{g}}_1\xrightarrow{\; \dot{\pi} \;} \mathfrak{g}_1 \longrightarrow 0
\end{equation}
of the Lie algebra. We compute the corresponding Lie algebra cocycle.\\
\begin{Proposition}[Lie Algebra Cocycle]
\label{Prop:cocycleisSchwinger}
\mbox{}\\
The Lie algebra cocycle for to the central extension \eqref{g1Extension} is
\begin{equation}\label{Schwingercocycle}\addtolength{\fboxsep}{5pt} \boxed{\begin{split}
c(X, Y) = \,& \trace(X_{-+} Y_{+-} - Y_{-+}X_{+-})\\
=\,& \frac{1}{4}\, \trace(\epsilon \, [\epsilon , X] [\epsilon , Y]) \end{split}}\end{equation}
\end{Proposition}
\noindent This is known as the \emph{Schwinger cocycle} or \emph{Schwinger term}.
\begin{proof}
Using formula \eqref{cocycleformula} we find:
\begin{align*} c(X , Y) \, = \, & \frac{\partial}{\partial t} \frac{\partial}{\partial s} \Bigl\vert_{t=s=0}\det \bigl[ e^{sX}_{++} \, e^{tY}_{++} ((e^{sX}e^{tY})_{++})^{-1} \bigr] 
-  \frac{\partial}{\partial t} \frac{\partial}{\partial s} \Bigl\vert_{t=s=0} \det \bigl[ e^{tY}_{++} \, e^{sX}_{++} ((e^{tY}e^{sX} )_{++})^{-1} \bigr]\\[1.3ex]
= & \trace\bigl(X_{++} Y_{++} \, -  (XY)_{++}\bigr) - \trace\bigl( Y_{++} X_{++} \, -  (YX)_{++}\bigr)\\[1.3ex]
= & \trace\bigl(X_{++} Y_{++}\, - X_{++}Y_{++} - X_{+-}Y_{-+}\bigr)- \trace\bigl( Y_{++} X_{++} \, -  Y_{++}X_{++} - Y_{+-}X_{+-}\bigr)\\[1.3ex]
= & \trace\bigl(- X_{+-}Y_{-+}\bigr) + \trace\bigl( Y_{+-}X_{-+}\bigr)
\end{align*}
The traces converge individually and put together we get
 \begin{equation*} c(X , Y) \, =  \trace\bigl(X_{-+}Y_{+-} - Y_{-+}X_{+-}\bigr). \end{equation*}
 A straight forward computation shows that this can also be expressed as 
\begin{equation*} c(X , Y) \, = \frac{1}{4}\, \trace(\epsilon \, [\epsilon , X] [\epsilon , Y]). \end{equation*}
\end{proof}

\subsection{The Central Extension of $\Ur$ and its Local Trivialization}
Having defined the central extension $\Gres(\HH)$ of $\Gl_{\rm res}(\HH)$, we set
%\begin{equation} \addtolength{\fboxsep}{5pt}\boxed{\Ures(\HH) := \bigl\lbrace [(U, r)] \in \Gres\, \bigr\rvert (U,r) \in \cu(\HH) \times \cu(\HH_+)\bigr\rbrace} \end{equation}
\begin{equation} \Ures(\HH) := \bigl\lbrace [(U, r)] \in \Gres\, \bigr\rvert (U,r) \in \cu(\HH) \times \cu(\HH_+)\bigr\rbrace \end{equation}
This is a central extension of $\Ur(\HH)$ by $\cu(1)$ and its complexification is $\Gres(\HH)$.\\

\noindent $\Ures(\HH)$ is a principle $\cu(1)$-bundle. But the section $\tau$ defined above does not restrict to a section in $\Ures(\HH)$ because for unitary $U$, $U_{++}$ need not be unitary, even if it's invertible. We can however use a polar decomposition to write 
\begin{equation} U_{++} = V_U \lvert U_{++} \rvert\;\; \text{with}\;\; V_U \in \cu(\HH_+)\; \text{unitary},\;  \lvert U_{++} \rvert = \sqrt{U^*_{++}U_{++} }.  \end{equation} 

\begin{Lemma}[Local Section of $\Ures$]\label{Uressection}
\mbox{}\\
On $W\cap \Ur(\HH)$, the map 
\begin{equation}\label{sigma} \sigma: U \longmapsto  [(U , V_U)] \end{equation}
defines a local section in $\Ures(\HH)$.
The local trivialization $\phi_{\Ur}$ induced by this section equals $\phi$ up to normalization. 
\end{Lemma} 
\begin{proof} For $U \in \Ur$, unitarity implies $U^*_{++}U_{++} + U^*_{+-}U_{+-} = \mathds{1}_{\HH_+}$.
As the odd-parts are Hilbert-Schmidt operators, $U^*_{+-}U_{+-}$ is trace-class. Therefore,
$U^*_{++}U_{++} \in Id_{\HH_+}+I_1(\HH_+)$ has a determinant. It follows that $\sqrt{(U^*_{++}U_{++})} = \lvert U_{++} \rvert $ is also in $Id_{\HH_+} + I_1(\HH_+)$ and we conclude 
\begin{equation*}U_{++} - V_U = V_U \bigl(\lvert U_{++} \rvert - \mathds{1}_{\HH_+} \bigr) \in I_1(\HH_+).\end{equation*}
Hence, $[(U , V_U)] \in \Ures(\HH)$ for all $U \in \Ur$. Furthermore, we find $\phi\bigl([(U, r)]\bigr) = (U, \lambda)$ with 
\begin{align*} \lambda & = \det( U^{-1}_{++}\, r ) = \det(\lvert U_{++} \rvert^{-1}) \det(V^*_U\, r) = \det\Bigl(\sqrt{(U^{-1}_{++}\, r )( U^{-1}_{++}\, r )^*}\Bigr)  \det(V^*_U\,  r)\\[1.5ex]
& =  \sqrt{\det( U^{-1}_{++}\, r )\,   \overline{\det( U^{-1}_{++}\, r )}}\,  \det(V^*_U\, r) = \lvert \lambda \rvert \det(V^*_U\, r).  \end{align*}
But  $\det(V^*_U\, r)$ is just the $\cu(1)$-component w.r.to the local trivialization $\phi_{\Ur}$ defined by $\sigma$. 
Hence we read off:
\begin{equation}
\begin{xy}
  \xymatrix{
       & \mathfrak{U} \in \Ures(\HH)\ar[dl]_{\phi}\ar[dr]^{\phi_{\Ur}}  &    \\
  \Ur\times \IC^{\times}\ni(U, \lambda)       &     & \bigl(U , \frac{\lambda}{\lvert\lambda\rvert}\bigr) \in \Ur \times \cu(1)
  }
\end{xy}
 \end{equation} 
In particular, for $r = V_U$ we find $\phi \circ \sigma (U) = \bigl(U , \frac{1}{\lvert \det(U_{++})\rvert} \bigr) \in \Ur \times \IC^{\times}$, showing that the local section $\sigma$ is smooth in $\Gres(\HH)$ (see Prop. \ref{Prop:dethol} for smoothness of the determinant). This finishes the proof.
 \end{proof}
\noindent In the local trivialization of $\Ures(\HH)$, the $\cu(1)$-component corresponds to the information about the phase of the lift of the unitary transformation. In this sense, we can think of the sections as ``gauging'' the phases of the lifts by picking out a reference lift.\\
%\newpage

\noindent It can be readily computed that the cocycle for the section $\sigma$ is, as one would expect, just the normalized version of the $\Gl_{\rm res}$-cocycle $\chi$ in \eqref{Liegroupcocycle}, i.e.
\begin{equation}\label{UresLiegroupcocycle} \chi_{\cu}(U,V) := \det[U_{++} V_{++} (UV_{++})^{-1}] \slash \,\lvert\, \det[U_{++} V_{++} (UV_{++})^{-1}]\, \rvert. \end{equation} 
For the corresponding central extension of Lie algebras, $ 0 \longrightarrow \IR \longrightarrow \widetilde{\mathfrak{u}}_{\rm res} \longrightarrow \mathfrak{u}_{\rm res} \longrightarrow 0 $,
we compute the cocycle
\begin{align*}  \;&\frac{\partial}{\partial t} \frac{\partial}{\partial s} \Bigl\vert_{t=s=0} \chi(e^{sX}, e^{tY}) \, \lvert \chi(e^{sX}, e^{tY}) \rvert^{-1} - \frac{\partial}{\partial t} \frac{\partial}{\partial s} \Bigl\vert_{t=s=0} \chi(e^{tY}, e^{sX}) \, \lvert \chi(e^{tY}, e^{sX}) \rvert^{-1}\\
\;=\; & \frac{\partial}{\partial t} \frac{\partial}{\partial s} \Bigl\vert_{t=s=0}\;\Bigl[ \chi(e^{sX}, e^{tY})  \bigl(\chi(e^{sX}, e^{tY}) \overline{\chi(e^{sX}, e^{tY})}\: \bigr)^{-\frac{1}{2}}
%- \; \stackrel{(X \leftrightarrow Y)}{\ldots}\\
-  \chi(e^{tY}, e^{sX}) \bigl(\chi(e^{tY}, e^{sX}) \overline{\chi(e^{tY}, e^{sX})}\, \bigr)^{-\frac{1}{2}}\Bigr] 
\end{align*}\vspace{2mm} 
$=\; c(X ,Y) - \frac{1}{2} \bigl( c(X ,Y) + \overline{c(X,Y)} \bigl) = i\,\mathfrak{Im} \bigl( c(X,Y) \bigr), \;\; \text{for}\; X,Y \in \mathfrak{u}_{\rm res}$.\\
But as all $X,Y \in \mathfrak{u}_{\rm res}$ are (anti-)Hermitian, we find
\begin{equation*}\overline{c(X,Y)} = c(Y^* , X^*) = c(Y , X) = - c(X, Y),  \end{equation*}
which means $i\, \mathfrak{Im} \bigl( c(X,Y) \bigr) =  c(X,Y)$ on $\mathfrak{u}_{\rm res}$.
Thus, the Lie algebra cocycle on $\mathfrak{u}_{\rm res}$ defined by $\sigma$ is just the Schwinger cocycle restricted to the subalgebra $\mathfrak{u}_{\rm res} \subset \mathfrak{g}_1$. 
%\vspace*{0.5cm}

\section{Non-Triviality of the Central Extensions}

\begin{Theorem}[Non-triviality of the Central Extension]
\mbox{}\\
The following equivalent statements hold true:
\begin{enumerate}[i)]
\item The central extension $ \Gres(\HH) \xrightarrow{\; \pi \;} \Gl_{\rm res}(\HH)$ and $\Ures(\HH)\xrightarrow{\; \pi \;} \cu_{\rm res}(\HH)$ are not trivial.
\item The Lie algebra extensions $\tilde{\mathfrak{g}}_1\rightarrow \mathfrak{g}_1$ and $\widetilde{\mathfrak{u}}_{\rm res} \rightarrow \mathfrak{u}_{\rm res}$ are not trivial.
\item The Schwinger cocycle \eqref{Schwingercocycle} is not trivial in $\mathrm{H}^2(\mathfrak{g}_1, \IC)$, respectively in $\mathrm{H}^2(\mathfrak{u}_{\rm res}, i\IR).$
\item There exists no continuous (local) section $\Gamma: \Ur(\HH) \rightarrow \widetilde{\cu}_{\rm res}(\HH)$ which is also a homomorphism of groups. The same is true for  $\Gl_{\rm res}(\HH)$. 
\item There exist no unitary representation of $\cu_{\rm res}(\HH)$ on the fermionic Fock space, lifting the projective representation \eqref{rhoGr}. The analogous statement is true for $\Gl_{\rm res}(\HH)$.
\end{enumerate}
\end{Theorem}

\begin{proof}[Proof of the Theorem]\footnote{The proof follows \cite{Wurz}, where it is presented in a very nice and complete way. However, we believe that the abstract mathematical language might conceal the ultimately simple nature of the proof, at least from a physicist's point of view. We have therefore tried to rephrase it in more elementary terms, without any reference to cohomology or the like.} For the equivalence of the statements, we have to recall the discussion  in Chapter 3. Note that any continuous Lie group homomorphism $\varphi$ between linear Banach Lie groups induces a unique continuous homomorphism $\mathrm{Lie}(\varphi) = \dot{\varphi}$ between their Lie algebras with $\exp \circ \mathrm{Lie}(\varphi) = \varphi \circ \exp$ (see Appendix  A.2). Thus, a continuous splitting map for the central extension of Lie groups would induce a splitting map for the central extension of Lie algebras, showing $ii) \Rightarrow i) \iff iv)$. $i) \Rightarrow ii)$ is Prop. \ref{PropTrivialityofExtensions} [Triviality of extensions]. $ii) \iff iii)$ is Thm. \ref{CEalgebraTriviality} [Triviality of Lie algebra extensions]. Finally, $i) \iff v)$ follows analogously to Prop. \ref{Corliftingrep} [Lifting projective representations] from the representation of $\Gl_{\rm res}(\HH)$ on the fermionic Fock space that will be constructed in the next chapter. Def. \ref{Def:Pfluckerembedding} will give an embedding of $\Gr$ into the projective Fock space $\mathbb{P}(\FF)$, known as the Pfl\"ucker embedding.\\

\noindent Now, we are actually going to prove iii), i.e. show that the Schwinger cocycle $c$ is not trivial. This means that there is no linear map $\mu: \mathfrak{u}_{\rm res} \rightarrow \mathbb{R}$ with $\mu([X,Y]) = c(X,Y),\; \forall X,Y \in \mathfrak{u}_{\rm res}$. This is proven, for example, if we can find $X,Y \in \mathfrak{u}_{\rm res}$ with $[X,Y] = 0$ but $c(X,Y) \neq 0$.\\

\noindent For completeness, we recall the argument why this proves that there exists no continuous section in $\Ur(\HH)$ which is also a homomorphism. 
In Lemma \ref{Uressection}, we have constructed a smooth section $\sigma: \Ur \rightarrow \Ures$. However, this section fails to be a homomorphism of groups, i.e. we get $\sigma(U)\sigma(V) = \kappa(U,V) \sigma(UV)$ with a Lie group cocycle $\kappa: \Ur \times \Ur \rightarrow \cu(1)$. 
Consequently, its derivative (at the identity) $\dot{\sigma}: \mathfrak{u}_{\rm res} \rightarrow \tilde{\mathfrak{u}}_{\rm res}$ fails to be a Lie algebra homomorphism. This is expressed by the \emph{Schwinger cocycle} 
\begin{equation*} c(X,Y) := [\dot{\sigma}(X),\dot{\sigma}(Y)] - \dot{\sigma}([X,Y]) \; \text{for}\; X,Y \in \mathfrak{u}_{\rm res}.\end{equation*}
Now suppose the section $\sigma$ is just a bad choice and there was in fact a different section $\Gamma: \Ur \rightarrow \widetilde{U}_{\rm res}$ which is also a Lie group homomorphism. Since two elements in the same fibre in $\widetilde{U}_{\rm res}$ differ only by a complex phase, the ``good'' section differs (multiplicatively) from $\sigma$ by a map $\lambda:\Ur \rightarrow \cu(1)$. As $\Gamma$ is a Lie group homomorphism, the corresponding Lie algebra map $\dot{\Gamma}= \dot{\lambda} + \dot{\sigma}$ is a continuous Lie algebra homomorphism.\\
This means:
\begin{align*}
0=[\dot{\lambda}(X) + \dot{\sigma}(X), \dot{\lambda}(Y) + \dot{\sigma}(Y)] &= [\dot{\sigma}(X),\dot{\sigma}(Y)] = \dot{\lambda}([X,Y]) +\dot{\sigma}([X,Y])\\[1.5ex]
&\Rightarrow \dot{\lambda}([X,Y]) = c(X,Y), \; \forall X,Y \in \mathfrak{u}_{\rm res}.
\end{align*}

For the first equality we have used that $\dot{\lambda}:\mathfrak{u}_{\rm res} \rightarrow \mathrm{Lie}(\cu(1)) =  \mathbb{R}$ maps into the center of $\widetilde{\mathfrak{u}}_{\rm res}$ and so all the commutators with $\dot{\lambda}$ vanish.
We observe: if a lift $\Gamma: \Ur \rightarrow \widetilde{\cu}_{\rm res}$ preserving the group structure exists, then there exists a linear map
$\mu\, (= \dot{\lambda}):\mathfrak{u}_{\rm res} \rightarrow \mathbb{R}$ with $\mu([X,Y]) = c(X,Y),\; \forall X,Y \in \mathfrak{u}_{\rm res}$.
Therefore, to prove that such a lift does \emph{not} exist, it suffices to find $X,Y \in \mathfrak{u}_{\rm res}$ with \begin{equation*} [X,Y] = 0\; \hspace{5mm} \text{but} \hspace{5mm}  c(X,Y) \neq 0. \end{equation*}
This would prove the statement for both the unitary and the general linear case.

\noindent The Hilbert space $\HH= \HH_+ \oplus \HH_-$ with the polarization given by the sign of the free Dirac Hamiltonian is somewhat difficult to handle. Fortunately, by unitary equivalence we can just as well consider any other polarized (separable, infinite-dimensional, complex) Hilbert space. For our purpose it is nice to work with the Hilbert space
$\mathcal{K} := L^2(S^1, \IC)$ with the polarization given by separation into Fourier components with positive and negative frequencies. The natural Hilbert-basis on $\mathcal{K}$ is the Fourier-Basis $(e_k)_{k \in \IZ}$ where $e_k(t) := e^{i2\pi k t} \in L^2(S^1, \IC)$.

\noindent Writing $\mathcal{K} \ni f = \sum\limits_{k \in \IZ} f_k e_k$ we have:
\begin{equation}P_+ f :=  \sum\limits_{k \geq 0} f_k e_k\,, \; \; P_- f := \sum\limits_{k <0} f_k e_k \end{equation}
On $\mathcal{K}$ we consider the class of operators given by multiplication with smooth functions. For $g \in C^{\infty}(S^1,\IC)$ and $f \in L^2(S^1, \IC)$ we write $M_g (f) = g \cdot f$.\\ In Fourier space, multiplication corresponds to convolution.
So, if $g = \sum\limits_{k \in \IZ} g_k\, e_k$ then 
\begin{equation} M_g (e_l) = \sum\limits_{k \in \IZ} g_{k-l} \, e_k \end{equation}
Be careful not to make the easy mistake to confuse the Fourier components of $g$ with those of the multiplication operator $M_g: \mathcal{K} \rightarrow \mathcal{K}$.\\
For $g,h \in C^{\infty}(S^1,\IC)$ we compute:
\begin{align}\label{comp1}
\notag \trace(M_{h-+} M_{g+-}) &= \sum\limits_{l <0} \langle e_l , M_h \sum\limits_{k \geq 0} g_{k-l} e_k \rangle =  \sum\limits_{l <0} \langle e_l ,\sum\limits_{k \geq 0} \sum\limits_{m < 0} h_{m-k}\,g_{k-l}\, e_m \rangle\\ 
& = \sum\limits_{l <0} \sum\limits_{k \geq 0} h_{l-k}\,g_{k-l} =- \sum_{l<0} l \, h_l\,g_{-l}.
\end{align}
Similarly:
\begin{equation}\label{comp2} \trace(M_{g-+} M_{h+-}) = \sum_{l>0} l \, h_l\,g_{-l}.\end{equation}
In particular, we conclude that
\begin{align*} \lVert [\epsilon , M_g] \rVert_2^2 =& 4 \cdot \trace\Bigl( (M_{g+-})(M_{g+-})^* + (M_{g-+})(M_{g-+})^*\Bigr)\\ =\; & 4 \cdot \Bigl(\sum\limits_{l\geq0} l \vert g_l \vert^2 +  \sum\limits_{l<0} (-l) \vert g_l \vert^2 \Bigr)\\
=\; & 4 \cdot \Bigl( \sum_{l \in \IZ} \lvert l \rvert   \lvert g_l \rvert^2 \Bigr) < \infty \; \text{for smooth g}.
\end{align*}
So all operators of this type have off-diagonal components in the Hilbert-Schmidt class, meaning that indeed
\begin{align*} \Bigl\lbrace M_g \;\bigl\lvert\; \lvert g(t) \rvert^2 \equiv 1 \Bigr\rbrace \subset \Ur(\mathcal{K})\;; \hspace{1cm} \Bigl\lbrace M_g\; \bigl\lvert \; g(t) \in \IR \Bigr\rbrace \subset \mathfrak{u}_{\rm res}(\mathcal{K}). \end{align*}

\noindent From \eqref{comp1} and \eqref{comp2} we also read off that the Schwinger cocycle is well-defined on operators of this type with
\begin{equation*}c(M_h, M_g) = \trace(M_{h-+} M_{g+-} - M_{g-+} M_{h+-}) = - \sum\limits_{l \in \IZ} l h_l g_{-l}
\end{equation*}
\begin{equation}\label{cocycleintegralformula} \;\addtolength{\fboxsep}{3pt} \boxed{= \frac{1}{2\pi i} \int\limits_0^1 h(t) \dot{g}(t)\, \mathrm{d}t}  \end{equation}
The rest is easy. Just take any two functions $g,h \in C^{\infty}(S^1, \IR)$ with $\int\limits_0^1 h(t) \dot{g}(t) \, \mathrm{d}t \neq 0$.\\ For instance, consider $g(t) = \cos(2\pi t)$ and $h(t) = \sin(2\pi t)$. Then we got
\item $M_g, M_h \in \mathfrak{u}_{\rm res}(h)$ with $[M_h , M_g] = 0$ but
\begin{equation*} c(M_h , M_g) = \frac{1}{2\pi i} \int\limits_0^1 h(t) \dot{g}(t)\, \mathrm{d}t = (-1)\int\limits_0^1 \sin^2(t)\, \mathrm{d}t \neq 0 \end{equation*} 
This shows that the Schwinger cocycle $c$ is non-trivial and completes our proof.\\
\end{proof}

\noindent In fact, we have proven a much stronger result that doesn't require continuity at all.

\begin{Theorem}[Non-Triviality, algebraic version]
\mbox{}\\
There exist no local homomorphism $\tau: \Ur(\HH) \to \Ures(\HH)$ with $\pi \circ \tau = Id$.\\
Consequently, there exists no homomorphism $\rho\tilde\;:\;  \Ur(\HH) \to \cu(\FF)$ that lifts the projective representation \eqref{rhoGr}, not even locally. The same is true for $\Gl_{\rm res}(\HH)$. 
\end{Theorem}
\begin{proof} Again, we present the proof for the unitary case. The general linear case follows analogously. Recall that for the central extension \begin{equation*}1 \longrightarrow \cu(1) \xrightarrow{\; \imath \;} \widetilde{\cu}^0_{\rm res}(\HH)\xrightarrow{\; \pi \;} \cu^{0}_{\rm res}(\HH) \longrightarrow 1\end{equation*} we have the 2-cocycle \eqref{UresLiegroupcocycle} 
\begin{equation*}\chi_{\cu}(U,V) := \det[U_{++} V_{++} (UV_{++})^{-1}] \slash \lvert \det[U_{++} V_{++} (UV_{++})^{-1}] \rvert \end{equation*} coming from the smooth local section \eqref{sigma}. Suppose there was an open neighborhood $\mathcal{O}$ of the identity in $\Ur(\HH)$ and a homomorphism $\tau: \mathcal{O} \to \Ures(\HH)$ with $\pi \circ \tau = Id$. Then, $\chi$ was trivial as a (local) algebraic group 2-cocycle, i.e. $\tau$ would give rise to a map $\lambda: \mathcal{O} \to \cu(1)$ with \begin{equation*}\lambda(UV) = \chi(U,V) \lambda(U) \lambda(V), \; \forall \, U,V \in \mathcal{O} \end{equation*}
(see the discussion at the beginning of \S. 3.4). Now, we can proof that such a map \textit{cannot} exist anywhere near $1 \in \Ur(\HH)$, by showing that every open neighborhood contains elements $U,V$ with $UV = VU$, but $\chi(U,V) \neq \chi(V,U)$. 

But in fact, this was already shown in the proof of the previous theorem, where we found  $X,Y \in \mathfrak{u}_{\rm res} = \mathrm{Lie}(\Ur(\HH))$ with $[X,Y] = 0$ and $c(X,Y) \neq 0$, where $c: \mathfrak{u}_{\rm res} \to i\IR$ was the Schwinger term, i.e. the Lie algebra cocycle corresponding to $\chi$. Because given such $X,Y \in \mathfrak{u}_{\rm res}$, $\;[X,Y] =0$ implies $[\exp(sX), \exp(tY)] = 0$ in $\Ur(\HH)$ for all $s,t \in \IR$ sufficiently small. And using the identity \eqref{cocycleformula} i.e.
\begin{equation*}c(X,Y) \, = \; \frac{\partial}{\partial t} \frac{\partial}{\partial s} \Bigl\vert_{t=s=0}\, \chi_{\cu}(e^{sX} , e^{tY})\, - \, \frac{\partial}{\partial t} \frac{\partial}{\partial s} \Bigl\vert_{t=s=0}\, \chi_{\cu}(e^{tY} , e^{sX}) \end{equation*}
we see that $\forall \epsilon >0 \; \exists s,t \in (0, \epsilon): \; \chi(e^{sX} , e^{tY}) \neq \chi(e^{tY} , e^{sX})$, since the contrary would imply $c(X,Y)=0$.
Note that we use continuity or differentiability only of the section $\sigma$, which allows us to apply \eqref{cocycleformula}. The final result, however, is purely algebraical.\\  
\end{proof}

\begin{Remark}(Embedding of Loop Groups)\\
In the related mathematical literature, the multiplication operators and the cocycle \eqref{cocycleintegralformula} considered in the proof of the previous theorem arise in the abstract context of embeddings of \emph{loop groups} into $\Gl_{\rm res}(\HH)$. If $K$ is a d-dimensional (compact) Lie group, then $C^{\infty}(S^1 , K)$ is an infinite-dimensional Lie group, called a \emph{loop group} and usually denoted by $LK$ or $Map(S^1 , K)$. Any representation $\rho: K \rightarrow \Gl(\IC^d)$ then provides an action of the loop group on the Hilbert space $\mathcal{K} = L^2(S^1 , \IC)$ via
\begin{equation*}C^{\infty}(S^1 , K) \times L^2(S^1 , \IC) =  LK \times \mathcal{K} \rightarrow \mathcal{K} \end{equation*}
\vspace{-7mm}
\begin{equation*}(\varphi , f) \longmapsto \rho(\varphi(t) ) \cdot f(t) \end{equation*}

\noindent Just as we did above (for d=1), this provides an embedding of the loop group $LK$ into $\Gl_{\rm res}(\mathcal{K}) \cong \Gl_{\rm res}(\HH)$. The central extension $\tilde{\Gl}_{\rm res}$ of $\Gl_{\rm res}$ induces a central of $LK$ and the corresponding Lie algebra cocycle takes a form analogous to \eqref{cocycleintegralformula}. Multiplication in $\IC$ then is just replaced by matrix multiplication and taking the trace in $Mat(d \times d , \IC)$. In our proof we have explicitly avoided any reference to loop groups or cohomology theory as usually found in the mathematical literature, since we feel that they unnecessarily obscure the otherwise simple nature of the proof.
\end{Remark}

%\noindent We are particularly interested in the action of $\Gl_{\rm res}$ on the Fock space. We will briefly discuss it in the language of the geometric construction. Recall that the "Stiefel manifold" $\mathrm{St}$ of admissible bases also has $\IZ$-connected components descending to the different charge-sectors of the fermionic Fock space. 
%After choosing a basis as above we defined $\HH_{\geq -c} := \spn(\lbrace e_k \rvert k \geq -c \rbrace )$.The connected component indexed by $c \in \IZ$ was then 
%\begin{equation*} \mathrm{St}^{(c)} = \lbrace w: \HH_{\geq -c} \rightarrow \HH \, \bigr\rvert\, w\lvert_{\HH_{\geq -c} \rightarrow \HH_{\geq -c} } \; \text{has a determinant} \rbrace \end{equation*}
%This corresponds to the Dirac sea class of the identity map $\mathrm{Id}: \HH_{\geq -c} \rightarrow \HH_{\geq -c}$. 

\chapter{Three Routes to the Fock space}

\section{CAR Algebras and Representations}
In this section we carry out the quantization of the Dirac field in the spirit of  the ``electron-positron picture''. We assume that the reader is somewhat familiar with creation/annihilation operators and the standard construction of the Fock space and review them just briefly. The rather hands-on construction is followed by a brief discussion of abstract $C^*$ - and CAR algebras. Fock spaces will then arise as representation spaces of irreducible representations of the CAR-algebra. This is quite an abstract mathematical machinery, but I think it's rewarding for different reasons:
\begin{enumerate}
\item It is probably the most common and most developed mathematical description of Fock spaces and ``second quantization''.
\item It constitutes, at least rudimentary, a rigorous mathematical formalization of what physicists are usually trying to say.
\item It provides a language in which the mathematical problems can be formulated in a very precise way and reveals an abstract perspective that can be fruitful from \mbox{time to time}.
\end{enumerate}
What do we mean by a ``fruitful perspective''? For instance, as a physicist with some training in (non-relativistic) quantum theory one is often used to think of ``the'' Fock space as a fundamental object (``the space of all physical states''). This makes it hard to grasp some of the issues we're facing in relativistic quantum field theory, for example the fact that the time evolution is ``leaving'' the Fock space (where else would it go?). Thus it can be helpful to think of the CAR-algebra as the more fundamental object and of different Fock spaces corresponding to different representations. The danger, however, is that this easily becomes a vain exercise in abstract mathematics, detached from the physical problems.

\subsection{The Field Operator}
On the one-particle Hilbert-space $\HH = \HH_+ \oplus \HH_-$ we have a \textit{charge conjugation} operator $\mathcal{C}$ mapping negative energy solutions of the free Dirac equation to positive energy solutions with opposite charge. The charge-conjugation is anti-unitary, in particular \textit{anti-linear}. The exact form of the operator depends on the representation of the Dirac algebra. In the so-called ``Majorana representation'', charge conjugation just corresponds to complex conjugation. In Appendix A.2 we give an intuitive, yet general derivation of the charge-conjugation operator.
%\newpage
\begin{Definition}[Fermionic Fock space]
\mbox{}\\
Let $\FF_+ := \HH_+$ and $\FF_- := \mathcal{C}\HH_-$. We define the \emph{fermionic Fock space} as \footnote{By the direct sums we implicitly understand the completion w.r.to the induced scalar product.} 
\begin{equation} \FF := \bigoplus\limits_{n,m = 0}^{\infty} \FF^{(n,m)} ; \;\; \FF^{(n,m)} := \sideset{}{^n}\bigwedge\FF_+ \otimes \sideset{}{^m}\bigwedge \FF_- \end{equation}
We can  split the Fock space into the different \textit{charge sectors}
\begin{equation} \FF := \bigoplus\limits_{c = 0}^{\infty} \FF^{(c)} ; \;\; \FF^{(c)} := \bigoplus\limits_{n - m = c} \FF^{(n,m)}. \end{equation}
The state 
\begin{equation*} \Omega := 1\otimes 1 \in \FF^{(0,0)} = \IC \otimes \IC\end{equation*} 
is called the \emph{vacuum state}.
\end{Definition}
\noindent Recall the definitions of the \textbf{creation and annihilation operators}.\\
On the ``particle-sector'' $\bigwedge \FF_+$:
\begin{align}
&  a(f): \FF^{(n+1,m)} \longrightarrow \FF^{(n,m)},\notag \\
& a(f) f_0\wedge \dots \wedge f_n := \sum\limits_{k=0}^{n}(-1)^k \langle f , f_k \rangle f_0\wedge\dots \wedge \widehat{f_k}\wedge\dots \wedge f_n\\
& a^*(f): \FF^{(n-1,m)} \longrightarrow \FF^{(n,m)},\notag \\
& a^*(f) f_1\wedge \dots \wedge f_{n-1} = f\wedge f_1 \wedge \dots \wedge f_{n-1} 
\end{align}
On the ``anti-particle sector''  $\bigwedge \FF_-$:
\begin{align}
&  b(g): \FF^{(n,m+1)} \longrightarrow \FF^{(n,m)},\notag \\
& b(g)\, \cc g_0\wedge \dots \wedge \cc g_n := (-1)^n \sum\limits_{k=0}^{n}(-1)^k \langle \cc g , \cc g_k \rangle \cc g_0\wedge\dots \wedge \widehat{\cc g_k}\wedge\dots \wedge \cc g_n\\
& b^*(g): \FF^{(n,m-1)} \longrightarrow \FF^{(n,m)},\notag \\
& b^*(g)\, \cc g_1\wedge \dots \wedge \cc g_{n-1} = \cc g\wedge \cc g_1 \wedge \dots \wedge \cc g_{n-1} 
\end{align}
The reader is probably familiar with the fact that $a$ and $a^*$, as well as $b$ and $b^*$, are formal adjoints of each other and satisfy the \textit{canonical anti-commutation relations}
\begin{equation}\begin{split} &\lbrace a(f_1) , a^*(f_2) \rbrace = a(f_1)a^*(f_2) + a^*(f_2)a^(f_1) = \langle f_1 , f_2 \rangle_{\HH} \cdot \mathds{1}, \;\;\;\forall f_1,f_2 \in \HH_+\\
 &\lbrace b(g_1) , b^*(g_2) \rbrace =  \langle \mathcal{C}g_1 ,\mathcal{C} g_2 \rangle_{\HH} \cdot \mathds{1} = \overline{\langle g_1 , g_2 \rangle} \cdot \mathds{1} = \langle g_2 , g_1 \rangle_{\HH} \cdot \mathds{1}, \;\forall g_1,g_2 \in \HH_-
 \end{split}\end{equation}
 and all other possible combination \textit{anti-commute}.\\

\noindent If $(f_j)_{j\in \IN}$ and $(g_k)_{k\in \IN}$ are ONB's of $\HH_+$ and $\HH_-$ respectively, the elements of the form
\begin{equation}\label{FockONB} a^{*}(f_{j_1})a^{*}(f_{j_2})\dots a^{*}(f_{j_n})b^{*}(g_{k_1})b^{*}(g_{k_2})\dots b^{*}(g_{k_m}) \Omega \in \FF^{(n,m)} \subset \FF \end{equation}
for $j_1 < \dots < j_n; k_1 < \dots < k_m$ and $n,m = 0,1,2,3, \dots$
form an ONB of the Fock space $\FF$.\\

\noindent So, the creation operators acting on the vacuum generate the dense subspace
\begin{equation*} \mathcal{D} := \Bigl\lbrace \text{finite linear combinations of vectors of the form}\; \eqref{FockONB} \Bigr\rbrace \end{equation*}
which is the convenient domain for the second quantization of bounded operators.\\

\begin{Definition}[Field Operator]\label{Def:fieldoperator}
\item For any $f \in \HH$ we define the \emph{field operator} $\Psi(f)$ on $\FF$ by
\begin{equation}\begin{split}\label{fieldoperator}
\Psi(f) := a(P_+ f) + b^*(P_- f)\\
\Psi^*(f) = a^*(P_+f) + b(P_- f)
\end{split}\end{equation}
\item We can view the field operator as an \emph{anti-linear} map $\Psi: \HH \longrightarrow \mathcal{B(F)}$. 
\item It satisfies the \textit{canonical anti-commutation relations} (CAR)
\begin{equation}\label{CAR}\fingbox{\begin{split}
& \bigl\lbrace \Psi(f) , \Psi^*(g) \bigr\rbrace = \langle f , g \rangle \cdot \mathds{1} \\
& \bigl\lbrace \Psi(f) , \Psi(g) \bigr\rbrace = \bigl\lbrace \Psi(f)^* , \Psi^*(g) \bigr\rbrace = 0
\end{split}}\end{equation} 
\end{Definition}

\subsubsection{Second Quantization of Unitary Operators}

\noindent Let $U: \HH \rightarrow \HH$ be a unitary transformation on the one-particle Hilbert space.
We want to lift it to a unitary transformation $\Gamma(U)$ on the Fock space $\FF$. 
In non-relativistic quantum mechanics, we would lift an operator $U$ to the n-particle anti-symmetric Hilbert-space $\sideset{}{^n}\bigwedge \HH$ ``product-wise'', i.e. to
\begin{equation}U \otimes U \otimes \dots \otimes U. \end{equation}
But the naive generalization 
\begin{equation*} f_1\wedge \dots \wedge f_n \otimes \cc g_1 \wedge \dots \wedge \cc g_m \, \longmapsto \, Uf_1\wedge \dots \wedge Uf_n \otimes \cc U g_1 \wedge \dots \wedge \cc U g_m \end{equation*}
make sense only if $U$ preserves the splitting $\HH = \HH_+ \oplus \HH_-$ i.e. only if  $U_{+-}=U_{-+} = 0$.
In general this is not the case and $U$ will mix positive and negative energy states. Physically, this leads to the phenomenon we call pair creation. Mathematically, this leads to trouble.\\

\noindent While it is difficult to say how the unitary transformation is supposed to act on the Fock space, there is a very natural action on the field operator $\Psi$, given by 
\begin{equation}\label{Bogo}\begin{split} & \Psi \longmapsto \beta_u(\Psi) := \Psi \circ U,\\
i.e.\; \;  \beta_u(\Psi)(f) = & \Psi(Uf) = a(P_+ Uf) + b^*(P_- Uf), \;  \forall f \in \HH .\end{split}\end{equation}
$\beta_U$ is called a \textit{Bogoljubov transformation}. It is easy to see that $\tilde{\Psi}:= \beta_U(\Psi)$ is still an antilinear map $\HH \rightarrow \mathcal{B(F)}$ satisfying the CAR \eqref{CAR}. This means that $\tilde\Psi$ is also a field operator, inducing another \textit{representation of the CAR-algebra} with the new annihilation operators \textit{defined} as
\begin{equation}\label{newannihilation}\begin{split}
c(f) := \tilde{\Psi}(f), \; \text{for}\; f \in \HH_+\\
d(f) := \tilde{\Psi}^*(g), \; \text{for}\; g \in \HH_-.
\end{split}\end{equation}
%\vspace*{2mm}
\begin{Definition}[Implementability of Unitary Transformations]\label{Def:Implementability}
\mbox{}\\ 
The unitary transformation $U \in \cu(\HH)$ is \emph{implementable} on the Fock space $\FF$ if there exists a unitary map $\Gamma(U): \FF \rightarrow \FF$ with
\begin{equation}\label{eq:implementation} \Gamma(U)\Psi(f) \Gamma(U)^* = \beta_U(\Psi)(f) = \Psi(Uf), \; \forall f \in \HH. \end{equation}
If $U$ is implementable, the implementation is unique up to a phase.
\end{Definition}

\noindent That the implementation can be determined  only up to a phase is obvious, because if $\Gamma(U)$ is an implementation of $U \in \cu(\HH)$, so is $e^{i\varphi}\,\Gamma(U)$ for any $e^{i\varphi} \in \cu(1)$.\\
%\newpage
\noindent Supposed that $\Gamma(U)$ is an implementation of $U$, we note that by \eqref{eq:implementation} it acts on a basis vector of the form 
\begin{align*}a^{*}(f_{1})a^{*}(f_{2})\dots a^{*}(f_{n})b^{*}(g_{1})b^{*}(g_{2})\dots b^{*}(g_{m})\, \Omega\\
= \;  \Psi^*(f_1)\Psi^*(f_2)\dots\Psi^*(f_n)\Psi(g_1)\Psi(g_2)\dots\Psi(g_m) \,\Omega \end{align*} 
in the following way
\begin{align*}
& \Gamma(U) \, \Bigl( \Psi^*(f_1)\Psi^*(f_2)\dots\Psi^*(f_n)\Psi(g_1)\dots\Psi(g_m) \,\Omega \Bigr)\\
=\, & \Gamma(U) \Psi^*(f_1)\Gamma(U)^*\Gamma(U)\Psi^*(f_2)\Gamma(U)^* \dots \Gamma(U)\Psi(g_{m-1})\Gamma(U)^*\Gamma(U)\Psi(g_m)\Gamma(U)^* \Gamma(U)\, \Omega\\
=\, & \Psi^*(U f_1)\Psi^*(U f_2)\dots\Psi^*(U f_n)\Psi(U g_1)\dots\Psi(U g_m)\;\bigl[\Gamma(U)\, \Omega \bigr] \\
=\, &c^{*}(f_{1})c^{*}(f_{2})\dots c^{*}(f_{n})d^{*}(g_{1})d^{*}(g_{2})\dots d^{*}(g_{m}) \, \widetilde{\Omega}
\end{align*}
with $\widetilde{\Omega} = \Gamma(U) \Omega$. This simple consideration yields an important result:

\begin{Proposition}[New Vacuum]
\mbox{}\\
A unitary transformation $U \in \cu(\HH)$ is implementable on $\FF$, if and only if there exists a vacuum for the new annihilation operators \eqref{newannihilation} defined by $\beta_U(\Psi)$ in the same Fock space.
That is, if there exists a normalized state $\widetilde{\Omega}\in \FF$ with
\begin{align}
c(f) \widetilde{\Omega} = \beta_U(\Psi)(f)\, \widetilde{\Omega} = 0, \;\; \forall f \in \HH_+\\
d(g) \widetilde{\Omega} = \beta_U(\Psi)^*(g)\, \widetilde{\Omega} = 0, \;\; \forall g \in \HH_-.
\end{align}
\end{Proposition}
\begin{proof} If $\Gamma(U)$ is an implementer of $U$ set $\widetilde{\Omega} = \Gamma(U) \Omega$. Then, 
\begin{align}
\beta_U(\Psi)(f) \widetilde{\Omega} = \Gamma(U) \Psi(f) \Gamma(U)^*\Gamma(U) \Omega = \Gamma(U) \Psi(f) \Omega = 0, \; \; \forall f \in \HH_+\\
\beta_U(\Psi)^*(g) \widetilde{\Omega} =\Gamma(U) \Psi(g) \Gamma(U)^*\Gamma(U) \Omega = \Gamma(U) \Psi(g) \Omega = 0, \;\; \forall g \in \HH_-
\end{align}
hence $\widetilde{\Omega}$ is a vacuum for $\beta_U(\Psi)$. 
Conversely, if there exists a vacuum $\widetilde{\Omega}$ for $\beta_U(\Psi)$, we can define a unitary transformation on $\FF$ by
\begin{align*}&a^{*}(f_{1})a^{*}(f_{2})\dots a^{*}(f_{n})b^{*}(g_{1})b^{*}(g_{2})\dots b^{*}(g_{m})\, \Omega\\
\longmapsto 
\,e^{i\phi} & c^{*}(f_{1})c^{*}(f_{2})\dots c^{*}(f_{n})d^{*}(g_{1})d^{*}(g_{2})\dots d^{*}(g_{m}) \, \widetilde{\Omega}
\end{align*}
with an arbitrary phase $e^{i\phi}$ and the computation above shows that this is an implementation of $U$ on $\FF$.\\  
\end{proof}

\noindent In the simplest case, when $U$ leaves the positive and negative energy subspace invariant, we note that
$c(f) = a(Uf), \; \forall f \in \HH_+$ and $d(g) = b(Ug),\; \forall g \in \HH_-$, so that $\widetilde{\Omega} = e^{i\varphi} \Omega$. Of course in this case, the phase is conveniently chosen to be 1. Consequently, the implementation is nothing else than
\begin{equation}\begin{split} & f_1\wedge \dots \wedge f_n \otimes \cc g_1 \wedge \dots \wedge \cc g_m = a^{*}(f_{1})\dots a^{*}(f_{n})b^{*}(g_{1})\dots b^{*}(g_{m})\, \Omega\\ 
\stackrel{\Gamma(U)}\longmapsto \; & c^{*}(f_{1})\dots c^{*}(f_{n})d^{*}(g_{1})\dots d^{*}(g_{m}) \, \Omega = Uf_1\wedge \dots \wedge Uf_n \otimes \cc U g_1 \wedge \dots \wedge \cc U g_m. \end{split}\end{equation}
We see that the Bogoljubov transformation of the field operator and its implementation (if it exists) are indeed the natural generalization of the product-wise lift of $U$ to the Fock-space.\\
%\newpage
\noindent Almost sneakily, we have established a duality in our thinking about the implementation problem that is quite exciting. Either we can formulate the problem as one about lifting a unitary transformation $U$ on $\HH$ to the Fock space $\FF$, as was our initial motivation. Or we can ask the equivalent question, whether the two representations of the CAR-algebra induced by $\Psi$ and $\beta_U(\Psi)$, respectively, are equivalent. So far, we have but hinted at this fact. A more detailed discussion will follow in the next section.\\

\noindent We can finally state our first version of the Shale-Stinespring theorem.

\begin{Theorem}[Shale-Stinespring, explicit version]\footnote{see \cite{Tha}, Thms. 10.6 \& 10.7}
 \mbox{}\\
A unitary operator $U \in \cu(\HH)$ is unitarily implementable on $\FF$ if and only if the operators $U_{+-}$ and $U_{-+}$ are Hilbert-Schmidt. In this case, an implementation $\Gamma(U)$ acts on the vacuum as

\begin{equation}\label{transformedvacuum}\addtolength{\fboxsep}{5pt} \boxed{\Gamma(U)\, \Omega := N e^{i\phi} \prod\limits_{l=1}^{L} a^*(f_l)\prod\limits_{m=1}^{M} b^*(g_m)\, e^{Aa^*b^*} \, \Omega }\end{equation}
where
\begin{itemize}
\item $e^{\phi} \in \cu(1)$ is an arbitrary phase 
\item $N = \sqrt{\det(1-U_{+-}U^*_{-+})} = \sqrt{\det(1-U_{-+}U^*_{+-})}$ is a normalization constant\footnote{Note that, in contrast to \cite{Tha}, we use the notation $U^*_{-+} = (U^*)_{-+}$ and not $(U_{-+})^* = (U^*)_{+-}$} 
\item $A:= (U_{+-})(U_{--})^{-1}$ is Hilbert-Schmidt
\item $\lbrace f_1, \ldots , f_L \rbrace$ is an ONB of $\ker U^*_{++}$ and
$\lbrace g_1, \ldots, g_M \rbrace$ an ONB of $\ker U^*_{--}$.
\end{itemize}
\end{Theorem}
\noindent Recall that $U_{++}$ and $U_{--}$ are Fredholm-operators, therefore $\ker U^*_{++}$ and $\ker U^*_{--}$ are indeed finite-dimensional subspaces. Also, $(U_{--})^{-1}$ is well-defined and bounded on $\im U_{--} = (\ker U^*_{--})^{\perp}$ and we extend it to 0 on $\ker U^*_{--}$. In this way, the operator $A:= (U_{+-})(U_{--})^{-1}$ is well-defined.\\  

\noindent In the explicit form of the transformed vacuum our result \eqref{chargetransform} about the net-charge "created" by $U$ becomes manifest. From \eqref{transformedvacuum} we can directly read of that $\Gamma(U) \Omega$ contains $M = \dim \ker(U^*_{--})$ particles of negative charge and $L = \dim \ker (U^*_{++})$ particles of positive charge. Hence the net charge is 
%\begin{align*} M - L & = \dim \ker (U^*_{--}) - \dim \ker(U^*_{++})\\
%& = \dim \ker (U_{++}) - \dim \ker(U^*_{++})\\
%& = \ind(U_{++}) \end{align*}
\begin{equation}\label{L-M}\begin{split} L -M  & = \dim \ker(U^*_{++}) - \dim \ker (U^*_{--}) \\
& = \dim \ker (U_{--}) - \dim \ker(U^*_{--})\\
& = \ind(U_{--}) =  - \ind(U_{++}).\footnotemark\end{split}\end{equation}
For the last equality we need a Lemma, which is proven in Appendix A.3.

\noindent The different sign in comparison to \eqref{chargetransform} comes from the unfortunate fact that the roles of $\HH_+$ and $\HH_-$ are interchanged in the two chapters.

\subsubsection{Second Quantization of Hermitian Operators}

\noindent Now we're looking at the second quantization of an arbitrary bounded operator $A: \HH \rightarrow \HH$, which is usually denoted by $\mathrm{d}\Gamma(A)$. Self-adjoint operators, as generators of the unitary group, are of particular interest. To $\sideset{}{^n}\bigwedge \HH$, we would lift an operator $A$ additively as: 
\begin{equation}\label{additivelift}A \otimes \mathds{1}\otimes \dots \otimes \mathds{1}\, + \,\mathds{1}\otimes A \otimes \mathds{1} \dots \mathds{1} \, +\, \dots +\mathds{1} \otimes \dots \otimes \mathds{1}\otimes A \end{equation}  
But again, the immediate generalization to $\FF = \bigwedge \mathcal{H}_+ \otimes  \bigwedge \mathcal{C (H_-)}$ makes sense only  if $A$ respects the polarization.
So again, we should look at the field operator to formulate a generalization of the additive lifting prescription.
We will demand the following conditions for the ``second quantization'': 
\begin{equation}\label{dgamma}\begin{split} & [\mathrm{d}\Gamma(A) , \Psi^*(f)] = \Psi^*(Af)\\ 
& [\mathrm{d}\Gamma(A) , \Psi(f)] = - \Psi(A^*f)
\end{split}\end{equation}
The minus-sign in the second line might look ad-hoc, but its relevance will become apparent soon.
Note that  if $\mathrm{d}\Gamma(A)$ satisfies \eqref{dgamma} so does $\mathrm{d}\Gamma(A) + c\,\mathds{1}_\FF$ for every constant $c \in \IC$.\\ 
Therefore, the implementation of a Hermitian operator, if it exists, is well defined only up to a constant. 
An operator $\mathrm{d}\Gamma(A)$ satisfying \eqref{dgamma} then acts on a basis state \eqref{FockONB} as
\begin{align*}
&\mathrm{d}\Gamma(U) \, \Bigl( \Psi^*(f_1)\Psi^*(f_2)\dots\Psi^*(f_n)\Psi(g_1)\dots\Psi(g_m) \Omega \Bigr)\\
= & \sum\limits_{j=1}^{n}  \Psi^*(f_1)\dots \Psi^*(f_{j-1})\Psi^*(Af_j)\Psi^*(f_{j+1})\dots\Psi^*(f_n)\,\Psi(g_1)\dots\Psi(g_m)\,\Omega\\
- &\sum\limits_{k=1}^{m} \; \Psi^*(f_1)\Psi^*(f_2)\dots\Psi^*(f_n)\,\Psi(g_1)\dots\Psi(g_{k-1})\Psi(A^*g_k) \Psi(g_{k+1})\dots\Psi(g_m)\,\Omega\\
+ & \; \Psi^*(f_1)\Psi^*(f_2)\ldots\Psi^*(f_n)\,\Psi(g_1)\Psi(g_2)\ldots\Psi(g_m) \bigl( \mathrm{d}\Gamma(A)\Omega \bigr)
%& \Psi^*(Af_1)\Psi^*(f_2)\dots\Psi^*(f_n)\Psi(g_1)\dots\Psi(g_m)\Omega  +  \Psi^*(f_1)\mathrm{d}\Gamma(A)\Psi^*(f_2)\dots\Psi^*(f_n)\Psi(g_1)\dots\Psi(g_m)\Omega\\
%= & \dots = \\
%= & \; \Psi^*(Af_1)\Psi^*(f_2)\dots\Psi^*(f_n)\Psi(g_1)\dots\Psi(g_m)\Omega  +  \Psi^*(f_1)\Psi^*(Af_2)\dots\Psi^*(f_n)\Psi(g_1)\ldots\Psi(g_m)\Omega + \ldots \\
% + & \; \Psi^*(f_1)\Psi^*(f_2)\dots\Psi^*(Af_n)\Psi(g_1)\ldots\Psi(g_m)\Omega -   \Psi^*(f_1)\Psi^*(f_2)\ldots\Psi^*(f_n)\Psi(A^*g_1)\dots\Psi(g_m)\Omega - \ldots \\
%- & \; \Psi^*(f_1)\Psi^*(f_2)\dots\Psi^*(f_n)\Psi(g_1)\dots\Psi(A^*g_m)\Omega\\
%+ & \; \Psi^*(f_1)\Psi^*(f_2)\ldots\Psi^*(f_n)\Psi(g_1)\ldots\Psi(g_m) \Bigl( \mathrm{d}\Gamma(A)\Omega \Bigr)
\end{align*}

\noindent We see that $\mathrm{d}\Gamma(A)$ is well-defined on the dense domain $\mathcal{D}$ if and only if the last summand $\mathrm{d}\Gamma(A) \Omega$ is well-defined. In this case, $\mathrm{d}\Gamma(A)$ can be regarded as a generalization of the additive-lift \eqref{additivelift}. But how do we find such an operator $\mathrm{d}\Gamma(A)$ on $\FF$, satisfying \ref{dgamma}?

\noindent Let $(e_k)_{k \in \mathbb{Z}}$ be a basis of $\mathcal{H}$, s.t. $(e_k)_{k \leq 0}$ is a basis of $\mathcal{H}_-$ and $(e_k)_{k > 0}$ a basis of $\mathcal{H}_+$.\\ 
A first educated guess for the second quantization of a bounded operator $A$ would be \\ 
$\mathrm{d}\Gamma'(A) := A\Psi^*\Psi$, where $A\Psi^*\Psi$ is short-hand notation for\footnote{Consult e.g. \cite{Tha} \S 10 for more details.}
\begin{equation*}\sum\limits_{i,j \in \IZ} \langle e_i , A e_j \rangle \Psi^*(e_j)\Psi(e_i).\end{equation*}
This expression is easily shown to be independent of the choice of basis. We can decompose it into creation and annihilation operators and write $A\Psi^*\Psi = Aa^*a + Aa^*b^* + Aba + Abb^*$. The vacuum expectation value of this operator can then be computed as 
\begin{align*}  \langle \Omega, \mathrm{d}\Gamma'(A) \Omega \rangle & =  \langle \Omega, A\Psi^*\Psi \Omega \rangle = \langle \Omega, Abb^* \Omega\rangle\\
& = \sum\limits_{i,j <0} \langle e_i , A e_j \rangle \langle \Omega, b(e_i)b^*(e_j) \Omega \rangle\\
& = \sum\limits_{i,j <0} \langle e_i , A e_j \rangle\, \delta_{ij} = \trace(A_{--}).
\end{align*}
So unless $A_{--}$ is trace-class, not even the vacuum state is in the domain of $A\Psi^*\Psi$ and thus, 
A more promising attempt is therefore
\begin{equation}\label{normalordering}\mathrm{d}\Gamma(A) =\, :A\Psi^*\Psi: \,= Aa^*a + Aa^*b^* + Aba - Ab^*b 
\end{equation}
\noindent which also satisfies \eqref{dgamma}. Physicists call this \textit{normal ordering} and denote it by two colons embracing expressions like $A\Psi^*\Psi$.
With this prescription, $:A\Psi^*\Psi:\Omega = Aa^*b^* \Omega$ and therefore
\begin{equation}\langle \Omega, \mathrm{d}\Gamma(A)\, \Omega \rangle =  \langle \Omega, :A\Psi^*\Psi: \Omega \rangle = 0.\end{equation}
This implies immediately \begin{equation}\label{dgamma-trace}:A\Psi^*\Psi: \;= A\Psi^*\Psi - \trace(A_{--})\cdot \mathds{1} \end{equation} whenever $\lvert\trace(A_{--})\rvert < \infty$,  so that in this case, $A\Psi^*\Psi$ and $:A\Psi^*\Psi:$ are both second quantizations of $A$ differing by a complex constant (or a real constant, if $A$ is self-adjoint).

\noindent Assuming that $A_{+-}: \HH_- \rightarrow \HH_+$ is compact, there exists an ONB $(u_j)_j$ of $\HH_+$ and $(v_j)_j$ of $\HH_-$ and singular values $\lambda_j \geq 0$  such that $A_{+-} = \sum\limits_{j=0}^{\infty} \lambda_j \langle v_j , \cdot \rangle u_j$.\footnote{$\lbrace\lambda_j^2\rbrace$ are the eigenvalues of the positive operator $(A_{+-})^*(A_{+-})$.}\\ 
Hence, $Aa^*b^* = \sum\limits_{j}\lambda_j a^*(u_j)b^*(v_j)$ and we compute
\begin{align*}\lVert Aa^*b^* \Omega \rVert ^2 = & \sum\limits_{k,j} \lambda_k\lambda_j \langle \Omega, b(v_k)a(u_k)a^*(u_j)b^*(v_j)\Omega \rangle\\
= & \sum\limits_{k,j} \lambda_k\lambda_j \langle \Omega, b(v_k)\bigl(\lbrace a(u_k), a^*(u_j) \rbrace - a^*(u_j)a(u_k) \bigr) b^*(v_j)\Omega \rangle\\
= & \sum\limits_{k,j} \lambda_k\lambda_j \Bigl(\delta_{jk}\langle \Omega, \bigl(\lbrace b(v_k), b^*(v_j) \rbrace - b^*(v_j)b(v_k)\bigr) \Omega \rangle - \langle \Omega, b(v_k)a^*(u_j)b^*(v_j)a(u_k)  \Omega \rangle \Bigr)\\
= & \sum\limits_{k,j} \lambda_k\lambda_j \delta_{jk} = \sum\limits_k \lambda_k^2 = \lVert A_{+-} \rVert_2^2
\end{align*}
This shows that $:A\Psi^*\Psi:$ is well-defined on the dense domain $\mathcal{D}$ if and only if $A_{+-}$ is of Hilbert-Schmidt type. For a Hermitian operator, the same must be true for the adjoint so we require $A_{-+}$ to be Hilbert-Schmidt as well.

 \noindent How does this relate to the discussion in the previous section? The relationship between unitary and self-adjont operators is claryfied by \textit{Stone's Theorem}, which says that every strongly continuous, one-parameter group of unitary operators is generated by a unique self-adjoint operator. We are in particular interested in the free time evolution $U_t = e^{-itD_0}$.\\
Let $\Psi_t := \beta_{U_t}(\Psi)$ be the time-evolving field operator. Then, for any fixed $f$ in the domain of $D_0$, $\Psi_t(f)$ is even norm-differentiable in t because 
\begin{align*} & \lim\limits_{h \rightarrow 0}\;  \bigl\lVert \frac{1}{h}[ \Psi_{t+h}(f) - \Psi_t(f)] - \Psi_t(-i D_0 \,f)\bigr \rVert =\lim\limits_{h \rightarrow 0}\; \bigl\lVert  \Psi \bigl(e^{iD_0t}[\frac{1}{h}(e^{iD_0h} - \mathds{1}) + iD_0] f\bigr) \bigr\rVert\\
 = & \lim\limits_{h \rightarrow 0}\; \bigl\lVert \bigl[\frac{1}{h}(e^{iD_0h} - \mathds{1}) + iD_0\bigr] f \bigr\rVert = 0 \end{align*} and thus 
\begin{equation} \frac{d}{dt}\, \Psi_t(f) = \Psi_t(-i D_0\,f) =  i \Psi_t(D_0\,f). \end{equation} 
The same is true for $\Psi^*_t$, but note that $\Psi_t^*$ is now linear and not anti-linear in f.\\
On a basis state of the form \eqref{FockONB} we therefore compute
\begin{align*} & i \frac{d}{dt}\Bigl\lvert_{t=0}\; \Bigl[ \Gamma(U_t)  \Bigl( \Psi^*(f_1)\Psi^*(f_2)\dots\Psi^*(f_n)\Psi(g_1)\dots\Psi(g_m)\, \Omega \Bigr)\Bigr]\\
= \; & i \frac{d}{dt}\Bigl\lvert_{t=0}\;  \Bigl[ \Psi_t^*(f_1)\Psi_t^*(f_2)\dots\Psi_t^*(f_n)\Psi_t(g_1)\dots\Psi_t(g_m) \,\Omega \Bigr]\\
= \; & \Psi^*(D_0 f_1)\Psi^*(f_2)\dots \Psi^*(f_n)\Psi(g_1)\dots\Psi(g_m) \Omega
+ \Psi^*(f_1)\Psi^*(D_0f_2)\dots\Psi^*(f_n)\Psi(g_1)\dots\Psi(g_m) \Omega + \cdots\\
+\;  & \Psi^*(f_1)\Psi^*(f_2)\dots\Psi^*(D_0 f_n)\Psi(g_1)\dots\Psi(g_m) \Omega
-  \Psi^*(f_1)\Psi^*(f_2)\dots\Psi^*(f_n)\Psi(D_0 g_1)\dots\Psi(g_m) \Omega - \cdots\\
-\;  & \Psi^*(f_1)\Psi^*(f_2)\dots\Psi^*(f_n)\Psi(g_1)\dots\Psi(D_0 g_m) \Omega\\
=\;& \mathrm{d}\Gamma(D_0)  \Bigl( \Psi^*(f_1)\Psi^*(f_2)\dots\Psi^*(f_n)\Psi(g_1)\dots\Psi(g_m)\, \Omega \Bigr)
\end{align*}  
This means:
\begin{equation} i \frac{d}{dt} \Gamma(e^{-iD_0t})\,=\, \mathrm{d}\Gamma(D_0)\,=\; :D_0\Psi^*\Psi: \end{equation}
on the dense domain where both sides are well-defined.
The operator $\mathrm{d}\Gamma(D_0)$ is called the \textbf{second quantization of the free Dirac Hamiltonian.}\\

\noindent If $(f_k)_{k \in \mathbb{N}}$ is a basis of $\HH_+$ and  $(g_k)_{k \in \mathbb{N}}$ a basis of $\HH_-$ then
\begin{equation}\label{D0semibdd} \fingbox{\begin{split}\mathrm{d}\Gamma(D_0)\,&=\; :D_0\Psi^*\Psi:\\
 = & \sum\limits_{i=1}^{\infty} \sum\limits_{j=1}^{\infty} \Bigl( \langle f_i, D_0 f_j \rangle a^*(f_i)a(f_j) - \langle g_j, D_0 g_i \rangle b^*(g_i)b(g_j) \Bigl)\end{split}}\end{equation}
which is obviously semi-bounded from below because $-D_0$ is positive definite on $\HH_-$.\\
Note that the anti-linearity of the charge-conjugation and thus of the Fock-space sectors $\bigwedge \FF_-$ makes all the difference. It is also the origin of the minus-sign in \eqref{dgamma}. Similarly, the global gauge transformations $e^{-it\mathds{1}}$ are generated by the \textit{charge-operator}
\begin{equation} Q := \mathrm{d}\Gamma(\mathds{1}) = \sum\limits_{i=1}^{\infty} \sum\limits_{j=1}^{\infty} \Bigl( a^*(f_i)a(f_j) -  b^*(g_i)b(g_j) \Bigl) \end{equation}
which is just the difference of the number operators on $\bigwedge\FF_+$ and $\bigwedge\FF_-$.\\
Of course, the generators of the free time evolution and of the global gauge transformations are easy to compute because both operators are diagonal w.r.to the polarization\\
$\HH = \HH_+ \oplus \HH_-$. In particular, we can set $\Gamma(U_t) \Omega = \Omega, \forall t \in \IR$.
However, the suggested relationship between $\Gamma$ and $\mathrm{d}\Gamma$ is found to be true in the general setting:

\begin{Theorem}[Generators of Unitary Groups]
\mbox{}\\
Let $A$ be a self-adjoint operator on $\HH$ with $[\epsilon, A]$ Hilbert-Schmidt.  Then $\mathrm{d}\Gamma(A) =\, :A \Psi^*\Psi:$ as defined in \eqref{normalordering} is essentially self-adjoint on $\FF$ and for its (unique) self-adjoint extension on the Fock space $\FF$ it is true that $e^{i\,:A \Psi^*\Psi:}$ is an implementation of $e^{iA}$.\\
In particular, we can choose phases for the implementations such that 
\begin{equation}\label{expdgamma} \Gamma(e^{iA}) = e^{i \,\mathrm{d}\Gamma(A)} = e^{i\, :A \Psi^*\Psi:} \end{equation}
\end{Theorem}
\begin{proof}\cite{CaRu} Proposition 2.1 ff.. This seems to be the first complete proof of this theorem.\\
\end{proof}
\vspace*{-4mm}\noindent Note that \eqref{expdgamma} fixes the phase of the implementation $\Gamma(e^{iA})$ only in a neighborhood of the identity where the exponential map is 1:1 from the Lie algebra on to the Lie group.\\

In physics, \textit{normal ordering} is often introduced in an ad-hoc and merely formal way, simply by ``writing all the annihilators on the right of the creation operators'', so that it gives zero when acting on the vacuum. 
This is why quite often there seems to be something mysterious or even dishonest about it. But in fact, the last theorem tells us that the second quantization in normal ordering merely defines an action of the Lie algebra $\mathfrak{u}_{\rm res}$ of $\Ur(\HH)$ on $\FF$. No big mystery there.\\

For A in the trace-class, $A\Psi^*\Psi$ and $:A\Psi^*\Psi:$ differ only by a constant multiple of the identity. In fact, they just correspond to different ``lifts'' of A to the central extension $\tilde{\mathfrak{u}}_{\rm res} \cong \mathfrak{u}_{\rm res} \oplus \IR$ that is acting on the Fock space. 
From the discussion in Chapter 3, we also know that as $\Gamma$ defines only a \textit{projective} representation of $\Ur(\HH)$, we are going to end up with a \textit{commutator anomaly} in the representation of the corresponding Lie algebra in form of a Lie algebra 2-cocycle. It should come as no surprise that we encounter the Schwinger-cocycle \eqref{Schwingercocycle} again.

\begin{Proposition}(Schwinger Cocycle)
\mbox{}\\ For $A, B \in \mathfrak{u}_{\rm res}$ we find
\begin{equation}\label{Schwingerinrepr}[\mathrm{d}\Gamma(A) , \mathrm{d}\Gamma(B)] = \mathrm{d}\Gamma([A,B]) + c(A,B)\mathds{1} \end{equation}
with the Schwinger cocycle $c(A,B) = \trace(A_{-+}B_{+-} - B_{-+}A_{+-})$ computed in \eqref{Schwingercocycle}.
\end{Proposition}
\begin{proof} Because $[\mathrm{d}\Gamma(A) , \mathrm{d}\Gamma(B)]$ and $\mathrm{d}\Gamma([A,B])$ are both second quantizations of $[A,B]$, they differ only by a constant multiple $c_{AB}$ of the identity
\footnote{In the language of the previous chapters: $[\mathrm{d}\Gamma(A) , \mathrm{d}\Gamma(B)]$ and $\mathrm{d}\Gamma([A,B])$ are just different lifts of $[A,B] \in \mathfrak{u}_{\rm res}$ to $\tilde{\mathfrak{u}}_{\rm res} \cong \mathfrak{u}_{\rm res} \oplus \IC$.}, i.e.
\begin{equation*}[\mathrm{d}\Gamma(A) , \mathrm{d}\Gamma(B)] = \mathrm{d}\Gamma([A,B]) + c_{AB}\cdot\mathds{1} \end{equation*}
Now $\mathrm{d}\Gamma([A,B])$ is defined precisely in such a way that the vacuum expectation value $\langle \Omega , \mathrm{d}\Gamma([A,B]) \Omega \rangle$ vanishes. Therefore,
\begin{equation*} c_{AB} = \langle \Omega ,[\mathrm{d}\Gamma(A) , \mathrm{d}\Gamma(B)] \Omega \rangle \end{equation*}
We take bases  $(f_j)_{j\in\mathbb{N}}$ of $\HH_+$ and $(g_k)_{k \in \mathbb{N}}$ of $\HH_-$ and compute:
\begin{align*}& \langle \Omega ,\mathrm{d}\Gamma(A)\mathrm{d}\Gamma(B)\, \Omega \rangle = \langle \Omega, Aba\, Ba^*b^*\, \Omega \rangle\\
= &\sum\limits_{j,k}\sum\limits_{l,m} (g_k, A f_j) (f_l, B g_m) \langle\Omega, b(g_k)a_(f_j)a^*(f_l)b^*(g_m) \Omega \rangle\\
= &\sum\limits_{j,k} (g_k, A f_j) (f_j, B g_k) \langle\Omega, b(g_k)a(f_j)a^*(f_j)b^*(g_k) \Omega \rangle\\
= & \sum\limits_{j,k} (g_k, A f_j) (f_j, B g_k) = \trace(A_{-+}B_{+-})
\end{align*}
Similarly, $\langle \Omega ,\mathrm{d}\Gamma(B)\mathrm{d}\Gamma(A) \Omega \rangle = \trace(B_{-+}A_{+-})$. And therefore 
\begin{equation*} c_{AB} = c(A,B) =  \trace(A_{-+}B_{+-} - B_{-+}A_{+-}). \end{equation*}
 \end{proof}
\vspace*{-1mm}
\noindent Note that on $\mathfrak{u}_{\rm res} \cap I_1(\HH)$, the proof simplifies considerably. Because then, we can use\\
 $:A\Psi^*\Psi:\, = A\Psi^*\Psi - \trace(A_{--})\mathds{1}$ and it is easy to check that $[A \Psi^*\Psi, B \Psi^*\Psi] =[A,B]\Psi^*\Psi$.\\ 
 Hence we find
\begin{equation*} [:A\Psi^*\Psi:\,,\, :B\Psi^*\Psi:] = [A,B]\Psi^*\Psi =\; :[A,B]\Psi^*\Psi: + \trace ([A,B]_{--})\cdot\mathds{1} \end{equation*}
 and \begin{align*} \trace ([A.B]_{--}) &= \trace \bigl( A_{-+}B_{+-} + A_{--}B_{--} - B_{-+}A_{+-} - B_{--}A_{--} \bigr)\\
 & = \trace \bigl( A_{-+}B_{+-} - B_{-+}A_{+-} \bigr) = c(A,B). \end{align*}
 
\noindent\textbf{In Conclusion:}\\
\noindent We have sketched the quantization of the Dirac theory in the most established (rigorous) way. It should be noted that at no point in this construction did an infinite number of particles appear and no reference to a Dirac sea or the like was made. But still, we have encountered all the difficulties that we have hinted at in the introductory chapter and that were suggested by the Dirac sea picture:

\begin{itemize}
\item Unitary transformations are not implementable as unitary operators on the Fock space unless they satisfy the Shale-Stinespring condition.
\item Even for implementable operators, the implementation is well-defined only up to a complex phase.
\item Second quantization of Hermitian operators has similar difficulties. Where it does exist, it gives rise to a commutator anomaly in form of the Schwinger term.
\end{itemize}
These results are the reason why a second quantization of the time evolution is, in general, impossible. The asymptotic case will turn out to be somewhat better behaved. But even then, at least the phase of the second quantized scattering operator $\bf{S}$ remains ill-defined.\\

\newpage
\subsection{$C^*$ - and CAR-  Algebras}
We take a quick look at Fock spaces and implementations of unitary transformations from the perspective of representation theory. This approach is very abstract but will reveal its beauty along the way. 

%\begin{Definition}($C^*$-algebra)
%\item A Banach algebra (over $\IC$) is a complex, associative algebra $B$ with 1, equipped with a norm $\lVert \cdot \rVert$ satisfying the following conditions
%\begin{enumerate}[i)]
%\item $\lVert 1 \rVert = 1$ 
%\item  $\lVert ab \rVert \leq \lVert a \rVert \cdot \lVert b \rVert,\; \forall a,b \in B$
%\item B is complete as a metric space w.r.to the norm $\lVert \cdot \rVert$
%\end{enumerate}

%\item a  *-algebra is an associative algebra $A$ over $\IC$ together with an anti-linear involution $*: A \rightarrow A$ satisfying $(a^*)^* = a \; \text{and}\; (ab)^* = b^*a^*,\; \forall a,b \in A$. 

%\item A $C^*$-algebra is a Banach algebra which is also a $*$-algebra, and whose norm $\lVert \cdot \rVert$ is compatible with the $*$-involution in the sense that $\lVert a^*a \rVert = \lVert a \rVert^2$.\\
%This implies $\lVert a^* \rVert$ = $\lVert a \rVert$.

%\item A $C^*$-homomorphism is an algebra homomorphism $h: A \rightarrow B$ between two $C^*$-algebras A and B satisfying $h(a^*) = h(a)^*,\; \forall a \in A$.\\
%It can be shown that any such homomorphism is continuous with norm $\leq 1$.
%\end{Definition}
\begin{Definition}[$C^*$-Algebra] \label{DefBanachalgebra}
\mbox{}\\
A unital \emph{Banach algebra} (over $\IC$) is a complex, associative algebra $B$ with 1, equipped with a norm $\lVert \cdot \rVert$ that makes it a complex Banach space and has the properties\\
%\begin{enumerate}[i)]
\begin{equation*}i)\; \lVert 1 \rVert = 1 \hspace{2cm} ii)\; \lVert ab \rVert \leq \lVert a \rVert \cdot \lVert b \rVert,\; \forall a,b \in B \end{equation*}

\noindent A \emph{$C^*$-algebra} is a complex Banach algebra $A$ together with an anti-linear involution\\
$*: A \rightarrow A$ satisfying: 
\begin{equation*} iii)\; (a^*)^* = a \hspace{1.2cm} iv)\; (ab)^* = b^*a^* \hspace{1.2cm} v)\; \lVert a^*a \rVert = \lVert a \rVert^2
\end{equation*}
$\forall a, b \in A$.  Properties $ii) - v)$ together imply $\lVert a^* \rVert = \lVert a \rVert$.
\end{Definition}

\begin{Definition}[$C^*$-homomorphism]
\mbox{}\\
A $C^*$-homomorphism is an algebra homomorphism $h: A \rightarrow B$ between two\\ 
$C^*$-algebras A and B satisfying $h(a^*) = h(a)^*,\; \forall a \in A$.\\
It can be shown that any such homomorphism is continuous with norm $\leq 1$.
\end{Definition}

\begin{Example}[$C^*$-algebra of bounded operators]
\mbox{}\\
Consider the Banach space $\mathcal{B}(\HH)$ of bounded operators on a Hilbert space $\HH$.\\
On $\mathcal{B}(\HH)$ we have a $*$-involution $T \mapsto T^*$, where $T^*$ denotes the adjoint of the bounded operator $T$. From the well-known properties of Hermitian conjugation and of the operator norm it follows that $\bigl(\mathcal{B}(\HH), \lVert\cdot\rVert , *\bigr)$ is a $C^*$-algebra.
In fact, \emph{every} $C^*$-algebra is isomorphic to a sub-algebra of $\mathcal{B}(\HH)$ for suitable $\HH$.
\end{Example}
\begin{Definition}[States of a $C^*$-algebra]
\mbox{}\\
A \emph{state} of a $C^*$-algebra A is a complex-linear map $\omega: A \rightarrow \IC$ satisfying\\
$\omega(1) = 1$ and $\omega(a^*a)\geq 0, \; \forall a \in A$.
\end{Definition}

\begin{Example}[Quantum States]
\mbox{}\\
For the $C^*$-algebra $A = \mathcal{B}(\HH)$ every $\Psi \in \HH$ defines a state $\omega_{\Psi}$ by  
\begin{equation*}  \omega_{\Psi}(T) := \frac{\langle \Psi , T\, \Psi \rangle}{\langle \Psi , \Psi \rangle}, \; \forall\, T \in \mathcal{B}(\HH)\end{equation*}
In the language of quantum theory, we would say that the \emph{state} represented by the Hilbert-space vector $\Psi \in \HH$ is characterized by all its ``expectation values''. 
\end{Example}

\begin{Definition}[CAR-Algebra]
\mbox{}\\
A \emph{CAR-map} into a C*-algebra $A$ is an anti-linear map $a: \HH \rightarrow A$ satisfying the CAR-relations
\begin{align} a(f)a(g)^* + a(g)a(f)^* = \lbrace a(f) , a^*(g) \rbrace =& \langle f , g \rangle \cdot 1\\
a(f)a(g)\; + \;a(g)a(f)\; =\, \lbrace a(f) , a(g)\, \rbrace =&\; 0 \end{align}
A CAR algebra over $\HH$ is a C*-algebra $\car$ together with a CAR-map $a: \HH \rightarrow \car$ having the following universal property: for every CAR-map $a': \HH \rightarrow \mathcal{B}$ into a C*-algebra $\mathcal{B}$ there is a unique C*-homomorphism $h: \car \rightarrow \mathcal{B} \; \text{s.t.}\; a' = h \circ a$.
\begin{equation}
\begin{xy}
  \xymatrix{
      \HH \ar[r]^{a} \ar[d]_{a'}& \car \ar@{-->}[ld]^h      \\
      \mathcal{B}
  }
\end{xy}
 \end{equation}  
Such a CAR-algebra can be explicitly constructed and by the universal property it is unique up to canonical isomorphism.
 \end{Definition}

It is probably more intuitive to think of the CAR Algebra as the free associative algebra generated by $\HH$ modulo the canonical anti-commutation relations \eqref{CAR}. Mathematicians, however, like to define it by the universal property as above which is rather abstract, but very elegant to work with. We might get a hint of how powerful the universal property is, by clarifying the relationship between field operators as introduced in the last section and representations of the abstract CAR-algebra defined above. Recall that a field operator is a map $\Psi: \HH \rightarrow B(\FF)$ from the Hilbert-space into the $C^*$-algebra of bounded operators on the Fock-Space, satisfying the CAR \eqref{CAR}. By the universal property of the CAR-algebra there exists a C*- homomorphism $\pi: \car \rightarrow  B(\FF)\; s.t. \;\Psi = \pi \circ a$. This homomorphism $\pi$ is then a representation of the CAR-algebra $\car$ on $\FF$.

 \begin{Proposition}[Fock Representation]
 \label{Prop:Fock Representation}
 \mbox{}
 \item The field operator $\Psi$ defined in \eqref{fieldoperator} induces  an irreducible
 %\footnote{cf. \cite{Tha}, Thm.10.2.}  
 representation of the CAR-algebra of $\HH$ on the Fock space 
 \begin{equation}\FF = \FF(\HH_+,\HH_-) :=  \bigwedge \HH_+ \otimes \bigwedge \overline{\HH_-}, \end{equation} called the \emph{Fock representation}. (Here we take, for simplicity, $\mathcal{C}$ to be the anti-unitary map onto the complex conjugated space).
 
\item Similarly, we can define for an arbitrary polarization $W \in \pol(\HH)$ the field operator $\Psi_W(f) = a(P_W f) + b^*(P_{W^\perp} f)$  and call the induced representation on $\FF(W,W^\perp) := \bigwedge W \otimes \bigwedge \overline{W^\perp}$ the Fock representation w.r.to the polarization $\HH = W \oplus W^\perp$.

\end{Proposition}
 
%If $\Psi$ is the standard field operator defined in XXX, $\pi$ is also called the wave-representation of the CAR-algebra.
\subsection{Representations of the CAR Algebra}
\label{subsection:CAR representations}

Mathematically, a ``Fock space'' can be understood as a Hilbert space, arising as the representation space of an irreducible representation of an abstract CAR-algebra (in the fermionic case), or a CCR-algebra (in the bosonic case). There are various ways to get irreducible representations of a CAR-algebra and they are generally not equivalent to each other.

\begin{Definition}[Equivalent Representations]
\mbox{}\\
Consider the CAR algebra $\car$ ober $\HH$. Two representations $\pi$ and $\pi'$ of $\car$ on Fock spaces $\mathcal{F}$ and $\mathcal{F'}$ are \emph{equivalent} if there exists a unitary transformation $T: \mathcal{G} \rightarrow \mathcal{G'}$ with $T\, \pi(a(f)) =  \pi'(a(f))\,T, \; \forall f \in \HH$.\\ 
\end{Definition}

\noindent We generalize the notion of ``implementability'' introduced in Def. \ref{Def:Implementability} to arbitrary representations of the CAR-algebra.\\

\begin{Definition}[Implementation of Unitary Transformation]\label{Def:Implementation}
\mbox{}\\
Let $U \in \cu(\HH)$ a unitary operator. If $a: \HH \rightarrow \car$ is a CAR-map, so is $a \circ U$. By the universal property of the CAR-algebra, there exists a unique $C^*$-homomorphism $\beta_U: \car \rightarrow \car$ with $\beta_U \circ a = a \circ U$.
$\beta_U$ is the \emph{Bogoljubov transformation} corresponding to $U$.
\item Let $\pi: \car \rightarrow \FF$ be a representation of the CAR-algebra on the Fock space $\FF$. \\
$U$ (or $\beta_U$) is called \emph{implementable} on $\FF$ if there exists a unitary operator $\tilde{U}$ on $\FF$ such that $\tilde{U} \pi(a(f)) = \pi(a(Uf)) \tilde{U}, \; \forall f \in \HH$. $\tilde{U}$ is then called an \emph{implementation} of $U$ on $\FF$.
\end{Definition} 

In other words: $U$ is implementable on the Fock space $\FF$ if and only if the two CAR-representations $\pi$ and $\pi\circ\beta_U$ are unitarily equivalent. Therefore, the ``implementation problem'' and the ``equivalence problem'' are closely related.
\subsubsection{GNS-representations}
\label{subsubsection:GNSrepresentation}
We have explicitly constructed the Fock representation on  $\FF = \bigwedge\HH_+ \otimes \bigwedge \overline{\HH_-}$, by defining creation and annihilation operators and finally the field operator $\Psi$ \eqref{fieldoperator}. This, however, was just a very special example of a Fock space and a representation of the CAR-algebra. The seminal paper of Shale and Stinespring is actually formulated for spin-representations of the infinite-dimensional Clifford algebra of $\HH$ which is closely related to the CAR-algebra $\car$ (\cite{SS}, see \cite{GrVa} for a more detailed construction).
One of the most important methods for obtaining representations of a CAR algebra are the so-called \textit{GNS constructions} (after I.M. Gelfand, M.A. Neumark and I.E. Siegel.)
Given a state $\omega$ on $\car$, the GNS-construction yields a Hilbert space $\HH_{\omega}$, a representation $\pi_{\omega}$ on $\HH_{\omega}$ and a unit vector $\Omega_{\omega} \in \HH_{\omega}$ (the ``GNS-vacuum''), such that
\begin{equation} \omega(a) = \sk{\Omega_{\omega}, \pi_{\omega}(a)\, \Omega_{\omega}}, \; \forall a \in \car \end{equation}
and such that the action of $\pi_{\omega} (\car)$ on $\Omega_{\omega}$ generates a dense subset of $\HH_{\omega}$.\\

\noindent In particular, given an orthogonal projection $Q$, i.e. a self-adjoint operator on $\HH$ with $Q^2 = Q$,  we can \textit{define} a state $\omega_{Q}$ by setting $\omega_{Q} (1) = 1$ and fixing the \textit{two-point functions}
\begin{equation}\label{quasifreestate} \omega_{Q}(a^*(f)a(g)) := \langle g , Q f \rangle_{\HH} \end{equation}
for $f, g \in \HH$ and $a$, the CAR-map defining $\car$.
Then, the corresponding GNS-construction $(\FF_Q, \pi_Q, \Omega_Q)$ satisfies
\begin{equation} \langle \Omega_{Q}, \pi_{V}(a^*(g)a(f)) \Omega_{Q}\rangle = \langle g , Q f \rangle_{\HH} =  \langle Q f , Q g \rangle_{\HH} \end{equation}
and, using $a(g)a^*(f) = \langle g , f \rangle_{\HH} \mathds{1} -  a^*(f)a(g)$, 
\begin{align} \langle \Omega_{Q}, \pi_{Q}(a(g)a^*(f)) \Omega_{Q}\rangle = \sk{f,g}_{\HH} -  \langle Q f , Q g \rangle_{\HH} = \langle (\mathds{1} - Q) f , (\mathds{1} - Q) g \rangle_{\HH}. \end{align}
In particular, we read off 
\begin{equation}\begin{split}
\pi_Q(a(g)) \Omega_{Q} = 0, \; \text{for}\; g \in \ker(Q)\\
 \pi_Q(a^*(f)) \Omega_{Q} = 0, \; \text{for}\; f \in \im(Q)
\end{split}\end{equation}
so that $\pi_Q \circ a : \HH \rightarrow \mathcal{B}(\FF_Q)$ acts like a field operator on the vacuum $\Omega_Q$.\\
Indeed, for $W \in \pol(\HH)$ and  $Q = P_W^\perp$, the two-point functions of the GNS-construction are the same as the two-point functions of the field operator $\Psi_W$ on $\bigwedge P_W \HH \oplus \bigwedge \overline{P_{W^\perp} \HH}$, because
\begin{equation*} \sk{\Omega , \Psi^*(f)\Psi(g) \Omega} = \sk{\mathcal{C}P_{W^\perp} f , \mathcal{C}P_{W^\perp g}} = \sk{g, P_{W^\perp f}}_{\HH}=\sk{g,Q f}_{\HH}.\end{equation*}
This implies that the two representations are actually equivalent. 
%The same is true, of course, for the Fock- and GNS representation w.r.to a different polarization of $\HH$.\\

\noindent Now, in the language of representation theory we can formulate another, very general version of the Shale-Stinespring theorem, also known as the theorem of Powers and St\o rmer.

\begin{Theorem}[Powers, St\o rmer 1969]
\label{Thm:PoSto}
\mbox{}\\
The GNS-representations $(\FF_P, \pi_P, \Omega_P)$ and  $(\FF_Q, \pi_Q, \Omega_Q)$  corresponding to orthogonal projections $P$ and $Q$ on $\HH$ are unitary equivalent if and only if $P - Q $ is a Hilbert-Schmidt.
\end{Theorem}
\begin{proof} \cite{PoSt}\\ 
\end{proof}

\noindent Applying this result to the ``implementation problem''  we find ones more that $U \in \cu(\HH)$ is implementable on $\FF = \bigwedge\HH_+ \otimes \bigwedge \mathcal{C}\HH_-$ if and only if $P_{-} - UP_{-} U^*$ is Hilbert-Schmidt i.e. if and only if $U_{+-}$ and $U_{-+}$ are Hilbert-Schmidt i.e. if and only if $U \in \Ur(\HH)$.

%``Quantization'' is one of the strange things that all physicists seem to understand, but is very hard to give precise meaning to \footnote{And if one tries, it might turn out that ``quantization'' is not even possible in the way one had always assumed}. With a mathematician, however, one can probably agree that a ``quantization'' is a representation of a $C^*$-algebra on a complex Hilbert space, fulfilling certain commutation relations. 

\newpage
\section{The Infinite Wedge Space Construction}
In the light of realization that the unitary time evolution cannot be implemented on a fixed Fock space, my esteemed teachers and colleagues Dirk Deckert, Detlef D\"urr, Franz Merkl and Martin Schottenloher developed a formulation of the external field problem in QED on \textit{time-varying} Fock spaces. The Fock spaces are realized as ``infinite wedge spaces'' constructed from ``Dirac seas'' over a chosen polarization class. This construction is nice for several reasons:

\begin{enumerate}
\item It is very ``down to earth'' and close to the physical intuition.

\item It is designed to highlight all the ambiguities and choices involved in the construction, rather than hiding them.

\item The formalism is very flexible and well-suited for the treatment of different polarization classes and the implementation of the time evolution on varying Fock spaces. 
\end{enumerate}

\noindent Of course, Deckert et.al. are consistently ignoring the geometric structure underlying the construction (as we will see). However, in doing so, it becomes clear that this geometric structure, as beautiful as it might be, is not essential for the understanding of the crucial (physical) problems. Whether it is helpful in finding a solution is a different question that will be discussed in subsequent chapters.\\

\noindent The basic idea of the infinite wedge space construction is this:
\begin{itemize}
\item For a fixed equal charge polarization class, we consider certain maps from an index-space $\ell$ into the Hilbert space $\HH$, with image in the respective polarization class. We call these maps \textit{seas}. Intuitively, they will represent states of the Dirac sea, i.e. states of infinitely many fermions. 
\item If we fix a basis $(e_k)_{k\in \IN}$ in $\ell$, we can think of a such a sea $\Phi: \ell \rightarrow \HH$ as the infinite exterior product $\Phi(e_0) \wedge \Phi(e_1) \wedge \Phi(e_2)\wedge \dots$
\item We generalize the finite-dimensional scalar product \eqref{slaterdet} (also known as Slater determinant) by the infinite-dimensional Fredholm determinant i.e. by
\begin{equation} \langle \Phi , \Psi \rangle = \det(\Phi^{*}\Psi) = \lim\limits_{N\rightarrow \infty} \det \bigl(\langle \Phi(e_i) , \Psi(e_j) \rangle \bigr)_{i,j \leq N}. \end{equation}
For this to be well-defined, we have to impose further restrictions on the permissible seas. This leads us to an equivalence relation, defining classes of seas. 
\item Finally, we get a linear space by taking the algebraic dual of such a  sea class and forming the completion with respect to before mentioned scalar product. This will be the Fock space realized as an infinite wedge space over the chosen Dirac sea class.\\ 
\end{itemize}  

\noindent For the remainder of this chapter, let $\ell$ be an infinite-dimensional, separable, complex Hilbert-space. $\ell$ will serve as an index space. A convenient choice is therefore $\ell = \ell^2(\IC)$.
\newpage
\begin{Definition}[Dirac Seas]
\item Let  $\seas$ be the set of all bounded, linear maps $\Phi: \ell \rightarrow \HH$ such that $\im(\Phi) \in \pol(\HH)$ and $\Phi^*\Phi: \ell \rightarrow \ell$ has a determinant, i.e. $\Phi^*\Phi \in Id + I_1(\ell)$.
\item Let $\iseas$ denote the subset of all $\Phi \in \seas$ that  are also isometries.
\item For any equal-charge polarization class $C \in \pol(\HH)\slash\approx_0$ let $\ocean(C)$ be the set of all $\Phi \in \iseas$ with $\im(\Phi) \in C$.
\end{Definition}

\begin{Definition}[Dirac Sea Classes]
\mbox{}\\
For $\Phi, \Psi \in \seas$ we introduce the relation 
\begin{equation*} \Phi \sim \Psi \iff \Phi^*\Psi \in Id + I_1(\ell) \end{equation*}
This defines an equivalence relation on $\seas$ (\cite{DeDueMeScho}, Cor. II.9) . The equivalence classes are called \emph{Dirac sea classes} and denoted by $\mathcal{S}(\Phi) \in \seas \slash\sim$. 
\end{Definition} 

\begin{Lemma}[Connection between $\sim$ and $\approx_0$]\label{lem:connection sim approx}
\mbox{}\\
  Given $C\in\pol(\HH)/{\approx}_0$ and $\Phi\in\ocean(C)$ we find
  \begin{align*}
    C=\{\im(\Psi)\;|\;\Psi\in\iseas  \text{ such that }\Psi\sim\Phi\}.
  \end{align*}
In other words: $\Psi \sim \Phi \;\text{in}\; \iseas \Rightarrow \im(\Psi) \approx_0 \im(\Phi) \;\text{in}\; \pol(\HH)$\\ 
and $W \in C = [\im(\Phi)]_{\approx_0} \Rightarrow \exists\, \Psi \in \iseas: \Psi \sim \Phi\, \wedge \, \im(\Psi) = W$. 
\end{Lemma}
\begin{proof} See Lemma II.12 in \cite{DeDueMeScho} or note that the statement follows as a reformulation of our Lemma \ref{Lemma:admbasis} (``Existence of admissible basis'') and the remark following Definition \ref{Def:admbasis}(``Admissible Basis''). A polar decomposition can be used to make the Dirac sea/admissible basis isometric. 
\end{proof}

\begin{Construction}[Formal Linear Combinations] \label{const:proto}
\mbox{}
\begin{enumerate}
\item
For any set $\mathcal{S}$, let $\IC^{(\cs)}$ denote the set of all maps $\alpha:\cs\to \IC$ for which the support $\{\Phi\in \cs\;|\;\alpha(\Phi)\neq 0\}$ is finite. For $\Phi\in \cs$, let $[\Phi]\in \IC^{(\cs)}$ denote the algebraic dual, i.e. the map satisfying $[\Phi](\Phi)=1$ and $[\Phi](\Psi)=0$ for $\Phi\neq\Psi\in \cs$.
Thus, $\IC^{(\cs)}$ consists of all finite formal linear
combinations $\alpha=\sum_{\Psi\in \cs}\alpha(\Psi)[\Psi]$ of elements of $\cs$ with complex coefficients.
\item
Now, let $\cs\in\seas /{\sim}$ be a Dirac sea class. \\
We define the map
$\sk{\cdot,\cdot}:\cs\times \cs\to \IC$, $(\Phi,\Psi)\mapsto\sk{\Phi,\Psi }:=\det(\Phi^*\Psi )$.\\
The determinant exists and is finite by definition of $\sim$. 

\item
For $\cs\in\seas /{\sim}$,
let $\sk{\cdot,\cdot}:\IC^{(\cs)}\times \IC^{(\cs)}\to \IC$
denote the sesquilinear extension of $\sk{\cdot,\cdot}:\cs\times \cs \to \IC$ to the linear space $\IC^{(\cs)}$ defined as follows:\\
For $\alpha,\beta\in \IC^{(\cs)}$,
\begin{equation}
\sk{\alpha,\beta}:=\sum_{\Phi\in \cs}\sum_{\Psi \in \cs} \overline{\alpha(\Phi)}\beta(\Psi )
\det(\Phi^*\Psi ).
\end{equation}
The bar denotes the complex conjugate. Note that the sums consist of at most finitely many nonzero summands. In particular, we have $\sk{[\Phi],[\Psi ]}=\sk{\Phi,\Psi }$ for $\Phi,\Psi \in \cs$.

The sesquilinear form $\sk{\cdot,\cdot}$ on $\IC^{(\cs)}$ is Hermitian and positive-semidefinite\footnote{see\cite{DeDueMeScho}, Lemma II.14 or our Proposition \ref{Hermitianform} ``Hermitian Form''.}.\\ 
Therefore, it defines a semi-norm on $\IC^{(\cs)}$ by $\lVert \alpha \rVert := \sqrt{\langle \alpha , \alpha \rangle}$.
\end{enumerate}
\end{Construction}
\begin{Definition}[Infinite Wedge Space]
\mbox{}\\
Let $\FF_{\cs}$ be the completion of $\IC^{(\cs)}$ with respect to the semi-norm $\lVert \cdot \rVert$.\\
$\FF_{\cs}$ is an infinite-dimensional, separable, complex Hilbert space.\\
We will refer to it as the \emph{infinite wedge space} over the Dirac sea class $\cs$.\\  

\noindent By $\bwdge: \cs \rightarrow \FF_{\cs}$ we denote the canonical map $\bwdge\Phi := [\Phi]$ coming from the inclusions \begin{equation*}\cs \hookrightarrow \IC^{\cs} \hookrightarrow \FF_{\cs}.\end{equation*}   
\end{Definition}
\noindent Note that the null-space $N_{\cs}:= \lbrace \alpha \in \IC^{(\cs)} \mid \lVert \alpha \rVert = 0 \rbrace$ is factored out in the process of completion. 

\begin{Construction}[The Left Operation]
\item $\cu(\HH)$ acts on $\seas$ from the left by $(U, \Phi) \rightarrow U\Phi$.\\
This extends to a well-defined map 
\begin{equation*} \cs \xrightarrow{U} U\cs:= \lbrace U \Phi \mid \Phi \in \cs \rbrace \end{equation*}
between Dirac sea classes.\\
\item For $U \in \cu(\HH)$, the induced left operation $\lop U: \IC^{(\cs)} \rightarrow \IC^{(U\cs)}$, given by
\begin{align*}
\lop U\left(\sum_{\Phi\in \cs}\alpha(\Phi)[\Phi]\right)
=\sum_{\Phi\in \cs}\alpha(\Phi)[U\Phi],
\end{align*}
is an isometry with respect to the Hermitian forms $\scpro$ on $\IC^{(\cs)}$ and $\IC^{(U\cs)}$. \\
Consequently, it extends to a unitary map $\lop U:\FF_{\cs} \rightarrow \FF_{U\cs}$ between infinite wedge spaces, characterized by $\lop U(\bwdge \Phi) = \bwdge(U\Phi)$.
\item This extends immediately to unitary maps between different Hilbert spaces $\HH$ and $\HH'$.
\end{Construction}

\begin{Construction}[The Right Operation]
\item Let $\mathrm{Gl}_{-}(\ell) := \lbrace R \in \mathrm{Gl}(\ell) \mid R^*R \in Id + I_1(\ell)\rbrace$.\\ 
$\mathrm{Gl}_{-}(\ell)$ acts on $\seas$ from the right by $(\Phi , R) \rightarrow \Phi R$.
This extends to a well-defined map 
\begin{equation*} \cs \xrightarrow{R} \cs R:= \lbrace \Phi R \mid \Phi \in \cs \rbrace \end{equation*}
between Dirac sea classes. For $\cs\in\seas/{\sim}$ and $R\in \Gl_{-}(\ell)$
we have an induced operation from the right, namely
$\rop R:\IC^{(\cs)}\to \IC^{(\cs R)}$
given by
\begin{align*}
\rop R\left(\sum_{\Phi\in \cs}\alpha(\Phi)[\Phi]\right)
=\sum_{\Phi\in \cs}\alpha(\Phi)[\Phi R].
\end{align*}
This map is angle-preserving w.r.to the Hermitian forms, i.e. an isometry up to scaling, since
\begin{equation*}\det\bigl( (\Phi R)^*(\Psi R) \bigr) = \det(R^*\Phi\Psi R) = \det(R^*R)\det(\Phi^*\Psi) \end{equation*}
$\forall R \in Gl_{-}(\ell)$ and $\Phi, \Psi \in \cs$ and therefore, for all $\alpha,\beta\in \IC^{(\cs)}$,
\begin{align*}
\sk{\rop R\alpha,\rop R\beta}
= \det(R^*R) \sk{\alpha,\beta}.
\end{align*}
In particular, we have $\rop R[N_\cs]\subseteq N_{\cs R}$. It follows that for every $R\in \mathrm{Gl}_{-}(\ell)$,
the operation $\rop R:\IC^{(\cs)}\to \IC^{(\cs R)}$
induces a bounded linear
map $\rop R:\FF_\cs\to \FF_{\cs R}$ between the infinite wedge spaces, characterized by
$\rop R(\mathsf{\Lambda} \Phi)=\mathsf{\Lambda}(\Phi R)$ for $\Phi\in \cs$. This map is unitary, up to scaling.\\ 
The definition extends immediately to maps between different index spaces $\ell$ and $\ell'$.
\end{Construction}

We can think of the right action as basis transformations on the Hilbert space or, with a somewhat greater leap of imagination, as ``rotations'' of the Dirac seas. The seas will stay in the same Dirac sea class if and only if the transformations are ``small'' in the sense that they vary from the identity only by a trace-class operator or, in other words, have a determinant. Pictorially speaking, we ``rotate'' a Dirac sea just ``a little'', if we don't stir too much deep down in the sea. This intuition is made precise in the following Lemma: 

\begin{Lemma}[Uniqueness up to a Phase]\label{Lemma:Uniquenessuptoaphase}
\mbox{}\\
Let $R \in \mathrm{Gl}_{-}(\ell)$ and $\cs \in \seas /{\sim}$ a Dirac sea class.
Then, $\cs R = \cs$ if and only if $R$ has a determinant, i.e. iff $R \in \mathrm{GL}^1(\ell)$.
In this case, the right operation $\rop R : \FF_{\cs} \rightarrow \FF_{\cs}$ corresponds to multiplication with $\det(R)$ on $\FF_{\cs}$. In particular, if $R \in \cu(\ell)\, \cap\, \Gl^1(\ell)$, the right operation $\rop R$ corresponds to multiplication by a complex phase $\in \cu(1)$.
\item Since the right-action is a well-defined map between Dirac sea classes, it follows immediately for another operator $Q \in \Gl_{-}(\ell)$ that $\cs R = \cs Q$ if and only if  $Q^{-1}R$ has a determinant. In this case, 
$\rop R = \det(Q^{-1}R) \rop Q$ on $\FF_{\cs}$. 
\end{Lemma}
\begin{proof} Let $\cs \in \seas /{\sim}$. For any $\Phi \in \cs,\; \Phi^*(\Phi R)$ has a determinant if and only if $R$ has a determinant, since $\Phi^*\Phi \in Id + I_1(\ell)$ by definition. Thus: 
\begin{equation*}\Phi \sim \Phi R \iff \cs = \cs R \iff R \in \mathrm{Gl}^1(\ell)\end{equation*}
Now, on $\IC^{(\cs)}$ with the semi-norm defined by $\scpro$:
\begin{align}
\|[\Phi R]-(\det R)[\Phi]\|^2
=\;&\det((\Phi R)^*(\Phi R))-(\det R) \det((\Phi R)^* \Phi)\nonumber\\
- &\,\overline{\det R}
\det(\Phi^*\Phi R)+\lvert\det R\rvert^2\det(\Phi^*\Phi)
\nonumber\\
=\;& 2\lvert\det R\rvert^2\det(\Phi^*\Phi)-2\lvert\det R\rvert^2\det(\Phi^*\Phi)
=0.
\end{align}
Therefore, $\rop R = \det(R)\cdot \mathrm{Id}$ on $\FF_{\cs}$.\\
\end{proof}
Recall that the construction of the Fock space as an infinite wedge space involves two consecutive choices. The first one (usually determined by the physics) is the choice of a polarization class $C \in \pol(\HH)$. Afterwards, we have a more or less arbitrary choice of a polarization class $\cs \in \ocean(C)/{\sim}$. As this is an equivalence class, it is uniquely determined by a ``reference polarization'' $\Phi \in \ocean(C)$.
This duality is reflected in the duality of left- and right- operations. The operations from the left are transformations between polarization classes. The unitary transformation induced by $U \in \cu(\HH)$ stays in the same polarization class if and only if it satisfies the Shale-Stinespring condition \eqref{SS}.
It preserves the charge, if and only if the $(++)$-component w.r.t the corresponding polarization has index 0.
The operations from the right, on the other hand, are transformations between different Dirac sea classes within the same ocean i.e. between seas with image in the same polarization class. We see how the mathematical structure at hand gives us a very natural way to handle unitary transformations between Fock spaces.

\noindent The left-action alone would correspond to the product-wise lift
\begin{equation*} \Phi(e_0)\wedge\Phi(e_1)\wedge\Phi(e_2)\wedge \cdots \; \xrightarrow{\;\;\lop U\;\;} \; U\Phi(e_0)\wedge U\Phi(e_1)\wedge U\Phi(e_2)\wedge \cdots \end{equation*}
But this alone is not very helpful, in general. In addition, we need a suitable right-operation $\rop R$ to rotate the seas back into the desired Dirac sea class. If $U$ preserves an equal-charge polarization class $C \in \pol(\HH)\slash \approx_0$, i.e.  $U \in \cu^0_{\rm res}(\HH, C ; \HH , C)$ a right-operation can be chosen in such a way as to make $ \rop R \lop U$ a unitary transformation on the Fock space $\FF_{\cs}$ for a fixed $ \cs \in \ocean(C)/{\sim}$. If, however, $U$ maps one polarization class $C$ into a different one, i.e. $U \in \cu^0_{\rm res}(\HH, C; \HH, C')$ with $C' \neq C,$ the best we can do is to implement $U$ as a unitary transformation between \textit{different} Fock spaces $\FF_{\cs}$ and $\FF_{\cs'}$ for $\cs \in \ocean(C)/{\sim}$ and $\cs' \in \ocean(C')/{\sim}$. The right-operation then has to be such that $U \cs R = \cs'$.\\ 

The previous Lemma tells us that the implementations are determined only up to a complex phase. This means that the group that acts on the Fock space contains an additional $\cu(1)$-freedom besides the information contained in $\cu^0_{\rm res}(\HH, C; \HH, C) = \Ures(\HH)$. Indeed, it corresponds to the central extension $\Ures(\HH)$ of $\Ur(\HH)$ by $\cu(1)$ studied in Chapter 4 (cf. [Cor.\ref{Ures on wedge spaces}]). These considerations lie at the heart of the following crucial theorem.
%\newpage
\begin{Theorem}[Lift of Unitary Transformations, \cite{DeDueMeScho} Thm. II.26] \label{unser abstrakter Shale Stinespring}
\mbox{}\\
For given polarization classes $C, C' \in\pol(\HH)/{\approx}_0$ %and $C'\in\pol(\HH)/{\approx}_0$, 
let $\cs\in \ocean(C)/{\sim}$ and $\cs'\in \ocean(C')/{\sim}$. Then, for any unitary map $U:\HH \rightarrow \HH$, the following are equivalent:
\begin{enumerate}
\item[i)] $U\in \cu_{\rm res}^0(\HH,C; \HH, C')$.

\item[ii)]
There is $R\in \cu(\ell)$ such that $U\cs R=\cs'$, and hence
$\rop R\lop U$ maps $\FF_\cs$ to $\FF_{\cs'}$.

In this case, if  $R' \in \cu(\ell)$ is another map with $U \cs R' = \cs'$, then \\$\rop {R'} \lop U = \det(R'R^*)\, \rop R \lop U$ with $\det(R'R^*) \in \cu(1)$. 
\end{enumerate}
\end{Theorem}
%\vspace*{2mm}
\begin{Corollary}[$\cu(\ell)$ acts transitively on oceans]
\label{Cor:U(l) acts transitively}
\mbox{}\\
Setting $U = \mathds{1}_{\HH}$ in the last theorem it follows immediately that 
  \begin{align*}
    \ocean(C)/{\sim}=\{\cs R\;|\; R\in \cu(\ell)\}.
  \end{align*}
for any given $C\in\pol(\HH)/{\approx_0}$ and $\cs\in\ocean(C)/{\sim}$.
\item In other words: $\cu(\ell)$ acts transitively on $\ocean(C)/{\sim}$ from the right.  
\end{Corollary}

\begin{Corollary}[Equivalence of Infinite Wedge Spaces]
\label{Cor:equivalence of infinite wedge spaces}
\mbox{}\\
Let $\cs, \cs' \in \ocean(C)/{\sim}$ two Dirac sea classes over the same polarization class $C \in\pol(\HH)/{\approx}_0$.  
There exists $R \in \cu(\ell)$ such that $\rop R: \FF_\cs \to \FF_{\cs'}$ is an isomorphism.\\
If $R'$ is another map providing such an isomorphism, then $R^*R' \in \cu^1(\ell)$ has a determinant and $\rop {R'} =  \det(R'R^* )\, \rop R$. In particular, $\FF_\cs = \FF_{\cs'}$ if and only if $R \in \cu^1(\ell)$.
\end{Corollary}

%and $C'\in\pol(\HH)/{\approx}_0$, 

\vspace*{2mm}
Theorem \ref{unser abstrakter Shale Stinespring} is a generalization of the Shale-Stinespring theorem in the language of the infinite wedge spaces.
Setting $C = C' = [\HH_+]$ it reproduces the well-known result that a unitary operator $U \in \cu(\HH)$ is implementable on a fixed Fock space $\FF$ over the polarization class $[\HH_+]$ if and only if $U \in \Ur(\HH)$.
Together with Lemma \ref{Lemma:Uniquenessuptoaphase}, the theorem implies that the implementation of a unitary operator (on a fixed Fock space or as a map between different Fock spaces) is unique up to a complex phase. So again, we have encountered the infamous geometric phase of QED. 

%\newpage
\section{The Geometric Construction}
In this section, we are going to follow the construction of the fermionic Fock space as presented, for example, in \cite{PreSe} or \cite{Mi}. The Fock space will be constructed from holomorphic sections in the dual of the determinant line-bundle over the complex (restricted) Grassmannian manifold. We will refer to this as the \textit{geometric construction} of the Fock space. It as quite an elegant approach for several reasons:
\begin{enumerate}
\item It reveals and exploits a remarkably elegant geometric structure that appears very naturally in our mathematical framework.
\item Polarizations, seas and their basis transformations are embraced by the geometric description, were they appear as base-manifold, fibre bundle and structure group. This provides an elegant and concise formalism, illustrating the relations between the objects involved. 
\item The central extensions $\Gres(\HH)$ of $\Gl_{\rm res}(\HH)$ and $\Ures(\HH)$ of $\Ur(\HH)$ constructed in Chapter 4 act most naturally on this Fock space.\\
\end{enumerate}

\noindent We will follow the pertinent literature in constructing the Fock space for the standard polarization $\HH = \HH_+ \oplus \HH_-$. This will serve as an exemplary illustration of the general case. The construction can be easily applied to arbitrary polarizations (and polarization classes), matching the generality of the infinite wedge space construction. This is outlined at the end of the section. We are also conforming to the mathematical literature in using $\HH_+$ instead of $\HH_-$ for the polarization that -- intuitively -- plays the role of the ``Dirac sea'' in the geometric construction of the Fock space. So the reader should note that the convention used in this section is in conflict with the physical convention used in \S 5.1, as well as in chapters 7 and 8.\\ 

We have already argued that polarizations, which are points on the infinite-dimensional Grassmanian manifold $\mathrm{Gr}(\HH)$, can be thought of as ``projective states'' of infinitely many fermions. Of course, we are not satisfied with a projective description, but aiming for a full-blown Hilbert space structure. In particular, we want to generalize the scalar product 
\begin{equation*}\langle v_1 \wedge \dots \wedge v_n , w_1 \wedge \dots \wedge w_n \rangle= \det ( \langle v_i , w_j \rangle)_{i,j} \end{equation*} 
from finite exterior products to infinite particle states. However, as we have seen before, while this expression is always well-defined in the finite-dimensional case, the infinite-dimensional limes exists only if we impose further compatibility-conditions on the choices of Hilbert-bases for different polarizations. In the infinite wedge space construction, this was done by identifying ``Dirac sea classes'', on which the Hermitian form is well-defined. In the geometric construction, the corresponding concept is that of \textit{admissible bases}. \\

\begin{Definition}[Admissible Basis]\label{Def:admbasis}
\mbox{}\\
\noindent Let $W \in \mathrm{Gr}^{0}(\HH)$. An \emph{admissible basis} for W is a bounded isomorphism $w: \HH_+\rightarrow W$ with the property that $P_+ \circ w$ has a determinant.
\end{Definition}
\newpage
\noindent It might be helpful to note:
\begin{itemize}
\item $W = \im(w) \in \mathrm{Gr}(\HH)$ iff $P_+w$ is Fredholm and $P_- w$ Hilbert-Schmidt.
\item $W = \im(w) \in \mathrm{Gr}^{0}(\HH)$ iff in addition $\ind(P_+ w) = 0$. 
\item $w$ is admissible basis of $W  \in \mathrm{Gr}^{0}(\HH)$ iff in addition $P_+ w \in Id + I_1(\HH_+)$.\\
\end{itemize}

\begin{Lemma}[Existence of Admissible Bases]\label{Lemma:admbasis}
\item Every $W \in \mathrm{Gr}^{0}(\mathcal{H})$ has an admissible basis.
\end{Lemma}
\begin{proof} Let $W \in \mathrm{Gr}^{0}(\mathcal{H})$. Set $\tilde{w}:= P_W\lvert_{\HH_+}: \, \HH_+ \rightarrow W$.
$P_+ \tilde{w}$ has a determinant, because $\mathds{1}_{\HH_+} - P_+ \tilde{w}  = \bigl(P_+ -P_+P_WP_+\bigr)\bigl\lvert_{\HH_+}$ and
$P_+ - P_+P_WP_+ = P_+P_{W^{\perp}}P_+ = (P_{W^{\perp}}P_+)^*(P_{W^{\perp}}P_+) \in I_1(\HH)$ (Lemma \ref{approx}). We also know that $\tilde{w}$ is a Fredholm operator of index 0\\ 
(bc. $\charge(W, \HH_+) = \ind( P_W\lvert_{\HH_+}) = 0$ ), which means that the kernel of $\tilde{w}$ and its cokernel in $W$ are of equal and finite dimension. Therefore, we can define $w: \HH_+ \rightarrow W$ by setting $w = \tilde{w}$ on $\ker(\tilde{w})^{\perp}$ and extending it to the whole $\HH_+$ in such a way that it maps $\ker(\tilde{w})$ to a complement of $\im(\tilde{w})$ in $W$. Thus, by construction, $\im(w)=W$ and $P_+\,w$ has a determinant as it differs from $P_+ \tilde{w}$ by a finite-rank operator only. Using a polar decomposition, we can also make $w$ unitary.\\ 
\end{proof}
 
\begin{Lemma}[Relationship between Admissible Bases]
\mbox{}\\
Let $w: \HH_+ \to W$ be an admissible basis for $W \in \Gr(\HH)$. Then, $w': \HH_+ \to W$ is another admissible basis for $W$, if and only if $w' = w \circ L$ for some $L \in \Gl^1(\HH_+)$.
\end{Lemma}
\begin{proof} If $w' = w \circ L$, then $w'$ is an isomorphism from $H_+$ to $W$ and $P_+ w' = (P_+ w)\,L$ has a determinant because $P_+ w$ and $L$ do. Therefore, $w'$ is an admissible basis for $W$.\\
Conversely, if $w'$, $w$ are two admissible bases for $W$, we set $L := w^{-1} \circ w' : \HH_+ \rightarrow \HH_+$.\\ Clearly, $L$ is invertible. Furthermore, $P_+w\, L = P_+ w' \in Id + I_1(\HH_+)$ and also\\
$P_+w \in Id + I_1(\HH_+)$, which implies  $L \in Id + I_1(\HH_+)$.\\ \end{proof}

\subsubsection{Arbitrary charges}

\noindent Recall that the Grassmanian $\mathrm{Gr}(\HH)$ has $\IZ$ connected components, corresponding to the relative charges of the polarizations (w.r.to $\HH_+$). Usually we are only interested in the connected component $\mathrm{Gr}^{0}(\HH)$ of the initial vacuum $\approx \HH_+$. The geometric construction over $\mathrm{Gr}^{0}(\HH)$ will yield the zero-charge sector of the fermionic Fock space. Note that the infinite wedge space construction also yields Fock spaces of ``constant charge''. We can easily extend the construction to a ``full'' Fock space by taking the direct (orthogonal) sum of infinite wedge spaces over polarization classes with different relative charges. However, such a construction is not canonical, as it involves a choice of a Dirac sea class for every such charge sector. Thus it should come as no surprise that the same is true for the geometric construction as well. It inherits arbitrary charges very naturally, however, the construction is not canonical. It involves a choice of a (polarized) basis.\\
%as we have encountered several times already. 

\noindent For the remainder of this section we fix an orthonormal basis $\lbrace (e_k)_{k\in\mathbb{Z}} \rbrace \; \text{of} \; \mathcal{H}$  such that $(e_k)_{k \le 0}$ is an ONB of $\mathcal{H}_-$ and $(e_k)_{k \geq 0}$ is an ONB of $\mathcal{H}_+$.
For any $n \in \IZ$ we write
\begin{equation*} \Hplusn := \spn(\lbrace e_k \,\rvert \, k \geq -n \rbrace ).\end{equation*}
By $P^{(n)}_+$ we denote the orthogonal projection $P_{\Hplusn}$ onto $\Hplusn$\\
From the construction of the complete $\Gres(\HH)$ in Section 4.1.2, we recall the shift-operator $\sigma$ defined by $\sigma(e_k) := e_{k+1}$. For any $d, n \in \IZ$, $\sigma^d$ maps $\Hplusn$ to $\HH_{\lbrace \geq -n + d\rbrace}$.\footnote{If we seek more generality, we can choose for each $n \in \IZ$ an arbitrary subspace $W_n \in \Gr$ with $\charge(W_n, \HH_+) = n$, and  a one-parameter group of unitary operators translating between them.}

\begin{Definition}[Admissible Basis for arbitrary Charges]\label{DefSt}
\mbox{}\\
Let $W \in \mathrm{Gr}(\HH)\; \text{with}\; \charge(W, \HH_+) = n$, i.e. $W$ lies in the connected component $\mathrm{Gr}^{n}(\HH)$ of  $\mathrm{Gr}(\HH)$.
An \emph{admissible basis} for $W \in \mathrm{Gr}^{n}(\HH)$ is a bounded linear map 
$w: \Hplusn \rightarrow H$ with $\im(w)=W$ s.t. $w: \Hplusn \rightarrow W$ is an isomorphism and $P^{(n)}_+ \circ w$ has a determinant.
\end{Definition}
\noindent We denote the set of all admissible bases by $\St$ and the subset of admissible bases for $\mathrm{Gr}^{n}(\HH), n \in \IZ$ by $\mathrm{St}^{(n)}(\HH)$.
On $\mathrm{St}^{(n)}(\HH)$, we let $\Gl^1(\HH_+)$ act from the right by
\begin{equation*} w \stackrel{L}\longmapsto w \circ \sigma^{-n}\circ L^{-1}\circ \sigma^{n} \end{equation*}
With a little abuse of notation, we can just write $w \circ L^{-1}$ for the right action and $P_+$ instead of $P^{(n)}_+$ for the orthogonal projection whenever the specific index is clear or irrelevant.
With this sneaky notation, we can forget about the different charges most of the time. The construction will look the same over any connected component.

\begin{Definition}[Stiefel Manifold]
\mbox{}\\
The set $\St$ of all admissible bases for the polarizations in $\mathrm{Gr}(\HH)$ carries the structure of an infinite-dimensional manifold, called the (restricted) \emph{Stiefel manifold} $\St$. The topology is given by the metric 
\begin{equation}\mathrm{d}(w,w') = \lVert P_+(w-w')\rVert_1 + \lVert P_-(w-w')\rVert_2. \end{equation}
The Stiefel manifold has $\IZ$-connected components, corresponding to the connected components of $\mathrm{Gr}(\HH)$. 
Actually, it follows from the previous discussion that $\pi: \St \to \mathrm{Gr}(\HH)$ with $\pi(w) := \im(W)$ is naturally a principle $\Gl^1(\HH_+)$-bundle over $\mathrm{Gr}(\HH)$. The fibre $\pi^{-1}(W)$ over a basepoint $W \in \mathrm{Gr}(\mathcal{H})$ is just the set of admissible bases for $W$ and $\Gl^1(\HH_+)$ acts on these fibers from the right as described above.
\end{Definition}

\begin{Proposition}[$\det$ is holomorphic]
\label{Prop:dethol}
\mbox{}\\
The (Fredholm-) determinant $\det: \Gl^1(\HH_+) \to \IC\setminus\lbrace0\rbrace$ defines a one-dimensional, holomorphic representation of the Lie group $\Gl^1(\HH_+)$.
\end{Proposition}
\begin{proof}Recall that the Lie group structure is given by the embedding 
\begin{align*} \Gl^1(\HH_+) \hookrightarrow \TC; \; 1 + A \mapsto A \end{align*}
of $\Gl^1(\HH_+)$ into the trace-class. Now note that it suffices to show complex differentiability in the identity, because for $(\mathds{1}+A), (\mathds{1}+B) \in \Gl^1(\HH_+)$, we have by the multiplication-rule
\begin{equation*}\det(\mathds{1}+A) - \det(\mathds{1}+B) = \det(\mathds{1}+B) \bigl(\det((\mathds{1}+A)(\mathds{1}+B)^{-1}) - \det(\mathds{1}) \bigr)\end{equation*}
with $(\mathds{1}+A)(\mathds{1}+B)^{-1} = (\mathds{1}+C) \in \Gl^1(\HH_+)$. And since $(\mathds{1}+C)(\mathds{1}+B) = (\mathds{1}+A)  \iff \mathds{1} +C + B + CB = \mathds{1} + A \iff C(\mathds{1}+B) = A - B$, we find $C \to 0 \iff B \to A$ in $\TC$.\\
Now, by definition:
\begin{equation*} \det(\mathds{1}+A) := \sum\limits_{k=0}^\infty \trace\bigl(\sideset{}{^k}\bigwedge (A)\bigr). \end{equation*}
If $(\lambda_n)_n$ are the eigenvalues of $A$ ($A$ is a compact operator), $\trace\bigl(\sideset{}{^k}\bigwedge (A)\bigr) = \sum\limits_{i_1 < \dots < i_k} \lambda_{i_1}\cdots\lambda_{i_k}.$ Using a polar decomposition, we can write $A = U \lvert A \rvert$, with a partial isometry $U$, and $\bigwedge^k(A) = \bigwedge^k(U) \bigwedge^k(\lvert A \rvert)$. The eigenvalues $(\mu_n)_n$ of the positive operator $\lvert A \rvert$ are the singular values of $A$. We derive:
\begin{align*}
\bigl\lVert \sideset{}{^k}\bigwedge (A) \bigr\rVert_1 & = \trace\bigl(\sideset{}{^k}\bigwedge(\lvert A \rvert)\bigr) = \sum\limits_{i_1 < \dots < i_k} \mu_{i_1}\cdots\mu_{i_k}\\
&\leq \frac{1}{k!} \sum\limits_{i_1,  \dots ,i_k} \mu_{i_1}\cdots\mu_{i_k} = 
\frac{1}{k!} \lVert A \rVert_1^k. \end{align*}
Therefore, we conclude:
\begin{equation*} \lvert \det(\mathds{1}+ A) - 1 - \trace(A) \rvert = \bigl\lvert \sum\limits_{k=2}^\infty \trace\bigl(\sideset{}{k}\bigwedge (A)\bigr) \bigr\rvert \leq \sum\limits_{k=2}^\infty \frac{\lVert A \rVert_1^k}{k!}  \leq \exp(\lVert A \rVert_1). \end{equation*} 
This means that $\det$ is continuously differentiable in $\mathds{1}$ with $\det'(\mathds{1})X = \trace(X)$ and hence holomorphic on the complex Banach Lie group $\Gl^1(\HH_+)$. 
\end{proof}

\begin{Definition}[Determinant Line Bundle]
\mbox{}\\
We have a principle-$\Gl^1(\HH_+)$-bundle $\St$ over $\Gr$ and a one-dimensional, holomorphic representation of $\Gl^1(\HH_+)$, given by $\det: \Gl^1(\HH_+) \to \IC$. We define the \emph{determinant bundle} as the corresponding \emph{associated line bundle} 
\begin{equation*} DET = DET\bigl(\St, \Gr, \det \bigr). \end{equation*}
%for the principle-$\Gl^1$-bundle $\St$ over $\Gr$ is called the \emph{determinant bundle}.\\
%In fact, the line-bunde $DET$ is even \emph{holomorphic} (\cite{PreSe}, \S 7.7).
As $\St$ has infinitely-many connected components, indexed by $\IZ$, so does $DET$. The connected component over $\mathrm{St}^{(n)}(\HH)$ will be denoted by $DET_n$.
\end{Definition}

\noindent In less technical terms, ``associated line-bundle'' means the following: Consider the set $\mathrm{St}(\HH) \times \mathbb{C} $ written as
\begin{equation*}\lbrace (W , w, \lambda) \mid \lambda \in \mathbb{C}, W \in \mathrm{Gr}(\HH),\; \text{w admissible basis for W}\, \rbrace. \end{equation*}
We introduce the equivalence relation
\begin{equation*}
[W, w, \lambda] \sim [W', w', \lambda] :\iff W' = W,\; w' = w \circ L\; \text{and}\; \lambda' = \det(L)^{-1} \lambda\end{equation*}
\vspace{-3mm}i.e. \vspace{-1mm}\begin{equation}[W, w, \det(L) \lambda] \sim [W, w\circ L , \lambda],\; \text{for}\; L \in \Gl^1(\HH_+).\end{equation}

\noindent The set of equivalence classes $ [W, w, \lambda]$, together with the projection\\
$\pi:[W, w, \lambda] \to W = \im(w)$ (induced by that in $\St$) defines a holomorphic line bundle
 \begin{equation*}
\bigl( \mathrm{St} \times \mathbb{C}\bigr) \slash \, \Gl^1 =: DET \xrightarrow{\;\pi\; } \mathrm{Gr}(\mathcal{H}).
\end{equation*}

\noindent Further on, we will usually drop the first entry and write $[w, \lambda]$ for $[\im(w) , w, \lambda]$.

We might get a better intuition for this construction by thinking about the fibre over $W \in \mathrm{Gr}(\mathcal{H})$ in $DET$ as the one-dimensional complex space spanned by a formal expression 
\begin{equation}\label{adbasiswedge} \pi^{-1}(W) \ni [W, w, 1] = [w,1]\; \simeq \; w_0 \wedge w_1 \wedge w_2 \wedge w_3 \wedge\dots \end{equation} where $\lbrace w_j \rbrace_{j \geq 0}$ is a basis of $W$.
So, intuitively, the determinant bundle over $\mathrm{Gr}(\HH)$ contains precisely the information that we expect to be encoded in the fermionic Fock space.\\ 
However, so far we don't even have a linear structure on the ``state-space'', except of course for the $\IC^1$-structure on the individual fibres. But there's no meaningful way of ``adding'' points from different fibers of the bundle.
Therefore, the idea is to consider \textit{sections} of the determinant line bundle which do naturally form an (infinite-dimensional) complex vector space $\Gamma(DET)$. Which section would then correspond to a state of the form \eqref{adbasiswedge} though? In analogy to the infinite wedge space construction, where we got a linear space by taking the algebraic dual of a Dirac sea class, we might think of the section that picks out the point $[w,1]$ over $W \in \mathrm{Gr}(\mathcal{H})$ and is zero everywhere else. But such a section is not even continuous and wouldn't do justice to the beautiful geometric structure we have so far. 

We can however take sections in the \textit{dual bundle} $DET^*$ and identify $w \in DET$ with the section $\xi_w \in \Gamma(DET^*)$ defined by $\xi_w ([z, \lambda]) := \lambda \, \det(w^* z)$.
Such a section is not only continuous, it is even \textit{holomorphic} with respect to the holomorphic structure induced by $\Gr$ (see also \cite{PreSe}, \S 10).\\ 

\begin{Remark}[Holomorphic sections]
\mbox{}\\
We denote by $DET^*$ the \textit{dual bundle} of the determinant-bundle $DET$ and by $\Gamma(DET^*)$ the space of holomorphic sections in $DET^*$.
A holomorphic section $\Psi$ of $DET^*$ is a holomorphic map $DET \rightarrow \IC$ which is linear in every fibre. This corresponds to a holomorphic map \begin{equation}
\psi: \mathrm{St} \rightarrow \IC\; \text{with} \; \psi(z \circ L) =\det(L)\cdot \psi(z), \; \forall L \in \Gl^1(\mathcal{H}_+),\end{equation} 
since then and only then is  $\Psi([z,\lambda]) := \lambda \cdot \psi(z)$ a well defined map $DET \rightarrow \IC$, as we see from $\Psi([z \circ L,\lambda / \det(L)]) = \frac{\lambda}{\det(L)}\,\psi(z \circ L) = \lambda\, \psi(z) = \Psi([z,\lambda])$.\\

\noindent Similarly, a holomorphic section $\Phi$ of $DET$ would correspond to a holomorphic map 
\begin{equation}\label{holsec2} \phi: \mathrm{St} \rightarrow \IC \; \text{with} \;  \phi(z \circ L) =  \det(L)^{-1} \, \phi(z)\, , \; \forall L \in \Gl^1(\mathcal{H}_+).\end{equation}
Due to the factor of $\det(L)^{-1}$ on the right-hand-side, points arbitrarily close in $\St$ can be mapped to points arbitrarily far in $\IC$, which contradicts the existence of a bounded derivative. Hence $DET$ doesn't allow \emph{any} holomorphic sections, except for the zero section.\\
\end{Remark}

\begin{Construction}[Embedding of $DET$ in $\Gamma(DET^*)$]
\mbox{}\\
Consider the map $\Phi: \mathrm{St} \times \mathrm{St} \rightarrow \IC$,\\
\[
\Phi(z,w) \; = \; \left\{
\begin{array}{ll}
\det(z^* w)
& \mbox{; for } \ind(P_+z) = \ind(P_+w) \\ \\
0
& \mbox{; else. }
\end{array}
\right. 
\]
This is well defined, since for $\ind(P_+z) = \ind(P_+w) = n$,
\begin{equation*}z^* w = z^*P^{(n)}_+ w + z^* P^{(n)}_{-} w =\end{equation*}
\vspace{-6.1mm}
\begin{align*}
\underbrace{(P^{(n)}_+z)^*}\limits_{\in Id + I_1(\HH_{\geq-n})}\;\underbrace{(P^{(n)}_+w)}\limits_{\in Id + I_1(\HH_{\geq-n})} +  \underbrace{(P^{(n)}_-z)^*}\limits_{\in I_2(\HH_{<-n}, \HH_{\geq-n})}\;\underbrace{(P^{(n)}_-w)}\limits_{\in I_2(\HH_{\geq-n}, \HH_{<-n})} \in Id + I_1(\Hplusn).
\end{align*}
Now, for any fixed $z \in \mathrm{St}(\HH)$, the map
\begin{equation*}w \longmapsto \Phi(z,w) =: \xi_z(w)\end{equation*}
is holomorphic with $\Phi(z, w \circ L) = \det(z^* w L) = \det(z^* w) \det(L)$ for $L \in \Gl^1(\mathcal{H}_+)$ 
and hence descends to a holomorphic section of $DET^*$, that will be denoted by the same symbol 
$\xi_z \in \Gamma(DET^*).$
This defines an \emph{anti-linear} map $\xi: DET \rightarrow \Gamma(DET^*)$ by
\begin{equation}\xi([z,\lambda]) := \bar{\lambda} \xi_z = \bar{\lambda} \Phi(z, \cdot). \end{equation}
Note that \begin{equation}\label{xicc} \xi_{z\circ L} = \overline{\det(L)} \, \xi_{z}, \; \forall L \in \Gl^1(\mathcal{H}_+) \end{equation}
thus $\xi$ has to be anti-linear, because  $[z \circ L, \lambda] = [z, \det(L) \, \lambda]$ in $DET$.
\item Deleting the zero-section from $DET$, we get an injection $\xi: DET^{\times} \rightarrow \Gamma(DET^*)$.
 \end{Construction}

\noindent This construction contains almost everything we need. We will construct the Fock space from the space of holomorphic sections of $DET^*$, with the Hermitian scalar product given by the determinant. But first, we need good coordinates to handle this.\\

\begin{Definition}[Pfl\"{u}cker Coordinates]\label{DefPflucker}
\mbox{}\\
Let $\lbrace (e_j)_{j\in\mathbb{Z}} \rbrace$ be the ONB of $\mathcal{H}$ chosen above.
\begin{enumerate}[i)]
\item We denote by $\csq$ the set of increasing sequences $S=(i_0, i_1, i_2. i_3, \dots) \in \IZ^{\IN_0}$\\
s.t. $i_{k+1} = i_k + 1$ for large enough k, i.e. the sequences $S \in \csq$ contain only finitely many negative indices and all but finitely many positive indices.
\item  For $S \in \csq$, we define the \emph{charge} c(S) to be the unique number $c \in \mathbb{Z}$ with $i_k = k-c$ for all k large enough.
\item  For $S=(i_0, i_1, i_2. i_3, \dots) \in \csq$ with $c(S) = n$ we define
\begin{equation*}H_S := \spn\bigl(\lbrace e_{i_k}\mid i_k \in S\rbrace\bigr) = \spn(e_{i_0}, e_{i_1}, e_{i_2}, \dots)\end{equation*} and
\begin{equation*}w_S: \Hplusn \rightarrow \mathcal{H}\,,\;  w_S(e_k) := e_{i_{(k+n)}}.\end{equation*}
This is an admissible basis for $H_S$ since by construction, it differs from the identity on $\Hplusn$ by a finite rank matrix only.
\item We denote by $\Psi_S$ the section $\xi_{w_S} \in \Gamma(DET^*)$, for $S \in \csq$.
\end{enumerate}
\end{Definition}
\newpage
\begin{Example}[Pfl\"{u}cker coordinates]
\mbox{}\\
\vspace{-5mm}\begin{itemize}
\item $S_0 = \mathbb{N} = (0, 1, 2, 3, \dots) \Rightarrow c(S_0) = 0$.
We call the corresponding state 
 \begin{equation*}\xi_{w_{S_0}} := \Psi_0 \approx e_0 \wedge e_1 \wedge e_2 \wedge e_3 \dots\end{equation*}
 the \emph{free vacuum state}.
\item   $S_1 = (-2, -1, 0, 1, 2, \dots) \Rightarrow c(S_1) = 2$.\\ 
This corresponds to the state $e_{-2} \wedge e_{-1} \wedge e_0 \wedge e_1 \wedge \dots$ where 2 negative energy states ($e_{-2}$ and $e_{-1}$) are occupied and all positive energy states $e_{j\geq0}$ are occupied.
\item  $S_2 = (-5, -1, 1, 3, 4, 5, \dots) \Rightarrow c(S_2)=0$.\\
This corresponds to the state $e_{-5} \wedge e_{-1} \wedge e_1 \wedge e_3 \wedge e_4 \dots$ with two negative energy states occupied ($e_{-5}$ and $e_{-1}$) and two ``holes'' in the positive spectrum ($e_0$ and $e_2$). Therefore, the net-charge of the state is zero.
\end{itemize}
\end{Example}

\noindent The role of the basis $\lbrace (e_j)_{j\in\mathbb{Z}} \rbrace$ can be understood as a choice of a complete set of one-particle states, used to characterize the particle content of the Fock space states ($\approx$ Dirac seas). Ideally, this choice is determined, or at least motivated, by physical properties e.g. as the spectral decomposition of a (or several commuting) self-adjoint operator(s) on $\mathcal{H}$.
Anyways, the sections $(\Psi_S)_{S\in \csq}$ will make a good basis for the Fock space.

\begin{Proposition}[Hermitian Form]\label{Hermitianform}
\mbox{}\\
Let $V \subseteq \Gamma(DET^*)$ be the complex vector space spanned by the sections\\
 $\lbrace \xi_z \in \Gamma(DET^*) \mid z = [(z,1)]\in DET\rbrace$. Then
\begin{equation} \langle \xi_z, \xi_w \rangle := \xi_z (w) := \Phi(z,w) \end{equation}
defines a Hermitian form on V, antilinear in the \emph{second} component.
This form is positive semi-definite.
The sections $\bigl\lbrace\Psi_S\bigr\rbrace_{S \in \csq}$ form a complete orthonormal set in $V$ w.r.to $\scpro$.
\end{Proposition}

\begin{proof} First we note that $\langle \xi_z, \xi_w \rangle = \Phi(z,w) = \det(z^*w)=\overline{\det(w^*z)} = \overline{\langle\xi_w,\xi_z\rangle}$.\\ 
Then we recall that $\Phi(\cdot,\cdot)$ is $\IC$-linear in the second entry and anti-linear in the first entry, but the mapping $w \mapsto \xi_w$ is $\IC$-anti-linear. In conclusion, $\langle\cdot,\cdot\rangle$ defines a sesquilinear form anti-linear in the second entry. It is easy to see that for $S,S' \in \csq$
\begin{equation*} \langle \Psi_S , \Psi_S' \rangle = \Psi_S(w_S') = \det(w^*_Sw_S')= \delta_{S S'}\end{equation*}
because $w^*_Sw_S = \mathds{1}_{\Hplusn}$, whereas $w^*_Sw_S'$ has non-trivial kernel if $S\neq S'$.
\item Now, in finite dimensions, if $A$ is a $n\times m$ matrix and $B$ a $m \times n$ matrix with $n \leq m$ then
\begin{equation*} \det(AB) = \sum\limits_{(i) = (1 \leq i_1\leq \ldots \leq i_n \leq m)} \det(AP_{(i)})\det(P_{(i)}B) \end{equation*}  
where $P_{(i)}$ denotes the projection onto $\spn(e_{i_1}, \ldots e_{i_n})$. As the Fredholm determinant of an operator is the limes of the determinants of its restriction to $n$-dimensional subspaces as $n \rightarrow \infty$, the above formula extends to the infinite-dimensional case yielding
\begin{align*}\langle \xi_z, \xi_w \rangle & = \Phi(z , w) = \det(z^*w)= \sum\limits_{S \in \csq,\, c(S) = d} \det((P_{\HH_S}z)^*)\det(P_{\HH_S}w)\\
& = \sum\limits_{S} \det((w^*_S z)^*)\det(w^*_S w) = \sum\limits_{S} \overline{\det((w^*_S z))}\det(w^*_S w)\\
&  = \sum\limits_{S\in\csq} \overline{\Psi_S(z)}\Psi_S(w) = \sum\limits_{S \in \csq} \langle \xi_z, \Psi_S \rangle\langle\Psi_S, \xi_w\rangle
\end{align*}
for $\ind(P_+z) = \ind(P_+w) = d$ and $0$ else.
This shows that $\lbrace \Psi \rbrace_{S \in \cs}$ is indeed a complete, orthonormal system in V.
Finally, for general $\Psi = \sum\limits_{\text{finite}} \alpha_n \xi_{z_n}\,\in V$ we compute
\begin{align*} \langle \Psi, \Psi \rangle & = \sum\limits_{n,m}\alpha_n\,\overline{\alpha_m}\,\langle \xi_{z_n}, \xi_{z_m} \rangle = \sum\limits_{n,m} \sum\limits_{S \in \cs}\alpha_n\,\overline{\alpha_m}\,\langle \xi_{z_n}, \Psi_S \rangle \langle \Psi_S, \xi_{z_m} \rangle\\
& = \sum\limits_{S \in \cs}\;\Bigr(\sum\limits_{n} \alpha_n \langle \xi_{z_n}, \Psi_S \rangle \Bigr)\; \Bigl(\sum\limits_{m} \overline{\alpha_m} \langle \Psi_S, \xi_{z_m} \rangle \Bigr)\\
& = \sum\limits_{S \in \cs}\;\Bigr\lvert\sum\limits_{n} \alpha_n \langle \xi_{z_n}, \Psi_S \rangle \Bigr\rvert ^2\, \geq 0
\end{align*}
Thus, the Hermitian form is positive semi-definite. This finishes the proof.\\
\end{proof}

\begin{Definition}[Fermionic Fock Space]\label{DefFgeom}
\mbox{}\\
Let $V_0:= \lbrace v \in V \mid \langle v,v \rangle = 0 \rbrace$ be the null-space of $(V, \scpro)$.
\item We define the \emph{fermionic Fock space} $\FF= \FF_{\rm geom}$ to be the completion of $V \slash V_0$ w.r.to $\scpro$.\\
$\FF_{\rm geom}$ is an infinite-dimensional, complex, separable Hilbert space.
\item  It can be written as the direct sum $\FF = \bigoplus\limits_{c \in \IZ} \FF^{(c)}$ of $\IZ$ ``charge-sectors'' built from holomorphic sections with support in the connected component of $DET^*$ over $\mathrm{Gr}^{c}(\HH)$.   
\item The sections $\lbrace \Psi_S \rbrace_ {S \in \csq}$ are a Fock basis of $\FF$.
They define the \emph{Pfl\"{u}cker coordinates}
\begin{equation}\begin{split} &\FF \longrightarrow \ell^2,\\&  \xi_z \longmapsto (\xi_z(w_S))_{S \in \csq} = (\langle \xi_z, \Psi_S \rangle)_{S \in \csq}.
\end{split}\end{equation}
\end{Definition}
\noindent In general, we write $\FF$ for ``the'' fermionic Fock space and $\FF_{\rm geom}$, if we want to emphasize that it is the Fock space from Def. \ref{DefFgeom}, obtained from the geometric construction, as opposed to different constructions presented before. 
  
%\begin{Remark}($V$ is dense in $\Gamma(DET^*)$)\\
%It can be shown that $V$ is dense in $\Gamma(DET^*)$ w.r.to the topology of uniform convergence on compact subsets of $\mathrm{Gr}(\HH)$.  (\cite{PreSe}, Chapter 10)
%\end{Remark}

\begin{Definition}[Pfl\"{u}cker Embedding]\label{Def:Pfluckerembedding}
\mbox{}\\
We have a natural embedding of $\mathrm{Gr}(\mathcal{H})$ into the projective Fock space $\mathbb{P}(\FF_{\rm geom})$, given by
\begin{equation}\begin{split}&\Gr \longrightarrow \mathbb{P}(\FF_{\rm geom}),\\
& W \longmapsto \IC\cdot \xi_w = \IC \cdot \xi( [w,1]) \end{split}\end{equation}
where w is an admissible basis for W, called the \emph{Pfl\"{u}cker embedding.}\\
This gives precise meaning to our often employed intuition that polarizations correspond to projective, decomposable states of infinitely many fermions.   
\end{Definition}

\begin{Proposition}[Action of $\GresO$ on $DET_0$]\label{GresDET}
\mbox{}\\
$\GresO$ has a natural action on $DET$ defined by
\begin{equation}
\begin{split}
\mu_0:\; &\GresO \times DET_0 \longrightarrow DET_0, \\
&([U , R] , [W, w, \lambda] ) \longmapsto [UW ,U w R^{-1}, \lambda].
\end{split}
\end{equation}
\end{Proposition}

\begin{proof} We have to show that the action $\mu_0$ is well-defined.
\begin{enumerate}[i)]
\item $DET_0$ \textit{is closed under the action of $\GresO$}:\\ We know that $W \in \mathrm{Gr}(\mathcal{H}) \Rightarrow UW \in \mathrm{Gr}(\mathcal{H})$ for $U \in \Gl_{\rm res}$. We still have to show:\\
w admissible basis for $W \Rightarrow UwR^{-1}$ admissible basis for $UW$.\\
Obviously, $\im(UwR^{-1} ) = UW$. Furthermore:
\begin{align*}P_+ U w R^{-1} &= P_+ U P_+ w R^{-1} + P_+ U P_- w R^{-1}\\
&= P_+ U P_+ P_+ w R^{-1} + P_+ U P_- P_- w R^{-1}\\
&= \underbrace{P_+ U P_+R^{-1}}\limits_{\in Id + I_1(\mathcal{H}_+)} \;  \underbrace{R P_+ w R^{-1} }\limits_{\in Id + I_1(\mathcal{H}_+)} + \underbrace{P_+ U P_-}\limits_{\in I_2(\mathcal{H}_+)} \; \underbrace{P_- w R^{-1}}\limits_{\in I_2(\mathcal{H}_+)}\\
&\in Id + I_1(\mathcal{H}_+)
\end{align*}
since the product of two Hilbert-Schmidt operators is trace-class.\\ 
Thus, $UwR^{-1}$ is an admissible basis for $UW$.

\item \textit{The definition is independent of the representative of $[W,w,\lambda]$}:\\
Let $(W,w,\lambda) \sim (W,w',\lambda')=(W, w\circ L, \lambda/ \det(L) )$.\\ 
We need $(UW, UwR^{-1}, \lambda) \sim (UW, UwLR^{-1}, \lambda/\det(L))$, for $U$ and $R$ as above.\\
This is true because $UwLR^{-1} = UwR^{-1} (RLR^{-1}) = UwR^{-1} L'$, with \\
$L'= RLR^{-1} \in \Gl^1(\mathcal{H_+}) \; \text{and}\; \det(L') = \det(RLR^{-1}) = \det(L)$.

\item  \textit{The definition is independent of the representative of $[(U,R)] \in \GresO$}:\\
Let $[(U,R)] = [(U, R')] \in \GresO$. This means $\det(RR'^{-1}) = 1.$ Setting $L:= RR'^{-1} \in \GL$, we find
$UwR'^{-1} = UwR^{-1}RR'^{-1} = UwR^{-1}L$, with $\det(L) = 1$ and thus $(UW, UwR^{-1}, \lambda) \sim (UW, UwR'^{-1}, \lambda)$. 
\end{enumerate}
\end{proof}

\begin{Construction}[Arbitrary Charges and the complete $\Gres$]
\mbox{}\\
So far, we have defined the action of $\GresO$ on $DET_0$ and in fact, this is all we need. Extending the action to arbitrary ``charges" is somehow tedious, yet pretty much straightforward. Again, we use the ``shift''-operator $\sigma$ and the structure of $\Gres$ as a semi-direct product $\IZ \ltimes \GresO$ with $\IZ$ generated by the action of $\tilde\sigma$ (see \S 4.1.2).
\item We define a $\IZ-$action on $DET$ by \begin{equation}\vartheta(n)([W , w, \lambda]) := [\sigma^n(W) , \sigma^n w \sigma^{-n} , \lambda]. \end{equation}
Note that $\vartheta(n)$ maps $DET_d$ into $DET_{d-n}$ for any $d \in \IZ$..
\item Now we extend the action $\mu_0$ defined above to an action
$\mu: \IZ \ltimes \GresO \, \times DET \rightarrow DET$ by
%\begin{equation}\begin{split} \mu\lbrace (n , [A , R]) \rbrace &:\,  DET_d \rightarrow DET_{d-n},\\
% = \; &\vartheta(-d) \circ \mu_0( \tilde{\sigma}^{d-n}([A , R])) \circ \vartheta(d+n) 
%\end{split}\end{equation}
\begin{equation}\mu\lbrace (n , [A , R]) \rbrace\bigl\lvert_{DET_d \rightarrow DET_{d-n}} := \; \vartheta(n-d) \circ \mu_0( \tilde{\sigma}^{d-n}([A , R])) \circ \vartheta(d). 
\end{equation}
In particular, $\GresO$ acts on $DET_d$ by 
\begin{equation} \mu\lbrace[A , R]\rbrace := \vartheta(-d) \circ \mu_0(\tilde\sigma^d [A,R]) \circ \vartheta(d). \end{equation}
This is almost the same as the action of $\GresO$ on $\mathrm{St}^{(d)}(\HH)$ descending to $DET$, only with $A$ acting from the left and $(R_{\sigma^d})^{-1}$ (instead of $R^{-1}$) acting from the right.
\item Admittedly, this is not very pretty. Therefore, we will be so daring and drop the $\IZ-$indices for the remainder of this chapter and use the simple notation for $\GresO$ acting on the ``zero-charge sector'', while still stating the results for the whole transformation group and the whole line bundle. If needed, the full expressions for arbitrary charges can be worked out in detail using the scheme we've just outlined.
\end{Construction}

\noindent Finally, we can tell how the central extension $\Gres(\HH)$ and $\Ures(\HH)$ defined in Chapter 4 act on the fermionic Fock space and everything fits together nicely.
\begin{Proposition}[Action of $\Gres$ on $\FF_{\rm geom}$]
\mbox{}\\
The action of $\Gres(\HH)$ on $\FF$ is given by the adjoint action of $\Gres$ on $DET$ as defined in Prop. \ref{GresDET}. This action restricts to a unitary representation of $\Ures(\HH)$ on the Fock space.\\
Explicitely, the action of $\Gres(\HH)$ on $\FF$ is defined by
\begin{equation}\begin{split} \mu^*:\;& \Gres \times \FF \longrightarrow \FF, \\
&([U,R] , \xi_z) \longmapsto \mu^*_{[U,R]}\xi_z = \xi_z \circ \mu^{-1}_{[U,R]}
\end{split}\end{equation}
i.e. for $w \in DET$:
\begin{align*} \mu^*_{[U,R]}\xi_z(w) & =  \xi_z( \mu^{-1}_{[U,R]} (w)) = \xi_z (U^{-1}wR)\\
&=\det(z^* U^{-1}wR) = \det(R z^* U^{-1}wR R^{-1})\\
&=\det((U^{-1*}zR^*)^*w) =: \xi_{z'}(w)
\end{align*}
with  $z' = (U^*)^{-1}zR^*$. In the unitary case, where $[(U,R)] \in \Ures(\HH) \subset \Gres(\HH)$, i.e.
$U \in \Ures(\HH)$ and $R \in \cu(\mathcal{H}_+)$, this simplifies to
\begin{equation} \xi_z \mapsto \mu^*_{[U,R]}(\xi_z)= \xi_{z'},\;\text{with} \; z' = U z R^*.\end{equation}
Clearly, this defines a representation of $\Ures(\HH)$ which is indeed unitary, because
\begin{equation*}\bigl\langle \mu^*_{[U,R]} \xi_z, \mu^*_{[U,R]} \xi_w \bigr\rangle = \det\bigl( (UzR^*)^*(UwR^*)\bigr) = \det(Rz^*wR^{-1}) = \det(z^*w) = \bigl\langle \xi_z , \xi_w \bigr\rangle. \end{equation*}
\end{Proposition}

\subsubsection{Generalization to arbitrary polarizations}
\label{subsubsection:generalization to arbitrary polarizations}
The geometric construction can be immediately generalized to arbitrary polarization, the only challenge being that it complicates our notation.

\noindent Let $\pol(\HH)$ be the set of all polarizations of $\HH$, i.e. of all closed, infinite-dimensional subspaces with infinite-dimensional orthogonal complements. In \S 2.1 we introduced the equivalence relation
\begin{equation*} V \approx W \iff  P_V - P_W \in I_2(\HH), \end{equation*}
defining a partition of $\pol(\HH)$ into disjoint polarization classes $C \in \pol(\HH)\slash_\approx$.\\

\noindent Given a fixed polarization class $C \in \pol(\HH)\slash_{\approx}$, we choose a subspace $W \in C$ corresponding to the splitting $\HH= W \oplus W^{\perp}$. This polarization distinguishes $I_2(W;W^{\perp})$ as a model-space for the manifold structure on $C=[W]_\approx$, analogous to that of $\Gr$ for $\HH_+$. The resulting Grassmann manifold $\Grass^0(\HH, W)$ is a homogeneous space for
\begin{align*} \Ur(\HH; [W], [W]) =\, & \lbrace U \in \cu(\HH) \mid UW \approx W \rbrace \\
= \,& \lbrace U \in \cu(\HH) \mid [P_W - P_{W^{\perp}} , U ] \in I_2(\HH) \rbrace.
\end{align*}  
This group inherits a Lie group structure from its representation on $W \oplus W^{\perp}$, with the topology induced by the norm $\lVert U \rVert_{\epsilon_W} = \lVert U \rVert + \lVert \,[P_W - P_{W^{\perp}} , U ] \,\rVert_2$.\\
Furthermore, we define in analogy to Def. \ref{DefSt} an admissible basis for $V \in \Grass^0(\HH, W)$ as a bounded isomorphism $w: W \to V$ with $P_W\,w \in Id + I_1(W)$ (and generalize to arbitrary charges). That is, the admissible bases of subspaces in $[W]_\approx$ are those compatible with the  identity map on $W$. The set of all such admissible bases is the Stiefel manifold $\mathrm{St}(\HH, W)$. In complete analogy to the construction above, we can now define the determinant bundle for $W$ as the associated line bundle 
\begin{equation*}DET_W = DET_W\bigl(\mathrm{St}(\HH, W),\Grass(\HH, W), \det: \GL^1(W) \to \IC \bigr) \end{equation*} 
and the Fock space $\FF_{\rm geom}^W$ as the space of holomorphic sections in $\mathrm{DET^*_W}$, equipped with the scalar product  $\langle \xi_z, \xi_w \rangle = \det(z^*w)$.

\newpage
\section{Equivalence of Fock Space Constructions}
The infinite wedge space construction of Deckert, D\"{u}rr, Merkl and Schottenloher is closely related to the geometric construction of the Fock space as the space of holomorphic section in the dual of the determinant bundle. The infinite wedge spaces, however, remain rather unimpressed by the underlying geometry and highlight the algebraic relations that seem more essential to the physical problems.

%In this section, we prove that the two constructions lead to equivalent descriptions of the fermionic Fock space. For simplicity, we focus mostly on the polarization class of $[\HH_+]$, i.e. on $\Gr$, but the results can be easily generalized to arbitrary polarizations.\\
The reader might have noticed that the admissible bases are for the geometric construction, what the seas are for the infinite wedge spaces. However, it is important to note one essential difference. For the construction of the infinite wedge space, we start with a polarization class and have the freedom to choose a suitable Dirac sea class. In the geometric construction, on the other hand, we have to fix a polarization within the polarization class (e.g. $\HH_+ \in \Gr$); this polarization then comes with a preferred choice for the class of admissible bases, induced by the identity map on that particular subspace.\footnote{Actually, in both cases, the choice determines only the charge-0 sector of the Fock space. Construction of the ``full'' Fock space requires a choice for every charge-sector, i.e. for every connected component of the polarization class.} By this means, the geometric Fock space genuinely contains a distinguished state, whereas all states in the infinite wedge spaces come with equal rights. In other words, the geometric construction, as opposed to the infinite wedge space construction, naturally yields a Fock space \textit{with vacuum}. (This is just the usual nomenclature, the preferred state may or may not actually play the role of a vacuum state in the physical description).
In particular, there are as many (different, but equivalent) geometric Fock spaces over a fixed polarization class, as there are polarizations, i.e. vacua. Two infinite wedge spaces $\FF$ and $\FF'$ over the same polarization class $C$ differ by right-action with an operator $R \in \cu(\ell)$, with $\FF = \FF' \iff R \in \cu^1(\ell)$ (see Cor. \ref{Cor:equivalence of infinite wedge spaces}).

So much for the abstract discussion. Let's now formulate the relationship between Dirac sea classes and admissible bases in a precise way.

\begin{Lemma}[Dirac Seas vs. Admissible Bases]
%Let $(e_k)_{k \in \IZ}$ be a polarized basis of $\HH$ as above.
\item Any isometric Dirac sea $\Phi_0 \in \iseas$ with $\im(\Phi_0) = \Hplusn$ gives rise to an\\ 
isomorphism
\begin{equation*}\bigl(\iseas\slash\sim\bigr)\,\ni S(\Phi_0) \xrightarrow{\; \cong\; } \mathrm{St}^{(n)}(\HH)\end{equation*}
between its Dirac sea class and the n-th connected component of the Stiefel manifold. 
\item More generally, let $\Phi_0 \in \iseas$ with $\im(\Phi_0) = W$. Then $\Phi_0$ induces an isomorphism between its Dirac sea class and the Stiefel manifold
 \begin{equation}\mathrm{St}^0(\HH,W) = \lbrace w: W \to \HH \mid P_W w \in Id + I_1(W)\rbrace \end{equation} corresponding to the polarization $W \in \pol(\HH)$.  
\end{Lemma}
\begin{proof} Given a $\Phi_0: \ell \rightarrow \Hplusn \subset \mathcal{H}$ as above, consider the map
\begin{equation}\label{StSea}\mathrm{St}^{(n)}(\HH) \rightarrow S(\Phi_0); \; w \mapsto w \circ \Phi_0,\end{equation}

\noindent This is a well-defined map into $S(\Phi_0)$: If $w$ is an admissible basis, $w \circ \Phi_0$ is a map from $\ell$ to $\mathcal{H}$ and indeed $w \circ \Phi_0 \sim \Phi_0$ in $\iseas$, because
\vspace{-4mm}
\begin{align*}P_+w = \Phi_0 \Phi_0^* w&\; \text{has a determinant} \\
\Rightarrow\, \Phi_0^* &\,w \, \Phi_0 : \ell \rightarrow \ell \; \text{has a determinant}\\
&\Rightarrow\,  w \circ \Phi_0 \sim \Phi_0\; \text{in} \; \iseas. \end{align*} 
It remains to show that \eqref{StSea} is bijective: Let $\Phi \sim \Phi_0 \in \seas$. Define \\
$w:= \Phi \circ \Phi^*_0\, \bigl\vert_{\Hplusn}$. $w$ is an isomorphism $\Hplusn \rightarrow \HH$ with $\im(w) =  \im(\Phi) \in \mathrm{Gr}(\mathcal{H})$ and clearly 
$w \circ \Phi_0 = \Phi$. Finally, $w$ is indeed an admissible basis i.e. in $\mathrm{St}^{(n)}(\HH)$, because
$P_+w = \Phi_0\Phi_0^* \,\Phi\Phi^*_0$ has a determinant$\iff \Phi_0^*\Phi$ has a determinant$\iff \Phi \sim \Phi_0$.\\
The general case works analogously. 
\end{proof}

\begin{Theorem}[Anti-Isomorphism of Fock spaces]
\mbox{}\\
Let $\Phi_0 \in \iseas$ with $\im(\Phi_0) = \Hplusn$. The infinite wedge space $\FF_{S(\Phi_0)}$ is \emph{anti-unitary} equivalent to $\FF^{(n)}_{\rm geom}$, the charge-n-sector of the fermionic Fock space constructed from $\Gamma(DET^*_n)$.
\end{Theorem}
\begin{proof} Consider the map $\mathfrak{f}:\; \FF^{(n)}_{\rm geom} \longrightarrow \FF_{S(\Phi_0)}$,  defined by \begin{equation} \xi_w \longmapsto {\textstyle\bigwedge} \bigl( w \circ \Phi_0 \bigr), \; \text{ for} \; w \in \mathrm{St}^{(n)} \end{equation}
and \textit{anti-linear} extension.
This is well defined, since for $L \in \GL$, we have\\ 
$\xi_{z\circ L} = \overline{\det(L)} \, \xi_{z}$ by \eqref{xicc}, and 
\begin{equation*}\bwdge\bigl(w \circ L \circ \Phi_0\bigr) = \bwdge\bigl(w \circ \Phi_0 \Phi_0^* L \circ \Phi_0\bigr) = \det(\Phi_0^* L \circ \Phi_0)\, \bwdge\bigl(w \circ \Phi_0\bigr) = \det(L) \, \bwdge\bigl(w \circ \Phi_0\bigr).\end{equation*}
It is also an (anti-) isometry, because
\begin{equation*}
\bigl\langle \bwdge \bigl(z \circ \Phi_0\bigr)  , \bwdge\bigl(w \circ \Phi_0\bigr) \bigr \rangle_{\FF_{S(\Phi_0)}} =\, \det(\Phi_0^* z^*w \Phi_0) = \det(z^* w) = \bigl\langle \xi_z , \xi_w \bigr\rangle_{\FF_{\rm geom}}. \end{equation*}
By construction of the infinite wedge space and by the previous Lemma, Dirac seas of the form $w \circ \Phi_0,\;  w \in \mathrm{St}^{(n)}(\HH)$ span the entire Fock space $\FF_{S(\Phi_0)}$.
We conclude that the two Fock spaces are anti-unitary equivalent.
\end{proof}

\begin{Corollary}[Action of $\widetilde{\cu}^0_{\rm res}$ on Infinite Wedge Spaces]\label{Ures on wedge spaces}
\item Let $\Phi_0$ and $\FF_{S(\Phi_0)}$ as above. We have a natural representation of $\Ures^0(\HH)$ on $\FF_{S(\Phi_0)}$ given by
\begin{equation}\label{Uresiso} \wedgeu: \Ures^0(\HH) \ni [U, R] \longmapsto \mathcal{L}_U \mathcal{R}_{R_{\Phi_0}^*} \end{equation}
with $R_{\Phi_0} := \Phi_0^*\,R\,\Phi_0 \in \cu(\ell)$,
i.e. $\bwdge \Phi \xrightarrow{[U,R]} \bwdge\bigl(U\, \Phi \, \Phi_0^*R^*\Phi_0\bigr)$.

\item Of course, this is the representation induced by the that on $\FF$ via the anti-automorphism $\mathfrak{f}$:

\begin{equation}
\begin{xy}
  \xymatrix{
     \FF^{(n)}_{\rm geom} \ar[d]^{\mathfrak{f}} \ar[r]^{\mu^*_{[U,R]}} \ & \FF^{(n)}_{\rm geom} \ar[d]^ {\mathfrak{f}} \\
     \FF_{S(\Phi_0)} \ar[r]^ {\mathcal{L}_U \mathcal{R}_{R_{\Phi_0}^*}}         & \FF_{S(\Phi_0)}
  }
\end{xy}
 \end{equation}  
\end{Corollary}
\begin{proof} It suffices to note that 
\begin{align*}\mathfrak{f} \circ \mu^*_{[U,R]} \, (\xi_w) =  \, \mathfrak{f}\,(\xi_{(U w R^*)}) = \bwdge\bigl(U w R^* \circ \Phi_0\bigr)
= \bwdge\bigl(U w \, \Phi_0\, \Phi_0^*\, R^* \Phi_0\bigr) = \mathcal{L}_U \mathcal{R}_{R_{\Phi_0}^*} \circ \mathfrak{f} \, (\xi_w).  
\end{align*}
\end{proof}

\section{CAR Representation on $\FF_{\rm geom}$}
In this section, we discuss the relationship between the geometric Fock space, or equivalently, the infinite wedge spaces, and the Fock spaces studied in Section 5.1 as representation-spaces of the abstract CAR-algebra, generated by creation and annihilation operators. We will provide a more rigorous formulation of our initial argument that a description in terms of ``particles'' and ``antiparticles'' is equivalent to formulations involving a ``Dirac sea'' of infinitely many particles.\\ 

\noindent Unfortunately, at this point we'll have to pay the price for bowing to different, contradictory conventions, when we used $\HH_+$ for the unperturbed Dirac sea, e.g. in \S 5.2, but also for the positive energy electron states, e.g. in \S 5.1. We will circumvent this difficulty, simply by setting $\FF_+ := \HH_-$ and $\FF_- := \mathcal{C} \HH_+$ in the construction of the Fock representation, i.e. in:
 \begin{equation*} \FF := \bigoplus\limits_{c}^{\infty} \FF^{(c)} ; \;\;\; \FF^{(c)} :=\bigoplus\limits_{n-m=c} \sideset{}{^n}\bigwedge\FF_+ \otimes \sideset{}{^m}\bigwedge \FF_- .\end{equation*}
Thus, compared to Section 5.1, the roles of $\HH_+$ and $\HH_-$ are interchanged. I hope the reader will excuse this little blemish.\\

%\noindent In this section, $\FF$ is strictly reserved for the Fock space just defined; the geometric Fock space will be denoted by $\FF_{\rm geom}$.

\begin{Construction}[Fock Space Isomorphism and Field Operator on Dirac Seas]
\item We fix a basis $\lbrace (e_k)_{k\in\mathbb{Z}} \rbrace \; \text{of} \; \mathcal{H}$ such that $(e_k)_{k \le 0}$ is an ONB of $\mathcal{H}_-$ and an $(e_k)_{k \geq 0}$ ONB of $\mathcal{H}_+$.
Let $\csq$ be the set of sequences as in Def. \ref{DefPflucker}. 
It consists of increasing sequences $S=(i_0, i_1, i_2, \ldots)$, containing only finitely many negative integers and all but finitely many positive integers. Recall that by Prop. \ref{Hermitianform}, the holomorphic sections $\lbrace\Psi_S\rbrace_{s \in \csq}$ defined in Def. \ref{DefPflucker}, iii), form an orthonormal basis of the Fock space $\FF_{\rm geom}$.\\

\noindent We define an isomorphism between $\FF_{\rm geom}$ and $\FF = \FF(\HH_+,\HH_-)$ by
\begin{align}\label{FgeomtoFock} \Psi_S \longmapsto \bigl(e_{i_0}\wedge \ldots \wedge e_{i_{n-1}}\bigr) \otimes\bigl(\mathcal{C}e_{j_0}\wedge\ldots\wedge\mathcal{C}e_{j_{m-1}}\bigr) \in \sideset{}{^n}\bigwedge\FF_+ \otimes \sideset{}{^m}\bigwedge \FF_-
\end{align}
for $S \cap \IZ^- = \lbrace i_0 < i_1 < \ldots < i_{n-1}\rbrace$ and $ \IN \setminus S = \lbrace j_0 < j_1<\ldots<j_{m-1}\rbrace$.\\

\noindent In particular, for $S=(0,1,2, \ldots), \Psi_S = \Psi_0$ is mapped to the vacuum state in $\FF$.\\ 

In other words, the states indexed by the negative integers in $S$ are mapped to ``electron states'' in $\FF_+$ and the states indexed by the positive integers \textit{missing} in $S$ (the ``holes'') are charge-conjugated and mapped to ``positron states'' in $\FF_-$.\\
\noindent Obviously, the assignment is 1--to--1 and isometric. It also preserves \textit{charge-preserving}, in the sense that states $\Psi_S$ with $c(S) = c$, spanning the charge-c-sector of $\FF_{\rm geom}$, are mapped into $\FF^{(c)}$, the charge-c-sector of $\FF$.\\

\noindent We can also introduce the equivalent of creation- and annihilation operators on $\FF_{\rm geom}$.\\ 
For this, it is convenient to use the intuitive notation
\begin{equation*} \Psi_S = e_{i_0} \wedge e_{i_1}\wedge e_{i_2} \wedge e_{i_3}\wedge \ldots \, .\end{equation*}
for $S = (i_0, i_1, i_2,i_3 \ldots) \in \mathbb{S}$.
\newpage
\noindent Then, we can define the \textit{field operator} (or rather its hermitian conjugate) $\mathbf{\Psi}^*$ on $\FF_{\rm geom}$ by 

\begin{equation} \mathbf{\Psi}^*(e_k):\, e_{i_0} \wedge e_{i_1}\wedge e_{i_2} \wedge e_{i_3}\wedge \ldots \longmapsto   e_k\wedge e_{i_0} \wedge e_{i_1}\wedge e_{i_2} \wedge e_{i_3}\wedge \ldots \end{equation}

\vspace*{2mm}
\noindent That is, $ \mathbf{\Psi}^*(e_k)$ maps $\Psi_S$ with $S = (i_0, i_1, i_2,i_3 \ldots)$ to zero, if $k\in \IZ$ is contained in $S$, otherwise to $(-1)^{j}\, \Psi_{S'}$ with $S'= (i_0, i_1, \ldots i_{j-1}, k, i_{j}, \ldots) \in \csq$.
By linear extension in the argument of $\mathbf{\Psi}^*(\cdot)$, we get a linear map $\mathbf{\Psi}^*: \HH \rightarrow \mathcal{B}(\FF_{\rm geom})$.\\

\noindent It is easy to see that under the isomorphism defined above, this field operator acts just  as the usual field operator \eqref{fieldoperator}, defined in terms of creation- and annihilation operators.\\
The same is true for the formal adjoint $\mathbf{\Psi}: \HH \rightarrow \mathcal{B}(\FF_{\rm geom})$, acting as

\[
 \mathbf{\Psi}(e_k) \; e_{i_0} \wedge e_{i_1}\wedge e_{i_2} \wedge e_{i_3}\wedge \ldots \; = 
 \; \left\{
\begin{array}{ll}
 (-1)^j\,e_{i_0} \wedge e_{i_1} \ldots \wedge e_{i_{j_-1}} \wedge \cancel{e_{i_j}}\wedge e_{i_{j+1}}\wedge \ldots
& \mbox{; if } k = i_j \\ \\
0
& \mbox{; if } k \notin S
\end{array}
\right. 
\]

\noindent In particular, $ \mathbf{\Psi}$ satisfies the canonical anti-commutation relations \begin{equation*}\lbrace  \mathbf{\Psi}(f),  \mathbf{\Psi}^*(g) \rbrace = \langle f, g \rangle_{\HH}\cdot \mathrm{Id}, \; \forall f,g \in \HH \end{equation*} and therefore defines a representation of the CAR-algebra on the Fock space $\FF_{\rm geom}$ which is equivalent  to the Fock representation on $\FF$.
\end{Construction}

\chapter{Time-Varying Fock spaces}
We have seen that a unitary operator on $\HH$ can be implemented on the fermionic Fock space if and only if it satisfies the Shale-Stinespring condition $[\epsilon, U] \in I_2(\HH)$ i.e. if and only if it belongs to $\Ur(\HH)$. In the previous chapter, this result appeared in different disguises depending on the construction of the Fock space, but there was no way around the fact itself. The dramatic conclusion seems to be that there is no well-defined time evolution in QED. By Ruijsenaar's theorem \ref{Ruij}, the Dirac-time evolution for a field $\sa$ will not be in $\Ur(\HH)$, unless the spatial component $\underline\sa$ of the vector potential vanishes. Physically, the reason is \textit{infinite particle creation}. Mathematically, this is reflected in the fact that a unitary transformation that doesn't satisfy the Shale-Stinespring condition, will leave the polarization class $[\HH_-] = \Gr$, over which the standard Fock space is constructed (here, we use again $\HH = \HH_+ \oplus \HH_-$ as the spectral decomposition and $\HH_-$ plays the role of the Dirac sea).\\ 

In the light of these results, Deckert et.al. concluded that the best we can do is to realize the time evolution as unitary transformations between \textit{different} Fock spaces over varying polarization classes. The first analysis along those lines was probably done by Fierz and Scharf in \cite{FS}, who studied time-dependent Fock space representations for QED in external fields. A similar construction, realized in more geometrical notions, is also known as the Fock space bundle in the mathematical literature (see for example  \cite{CaMiMu}).\\

Note that in this work we are not focusing on finding the weakest regularity condition on the external vector potentials. For simplicity, we always assume smooth fields with compact support in space and time, but in fact most of the results can be generalized a a bigger class of interactions, where suitable fall-off properties for the fields are assumed. The reader may consult the cited publications for notes on those possible generalizations. 

%\noindent\textbf{Remark:} In this work we do not focus on finding the weakest regularity condition on the external vector potentials. For simplicity, we always assume smooth fields with compact support in space and time. The reader may consult the cited publications for instruction about possible generalizations of the results to bigger classes of interactions. 
\newpage
\section{Identification of Polarization Classes}

Given an external field in a form of  4-vector potential 
\begin{equation*}\sa = (\sa_{\mu})_{\mu=0,1,2,3} = (\sa_0,- \underline{\sa}) \in C_c^{\infty} (\mathbb{R}^4, \mathbb{R}^4), \end{equation*} let $U^\sa = U^\sa(t,t'), \; t,t' \in \mathbb{R}$ be the unitary Dirac time evolution determined by the Hamiltonian \begin{equation}\label{HamiltonianA} H^\sa = D_0 + e\sum\limits_{\mu=0}^3\alpha^\mu \sa_\mu.\end{equation}

\noindent The question is the following: How can we determine polarization classes $C(t)$ such that
\begin{equation} U^\sa(t_1,t_0) \in \Ur\bigl(\HH, C(t_0) ; \HH, C(t_1)\bigr)\end{equation}
for all $t_1, t_0 \in \IR$ ?
The answer is a generalization of Ruijsenaars theorem and was provided in a very neat, systematic way by Deckert, D\"urr, Merkl and Schottenloher in \cite{DeDueMeScho}. We give a brief summary of their results.\\

%\newpage
\noindent The unitary time evolution $U^\sa(t_1,t_0), \; t_0,t_1 \in \IR$ satisfies the equation
\begin{equation}\label{EOMforU} \frac{\partial}{\partial t} U^\sa(t,t_0) = -i H^\sa U^\sa(t,t_0) = -i (D_0 + V^{\sa}(t) \bigr)\,  U^\sa(t,t_0) \end{equation}

%\noindent Recall that the Dirac equation has the form
%\begin{align}\label{eqn:Dirac}
%i\partial_t
%\Psi(t)&=H^{\sa}\psi(t) =\bigl(D_0 + V^{\sa}(t) \bigr)\Psi(t)
%\end{align}
\noindent with the interaction-term
\begin{align}
V^\sa& = e\sum_{\mu=0}^3\alpha^\mu A_\mu.
\end{align}
\noindent In momentum representation, i.e. after taking the Fourier transform of the equation, $D_0$ acts as a multiplication operator by the energy $E(p) = \sqrt{\lvert\underline{p}\rvert^2 + m^2}$ and the interaction potential takes the form
\begin{align}
V^A(p,q)&= e\sum_{\mu=0}^3\alpha^\mu \widehat{A}_\mu,
\end{align}
where the $\widehat A_\mu, \, \mu=0,1,2,3,$ now act as convolution operators, i.e.
\begin{equation}\label{eqn:fourier transform A}
(\widehat A_\mu\psi)(p)=
\int_{\IR^3}\widehat A_\mu(p-q)\psi(q)\,dq,\quad p\in\IR^3,
\end{equation}
for $\psi\in\HH$ and $\widehat A_\mu$ the Fourier transform of $A_\mu$.\\

\noindent From \eqref{EOMforU} we can deduce:

\begin{equation*} \frac{\partial}{\partial t} \bigl[U^0(t_1,t) U^\sa(t,t_0)\Bigr] = -i U^0(t_1,t)\bigl[H^{\sa(t)} - H^0\bigr]  U^\sa(t,t_0) = -i U^0(t_1,t) V^{\sa}(t) U^\sa(t,t_0). \end{equation*}
\mbox{}\\
\noindent This is equivalent to the integral equation:
\newpage

\begin{equation}\label{DiracIntegralForm} 
U^\sa(t_1,t_0) = U^0(t_1,t_0) -i  \int\limits_{t_0}^{t_1} U^0(t_1,t) V^{\sa}(t) U^\sa(t,t_0) \mathrm{d}t.
\end{equation}

\noindent We can read this as a fixed-point equation for $U^\sa(t_1,t_0)$, leading to the iteration 
\begin{align} U^\sa(t_1,t_0) &= U^0(t_1,t_0) -i \int\limits_{t_0}^{t_1} U^0(t_1,t) V^{\sa}(t) U^0(t,t_0)\, \mathrm{d}t \\ 
&- \int\limits_{t_0}^{t_1} \int\limits_{t_0}^{t}U^0(t_1,t) V^{\sa}(t) U^0(t_1,t') V^{\sa}(t') U^0(t',t_0) \, \mathrm{d}t' \mathrm{d}t\;  + ... \end{align}
known as the \textit{Born series}. We need to control the \textit{odd} part  $\,U^\sa_{odd} = U^\sa_{+-} + \, U^\sa_{-+}$ of the time evolution that is responsible for ``pair creation'' and for the time evolution leaving the polarization class. The free Dirac time evolution $U^0(t_1,t_0)= e^{-iD_0(t_1-t_0)}$ is diagonal w.r.to the polarization $\HH=\HH_+ \oplus \HH_-$, so that the relevant contributions in the perturbation expansion arise from

 \begin{equation}\label{Zodd}\int U^0 [V^{\sa}_{+-} + V^{\sa}_{-+}] U^0. \end{equation} 
 
\noindent Now, for a time-independent electromagnetic potential $A \in C_c^\infty(\IR^3,\IR^4)$, Deckert et. al. define a bounded operator $Q^\sa$ with integral kernel (in momentum-space)
 
\begin{equation}\label{eqn:operator_Q}
 Q^A(p,q):=\frac{V^A_{-+}(p,q)-V^A_{+-}(p,q)}{(E(p)+E(q))}.
\end{equation}
It can be readily computed that \eqref{Zodd} can be written as
\begin{equation}\begin{split} & i \int\limits_{t_0}^{t_1} U^0(t_1,t) \bigl[V^{\sa}(t)_{+-} + V^{\sa}(t)_{-+}\bigr] U^0(t,t_0) \mathrm{d}t\\
=& \; \;  Q^\sa(t_1) U^0(t_1,t_0) -  U^0(t_1,t_0)Q^\sa(t_0) - \int\limits_{t_0}^{t_1} U^0(t_1,t) \dot{Q}^\sa(t) U^0(t,t_0) \mathrm{d}t \end{split}\end{equation}
with $Q^\sa(t) = Q^{\sa(t)}$ (\cite{DeDueMeScho}, Lemma III.6). Now the idea is to use the $Q$-operators to ``renormalize'' the time evolution in such a way that the Hilbert-Schmidt norm of the odd-part remains finite. Equivalently (in a sense that will be discussed later), we can use these operators to identify the polarization classes into which the (unaltered) Dirac time evolution is mapping. Indeed, since $Q^{\sa(t)}$ is skew-adjoint, $e^{Q^{\sa(t)}}$ is unitary and it can be shown that
\begin{equation} e^{-Q^\sa(t_1)} \, U^\sa(t_1,t_0)\, e^{Q^\sa(t_0)} \in \Ur(\HH) \end{equation}
for all $\sa \in C^\infty_c(\IR^4, \IR^4)$ and all $t_0,t_1 \in \IR$, i.e. the odd parts are Hilbert-Schmidt operators.

\noindent More generally, Deckert et. al. proof the following crucial theorem: 
\newpage

%Indeed, for $\sa = \sa(t) \in C_C^{\infty} (\mathbb{R}^4, \mathbb{R}^4)$, $\bigl(e^{Q^A(t)}\bigr)_{t \in \IR}$ is a renormalization in the sense of Definition \ref{Defrenormalization}.\\
%In addition, this renormalization has the eminently helpful property that the polarization classes
%\begin{align}
  %C(A):= [e^{Q^{A}}\HH_-] \bigl(= e^{Q^{A}}C(0)=\{e^{Q^{A}}V\;|\;V\in C(0)\} \bigr)
%\end{align} 
%can be shown to depend on the spatial part of the 4-vector potential only! This justifies the notation
%$C(A) = C(\underline{A})$ introduced above and provides a very lucid proof of Ruijsenaars Theorem (\ref{Ruij}).\\
%We will state the essential properties of this renormalization. The proofs are technically somehow involved and can be found in detail in \cite{DeDueMeScho}.
%\noindent\textbf{Notation}: Every external field $\sa \in C^\infty_c(\IR^4,\IR^4)$ has compact support in time, i.e. there exists $T\gg 0$ such that $\sa(t)$ vanishes outside the time-interval $[T , -T]$. In particular, $U^{\sa}(t , -r) = U^\sa(t, -T) e^{iD_0(r+T)}, \forall r > T$ and thus  $[U^{\sa}(t , -r) \HH_+]_{\approx_0} = [U^{\sa}(t , -T) \HH_+]_{\approx_0}, \forall r >T$.
\begin{Theorem}[Identification of Polarization Classes]\label{Thm:polarizationclasses}
\mbox{}\\
Let $\sa\in C^\infty_c(\IR^4,\IR^4)$. The operators $e^{Q^{\sa(t)}}$ have the following properties:
\begin{enumerate}
\item[i)] Setting $C[\sa(t)]:= [e^{Q^{\sa(t)}} \HH_-]$ it is true that 
\begin{equation} U^\sa(t_1,t_0) \in \cu^0_{\rm res}\bigl(\HH; C[\sa(t_0)], C[\sa(t_1)]\bigr)\end{equation}
for all $t_0, t_1 \in \IR$. 

\item[ii)] For two \emph{static} potentials $\sa = (\sa_0,- \underline{\sa})$ and  $\sa' = (\sa'_0,- \underline{\sa'})  \in C_c^{\infty} (\mathbb{R}^3, \mathbb{R}^4)$ we find
\begin{equation}\label{poldeterminedbyA} 
[e^{Q^\sa} \HH_-]_{\approx_0} = [e^{Q^{\sa'}} \HH_-]_{\approx_0} \iff \; \underline{\sa} = \underline{\sa'}. 
\end{equation}
\end{enumerate}
It follows that the polarization class is completely determined by the spatial component of the vector potential $A$ at any fixed time. 
\end{Theorem}
\noindent In more physical terms, the theorem says that the polarization classes $C[\sa(t)]$
\begin{itemize}
\item depend only on the magnetic part of the interaction.
\item depend on $\sa$ instantaneously in time and \textbf{not} on the history of the system.
\end{itemize}
\noindent This justifies the notation $C[\underline{\sa}(t)]$ for the respective polarization classes, or even $C[\underline{\sa}]$ for a \textit{static} field $\sa$. (Note that in our convention, $C[\underline{\sa}] = C[\sa = (0, -\underline{\sa})]$.)\\
\noindent We observe that Ruijsenaars Theorem \ref{Ruij} follows immediately as a special case of \eqref{poldeterminedbyA}.\\ Furthermore, the following important result now follows as a simple corollary:
 
\begin{Theorem}[Implementability of the S-matrix]\label{Liftbarkeit}
\mbox{}\\
Let $\sa\in {\mathcal C}^\infty_c(\IR^4,\IR^4)$ and  $U^\sa_I$ the corresponding unitary time evolution in the interaction picture (cf. \S 8.1.2). Then, the scattering matrix
%\begin{equation} S = \lim_{\substack{t\rightarrow \infty \\ t' \rightarrow -\infty}} e^{itD_0}\,U^\sa(t , t')\,e^{-it'D_0} \end{equation} 
\begin{equation} S := \lim_{\substack{t\rightarrow \infty }} e^{itD_0}U^\sa(t , -t)e^{itD_0} = \lim_{\substack{t\rightarrow \infty }} U^\sa_I(t , -t) \end{equation} 
is in $\Ur(\HH)$ and can be implemented as a unitary operator on the (standard) Fock space. 
 \end{Theorem}

 \begin{proof}Since the interaction potential has compact support and thus, in particular, compact support in time, we have $\sa(t) = 0$ and hence $e^{Q^{\sa(t)}} = \mathds{1}$ for all $\lvert t \rvert$ large enough. In particular, $[e^{{Q^\sa(t)}} \HH_-]_{\approx_0} = [\HH_-]_{\approx_0}$ for all $\lvert t \rvert$ large enough.
By the previous theorem it follows that 
\begin{equation*}U_I^\sa(T , -T) = U^0(0,T)U^\sa(T,-T)U^0(-T,0) \in \Ur^0(\HH; [\HH_-], [\HH_-]) = \cu^0_{\rm res}(\HH) \end{equation*} for all T large enough and hence $S \in \cu^0_{\rm res}(\HH) $.\\
 \end{proof}

\newpage 
\section{Second Quantization on Time-varying Fock Spaces}
We present the recipe for second quantization of the time evolution on time-varying Fock spaces according to Deckert et. al., as described in \cite{DeDueMeScho}. The main ingredient is Theorem \ref{unser abstrakter Shale Stinespring}, the abstract version of Shale-Stinespring.
%We will present the method mainly in the language of infinite wedge spaces, which is best suited for this task. However, it can be applied to the geometric constructions with only little modifications.\\

\noindent Let $\sa = (\sa_{\mu})_{\mu=0,1,2,3} = (\sa_0,- \underline{\sa}) \in C_c^{\infty} (\mathbb{R}^4, \mathbb{R}^4)$
be an external field and $U^\sa(t,t')$ the corresponding Dirac time evolution.\\

\begin{itemize}
\item Let $C(t) =  C[\sa(t)] \in \pol(\HH)\slash\approx_0$ be the polarization classes identified above. By \ref{unser abstrakter Shale Stinespring}, $C(t)$ is uniquely determined by the spatial component of the vector potential at that time and
\begin{equation*}
U^\sa(t,t') \in \cu^0_{\rm res}\bigl(\HH; C(t'), C(t)\bigr), \; \forall t,t' \in \IR.\end{equation*}
\item For every $t \in \IR$, choose a Dirac sea class $\cs(t) \in \ocean(C(t))\slash \sim$.\\
For the geometric construction, we choose polarizations $W(t) \in C(t)$, instead.\\
We should demand that $\cs(t)$, respectively $W(t)$ depends only on $\sa(t)$ or even just on $\underline{\sa}(t)$. In particular, it is reasonable to demand that we stay in the initial (``standard'') Fock space, whenever the interaction is turned off.

\item We construct a family of Fock spaces $(\FF_{\cs(t)})_t$ as infinite wedge spaces over the Dirac sea classes $\cs(t)$. Equivalently, we can use the geometric construction where admissible bases are those compatible with the identity map on $W(t)$. 

\item By Theorem \ref{unser abstrakter Shale Stinespring}, we can implement $U^\sa(t_1,t_0)$ as a unitary map between the Fock spaces $\FF_{\cs(t_0)}$ and $\FF_{\cs(t_1)}$. There exists $R \in \cu(\ell)$ with $U^\sa(t_1,t_0) \cs(t_0) R = \cs(t_1)$ and therefore \begin{equation}\lop{U^\sa(t_1,t_0)} \rop R: \FF_{\cs(t_0)}\rightarrow \FF_{\cs(t_1)}\end{equation}
is a unitary map between the Fock spaces $\FF_{\cs(t_0)}$ and $\FF_{\cs(t_1)}$.

\item By Lemma \ref{Lemma:Uniquenessuptoaphase}, two right operations implementing $U^\sa(t_1,t_0)$ differ by an operator in $\cu(\ell) \cap Id + I_1(\ell)$. The induced transformation between $\FF_{\cs(t_0)}$ and $\FF_{\cs(t_1)}$ can differ by the determinant of such an operator, i.e. by a complex phase $\in \cu(1)$.\\ 
However, \textit{transition probabilities are well-defined}: if we have an ``in-state'' $\Psi^{in} \in \FF_{\cs(t_0)}$ and an ``out-state'' $\Psi^{out} \in \FF_{\cs(t_1)}$, the transition probability \begin{equation}
\lvert \langle \Psi^{out} , \lop{U^\sa(t_1,t_0)} \rop R \,\Psi^{in} \rangle \rvert^2 \end{equation}
is independent of the choice of $R \in \cu(\ell)$.\\
\end{itemize}

\noindent This is the scheme for implementation of the unitary time evolution on time-varying Fock spaces in its most general form. The formal clarity of the results (that we might or might not have been able to convey), could seduce us into too much optimism. In fact, the construction, as it stands so far, is neither practicable nor physically meaningful. 
In practice, we would rather make a ``global'' choice of Fock spaces, preferably one that depends only on the external potential, locally in time. For example, starting with any sea $\Phi \in \ocean(\HH_-)$ we can set 
\begin{equation}\cs(t) := [e^{Q^\sa(t)}\Phi] \in \ocean\bigl(C[\underline{\sa}(t)]\bigr)/\sim \end{equation}
which depends only on $\sa(t)$ at time $t$. By Thm. \ref{Thm:polarizationclasses}, $\bigl(\FF_{\cs(\sa(t))}\bigr)_t$ defines a suitable family of Fock spaces for any external field $\sa$. 
Note, however, that this choice is merely a convenient one and not motivated by physical insight.

Also note that the implementations we obtain will not readily yield a time evolution. That is, arbitrary choices of the right-operations by $R = R(t_1,t_0)$ will in general not be ``compatible'' with each other, so that for $t_2 \geq t_1 \geq t_0$ we are going to find that 
\begin{equation*}   \lop{U^\sa(t_2,t_1)} \rop{R(t_2,t_1)} \lop{U^\sa(t_1,t_0)} \rop{R(t_1,t_0)} \neq \lop{U^\sa(t_2,t_0)} \rop{R(t_2,t_0)}. \end{equation*}

\noindent This problem is discussed in Section 8.2 under the more general theme of \textit{causality}. There, it is proven that compatible choices, preserving the semi-group structure of the time evolution, are indeed possible.\\

% In particular, we have to realize that the construction contains \textit{too many freedoms} to allow for a straight forward physical interpretation in analogy to the well understood multi-particle theory in non-relativistic quantum mechanics. 

\subsection{Particle Interpretation} 
\label{subsec:Particle Interpretation}

The most important question, however, concerns the physical content of the presented solution. It is yet unclear what \textit{physical} quantities are meaningful in the context of time-varying Fock spaces and how physical states represented by vectors in different Fock spaces can be compared. Considering our discussion up to this point, it seems natural to ask about the particle content of a given state:

\begin{quote}\textit{``How many particles and anti-particles are present at time t} ?''. \end{quote}

\noindent It is clear, though, that the sea classes $\cs(t)$ and the resulting infinite wedge spaces alone do not provide enough structure to define how the Dirac sea is filled i.e. how many electrons and positrons a given state contains. Thus, to get a well-defined answer, we have to distinguish instantaneous ``vacua'' to compare our states to. Recall that the construction of the geometric Fock space already contains, or rather requires such a distinguished state by choice of a polarization within the correct polarization class, defining the Stiefel manifold of admissible bases. In the infinite wedge spaces, there is no such preferred state and a vacuum always constitutes an additional choice. It is however \textit{convenient} to fix the various Fock spaces by choosing suitable representatives of the Dirac sea classes, i.e. seas with images in the right polarization class. These ``reference states'' then may or may not represent a physical vacuum.\\

It is usually convenient to specify the (instantaneous) vacuum states only projectively, i.e. by identifying the Dirac sea as a polarization $W(t) \in C[\underline{\sa}(t)], \; \forall t \in \IR$.
%, where $C(t)$ are the polarization classes identified in Thm. \ref{Thm:polarizationclasses}, i) [Identification of Polarization classes]. 
Such a polarization $W \in \pol(\HH)$ specifies a unique geometric Fock space $\FF_{\rm geom}^W$ (see Section \ref{subsubsection:generalization to arbitrary polarizations}). It also defines a unique GNS-representation of the CAR-algebra with a GNS-vacuum (see Section \ref{subsection:CAR representations}) or, more down to earth, a unique Fock-representation via the field operator $\Psi_W(f) = a(P_W f) + b^*(P_{W^\perp} f)$ (see Prop. \ref{Prop:Fock Representation}). For the infinite wedge space construction, it is not sufficient to specify the vacuum state only projectively. A fixed polarization within a polarization class does not specify an infinite wedge state or even a Fock space in a canonical way. Two seas $\Phi$ and $\Psi$ with $\im(\Phi) = \im(\Psi) = W$ can differ by a transformation $R \in \cu(\ell)$ (acting from the right) with $\Phi \sim \Psi=\Phi R$ if and only if $R \in \cu^1(\ell)$ (cf. Lemma \ref{Lemma:Uniquenessuptoaphase}). The resulting Fock spaces are however isomorphic by
$\rop R: \FF_{\cs(\Phi)} \to \FF_{\cs(\Psi)}$.\\

Without any additional insight into the fully-interacting theory and the microscopic dynamics of the Dirac sea, the only suggested choice for instantaneous vacua is given by spectral decomposition with respect to the Hamilton including the external interaction potential. That is, given an external field $\sa \in C_c^{\infty} (\mathbb{R}^4, \mathbb{R}^4)$, we denote by $P^{\sa(t)}_+$ and $P^{\sa(t)}_-$ the orthogonal projections onto the positive and negative spectral subspaces of the (static) Hamiltonian
\begin{equation*}H^\sa(t) = D_0 + e\sum\limits_{\mu=0}^3\alpha^\mu \sa_\mu(t) \Bigl\lvert_{t=const.}. \end{equation*}
We assume that $0$ is not in the spectrum of $H^\sa(t)$ or at most an isolated eigenvalue of finite-multiplicity and include the corresponding eigenspace into $P^{\sa(t)}_-(\HH)$, let's say, to get an unambiguous prescription.
Then, the following is true:

\begin{Theorem}[Fierz, Scharf 1979]
\label{Thm:Scharf Fierz}
\mbox{}\\
Under the conditions formulated above, it is true that 
\begin{equation*} P^{\sa(t)}_-\,(\HH) \in C[\underline{\sa}(t)],\; \forall t \in \IR \end{equation*}
or, equivalently, 
%\begin{equation*}P^{\sa(t)}_- - U^\sa(t, -\infty) P_{-}U^{\sa *}(t, -\infty)\in I_2(\HH), \end{equation*}
\begin{equation*}P^{\sa(t)}_- - P_{U^\sa(t, -\infty) \HH_-} \in I_2(\HH), \end{equation*}
where by $t'=-\infty$ we mean a large negative $t'$ outside the support of  $t \to \sa(t)$.
\end{Theorem}

Since we know from Thm. \ref{Thm:polarizationclasses} [Identification of Polarization classes] that the polarization class depends on the spatial part of the field only, we may also take the spectral decomposition for the Hamiltonian with $\sa^0$ set to zero. Then, we have (using $\beta^2 = \mathds{1}$ and $\lbrace \alpha_k , \beta \rbrace = 0$):
\begin{align*} (H^{\underline{\sa}})^2 = & \bigl[ -i \underline{\alpha}\cdot (\underline{\nabla} - ie \underline{A}) + \beta m \bigr]^2\\
= &   \bigl[ -i \underline{\alpha}\cdot (\underline{\nabla} - ie \underline{A})\bigr]^2 + m^2 \, \mathds{1}\\ 
\geq&  m^2 \, \mathds{1} \end{align*}
so that $\mathrm{Spec}(H^{\underline{\sa}}) \subseteq (-\infty , -m] \cup [+ m , + \infty)$, i.e. we have a nice mass-gap between the Dirac sea and the excited electron states.

The problem with both alternatives is that the vacuum is obviously \textit{not Lorentz-invariant}. A Lorentz transformation will change the external field and therefore the energy spectrum of the Hamiltonian. This would mean that a vacuum state, defined with respect to a distinct reference frame, would look like a multi-particle state for an observer in a different reference frame. It is on this basis that Scharf and Fierz finally reject the particle interpretation for a time-dependent interaction and conclude that ``the notion of particles has only asymptotic meaning" (\cite{FS}, p. 453). We will discuss this conclusion in the final chapter.

\subsection{Gauge Transformations}
Given a unitary time evolution $U^\sa(t,t')$, the polarization classes $C(t) = C[\underline{\sa}(t)]$ are determined by the spatial components of the A-field, defining the interaction potential. But this quantity is blatantly not gauge-invariant. Therefore, neither the Shale-Stinespring condition $U^\sa(t,t') \in \Ur(\HH)$, nor the procedure of second quantization on time-varying Fock spaces is gauge-invariant. This might seem rather suspicious at first, since QED, if anything, is famous for being a gauge-theory. But the right conclusion to draw is that it's unclear, how the presumed gauge invariance of the second quantized theory has to be understood. It turns out that the problem with gauge transformations is the same as with the time evolution. Smooth gauge transformations
%Under a gauge transformation 
%\begin{equation} \sa_\mu \rightarrow \sa_\mu(t, \underline{x}) + \partial_\mu g(t,\underline{x})\end{equation}
\begin{equation}\label{eq:gaugetransformation}\mathcal{G}\ni g: \Psi(x) \rightarrow e^{i\, \Lambda_g(x)} \Psi(x),\;  \Lambda_g \in C^{\infty}_c(\IR^3, \IR)\end{equation} 
do \textit{not} satisfy the Shale-Stinespring condition $[\epsilon, g] \in I_2(\HH)$ unless they are constant, and are therefore not implementable as automorphisms on a fixed Fock space \cite{MiRa}.\\ 

\noindent More precisely, the following is true:

\begin{Theorem}[Gauge Transformations, \cite{DeDueMeScho}Thm. III.11]\label{Thm:Gaugetransformations}
\item Let $\mathcal{G}$ denote the space of smooth gauge transformations as in \eqref{eq:gaugetransformation}
\item By [Thm. \ref{Thm:polarizationclasses}, ii)] it is justified to write $[e^{Q^\sa} \HH_-] =: C[\underline{\sa}]$ etc. for the polarization classes, as they are completely determined by the spatial part of the A-field. 
\\Then, for all $g \in \mathcal{G}$ it is true that
\begin{equation}\fingbox{g = e^{i\,\Lambda_g} \in \cu^0_{\rm res}\bigl(\HH; C[\underline{\sa}], C[\underline{\sa} + \nabla \Lambda_g] \bigr)}\end{equation}
for any $\underline{\sa} \in C^\infty_c(\IR^3,\IR^3)$. 
\end{Theorem}

\noindent Of course, this is what \textit{must} be true in order to be consistent with our previous results.\\
\begin{proof} Let $g \in \mathcal{G}$ a smooth gauge transformation, corresponding to the multiplication operator $e^{i\,\Lambda_g}$ on $L^2(\IR^3, \IC^4)$. Let $f: \IR \to [0,1]$ be a smooth cut-off function with $f(0) = 0, f(1) = 1$ and compact support in, let's say, $[0,2]$. To apply the machinery developed so far, we make $g$ time-dependent by setting $g(t) := e^{i\,f(t) \Lambda_g}$. Let $\sa(t,\underline{x}) = (0, -\underline{\sa}(\underline{x}))$ a static vector-potential. Under the gauge transformation $g$, the vector potential transforms as $\sa_\mu \to \sa_\mu - \partial_\mu (f(t)\Lambda_g)$, so that $\sa$ becomes
\begin{equation*} \tilde{\sa}(t, \underline{x}) = \bigl(-f'(t)\Lambda_g(\underline{x}) , -\underline{\sa}(\underline{x}) - f(t)\underline\nabla\Lambda_g(\underline{x})\bigr). \end{equation*}
The time evolution transforms as 
\begin{equation*} g(t_1) U^\sa(t_1,t_0) g^*(t_0) = U^{\tilde{\sa}}(t_1,t_0), \; t_0 \leq t_1 \in [0,1]. \end{equation*}
Note that the gauge transformation is fully ``turned on'' at $t=1$ and ``turned off'' at $t=0$. 
By Thm. \ref{Thm:polarizationclasses} [Identification of Polarization classes], we have
\begin{align*}& U^\sa(1,0) \in  \cu^0_{\rm res}\bigl(\HH; C[\underline{\sa}], C[\underline{\sa}]\bigr)\\
&U^{\tilde{\sa}}(1,0) \in  \cu^0_{\rm res}\bigl(\HH; C[\underline{\sa}], C[\underline{\sa} + \underline\nabla\Lambda_g]\bigr)\end{align*}
and therefore, by \eqref{eq:Urescomposition}:
\begin{equation*}  e^{i\Lambda_g} = g(1) = U^{\tilde{\sa}}(1,0)U^\sa(0,1) \in \cu^0_{\rm res}\bigl(\HH; C[\underline{\sa}], C[\underline{\sa} + \underline\nabla\Lambda_g)\bigr].\end{equation*}
\end{proof}

\noindent In the spirit of this chapter, we can therefore realize gauge transformation not as unitary operations on a Fock space, but as a transformation of the Fock spaces themselves. In this sense, the theory might become ``gauge-covariant'' rather than ``gauge-invariant'' in the naive sense. We will come back to the problem of gauge invariance in Section \ref{Sec:Gauge Invariance}.\\

\begin{Remark}[Renormalized Determinants]
\mbox{}\\
Mathematically, the gauge transformation group $\mathcal{G}$ and its representations have been studied very thoroughly. We know that smooth gauge transformations $g \in \mathcal{G}$ cannot be in $\Gl_{\rm res}(\HH)$, i.e. their odd-part w.r.to the splitting $\HH = \HH_- \oplus \HH_-$ is not in the Hilbert-Schmidt class (unless they are constant). However, one can show that they are contained in a higher Schatten class; actually, the Schatten index depends on the dimension of the space-time manifold. In \cite{MiRa} it is proven that on $d+1$-dimensional space-time, for odd $d$,
\begin{equation} \addtolength{\fboxsep}{5pt}\boxed{[\epsilon, g] \in I_{2p}(\HH)\; \forall\,g \in \mathcal{G}\, \Rightarrow p \geq (d+1)/2.}\end{equation} 
Hence, $\mathcal{G}$ embeds in a smaller group \begin{equation*}\Gl_p(\HH):= \lbrace A \in \Gl(\HH)\mid [\epsilon, A] \in I_{2p}(\HH)\rbrace \subset \Gl(\HH).\end{equation*}
$\Gl_1(\HH) = \Gl_{\rm res}(\HH)$ is good enough in 1+1-dimensions. For higher dimensions, in particular 3+1 dimensional QED, it is possible in a certain sense to generalize our construction to arbitrary Schatten index $p$. In particular,  it is possible to extend the definition of the Fredholm determinant to operators in $Id + I_p(\HH)$ by including a suitable ``renormalization'' term. Then, one can construct central extensions $\widetilde{\Gl}_p(\HH)$ of $\Gl_p(\HH)$ that reduce to $\Gres(\HH) \rightarrow \Gl_{\rm res}$ for $p=1$ and act on smooth\footnote{The action is not any more holomorphic for $p>1$, see \cite{Mi}.} sections in the dual of a generalized determinant bundle $DET_p$. See \cite{Mi} for details on these constructions. But apart from being rather tedious, this solution is unsatisfying for a different reason: The action of $\mathcal{G} \subset \widetilde{\Gl}_p(\HH)$ is \textit{not unitary}. Actually, it is known that $\widetilde{\Gl}_p(\HH)$ does not have any faithful unitary representations for $p > 1$ (see \cite{Pick}).
\end{Remark}

%\subsubsection{Remark: The Fock bundle}
%\mbox{}\\
%It is possible to fit the concept of time-varying Fock spaces into quite an elaborated geometric structure as well, by introducing a manifold structure on $\A$, the space of (time-independent) vector potentials, and constructing the \textit{Fock bundle} over $\A$. This is a principle fibre bundle, where the fibre over any $A \in \A$ is an entire Fock space. $\Ures(\HH)$ is the structure group acting transitively on every fibre (i.e. on the Fock spaces themselves)
%whereas the time evolution and gauge transformations act as unitary bundle-maps between different fibres. For an outline of the construction see \cite{Mi} or \cite{CaMiMu}.

%Results from the theory of bundle gerbes show that the Fock bundle can be pushed down to a bundle over $\A \slash \mathcal{G}_e$ where $\mathcal{G}_e$ is the space of based, smooth, compactly supported gauge transformations (\cite{CaMiMu97}). In this sense, the description could be made gauge-invariant. This does seem promising indeed, but the mathematical treatment is quite abstract and not very physicist-friendly. We therefore suggest that one should try to reformulate the construction in terms closer to the physical intuition and compare it to the results presented in this work. See in particular Section 8.4 for our result on gauge invariance.\\

\newpage 
\thispagestyle{empty}
\quad 
\newpage

\chapter{The Parallel Transport of \\ Langmann \& Mickelsson}     
We have seen that the central extensions  $\Gres(\HH)\rightarrow\Gl_{\rm res}(\HH)$ and $\Ures(\HH)\rightarrow\Ur(\HH)$ carry the structure of principle fibre bundles. This allows us to equip them with a \textit{bundle-connection} as an additional geometric structure. Our aim is to use parallel transport with respect to this connection to lift the unitary time evolution to the Fock space, or at least to fix the phase of the second quantized S-matrix in a well-defined manner.
We will follow \cite{LaMi} and define a suitable connection not on $\Ures(\HH)$ but on its complexification $\Gres(\HH)$.
However, the construction restricts easily to the unitary case.\\

Recall that a connection $\Gamma$ on a principle bundle is a distribution in its tangent bundle, distinguishing tangent vectors that are called \textit{horizontal} (to the base-manifold). This distribution is complementary to the space of \textit{vertical} vector fields, defined as the kernel of $D\pi$, where $\pi$ is the projection onto the base-manifold and $D\pi$ its differential map. That is, vertical vectors are vectors along the fibres of the bundle.
A connection is a geometric structure that allows us to lift paths from the base-manifold (here: $\Ur$ or $\Gl_{\rm res}$) to the bundle in a unique way, by demanding that the tangent vectors to the lifted path are always horizontal.
These \textit{horizontal lifts} define parallel transport in the principle bundle.\footnote{If $\pi: P \rightarrow M$ is a principle bundle over $M$ with a connection, the parallel transport of $p \in P$ along a curve $\gamma$ in $M$ starting in $\pi(p)$ is the end-point of the horizontal lift of $\gamma$ with starting point $p \in P$.}  Furthermore, recall that connections on a principle bundle are in one-to-one correspondence with \textit{connection one-forms} with values in the Lie algebra of the structure group. For every connection $\Gamma$ there exists one -- and only one -- connection form with $\Gamma = \ker(\Theta)$. A nice treatment of connections and parallel transport on principle bundles can be found, for instance, in \cite{KoNo}.\\

Throughout this chapter, a little care is required with the fact that we're working on infinite-dimensional manifolds. But Banach manifolds, as the ones we're dealing with, are generally pretty well-behaved. One essential point is that on Banach manifolds (as opposed to manifolds modeled on infinite-dimensional Fr\'echet spaces, for example) the Inverse Mapping Theorem holds and thus the local calculus works almost exactly as in the finite-dimensional case.\\ 
Note, however, that we avoid all expansions in local coordinate frames, as are often used for the analogous proofs in the finite-dimensional case. Fundamentals of Differential Geometry on Banach manifolds are presented in \cite{Lang}.

%\newpage
\section{The Langmann-Mickelsson Connection}
\subsection{The Maurer-Cartan Forms}
Let $G$ be a Lie group. The group operation on $G$ defines smooth actions of $G$ on itself by multiplication from the left or from the right.
For $g \in G$ we denote by $L_g$ the action from the left, i.e. 
\begin{equation}L_g: G \rightarrow G, \;   L_g(h) := g\cdot h\end{equation} and by $R_g$ the action from the right 
\begin{equation}R_g: G \rightarrow G, \;  R_g(h) := h \cdot g. \end{equation}
The differentials of these maps (the push forwards) as maps between tangent-spaces are then isomorphisms:
\begin{equation}\begin{split} (L_g)_{*}\, &: T_hG \rightarrow T_{gh}G \\
(R_g)_{*}\,&: T_hG \rightarrow T_{hg}G.\end{split}\end{equation}
The corresponding pull-back maps on differential forms will be denoted by $(L_g)^{*}$ and $(R_g)^{*}$.\\
%\begin{equation}\begin{split} (L_g)^{*}\, &: T^{*}_{gh}G \rightarrow T^{*}_{h}G \\
%(R_g)^{*}\,&: T^{*}_{hg}G \rightarrow T^{*}_{h}G \end{split}\end{equation}
Recall that a vector field $X \in \mathfrak{X}(G)$ is called \textit{left-invariant} if $(L_g)_{*} X_h = X_{gh}, \; \forall g,h \in G$ and \textit{right-invariant} if $(R_g)_{*} X_h = X_{hg}, \; \forall g,h \in G$.\\

\noindent $G$ can also act on itself by conjugation:
\begin{equation}c_g : G  \rightarrow G, \;   c_g(h) := L_gR_{g^{-1}} h = g h g^{-1}. \end{equation}
The map \begin{equation} Ad:G \rightarrow \mathrm{GL}(\mathrm{Lie}(G)), \; g \mapsto \dot{c_g} \end{equation}
yields a natural representation of G on it's own Lie algebra called the adjoint representation. The kernel of $Ad$ equals the center of G.\footnote{In a Matrix representation, $G \hookrightarrow \Gl(V), \mathrm{Lie}(G) \hookrightarrow End(V)$ and $Ad_g(X) = g X g^{-1}$. Thus, one can think of the adjoint representation in the following way: If the elements of a vector space V (or a vector bundle with fibre V) transform under a symmetry $g \in G$ via an action of G on V, the endomorphisms of V transform by the corresponding adjoint action.}\\

\noindent The Lie algebra $\mathfrak{g}:= \mathrm{Lie}(G)$ of a Lie group $G$ is usually defined as the space of \textit{left-invariant} vector fields with the Lie bracket given by the commutator of vector fields. This space is canonically isomorphic to the tangent space above the identity, i.e. $\mathfrak{g} \cong T_eG$, for $e$, the neutral element in $G$. In the following, we make use of this identification.

\begin{Definition}[The Maurer-Cartan Forms]
\mbox{}\\
The left Maurer-Cartan form, also called the canonical left-invariant one-form, is a one-form on $G$ with values in the Lie algebra. It's defined by:
\begin{equation}\begin{split} \omega_L(g) \, :\, T_g&G \longrightarrow T_eG \cong \mathfrak{g}\,;\\
& V \longmapsto (L_{g^{-1}})_{*}\, V  \end{split}\end{equation}
Analogously, we define the right Maurer-Cartan form by:
\begin{equation}\begin{split} \omega_R(g) \, :\, T_g&G \longrightarrow T_eG \cong \mathfrak{g}\,;\\
& V \longmapsto (R_{g^{-1}})_{*}\, V  \end{split}\end{equation}
\end{Definition} 
%\noindent We will often drop the subscript ``L'' or ``R'' if it is clear which form we're referring to. \\
\noindent  Obviously, the Maurer-Cartan forms are constant on loft-invariant, respectively, on right-invariant vector fields. Furthermore, $\omega_L (e) = \omega_R(e) = Id_{T_eG}$.\\

\noindent In a matrix representation, they can be written as
\begin{equation} \omega_L = g^{-1} \mathrm{d}g \;\;\text{and}\;\; \omega_R = \mathrm{d}g\, g^{-1},  \end{equation}
respectively.
We state a few less obvious properties of the Maurer-Cartan forms. 
%Let $\lbrace V_i \rbrace_{i\in I}$ be a basis of $T_eG$. This defines a set of linearly independent left-invariant vector fields $\lbrace X_i \rbrace_{i\in I}$ by $X_i(p) = (L_p)_{*} V_i$, which spans the tangent bundle $TG$ as a $C^{\infty}$-module.\\ 
%Let $\lbrace c_i \rbrace_{i\in I}$ be the corresponding frame of dual forms.

\begin{Lemma}[Properties of the left Maurer-Cartan form]
\mbox{}\\
The left Maurer-Cartan form $\omega_L$ has the following properties:
\begin{enumerate}[i)]
\item $\omega_L$ is left-invariant, i.e. $(L_g)^{*} \omega_L = \omega_L$
\item $(R_g)^{*}\omega_L = Ad_{g^{-1}}\,  \omega_L, \,\; \forall g \in G$
%\item $ \omega_L = \sum\limits_{i} V_i \otimes c^{i}$, in a basis as defined above.
\item $ d\omega_L + \frac{1}{2}[\omega_L \wedge \omega_L] = 0$
\end{enumerate}
\vspace{2mm}
\noindent The right Maurer-Cartan form $\omega_R$ has the following properties:
\begin{enumerate}[I)]
\item $\omega_R$ is right-invariant, i.e. $(R_g)^{*} \omega_R = \omega_R$
\item $(L_g)^{*}\omega_R = Ad_{g}\,  \omega_R, \,\; \forall g \in G$
\item $ d\omega_R - \frac{1}{2}[\omega_R \wedge \omega_R] = 0$
\end{enumerate}
\end{Lemma}

\begin{proof} Let $\omega$ denote the left Maurer-Cartan form.
\begin{enumerate}[i)]
\item For $ V \in T_hG$ and $g \in G$ we have \begin{equation*}(L_g^{*} \, \omega_{gh}) (V) =  \omega_{gh}((L_g)_{*} V) = L_{(gh)^{-1}*} L_{g*} V = L_{h^{-1}*}L_{g^{-1}*}L_{g*} V = L_{h^{-1}*}(V) = \omega_h(V)\end{equation*}
\item $(R_g^{*}\, \omega_{hg}) (V) = \omega_h (R_{g*} V) = L_{(hg)^{-1}*}  R_{g*} V = L_{g^{-1}*} L_{h*} R_{g*} V = Ad_{g^{-1}}\,  \omega_h (V) $
%\item The $\mathcal{C}^{\infty}$ - module $\mathfrak{X}(G)$ is spanned by the left-invariant vector fields $\lbrace X_i \rbrace_{i\in I}$ and on those: $\sum\limits_{i} V_i \otimes c^{i} (X_k) = V_k = \omega(X_k)$
\item $d\omega(X,Y) = X(\omega(Y)) - Y(\omega(X)) - \omega[X,Y],$ for $X, Y \in \mathfrak{X}(G)$. \\So, on left-invariant vector fields the identity is true, because \\
$d\omega(X,Y) = -\omega[X,Y] = -[\omega(X), \omega(Y)] = - \frac{1}{2}[\omega(X) \wedge \omega(Y)]$. But
the expression is bilinear in smooth functions (or in other words: tensorial) and thus depends on the vectors only pointwise. Therefore, the identity holds for all vector fields.
\end{enumerate}
%The Lie algebra structure on Lie($G$) can be characterized by the structure constants $c^{i}_{jk}$, defined by $[X_j , X_k] = \sum\limits_{i} c^{i}_{jk} X_i$. Furthermore, we have the so-called Maurer-Cartan structure equations $dc^{i}(X_j , X_k) = -  c^{i}_{jk}$. 
%\footnote{$dc^{i}(X_j , X_k) = \mathcal{L}_{X_j} c^{i}(X_k) - \mathcal{L}_{X_k} c^{i}(X_j) - c^i([X_j , X_k]) = - c^{i}(\sum\limits_{l} c^{i}_{jk} X_l) = - c^{i}_{jk}$}\\
%It follows that
%\begin{align*} d\omega &= d (\sum\limits_{i} V_i \otimes c^{i}) = \sum\limits_{i} (V_i \otimes d c^{i})  = \sum\limits_{i} V_i \otimes \sum\limits_{j,k} (-  \frac{1}{2} c^{i}_{jk} c^{j}\wedgec^{k})\\
%& = - \frac{1}{2}\sum\limits_{i,j,k} c^{i}_{jk}\, V_i \otimes c^{j}\wedgec^{k} = - \frac{1}{2}\sum\limits_{j,k} [V_j , V_k] \otimes c^{j}\wedgec^{k} = -\frac{1}{2} [\omega \wedge \omega]\end{align*}
\noindent The proofs for the right-invariant form work analogously. The different sign in $III)$ as compared to $iii)$ comes from the fact that the commutator on right-invariant vector fields equals \textit{minus} the commutator on the corresponding left-invariant vector fields which -- by convention --  defines the Lie bracket.\\
\end{proof}
\noindent As the Lie algebra of a Lie group $G$ is, by convention, usually defined as the vector space of left-invariant vector fields on $G$, the left Maurer-Cartan form is the more canonical object. However, it will turn out that the right-invariant form is better suited to our purposes.

%\section{The Langmann-Mickelsson Connection}
\subsection{Connection Forms}
\noindent Now, we have all ingredients for defining a connection on the principle $\IC^{\times}$-bundle\\ 
$\Gres \xrightarrow{\; \pi \;} \Gl_{\rm res}$.
Recall that the central extension of Lie groups \eqref{GLresExtension}
\begin{equation*} 1 \longrightarrow \IC^{\times} \xrightarrow{\; \imath \;} \Gres(\HH) \xrightarrow{\; \pi \;} \Gl_{\rm res}(\HH) \longrightarrow 1\end{equation*}
induces a central extension of the corresponding Lie algebras
\begin{equation}0 \longrightarrow \IC \xrightarrow{\dot{\imath}} \tilde{\mathfrak{g}}_1 \xrightarrow{\dot{\pi}} \mathfrak{g}_1 \longrightarrow 0. \end{equation}
Furthermore, every (local) trivialization of $\Gres$ about the identity defines an isomorphism of vector spaces $ \tilde{\mathfrak{g}}_1 \cong \mathfrak{g}_1 \oplus \IC$. In particular, we have the local trivialization $\phi$ defined in \eqref{localtriv} as, arguably, the most natural choice and we use it to identify  
\begin{equation*} \dot{\phi}:\, \tilde{\mathfrak{g}}_1 \xrightarrow{\;\cong\;} \mathfrak{g}_1 \oplus \IC. \end{equation*}

\noindent With all this in our hands, a quite natural connection on $\Gres \xrightarrow{\; \pi \;} \Gl_{\rm res}$ can be readily defined.
The (left) Maurer-Cartan form $\omega_L$ on $\Gres$ is a one-form on $\Gres$ with values in the Lie algebra $\tilde{\mathfrak{g}}_1 = \mathrm{Lie}(\Gres)$. A connection, however, is defined by a one-form with values in the Lie algebra of the structure group which in our case is Lie($\IC^{\times}$) = $\IC$. Thus, we simply set $\Theta_{LM} := \prc \circ \omega_L$, where $\prc$ is the projection onto the $\IC$-component with respect to the splitting $\widetilde{\mathfrak{g}} \cong \mathfrak{g} \oplus \IC$. This is indeed a connection one-form.

\begin{Definition}[Langman-Mickelsson Connection]
\mbox{}\\
Let $\omega_L$ be the left Maurer-Cartan form on the Lie group $\Gres$.
The one-form
\begin{equation}\label{LMconnection}
\Theta_{LM} := \prc \circ \omega_L \in \Omega^1(\,\Gres; \IC)
\end{equation}
is a connection one-form on $\pi: \Gres(\HH)\rightarrow \Gl_{\rm res}(\HH).$\\ We will call it the \emph{Langmann-Mickelsson connection}.
\end{Definition}

In \cite{LaMi},  Langmann and Mickelsson use this connection for a second quantization of the S-matrix by parallel transport in the $\Gres$-bundle. We propose a slightly modified version, using the right-invariant Maurer-Cartan form, instead of the left-invariant form. The advantage will become clear when we carry out the lift of the time evolution. We prove that the one-forms defined in this way are indeed connection forms. The proof for the left-invariant form is almost identical.

\begin{Proposition}[Connection on $\Gres$]
\mbox{}\\
Let $\omega_R$ be the right Maurer-Cartan form on the Lie group $\Gres$.
Then,
\begin{equation}\label{connection}
\Theta := \prc \circ \omega_R \in \Omega^1(\,\Gres; \IC)
\end{equation}
defines a connection one-form for the principle $\IC^{\times}$ bundle $\pi: \Gres(\HH)\rightarrow \Gl_{\rm res}(\HH)$.
\end{Proposition}
\begin{proof} 
For an element $x$ in the Lie algebra $\IC$ of  $\IC^{\times}$, we denote by $x^*$ the \textit{fundamental vector field}, generated by $x$ via its exponential action on the Lie group, i.e. by the flow \\
$\Theta_t(p) = p \cdot \exp(tx)$.  We need to show (cf. \cite{KoNo}, Prop. 1.1): 
\newpage

\begin{enumerate}[i)]
\item $\Theta(x^*) = x$ for all $x$ in $\mathrm{Lie}(\IC^{\times}) \cong \IC$.
\item $R_{\lambda}^* \Theta = Ad_{\lambda^{-1}} \Theta = \Theta$ for all $\lambda$ in the structure group $\IC^{\times}$.
\end{enumerate}
In our case, the structure group happens to be abelian, so that $Ad_{\lambda^{-1}}\equiv Id$.\\
ii) is trivial, because $\omega_R$ is invariant even under the right-action of the entire $\Gres$ on itself, in particular under the right-action of $\imath(\IC^{\times}) \subset \Gres$.  Thus: 
\begin{equation*}R_{\lambda}^*\, \Theta = \prc \circ R_{\lambda}^*\,\omega_R = \prc \circ  \omega_R = \Theta = Ad_{\lambda^{-1}} \Theta, \; \forall \lambda \in \IC^{\times}.\end{equation*}
For the first property, let $x \in \IC$. Then:
\begin{align*} &x^*(p) =\; \frac{d}{dt}\Bigl\vert_{t=0}\, p \cdot \imath(\exp(tx) ) = (L_p)_* \imath_*(x)\\[1.3ex]
\Rightarrow\; & \omega_R(x^*)\vert_p = (R_{p^{-1}})_* (L_p)_*\, \imath_*\, x  = \imath_*\, x\\[1.3ex]
\Rightarrow\; & \Theta(x^*)\vert_p = \prc \circ \imath_* \, x = x, \; \forall\, p \in \Gres(\HH).
%\Rightarrow\;  &  \omega_R(x^*)\vert_p = \frac{d}{dt}\Bigl\vert_{t=0}  p \cdot i(\exp(tx) )\cdot p^{-1} = \frac{d}{dt}\Bigl\vert_{t=0}  p \cdot i(\exp(tx) )\cdot p^{-1}   
%\Rightarrow\;  & \Theta(x^*)\Bigl\lvert_p = \frac{d}{dt}\Bigl\vert_{t=0} \mathrm{pr}_{\IC^{\times}}\circ \rho\,\bigl( p \cdot i(\exp(tx) )\cdot p^{-1} \bigr) 
\end{align*}
where we have used that $\imath$ maps into the center of $\Gres(\HH)$ and also that \\
$\Theta \circ \imath (c) = (\mathds{1} , c) \in \Gl_{\rm res} \times \IC^{\times}$, which implies $\prc \circ (\imath_*)\vert_e \, = Id_{\mathfrak{g}_1}$.\\
\end{proof}
\noindent Thus, $\Theta$ is a connection one-form and defines the connection 
\begin{equation}\Gamma_{\Theta}:= \ker{\Theta} \subset \mathrm{T}\Gres\end{equation} 
on the principle bundle $\Gres(\HH)$.\\

We can define an analogous connection on $\Ures(\HH)\rightarrow\Ur(\HH)$. However, we follow \cite{LaMi} in discussing the connection on the complexification $\Gres(\HH) \rightarrow \Gl_{\rm res}(\HH)$. Apart from being more general, this has the advantage that the the local section in $\Gres(\HH)$ is ``nicer'' which simplifies computations in local coordinates. Anyways, our construction restricts immediately to the unitary case. 

%\begin{Lemma}[Uniqueness of the Connection]
%\mbox{}\\
%$\Theta$ defined in \eqref{connection} determines the unique connection on $\pi: \Gres \rightarrow \Gl_{\rm res}$, right-invariant under the action of $\Gres$ on itself.\\
%Analogously, the Langmann-Mickelsson connection defined by \eqref{LMconnection} is the unique left-invariant connection on $\pi: \Gres \rightarrow \Gl_{\rm res}$, w.r.to the action of $\Gres$ on itself.
%\end{Lemma}
%\begin{proof} As a one-form depends on vector fields only pointwise, it is completely determined by its action on right-invariant vector fields. But a connection one-form $\tilde{\Theta}$ defining a right-invariant connection on the bundle must be constant on right-invariant vector fields.  Also, it must map fundamental vector fields (generated by the action of the Lie algebra via exp) to the corresponding generator. Consequently, any connection one-form which is invariant under the right-action of $\Gres$ on itself, equals $\Theta$.  The same is true for the left-invariant case and the Langmann-Mickelsson connection.
%\end{proof}

\begin{Proposition}[Curvature of the Connection]\label{Prop:curvatureSchwinger}
\mbox{}\\
The curvature $\Omega$ of the connection $\Theta$ corresponds to the Schwinger cocycle \eqref{Schwingercocycle} in the following sense: 
If $\widetilde{X}, \widetilde{Y}$ are right-invariant vector fields on $\Gres$, identified with elements of the Lie algebra $\tilde{\mathfrak{g}}_1$ of $\Gres$ then
\begin{equation} \Omega({\widetilde{X}, \widetilde{Y}}) =  c(X, Y) \end{equation}
for $X:= \mathrm{D}\pi (\widetilde{X}), Y:=\mathrm{D}\pi ( \widetilde{Y} ) \in \mathfrak{g}_1$.\footnote{The analogous computation for $\Theta_{LM}$ would yield \textit{minus} the Schwinger cocycle. In \cite{LaMi} there might be a sign error or a different convention for the cocycle might be used.}
% $X:= \mathrm{pr}_{\mathfrak{g}}(\widetilde{X}), Y:=\mathrm{pr}_{\mathfrak{g}}( \widetilde{Y} ) \in \mathfrak{g}_1$
Since $\Omega$ is a two-form and  $\mathfrak{X}(\Gres)$ is spanned (as a $C^\infty$-module) by the right-invariant vector fields, the curvature of the bundle is completely determined by this identity.
\end{Proposition}
\begin{proof}
Let $\widetilde{X} , \widetilde{Y}$ be left-invariant vectorfields corresponding to $(X , \lambda_1) , (Y , \lambda_2) \in \mathfrak{g}_1 \oplus \IC \cong \tilde{\mathfrak{g}}_1$, respectively. 
The curvature 2-form can be computed as $\Omega = \mathrm{d}\Theta + \frac{1}{2}[\Theta \wedge \Theta]$.
Since $\Theta$ takes values in the abelian Lie algebra $\IC$, the second summand is zero and $\Omega$ equals $\mathrm{d}\Theta$.  
\newpage
\noindent We recall from Section 3.5 and Prop. \ref{Prop:cocycleisSchwinger} that under the isomorphism $\tilde{\mathfrak{g}_1} \cong \mathfrak{g}_1 \oplus \IC$ induced by the local trivialization $\phi$, the Lie bracket on $\tilde{\mathfrak{g}}_1$ is the Lie bracket on $\mathfrak{g}_1$ plus the Schwinger-cocycle. Thus, we get:
\begin{align*} \Omega(\widetilde{X} , \widetilde{Y}) = \,& \mathrm{d}\Theta(\widetilde{X} , \widetilde{Y}) = \prc(\mathrm{d}\omega(\widetilde{X} , \widetilde{Y})) =  \prc([\omega(\widetilde{X}) , \omega(\widetilde{Y}) ] )\\[1.5ex]
= \,&  \prc( [\widetilde{X} , \widetilde{Y}] ) =   \prc \bigl( ( [X , Y] , c(X,Y) \bigr) =  c(X , Y). \end{align*}
\end{proof}

\subsection{Local Formula}
\noindent Let's compute an explicit formula for the connection w.r.to the coordinates defined by the local section \eqref{tau}.
Let $\gamma: t \mapsto [(g(t) , q(t))],\; t \in [-\epsilon, \epsilon]$ be a $C^1$-curve in $\Gl_{\rm res}(\HH)$. % staying in the domain $U$ of the chart $\Phi: U \rightarrow \Gl_{\rm res} \times \IC^{\times}$.\\ 
We write $[(g , q)]$ for $\gamma(0)= [(g(0), q(0)]$ and $\phi$ for the local trivialization \eqref{localtriv} on $\pi^{-1}(U)$.
Note that the connection form $\Theta$ can be expressed as:
\begin{equation*}\Theta_g = \prc(\mathrm{d}g\, g^{-1}) = \prc \circ DR_g^{-1} \circ Id_{T_g\Gres}.\end{equation*}
Thus, we compute:
\begin{align*} 
\Theta(\dot{\gamma}(0)) = & \frac{d}{dt}\Bigl\vert_{t=0} \prc \circ \phi \bigl( [g(t)g^{-1} , q(t)q^{-1}] \bigr)\\[1.5ex]
= & \frac{d}{dt}\Bigl\lvert_{t=0} \det \bigl[ (g(t)g^{-1})^{-1}_{++} (q(t)q^{-1}) \bigr]\\[1.5ex]
= & - \trace \bigl[(\dot{g}(0)g^{-1})_{++} - \dot{q}(0)q^{-1} \bigr].
\end{align*}
Writing \begin{align*}
  g (t) =
   \begin{pmatrix}
    a(t) & b(t)\\
    c(t) & d(t)
  \end{pmatrix} \; \text{and}\;   g^{-1}(t) = \begin{pmatrix}
    \alpha(t) & \beta(t) \\
    \gamma(t) & \delta(t) \end{pmatrix}\\
\end{align*}
w.r.t. the splitting $\HH = \HH_+ \oplus \HH_-$, the formula becomes 
  \begin{equation}\label{connectionformula}\addtolength{\fboxsep}{5pt} \boxed{\Theta\,=\,- \trace\bigl[ \mathrm{d}a\,\alpha + \mathrm{d}b\,\gamma - \mathrm{d}q\,q^{-1} \bigr].}\end{equation} 
%in the domain U of the local section (???).\\
The corresponding expression for the left-invariant connection can be computed analogously and equals

\begin{equation}\fingbox{ \Theta_{LM}\,=\, - \trace\bigl[ \alpha\, \mathrm{d}a + \beta\,\mathrm{d}c\ - q^{-1}\,\mathrm{d}q \bigr].}\end{equation}

\newpage
\section{Parallel Transport in the $\Gres(\HH)$-bundle}
We are actually interested in \textit{parallel transport} in the $\Gres-$bundle. The bundle-econnection defines a ``horizontal'' distribution in the tangent bundle of $\Gres(\HH)$ which gives us the notion of \textit{horizontal lifts} of paths from $\Gl_{\rm res}(\HH)$ to $\Gres(\HH)$. A horizontal lift is a path in $\Gres(\HH)$ that projects to the original path in $\Gl_{\rm res(\HH)}$ and whose tangent vector is always ``horizontal'' to the base manifold. So, if we think of the unitary time evolution as a differentiable path in $\Ur(\HH) \subset \Gl_{\rm res}(\HH)$, the horizontal lift of that path will correspond to a continuous (even differentiable) choice of implementations on the fermionic Fock space.
%We will show that, if done right, this lift will preserve the composition property of the time evolution. The original hope of Langmann and Mickelsson was probably that the parallel transport will also fix the \textit{phase} of the implementations (in particular of the S-matrix) in a more or less unique way. We will see that this hope was too optimisticthough.\\

We can use the local expression \eqref{connectionformula} for the connection form to compute an explicit formula for parallel transport inside the domain $W$ of the local section \eqref{tau}:

\noindent Let $g(t)$ be a path in $U \subset \Gl_{\rm res}(\HH)$, $-T \leq t \leq T$, with $g(-T) = \mathds{1}$.
The lift $\tilde{g}(t) = [(g(t) , q(t))]$ in $\Gres(\HH)$ is horizontal if and only if
\begin{equation*}\trace\bigl[ \dot{a}(t)\,\alpha(t) + \dot{b}(t)\,\gamma(t) - \dot{q}(t)\,q^{-1}(t) \bigr] \equiv 0. \end{equation*}
Formally, this implies \begin{equation*}\trace\bigl( \dot{q}(t)\,q^{-1}(t)\bigr) = \trace\bigl(\dot{a}(t)\,\alpha(t) + \dot{b}(t)\,\gamma(t)\bigr).\end{equation*}
Identifying the LHS as the logarithmic derivative of $\det(q(t))$ we can write
\begin{equation*}\det(q(T)) = \exp \bigl[ \int\limits_{-T}^{T}  \trace(\dot{a}(t)\,\alpha(t) + \dot{b}(t)\,\gamma(t))\, \mathrm{d}t \bigr].\end{equation*}
Also, formally: \hspace{1.1cm}$\det(a(T)) = \exp \bigl[ \int\limits_{-T}^{T}  \trace(\dot{a}(t)\,a^{-1}(t))\, \mathrm{d}t \bigr]$.\\
Individually, the traces do not converge but put together the trace converges and gives:
\begin{equation}\label{ptformula}\addtolength{\fboxsep}{3pt}\boxed{ \det\bigl[ a^{-1}(T) q(T) \bigr] =  \exp \Bigl[ \int\limits_{-T}^{T}  \trace\bigl[ \dot{a}(t)\,( \alpha(t) - a^{-1}(t)) + \dot{b}(t)\,\gamma(t)\bigr]\, \mathrm{d}t \Bigr]. }\end{equation}
In the local trivialization $\phi$ on $W \subset \Gl_{\rm res}^0(\HH)$, this is precisely the $\IC$-component of the horizontal lift $\tilde{g}(T)$ in \eqref{localtriv}. 
In other words, parallel transport in the $\Gres(\HH)$-bundle corresponds to multiplication by the right-hand-side of \eqref{ptformula} in local coordinates. Note that the expression \eqref{connectionformula} is valid everywhere, while \eqref{ptformula} makes sense only in the neighborhood $W$ of the identity, where $a$ is invertible.

\noindent The corresponding expression for the Mickelsson-Langmann connection is 
\begin{equation}\label{ptformulaLM} \det\bigl[ a^{-1}(T)\, q(T) \bigr] = \exp\Bigl[\int\limits_{-T}^{T} \trace\bigl[ (\alpha(t) - a^{-1}(t) )\dot{a}(t) + \beta(t)\dot{c}(t)\bigr] \mathrm{d}t \Bigr]. \end{equation}

\noindent For a unitary path, these factor corresponds to the phase of the lift of $g(T)$ up to normalization (cf. Lemma \ref{Uressection}). We have said that the structure defined on $\Gres$ restricts directly to the unitary case. Still, in case someone suspects hand-waving here, we show by explicit computation that the lift of a unitary path doesn't leave
$\Ures(\HH) = \Gres(\HH) \cap \bigl(\cu(\HH) \times \cu(\HH_+)\bigr)$.\\

\begin{Lemma}[Parallel Transport stays in $\Ures$]
\mbox{}\\
Let $u(t)\, , t \in I$ be a (piecewise $C^1$) path in $\Ur \subset \Gl_{\rm res}$ and $\widetilde{u}(t)$ a horizontal lift of $u(t)$ with initial conditions $\widetilde{u}(0) = [U , r] \in \Ures(\HH) \subset \Gres(\HH)$.\\ Then $\widetilde{u}(t) \in \Ures(\HH) \; \forall t \in I$, i.e. the horizontal lift remains unitary. 
\end{Lemma}

\begin{proof} If $\; \widetilde{u}(t) = [u(t) , q(t)]$ is the horizontal lift, consider the path 
\begin{equation*}\hat{u}(t) := [u(t) , p(t)], \; \text{with} \; p(t) := (q^*(t))^{-1} \end{equation*} 
Clearly, $\pi \circ \hat{u}(t) = u(t)$ and $\hat{u}(0) = \widetilde{u}(0) =  [U , r] \in \Ures$. 
Furthermore, by \eqref{connectionformula}, using the fact that $u(t)$ is unitary:
\begin{align*} 0 =\;& \trace\bigl[(\dot{u}(t)u^{-1}(t) )_{++} -  \dot{q}(t)q^{-1}(t) \bigr]\\ 
=\;& \trace\bigl[(u(t)\dot{u^*}(t) )_{++} - q^{*-1}(t)\dot{q^*}(t) \bigr]^{c.c.}\\
=\;& \trace\bigl[(u(t)\dot{u^*}(t) )_{++} - p(t)\dot{p^{-1}}(t) \bigr]^{c.c.}
\end{align*}
And using 
\vspace*{-3mm}\begin{align*} & 0 = \frac{d}{dt}\bigl(u(t)u^*(t)\bigr) = \dot{u}(t)u^*(t) + u(t)\dot{u^*}(t) \\
&0 = \frac{d}{dt}(p(t)p^{-1}(t)) = \dot{p}(t)p^{-1}(t) + p(t)\dot{p^{-1}}(t), \end{align*}
we derive
\begin{equation*} \trace\bigl[(\dot{u}(t)u^{-1}(t) )_{++} -  \dot{p}(t)p^{-1}(t)  \bigr] = 0. \end{equation*}
Thus, $\hat{u}(t)$ is also a horizontal lift of $u(t)$ satisfying the same initial condition.
By uniqueness of the horizontal lift it follows that $\hat{u}(t) \equiv  \widetilde{u}(t)$ and thus
 $\widetilde{u}(t) \in \Ures(\HH) \subset \Gres(\HH), \; \forall t \in I.$\\
\end{proof}

\section{Classification of Connections}
A connection is always an additional geometric structure on the bundle and as such constitutes a particular choice. We have already restricted this choice to \textit{right-inavariant} connections. Note that every connection is by definition invariant under the fibre-preserving right-action of the structure group on the principle bundle, denoted by $r_c$ for $c \in \IC^{\times}$. But our condition is much stronger: we demand that it is invariant under the right-action of $\Gres$ on itself as a Lie group (denoted by a capital $R$). That is, if $\tilde{\gamma}(t)$ is a horizontal path in $\Gres$, so is $R_{\tilde{g}}\, \tilde{\gamma}(t)=  \tilde{\gamma}(t) \cdot \tilde{g}$ for any $\tilde{g} \in \Gres$.
We will see that this assures that the horizontal lift of a unitary time evolution preserves the semi-group structure, i.e. is indeed a time evolution on the fermionic Fock space. However, the condition of right-invarience does not specify a unique connection. This is evident already from our explicit construction of the Langmann-Mickelsson connection one-form. It involved a projection onto $\IC$ that was given with respect to an isomorphism $\tilde{\mathfrak{g}}_1 \rightarrow \mathfrak{g}_1 \oplus \IC$. This isomorphism however is not canonical, but induced from the local trivialization $\phi$ or, equivalently, from the local section $\sigma$. A different section will define a different isomorphism and consequently a different right-invariant connection. 

Generally, every local section $\upsilon: \Gl_{\rm res} \rightarrow \Gres$ around the identity with $\upsilon(e) = e$
 ($e$ being the neutral element in the group, i.e. the identity) induces a local trivialization of $\GresO$ by 
 \begin{equation}Ê\phi_{\upsilon}^{-1}: \Gl^0_{\rm res}(\HH) \times \IC^{\times} \to \GresO,\, (g, \lambda) \mapsto \upsilon(g) \cdot \lambda. \end{equation}
 The differential $D_e{\phi_\upsilon}: \T_e\Gres \to \T_e\Gl_{\rm res} \oplus \IC$ induces a linear isomorphism \begin{equation}\dot{\phi_\upsilon}: \tilde{\mathfrak{g}}_1 \xrightarrow{\;\cong\;} \mathfrak{g}_1 \oplus \IC,\end{equation} 
which becomes a Lie algebra isomorphism, if the Lie bracket on $\mathfrak{g}_1 \oplus \IC$ is defined as in \eqref{Lieklammer} to include the algebra 2-cocycle induced by $\dot{\upsilon}$. Then, the canonical right-invarient Maurer-Cartan form, followed by the projection defined by this isomorphism, defines a right-invariant connection one-form. Conversely, it is true that every right-invariant connection on $\Gres(\HH)$ comes from a one-form of this kind.\\ 
 
\noindent We can understand all this better from a more abstract, geometrical perspective. Remember the geometric interpretation of a connection as a distribution in the tangent-bundle $\mathrm{T}\Gres$, distinguishing \textit{horizontal} vectors complementary to $\ker(\mathrm{D}\pi)$, which is spanned by vector fields with flow-lines along the fibres of the principle bundle.\\ 
Formally, a connection is a smooth subbundle $\Gamma \subset \mathrm{T}\Gres$ satisfying $\forall\,p \in \Gres(\HH)$:

\begin{equation}\begin{split}\label{rightinvariantconnection} &i) \;\;\; \mathrm{T}_p \Gres = \ker(D_p\pi) \oplus \Gamma_p \\
& ii) \; \; (r_c)_* \Gamma_p = \Gamma_{pc}, \; \forall c \in \IC^{\times}
\end{split} \end{equation}
A right-invariant connection has to satisfy in addition\footnote{This property can actually be regarded as a strenghtened version of ii), although, strictly speaking, the actions involved are not really the same.}

\begin{equation}
iii)\;\;  (R_g)_* \Gamma_p = \Gamma_{pg}, \; \forall g \in \Gres(\HH).
\end{equation}

What this tells us is that a right-invariant connection is uniquely determined by the choice of a subspace $\Gamma_e$ complementary to $\ker D_e\pi$ in $\mathrm{T}_e{\Gres}$. Sine then, by iii), we have 
\begin{equation*} \Gamma_p = (R_p)_*\, \Gamma_e, \; \forall p \in \Gres(\HH).\end{equation*} 

\noindent We may think of such a subspace $\Gamma_e$ as an imbedding of $\mathrm{T}_e\Gl_{\rm res} \cong {\mathfrak{g}}$ into $ \mathrm{T}_e{\Gres} \cong \tilde{\mathfrak{g}}$. And such an imbedding can always be realized by the differential $D_e\upsilon: \mathrm{T}_e\Gl_{\rm res} \rightarrow \mathrm{T}_e{\Gres}$ of a local section $\upsilon: \Gl_{\rm res} \to \Gres$. This brings us back to the relationship between local sections and bundle connections discussed above for the connection one-forms.

\begin{Lemma}[Right-invariant Connections and local Sections]
\mbox{}\\ 
Let $\Gamma$ be a connection on $\Gres$, invariant under the right-action of the Lie group on itself.
Let $\upsilon: \Gl_{\rm res} \rightarrow \Gres$ be a local section around the identity with $\upsilon(e) = e$ and 
$D_e\upsilon (\mathrm{T}_e\Gl_{\rm res}) = \Gamma_e$. Then $\Gamma$ corresponds to the connection one-form $\Theta_{\upsilon} := \prc^{\upsilon} \circ \omega_R$. 
\end{Lemma}

\begin{equation*}
\begin{xy}
  \xymatrix{
\upsilon: \Gl_{\rm res} \to \Gres \ar[d]_{\mathrm{D}\upsilon} \ar@{<->}[r] & \phi_{\upsilon}: \Gres \to \Gl_{\rm res} \times \IC^{\times} \ar[d]^{\dot\Phi_\upsilon}   \\
 \mathrm{T}_p \Gres \cong \ker(D_p\pi) \oplus D_e\upsilon (\mathrm{T}_e\Gl_{\rm res})\ar[d]\ar@{<-->}[r]     & \tilde{\mathfrak{g}} \cong \mathfrak{g} \oplus \IC\ar[d] \\
 \Gamma\bigl\lvert_p = (R_p)_* D_e\upsilon (\mathrm{T}_e\Gl_{\rm res})  \ar@{<->}[r] & \Theta^{\upsilon} = \mathrm{pr}^{\upsilon}_{\IC} \circ \omega_R
  }
\end{xy}
 \end{equation*}

\begin{proof} 
Note that $\im(D_e\upsilon) = \ker(\prc^{\upsilon})$. By assumption, $X \in \mathrm{T}_e\Gres$ is horizontal, i.e. $X \in \Gamma_e$, if and only if $X \in \im(D_e\upsilon)$. And this holds,  if and only if $X  \in \ker\bigl(\prc^\upsilon\bigl) =  \ker\bigl(\prc^\upsilon \circ \omega_R(e)\bigl)$. Since both, the distribution $\Gamma$ and the kernel of $\prc^{\upsilon} \circ \omega_R$ are invariant under the right-action of $\Gres$ on itself, the identity holds everywhere.\\
\end{proof}

\begin{Theorem}[Uniqueness of the Connection]\label{Thm:UniquenessofConnection}
\mbox{}\\ 
The connection $\Gamma_\Theta$ defined in \eqref{connection} is the unique connection on $\Gres(\HH)$ which is invariant under the right-action of the Lie group on itself and whose curvature equals the Schwinger cocycle $c$ in the sense of Prop. \ref{Prop:curvatureSchwinger}.
Horizontal lift from $\mathfrak{g}_1 \cong \mathrm{T}_e\Gl_{\rm res}$ to $\tilde{\mathfrak{g}}_1 \cong \mathrm{T}_e\Gres$ with respect to this connection, corresponds to the second quantization $\mathrm{d}\Gamma$, defined by normal ordering, in the Fock representation on $\FF = \bigwedge \mathcal{H}_+ \otimes  \bigwedge \mathcal{C (H_-)}$.
\end{Theorem}
\begin{proof}
Recall that the connection $\Gamma_{\Theta}$ comes from the local trivialization defined by the section $\tau$ \eqref{tau}.  If we take a connection $\Gamma'$, different from $\Gamma_{\Theta}$ but also invariant under the right-action of $\Gres$, the previous Lemma tells us that it comes from a local section $\upsilon$ with $D_e\upsilon \neq D_e\tau$. From the discussion of central extensions of Lie algebras in \S 3.5. we know that this means that the Lie algebra cocycle $c$ corresponding to $\upsilon$ differs from the Schwinger cocycle, corresponding to $\tau$, by a homomorphism (a ``coboundary'' in the sense of cohomology) 
\begin{equation*} \mu  = D_e\upsilon  - D_e\tau = \dot{\upsilon} - \dot{\tau}: \mathfrak{g}_1 \rightarrow \IC\end{equation*}  i.e. $c(X,Y) = c(X,Y) - \mu([X,Y])$ for $X,Y \in \mathfrak{g}_1$. Therefore, in the sense of Prop. \ref{Prop:curvatureSchwinger}, the curvature of the connection $\Gamma'$ differs from the Schwinger cocycle by $\mu([\cdot,\cdot])$.\\

Finally, comparison with \eqref{Schwingerinrepr} tells us that under the isomorphism \eqref{FgeomtoFock}, horizontal lift w.r.to the connection $\Gamma_{\Theta}$ corresponds to the second quantization prescription $\mathrm{d}\Gamma$ on\\ $\FF = \bigwedge \mathcal{H}_+ \otimes  \bigwedge \mathcal{C (H_-)}$, since in both cases, the resulting cocycle is the Schwinger-term.\\
% \eqref{Schwingercocycle}. \\
\end{proof}

\noindent In this sense, the connection we have defined is ``unique''. Of course, it's in no way necessary to demand that the curvature of the connection equals the Schwinger cocycle. We can just agree that it's nice if it does. 
%More importantly, though, comparison with \eqref{Schwingerinrepr} tells us that horizontal lift from $\mathfrak{g}_1 \cong \mathrm{T}_e\Gl_{\rm res}$ to $\tilde{\mathfrak{g}}_1 \cong \mathrm{T}_e\Gres$ with respect to the connection $\Gamma_{\Theta}$, corresponds to the second quantization $\mathrm{d}\Gamma$ by normal ordering in the Fock representation on $\FF = \bigwedge \mathcal{H}_+ \otimes  \bigwedge \mathcal{C (H_-)}$.
\noindent Also note that the arguments used in the proof of Thm. \ref{Thm:UniquenessofConnection}, together with the non-triviality of the Schwinger cocycle, prove that there exists no flat right-invariant connection on $\Gl_{\rm res}(\HH)$.

\chapter{Geometric Second Quantization}

Parallel Transport in the $\Gres(\HH)$-bundle was suggested by E.Langmann and J.Mickelsson as a method to fix the phase of the second quantized scattering matrix in external-field QED. We will call this method \textit{Geometric Second Quantization}. The bundle connection allows us to lift -- in a unique way -- paths from the base manifold $\Ur(\HH)$ to the central extension $\Ures(\HH) \subset \Gres(\HH)$ that carries the information about the geometric phase. It is important that these horizontal lifts are determined by the connection, i.e. the geometric structure of the bundle only. We will prove that this fact ensures that the geometric second quantization is \textit{causal}, i.e. preserves the causal structure of the one-particle Dirac theory.\\

We present the method of geometric second quantization in a more general setting, as a method of second quantization of the whole unitary time evolution. However, the time evolution will always require \textit{renormalization} and the physical significance of the renormalized time evolution is unclear. We will discuss this problem in the final chapter of this work.\\
%Why is it important to fix the phase of the second quantized operators?\\
%For once, as a matter of principle we should insist that if we want to take a physical state seriously, it must be represented by a well-defined mathematical object in the theory.\\
%But apart from that, the phase of the second quantized S-matrix does play a role in the further formulation of the theory.\\

One might think that the phase of the S-matrix is of little relevance in a quantum theory, or simply a consequence of the gauge-freedom in QED. But this is not the case. In mathematically rigorous formulations of the theory, the phase of the second quantized scattering matrix does appear as a relevant quantity and it seems to be the ill-definedness of this quantity that leads to (additional) divergences in perturbation theory.
It is well known that the \textit{vacuum polarization} in QED is ill-defined and requires various renormalizations to be made finite (see \cite{Dys} for a very ``honest'' computation). One reason for this is that the current density, usually defined as
\begin{equation}\label{currentdensitiyOLD} j^\mu(x) \, = e \overline{\Psi}(x) \gamma^\mu \Psi(x), \end{equation}
is not a well-defined object in the second quantized theory.
The better definition can be given in terms of the second quantized scattering operator $\mathbf{S}$ by 
\begin{equation}\label{currentdesity}\fingbox{ j^\mu(x) := i \mathbf{S}^* \frac{\delta}{\delta \sa_\mu(x)} \mathbf{S}[\sa]}\end{equation}
which equals $e : \overline{\Psi}(x) \gamma^\mu \Psi(x):$ in first order perturbation theory (see \cite{Scha}, \S 2.10). Herby, $\mathbf{S}[\sa]$ denotes the map sending $\sa \in C^\infty_c(\IR^4, \IR^4)$ to the second quantized scattering operator in $\Ures(\HH)$, corresponding to the interaction defined by $\sa$. The current-density itself has to be understood as an \textit{operator-valued distribution}. The \textbf{vacuum-polarization} is then well-defined as the vacuum expectation value
\begin{equation*} \addtolength{\fboxsep}{5pt}\boxed{\bigl\langle \Omega, j^\mu(x) \Omega \bigr\rangle = i \, \bigl\langle \mathbf{S}\, \Omega, \frac{\delta\mathbf{S}}{\delta A_\mu(x)}\, \Omega \bigr\rangle} \end{equation*}
or, more precisely, as a distribution evaluated at the test-function $A_1$,
\begin{align*}  \bigl\langle \Omega, j[A](A_1) \Omega \bigr\rangle = & i\, \frac{\partial}{\partial \epsilon}\bigl\lvert_{\epsilon=0}  \bigl\langle \Omega, S^{-1}(A) S(A + \epsilon A_1)\Omega \bigr\rangle\\
= & i\, \frac{\partial}{\partial \epsilon}\bigl\lvert_{\epsilon=0}  \log \bigl\langle \Omega, S^{-1}(A) S(A + \epsilon A_1) \Omega \bigr\rangle.
\end{align*}
Here, the phase of the S-operator enters explicitely. If we separate the phase-freedom, we find \begin{equation}
\mathbf{S}[\sa] = \tilde{\mathbf{S}}[\sa] e^{i\varphi[\sa]} \Rightarrow  \frac{\delta\mathbf{S}}{\delta A_\mu(x)} = i  \frac{\delta \varphi}{\delta A_\mu(x)} \mathbf{S} + e^{i\varphi}  \frac{\delta\tilde{\mathbf{S}}}{\delta A_\mu(x)} \end{equation}
and for the vacuum-polarization:
\begin{equation}\label{eq:Strom von Phase}\bigl\langle \Omega, j^\mu(x) \Omega \bigr\rangle = i \Bigl[\bigl\langle \Omega, i  \frac{\delta\, \varphi[\sa]}{\delta A_\mu(x)} \,\Omega \bigr\rangle + \bigl\langle \tilde{\mathbf{S}}\, \Omega, \frac{\delta \tilde{\mathbf{S}}}{\delta A_\mu(x)}\, \Omega \bigr\rangle \Bigr] \end{equation}
The second term on the right-hand side is well-defined, but the first term obviously requires a well-defined prescription for the phase of the S-matrix. We will try to provide this now.

\section{Renormalization of the Time Evolution}
Our motivation for studying the parallel transport is the second quantization of the Dirac time evolution. Given an external field $\sa= (\sa_0,- \underline{\sa}) \in C_c^{\infty} (\mathbb{R}^4, \mathbb{R}^4)$,
we have the corresponding unitary time evolution $U^\sa(t, t')$. It is usually more convenient to study the time evolution in the \textit{interaction-picture} which is related to $U^\sa(t, t')$ by
\begin{equation} U^\sa_I(t,t') = U^0(0,t)\,U^\sa(t, t')\,U^0(t',0); \;  t,t' \in \IR.\end{equation}
A \textit{second quantization} of the time evolution between $t_1$ and $t_2 $ corresponds to a \textit{lift} of the path 
\begin{equation}\label{Upath} s \mapsto U^\sa_I(t_1+s, t_1), \; s \in [0 , t_2-t_1]\end{equation}
to the group $\Ures(\HH)$ that acts on the Fock space. This would provide a well-defined prescription for the implementation of the time evolution on the Fock space, including phase. In particular, since the interaction has time-support contained in some compact interval $[-T, T]$ for $T$ large enough, 
\begin{equation}\label{Spath}
t \mapsto U^\sa_I(t , -T), \; t \in [-T, T] \end{equation} 
is a path from the identity $\mathds{1}_{\HH}$ to the S-Matrix $S= U^\sa_I(T, -T) = U^\sa_I(\infty , -\infty)$.\\

\noindent Now a bundle connection on $\Ur(\HH)$ or $\Gres(\HH)$, respectively, as introduced in the previous chapter,  does precisely that: it defines unique lifts of (smooth) paths from $\cu_{\rm res}(\HH)$ to the principle bundle $\Ures(\HH)$. However, as the theorem of Ruijsenaars (see Thm. \ref{Ruij} and especially Thm. \ref{Thm:polarizationclasses}) tells us, the paths \eqref{Upath} are typically NOT in $\Ur(\HH)$; in fact they will leave $\Ur(\HH)$ as soon as the spatial component $\underline{\sa}$ of the interaction potential becomes non-zero. Therefore, in order to be able to apply the method of parallel transport, Langmann and Mickelsson introduced a \textit{renormalization} of the time evolution, such that the transformed time evolution $U^\sa_{ren}(t , t')$ stays in $\Ur(\HH)$, for all $t,t' \in \IR$ and such that $U^\sa_{ren}(T , -T) = U^\sa_I(T , T) = S$.\\

\noindent Concretely, they prove the following:

\begin{Theorem}[Langmann,Mickelsson 1996]
\mbox{}\\
Let $\sa\in {\mathcal C}^\infty_c(\IR^4,\IR^4)$ a 4-vector potential and $U^\sa(t,t')$ the corresponding Dirac time evolution. There is a family of unitary time evolutions $\T_t(\sa), t \in \mathbb{R}$ such that the modified time evolution 
%\begin{equation*}U^\sa_{ren} (t,t') := e^{iD_0t} \T_t^{-1}(\sa)\, U^\sa (t,t')\, T_{t'}(\sa)e^{-iD_0t}
%\footnote{Note that the roles of $T$ and $T^{-1}$ are interchanged in out convention compared to \cite{LaMi}} \end{equation*}
\begin{equation*}\T_t^{-1}(\sa)\, U^\sa (t,t')\, \T_{t'}\;  \end{equation*}
belongs to $\Ur(\HH)$ for all $t,t' \in \IR$.\footnote{Note that the roles of $\T$ and $\T^{-1}$ are interchanged in our convention as compared to \cite{LaMi}.}
% and, in the interaction picture, is differentiable in t w.r.to the differentiable structure of $\Ur(\HH)$
\item Moreover, $\T(\sa)$ can be chosen such that $\T_t(\sa) = \mathds{1}$ if $\sa(t) =0$ and $\partial_t \sa(t) = 0$.%\footnote{see the explicit expression of $\T(\sa)$ in \cite{Mi98}.}
\end{Theorem}
\noindent An explicit expression for the operators $\T_t(\sa)$ is given in \cite{Mi98}. With the abbreviations $\slashed{\sa} = \sum_{\mu} \alpha^\mu A_\mu$ and $\slashed{E}:= \partial_t \slashed{\sa} - \slashed{p}\sa_0 + [\sa_0,\slashed{\sa}]$ it is defined by
\begin{equation}\label{LaMirenormalization} \T^*_t(\sa) = U^0(t,0) \exp\Bigl(\,\frac{1}{4} \bigl[D_0^{-1}, \slashed{\sa}\bigr] - \frac{1}{8} \bigl[D_0^{-1} \slashed{\sa}D_0^{-1}, \slashed{\sa}\bigr] - \frac{i}{4}D_0^{-1}\slashed{E}D_0^{-1}\Bigl) U^0(0,t). \end{equation} 

%\begin{Theorem}\label{Liftbarkeit}(Implementability of the S-matrix)\\
%Under the conditions of the previous Lemma with $U^\sa$ the unitary time evolution in the interaction picture, the scattering matrix
%\begin{equation} S = \lim_{\substack{t\rightarrow \infty \\ t' \rightarrow -\infty}} e^{itD_0}\,U^\sa(t , t')\,e^{-it'D_0} \end{equation} 
%\begin{equation} S = \lim_{\substack{t\rightarrow \infty \\ t' \rightarrow -\infty}} U^\sa(t , t')\end{equation} 
%is assured to be in $\Ures(\HH)$. Therefore, it can be lifted to a unitary transformation $\mathcal{S}: \FF \rightarrow \FF$ on the Fock space. This lift is unique up to a phase.   
 %\end{Theorem}
 %\begin{proof} By the previous Lemma, the renormalized time evolution satisfies\\
%$U'(t , t') \in \Ures, \; \forall t,t' \in \IR$ and $U'(t , -t) = U(T , -T), \; \forall t \geq T$.
%Since $S = \lim\limits_{t \rightarrow \infty} U(t , -t)   = U(T,-T) = U'(T , -T)$, the claim follows immediately.
 %\end{proof}
 
\noindent In the original paper of Langmann and Mickelsson, the renormalization appears as a mere technical tool - the meaning of the unitary transformations $\T_t(\sa)$ is not discussed. Apart from the question of differentiability, which will be the focus of the next section, this meaning becomes more evident by the following considerations:\\

\noindent Let $U(t,t')$ by the unitary evolution for a fixed external field $\sa$ and let $\T(t), \, t \in \mathbb{R},$ a family of unitary operators such that 
$\T(t)^{-1} U(t,t') \T(t') \in \cu^0_{\rm res}(\HH), \; \forall t,t' \in \IR$ and $\T(t) = \mathds{1}$ for $\lvert t \rvert$ large enough. In particular, for $t_0 \ll 0$ outside the time-support of $\sa$, this means
\begin{align*} &\T^{-1}(t)\, U(t, t_0) [\HH_-] = [\HH_-] \in \pol(\HH)/\approx_0\\[1.5ex]
\iff \; & U(t, t_0) [\HH_-] = \T(t) [\HH_-] \in \pol(\HH)/\approx_0 \\[1.5ex]
\iff \; & U(t, t_0) \in \cu^0_{\rm res}(\HH; [\HH_-], [\T(t) \HH_-] ).
\end{align*}
Hence, by the composition property \eqref{eq:Urescomposition},
\begin{equation}U(t, t') =  U(t, t_0)U(t_0, t') \in \cu^0_{\rm res}(\HH; [\T(t')\HH_-],  [\T(t)\HH_-] ).\end{equation}
Ergo, the operators $\T(t)$ identify the correct polarization classes, between which the Dirac time evolution is mapping. In other words, the renormalization satisfies
\begin{equation} [\T_t(\sa)\HH_-]_{\approx_0} = C[\sa(t)] \in \pol(\HH)\slash_{\approx_0}, \end{equation} 
where $C[\sa(t)] = C[\underline\sa(t)]$ are the polarization classes identified in Thm. \ref{Thm:polarizationclasses}.\\

\noindent So we see that we can interpret the renormalization in two different ways:
\begin{enumerate}
\item We can use  the $\T_t$-operators  to renormalize the unitary time evolution, i.e. transform it back to $\Ur(\HH)$ and implement it on the standard Fock space. We will see that a smoothening renormalization is actually equivalent to a renormalization of the Hamiltonian, i.e. it results in a modification of the interaction-potential that makes it well-behaved and prevent the creation of infinitely many particles.
\item We can use the ``renormalization''  to identify the correct polarization classes and implement the time evolution as unitary transformations between time-varying Fock spaces as explained in Chapter 6. Obviously, the unitary operators contain even more information that we can use to identify \textit{instantaneous vacuum states}. For the geometric construction, this means that we identify 
\begin{equation} W(t) := \T(t) \HH_- \end{equation}
as the new projective vacuum at time $t$ and use it to construct the geometric Fock space $\FF_{\rm geom}^{W(t)}$.  
In the language of the infinite wedge spaces, we start with a Dirac sea  $\Phi_0: \ell \rightarrow \HH_- \in \ocean(\HH_-)$  corresponding to the free vacuum state in the initial Fock space $\FF_{\cs_0}$ with $\cs_0 = \cs(\Phi_0)$. Then, $\T(t)\Phi_0$ is a Dirac sea with image in the polarization class $[U(t, t_0)\HH_-]_{\approx_0}$ and we can implement $U(t, t_0)$ as a unitary map between the Fock spaces $\FF_{\cs_0}$ and $\FF_{\cs_t}$, where $\cs_t = \cs(\T(t)\Phi_0)$.

The corresponding physical picture is that the Dirac sea is being ``rotated'' with time, thereby changing our notion of what we call the vacuum, i.e. which configuration of the Dirac sea we perceive as ``empty''.\\ 
\end{enumerate}

\noindent With this understanding it becomes clear that such a renormalization is \textit{not at all unique} but represents a very particular choice. It seems that this point was not evident to Langmann and Mickelsson by the time of their '96 publication, although they pick up the issue in a later publication \cite{Mi98}. In \cite{LaMi}, however, it is not discussed whether and how the results depend on the particular choice of the renormalization. Bad news is that, as it turns out, the entire freedom of the geometric phase is now contained in the freedom of choice of a renormalization. It merely gets a different name: geometrically, it is described by the \textit{holonomy group} of the principle bundle. We will make this more precise in Section 8.3.\\

\noindent To study renormalizations more thoroughly, we propose a general definition:

\begin{Definition}[Space of Vector Potentials]
\item Let $\A$ be the space of 4-vector potentials (equipped with a suitable topology). 
\item In our context, $\A =  C^\infty_c(\IR^4, \IR^4)$ and we write $\A \ni \sa = (\sa_{\mu})_{\mu=0,1,2,3} = (\sa_0,- \underline{\sa})$.
\item On a general space-time $\IR \times M$, with $M$ a compact manifold without boundary, $\A$ corresponds to the space  $\A = \Omega^1(\IR \times M, \IR)$ of smooth connection one-forms.\footnote{More generally, the one-forms take values in the Lie algebra of the gauge group which for QED is just $\mathrm{Lie}(\cu(1)) = \IR$.} 
\item By $\sa(t)$ we always mean the function $\sa(t, \cdot) \in C^\infty_c(\IR^3, \IR^4)$ for fixed $t \in \IR$.
\end{Definition} 
\newpage
\begin{Definition}[Renormalization]\label{Def:Renormalization}
\item We call a mapping $\T: \IR \times \A \rightarrow \cu(\HH)$ a \emph{renormalization} if it satisfies
\begin{enumerate}[i)]
\item $\T_t(0) \equiv \mathds{1}$
\item $\T(t , \sa) = \T_t(\sa) \in \cu^0_{\rm res}\bigl(\HH; [\HH_-], C[\sa(t)] \bigr), \; \forall t \in \IR$
\item $\T_t(\sa)$ depends only on $\sa(t)$ and $\partial^k \sa(t)$ for $k=0,1,...,n$ and some $n \in \IN$ 

\hspace*{-0.7cm}We call $\T$ a \emph{smoothening renormalization}, if it has the additional property
\item For any $\sa \in \A$, the \emph{renormalized (interaction picture) time evolution} 
\begin{equation}\fingbox{U^{\sa}_{ren}(t , s) =  e^{itD_0}\T^*(t) \, U^\sa (t,s)\, \T(s) e^{-isD_0}}\end{equation} is continuously differentiable in $t$ w.r.to the differentiable structure of $\Gl_{\rm res}(\HH)$. 
\end{enumerate}
\end{Definition}
%\newpage
\noindent Note: \begin{itemize} 
\item In our definition, the renormalized time evolution is an interaction-picture time evolution.
\item i) and iii) together imply that $\T_t(\sa) = \mathds{1}$, whenever $\sa$ vanishes in some time-interval around $t$. In particular, this assures that a renormalization doesn't alter the S-operator for compactly supported interactions.
\item iii) formulates a requirement of \textit{causality}. It states that the renormalization {depends only on the A-field, locally in time}. In particular, if the renormalization depends only on $\sa(t)$ and not its time-derivatives, it makes sense to regard it as determining one vacuum state (and therefore one Fock space) over the polarization class $C[\sa]$ for every $\sa \in \A$. Then we would have a ``global'' choice, suitable for \textit{any} Dirac time evolution, which is of course different (and arguably better) than choosing a family of Fock spaces for a \textit{fixed} time evolution. We will come back to this in Section \ref{Subsec:Notes on Renormalizations}.
\end{itemize}
%\noindent Note that by our definition, the renormalized time evolution is an interaction-picture time evolution. 

\begin{Example}[Renormalizations]
\mbox{}\\
\vspace*{-4mm}\begin{itemize}
\item The renormalization $\T(\sa)$ constructed in \cite{LaMi} is a smoothening renormalization in the sense of Def. \ref{Def:Renormalization} (with n=1).
\item The operators $e^{Q^{\sa(t)}}$ introduced in \S 6.1. provide a renormalization (with n=0) which is not smoothening (cf. \S \ref{Subsec:Notes on Renormalizations}).
\end{itemize}
%Unfortunately, to our knowledge these are the only two examples that have been explicitly constructed so far.
\end{Example}

%\begin{Definition}[Renormalization of the Time Evolution]\label{Defrenormalization}
%\mbox{}\\
%Let $\sa\in {\mathcal C}^\infty_c(\IR^4,\IR^4)$ a vector potential and $U^\sa$ the corresponding Dirac time evolution. We propose the following nomenclature:\\
%We call a family of unitary transformations $\T(t) , \, t \in \IR$ a \emph{quasi-renormalization} for the time evolution if it satisfies
%\begin{enumerate}[i)]
%\item $[\T(t)\HH_+]  = [U^\sa (t , -\infty) \HH_+] \in \pol(\HH)\slash\approx_0,\; \forall t\in \IR$ 
%\item$\T(t)$ has compact support, i.e. $\T(t) = \mathds{1}$ for all $\vert t \vert$ sufficiently large 
%\end{enumerate}

%\noindent Furthermore, we want propose the following nomenclature:\\
%A family of unitary operators satisfying i) and ii) is called 
%\begin{itemize}
%\item \emph{local in time} if $\T(t) = \mathds{1}$ whenever $A(t)=0$.
%\item \emph{minimally dependent on A} if $\T(t) = \mathds{1}$ whenever $\underline{A}(t) = 0$ 
%\item \emph{causal} if 
%\end{itemize}

%Still, the fixing of the phase via parallel transport constitutes some progress, as it allows us to study the geometric phase from a different perspective and maybe reduce the freedom by finding reasonable physical conditions for the renormalization procedures. We will give a hint on how this could be done in xxx.
%Langmann and Mickelsson pursue a different strategy in ???, where they propose a different connection for which they show that parallel transport is at least infinitesimally path-independent. 

\begin{Remark}[Interpolation Picture]
\mbox{}\\
The most obvious way to transform the time evolution back to $\Ur(HH)$ is by the time evolution itself, i.e. to set $\T_t(\sa) = U^\sa(t , - \infty)$. Then, the renormalized time evolution is
\begin{equation*} U^{*}_I(t , - \infty) U_I(t , t') U_I(t', - \infty) =\mathds{1}, \; \forall t,t' \in \IR, \end{equation*}
where, $t=-\infty$ can be understood as a large negative $t$ outside the time-support of the interaction. This approach is also known as ``interpolation picture''. The only problem with the interpolation picture: nothing's happening. The polarization, i.e what we call ``particles'' and ``antiparticles'', evolves in just the same way as the states themselves: 
\begin{equation}\HH = U(t, -\infty)\HH_+ \oplus  U(t, -\infty)\HH_-.\end{equation}
What we end up caling the ``vacuum'' at time $t$ is exactly the state into which the original ($t \rightarrow -\infty$) vacuum has evolved. There is no particle creation or annihilation -  an empty universe remains empty. For the obvious reasons, we are not satisfied with that. From theorem \ref{Liftbarkeit}, we know that at least the S-matrix is in $\Ur(\HH)$, hence we can make sense of the particle/antiparticle - picture at least asymptotically. Before the interaction is switched on and after it's switched off, we can determine the particle content with respect to the same vacuum and the question:  `how many particles and anti-particles were created?' has a well-defined answer. We demand of our renormalization to grant us this, at least.
\end{Remark}

\subsection{Renormalization of the Hamiltonians}
Now we shift our focus to the question of differentiability which makes all the difference between a renormalization and a smoothening renormalization. To make use of the geometric structure introduced in Chapter 7 and apply the method of parallel transport, we need the time evolution to be differentiable - notably with respect to the differentiable structure on $\Gl_{\rm res}(\HH)$ (or $\Ur(\HH)$, respectively), induced by the norm \eqref{EpsilonNorm}
%$\lVert \cdot \rVert_{\epsilon} = \lVert\cdot\rVert + \lVert [\epsilon, \cdot]\rVert_2$ 
on the Banach-algebra $\mathcal{B}_{\epsilon}(\HH)$. This notion of differentiability is very strong. The interaction picture time evolution is generally differentiable in the operator norm, but now we need to control the Hilbert-Schmidt norms as well. We will see that this requirement is not harmless.

\subsubsection{Renormalization of the Interaction Hamiltonians}

\noindent Recall that if $U^\sa(t,t')$ solves the equation of motion
\begin{equation}
 \left\{
\begin{array}{ll}
i \,\partial_t\, U^\sa(t,t') = H^{\sa}(t)\,U^\sa(t,t') \\ \\
 U^\sa(t',t') = \mathds{1} 
\end{array}
\right. 
\end{equation}
with the Hamiltonian 
\begin{equation*} H^{\sa}=  D_0 + e\sum\limits_{\mu=0}^3\alpha^\mu A_\mu = D_0\,+\, V^\sa(t), \end{equation*}
the corresponding interaction-picture time evolution
 \begin{equation*} U^\sa_I(t,t') = e^{itD_0}U^\sa(t,t')e^{-it'D_0} \end{equation*} is a solution to the equivalent equation
\begin{equation}\label{DiracIP}
 \left\{
\begin{array}{ll}
i\, \partial_t\, U^\sa_I(t,t') = H_I(t)\,U^\sa_I(t,t')\\ \\
 U^\sa_I(t',t') = \mathds{1}
\end{array}
\right. 
\end{equation}
with \begin{equation} H_I(t) = e^{itD_0}V^\sa(t) e^{-itD_0}. \end{equation} 
This is called the \textit{interaction picture}. The interaction picture is a hybrid between the Schr\"odinger- and Heisenberg-picture. The basic idea is that in the interaction-picture, operators evolve according to the \textit{free} time evolution. The evolution of the states on the other, is then generated by the interaction-part of the Hamiltonian only. One advantage of this procedure is that $V_I(t)$, in contrast to $H^{\sa}(t)$, is a bounded operator and thus the solution of \eqref{DiracIP} is given for all finite times by the norm-convergent \textit{Dyson series}, defined by
%\begin{framed}
\begin{equation}\label{Dyson}\begin{split}&U^\sa_I(t,t') = \sum\limits_{n=0}^{\infty} U_n(t,t') ,\\U_0(t,t') \equiv \mathds{1},\;\;\; & U_{n+1}(t,t') = -i \int\limits_{t'}^t H_I(s)U_n(s,t') \mathrm{d}s.\end{split}\end{equation} 
%\end{framed}
\noindent Now let $\T_t = \T_t(\sa)$ be a smoothening renormalization for this time evolution and consider the modified (Schr\"odinger-picture) time evolution
\begin{equation*} U'(t , t') = \T^*_t \,U^\sa(t , t')\, \T_{t'} = e^{-itD_0}U^\sa_{ren}(t,t')e^{it'D_0}.\end{equation*}
Differentiation with respect to $t$ yields:
\begin{align*} i \,\partial_t \,U'(t , t') =\; & i \,(\partial_t\, \T^*_t)\, U'(t , t')\, \T_{t'} + \T^*_t \,H^{\sa}(t)\, U'(t , t')\, \T_{t'}\\[1.5ex] 
= \; & \bigl[ i\, (\partial_t \,\T^*_t ) \,\T_t + \T^*_t \,H^{\sa}(t)\, \T_t \bigr]\, \T^*_t\, U'(t , t')\,\T_{t'} \\[1.5ex]
=\; & \bigl[ - i\, \T^*_t\, (\partial_t\, \T_t) + \T^*_t\, H^{\sa}(t) \,\T_t \bigr]\,U'(t , t').
%=: & \bigl[ D_0 + i Z^{A}_{ren} \bigr]\, U_{ren}^\sa(t , t')
\end{align*}
We define
\begin{equation*}
\Bigl[ - i\,\T^*_t (\partial_t \T_t) + \T^*_t H^{\sa}(t)  \T_t \Bigr] = : \Bigl( D_0 + i V^{\sa}_{ren} \Bigr)
\end{equation*}
\vspace{1mm}
with 
\begin{equation}\label{Zren}
\addtolength{\fboxsep}{3pt}\boxed{V^{\sa}_{ren} = \Bigl[\T^*_t  V^\sa \T_t  + \T^*_t [D_0 , \T_t] - i\T^*_t(\partial_t \T_t)\Bigr].}
\end{equation}\vspace{1mm}

\noindent Thus the renormalized time evolution is generated by the Hamiltonian $H^{\sa}_{ren} = D_0 + V^{\sa}_{ren}$\\
with the ``renormalized'' interaction \eqref{Zren}.\\

\begin{Theorem}[Generators of renormalized Time Evolution]\label{Thm:Urengenerators}
\mbox{}\\
Let $V(t) = V^\sa_{ren}(t)$ be a (renormalized) interaction potential and $h(t) = e^{itD_0}V(t) e^{-itD_0}$.\\
Let $U(t,t') (= U^\sa_{ren}(t,t') )$ be a solution of
\[ \left\{
\begin{array}{ll}
i\, \partial_t\, U(t,t') = \; h(t)\,U(t,t') \\ \\
U(t',t')\; =\; \mathds{1}
\end{array}
\right.
\]
in the operator-norm on $\cu(\HH)$, given by the Dyson-series \eqref{Dyson}.
\begin{enumerate}[i)]
\item If $U(t,t')$ is a solution in $\Ur(\HH)$, i.e. a solution w.r.to the differentiable structure induced by the norm $\lVert \cdot \rVert_{\epsilon}$, then 
\begin{equation*} [\epsilon, V(t)] \in I_2(\HH),\; \forall t \in \IR, \end{equation*} 
i.e. $V(t)$ and $h(t)$ are in the Lie algebra $\mathfrak{u}_{\rm res}$ of $\Ur(\HH)$.  

\item Conversely, if $[\epsilon, V(t)] \in I_2(\HH)\, \forall t \in \IR$, then $U(t,t')$ is a solution in $\Ur(\HH) \subset \Gl_{\rm res}(\HH)$, if additionally we assume  
\begin{equation}\label{intforh}\int_{\IR} \lVert [\epsilon , V(t)] \rVert_2\, \mathrm{d}t < \infty. \end{equation}
\end{enumerate}
\end{Theorem}

\noindent What is actually proven in \cite{LaMi} for the renormalization constructed by Langmann and Mickelsson is that the \textit{renormalized interaction} stays in $\mathfrak{u}_{\rm res}$. It follows that their unitary transformations are indeed smoothening renormalization in the sense of our Definition \ref{Def:Renormalization}.\\

\noindent Note that for $U(t,t') \in \Ur(\HH)$ it would suffice that 
\begin{equation}\int\limits_{t'}^t \bigl[ \epsilon\,,\, h(t) \bigr] \mathrm{d}t\,= \int\limits_{t'}^t \bigl[ \epsilon\,,\, e^{iD_0t} V(t) e^{-iD_0t}\bigr] \mathrm{d}t < \infty \; \forall t,t' \end{equation} which is much less restrictive than \eqref{intforh}.\\

\noindent The proof of the theorem requires some estimates.

\begin{Lemma}[Estimates]\label{Lem:Estimates}
\mbox{}\\
For the terms in the Dyson series \eqref{Dyson} we get the norm-estimates
\begin{equation}\lVert U_n(t,t') \rVert \leq \frac{1}{n!}\Biggl( \int\limits_{t'}^t \lVert V(s) \rVert \mathrm{d}s \Biggr)^n, \; \forall n \geq 0 \end{equation}
and 
\begin{equation}\begin{split}
& \lVert [\epsilon, U_1(t,t')] \rVert_2  \leq \int\limits_{t'}^{t} \lVert[\epsilon, V(s)]\rVert_2 \; \mathrm{d}s\\
\lVert [\epsilon ,  U_n(t,t')] \rVert_2 \leq & \frac{1}{(n-2)!} \int\limits_{t'}^t \lVert [\epsilon, V(s)] \rVert_2\; \mathrm{d}s \Biggl( \int\limits_{t'}^t \lVert V(r) \rVert \mathrm{d}r \Biggr)^{n-1}, \; \forall n \geq 1
\end{split}\end{equation}
\end{Lemma}
\begin{proof} See Appendix A.3.\\ \end{proof}

\begin{proof}[Proof of the Theorem] First we note that $V(t)$ is Hermitian with $[\epsilon, V(t)] \in I_2(\HH)$ if and only if $h(t) = H_I(t) = e^{itD_0}V(t) e^{-itD_0}$ is, because $e^{itD_0}$ is unitary and diagonal w.r.to the polarization $\HH = \HH_+ \oplus \HH_-$.
Now suppose $U(t,t')$ is in $\Ur(\HH)$ for all $t,t'$ and differentiable in the norm $\lVert \cdot \rVert_{\epsilon}$, solving $ i\, \partial_t\, U(t,t') = \; h(t)\,U(t,t') $. 
Then, for any fixed $t$:
\begin{equation*} - i\,h(t) = \,  \frac{d}{ds}\Bigl\lvert_{s=0}U(t+s,t')U^*(t,t') \in T_e\Ur \cong (-i)\cdot \mathfrak{u}_{\rm res}.\end{equation*}
Conversely, if $\int_{\IR}\rVert [\epsilon, V(t)]\rVert_2\, \mathrm{d}t  < \infty$, the Hilbert-Schmidt norm estimates in the previous Lemma show that $U(t,t') \in \Ur(\HH), \;\forall t,t'$ and 
\begin{equation*} U(t,t') = \mathds{1} -i \int_{t'}^t h(s) \mathrm{d}s + \mathcal{O}(\lvert t-t'\rvert^2),\end{equation*}
so that indeed $ i\, \partial_t\,U(t,t') =  h(t)\,U(t,t')$ (as long as $s \mapsto h(s)$ is continuous in $\mathfrak{u}_{\rm res}$).\\
\end{proof}
\newpage
\begin{Corollary}[Time Evolution always requires renormalization]\label{Cor:badnews}
\mbox{}\\
Let $\sa = (\sa_\mu)_{\mu=0,1,2,3} = (\Phi , - \underline{\sa}) \in \A$.
The interaction potential
\begin{equation} V^\sa(t) = e\, \alpha^{\mu} A_{\mu} = - e\, \underline{\alpha} \cdot \underline{A}  +  e\, \Phi \end{equation}
is not in $\mathfrak{u}_{\rm res}$, unless $\underline{\sa}\equiv 0$ and $\underline\nabla{\Phi}\equiv 0$.
Consequently, $U^\sa_I(t,t')$ is not differentiable in the $\lVert\cdot\rVert_{\epsilon}$-norm, except for those cases.
\end{Corollary}
\noindent Remark:  In general, only the weaker condition \begin{equation} \int\limits_{t'}^t [\epsilon, e^{iD_0t} \Phi(t,\underline{x}) e^{-iD_0t}] < \infty \end{equation} is satisfied (\cite{Ruij77}), which implies $U^\Phi(t,t') \in \Ur(\HH)$ but not the required differentiability.
\begin{proof}We apply Theorem \ref{Thm:Urengenerators}: We know that $U^\sa_I(t,t')$ is not even in $\Ur(\HH)$, unless $\underline{\sa}\equiv 0$, and hence $V^\sa(t) \notin \mathfrak{u}_{\rm res}$, unless $\underline{\sa}\equiv 0$.
Now consider a purely electric potential $V(t) = e\,\Phi$, where $\Phi(x) = \Phi(x) \cdot \mathds{1}_{\IC^4}$ has to be understood as the multiplication operator in $L^2(\IR^3, \IC^4)= \HH$. Naturally:
\begin{equation*} \Phi \in \mathfrak{u}_{\rm res} \iff \exp(i\Phi) \in \Ur(\HH).\end{equation*}
But we know when the latter is the case. $e^{i\Phi(x)}$ is a gauge transformation and thus Thm. \ref{Thm:Gaugetransformations} tells us that it is in $\Ur(\HH)$ if and only if $\underline\nabla \Phi(t,\underline{x}) \equiv 0$. This finishes the proof.\\
\end{proof}

\noindent This observation reveals somewhat of a shortcoming of the method of parallel transport: Even for purely electric potentials, when the time evolution actually stays in $\Ur(\HH)$, i.e is implementable on the standard Fock space, it requires us to apply a smoothening renormalization, in order to get a differentiable path in $\Ur(\HH)$.\\

\subsection{Notes on Renormalizations}
\label{Subsec:Notes on Renormalizations}

We have proposed the general definition \ref{Def:Renormalization} of a ``renormalization'' in order to establish a common framework for discussing and comparing the constructions of Langmann and Mickelsson \cite{LaMi} and Deckert et.al. \cite{DeDueMeScho}. We have also explained how such a renormalization can be used to translate between the second quantization of the renormalized time evolution and the implementation on time-varying Fock spaces (see also the proof of  Thm. \ref{Thm:Time evolution on time-varying Fock spaces} for an actual application of this duality). However, our analysis also shows in what sense this treatment is justified and when we have to differentiate. In particular cases we will have to be very careful with the way the renormalization is actually used. 

For example, we need a \textit{smoothening} renormalization for second quantization of the Hamiltonians, or for applying the method of parallel transport. If however we want to translate the results into the language of time-varying Fock spaces and use the same renormalization to identify instantaneous Fock spaces (respectively vacua), these choices will be restricted by the condition that the renormalized time evolution be differentiable in $\Ur(\HH)$ which is not a sensible requirement any more in the context of time-varying Fock spaces. In particular, the last corollary tells us that we might have to change vacua/Fock spaces, even if the time evolution actually stays in $\Ur(\HH)$.\\
On the other hand, if we are mainly interested in renormalizations as a mean to identify the correct polarization classes and pick out instantaneous vacua/Fock spaces, we have been very generous with the kind of choices that we allow, because we know that the polarization classes actually depend on the spatial part of the electromagnetic potential at fixed time $t$, only,  i.e $C(t) = C[\underline\sa(t)]$. In this context, it could make more sense to use a ``renormalization'' which depends only on the spatial part of the A-field, locally in time. This would correspond to a ``global'' choice of Fock spaces, respectively vacua, rather than allowing different choice for every single time evolution. One could call such a renormalization \textit{minimal}, because it requires only the minimal amount of information from the A-field. We formalize this in the following definition:

%from Thm. \ref{Thm:polarizationclasses} [Identification of Polarization classes],  

\begin{Definition}[Minimal Renormalization]
\mbox{}\\
Let $\A_3 = C^\infty_c(\IR^3, \IR^3)$ (or $\A_3 = \Omega^1(M, \IR)$) be the space of \emph{static,  space-like} vector potentials. We call a map $\T: \A_3 \rightarrow \cu(\HH)$ satisfying %\begin{equation*}T(\underline{\sa}) \in \cu^0_{\rm res}\bigl(\HH, [\HH_+], \HH, C(\uderline\sa)\bigr)\; \text{and}\; T(0) = \mathds{1}\end{equation*}
\begin{enumerate}[i)]
\item $\T(0) = \mathds{1}$
\item $\T(\underline{\sa}) \in \cu^0_{\rm res}\bigl(\HH; [\HH_-], C[\underline\sa]\bigr), \; \forall \underline\sa \in \A_3$
\end{enumerate}
a \emph{minimal renormalization} or a \emph{global choice of vacua}.
\end{Definition}
\noindent We can regard the set of polarization classes $C[\underline{\sa}]$ as a formal $\cu^0_{\rm res}$-bundle over $\A_3$.
A minimal renormalization $\T$ would then correspond to a global section in this bundle, determining one vacuum -- and thereby one geometric Fock space -- over every polarization class $C[\underline{\sa}]$.
Similarly, for the infinite wedge space construction, if we have chosen a sea $\Phi_0 \in \ocean(\HH_-)$ corresponding to the free vacuum state, the assignment $\underline\sa \longmapsto \cs\bigl(\T(\underline\sa) \Phi_0\bigr)$ defines a global choice of Dirac sea classes, and thereby a global choice of infinite wedge spaces over the corresponding polarization classes. For a mathematical treatment of the \textit{Fock space bundle} see \cite{CaMiMu} and \cite{CaMiMu97}, for example. This quite abstract construction however remains to be compared and reconciled with the constructions and results in the present work.

\begin{Example}[Minimal Renormalizations]
\mbox{}
\begin{enumerate}
\item For $\underline{\sa} \in \A_3$, we can set $\sa = (0, - \underline{\sa}) \in C^\infty_c(\IR^3, \IR^4)$ and 
$\T(\underline\sa) := e^{Q^\sa}$. By Thm. \ref{Thm:polarizationclasses} [Identification of Polarization classes], this defines a minimal renormalization in the sense of the previous definition.
\item For $\underline{\sa} \in \A_3$, we can take $\T(\underline\sa)$ to be any unitary transformation sending $\HH_-$ to $P^{\underline\sa}_-(\HH)$, where $P^{\underline\sa}_{\pm}$ are the orthogonal projectors corresponding to the spectral decomposition w.r.to the Hamiltonian with static vector potential $\sa = (0, - \underline{\sa})$. By Thm. \ref{Thm:Scharf Fierz} [Scharf, Fierz], this defines a minimal renormalization in the sense of the previous definition. The resulting choice of instantaneous vacua is also known as the \textit{Furry picture}.\\
\end{enumerate}
\end{Example}

Is it possible to define a \textit{smoothening} renormalization which is ``minimal'' in this sense? The answer is \textbf{no}. And the reason is Corollary \ref{Cor:badnews}, saying that even for a purely electric potential, the time evolution requires renormalization to become differentiable in $\Gl_{\rm res}(\HH)$. Is it possible, at least, to find a smoothening renormalization which depends only on the A-field itself, locally in time, and not on any time-derivatives? Again, the answer is no, as we will show in the following proposition.

\begin{Proposition}[Strictly Causal Renormalizations]
\mbox{}\\
Let $\T: \A \to \cu(\HH)$ be a renormalization, such that for every fixed $t \in \IR$, $\T_t(\sa)$ is only a function of $\sa(t) \in C^\infty_c(\IR^3, \IR^4)$. Then, $\T$ is not smoothening.
%This justifies the notation $\T(\sa(t))$ for $\T_t(\sa)$. $\T$ is a \emph{smoothening} renormalization in the sense of [Def.\ref{Def:Renormalization}], if and only if it satisfies
%\begin{align} 1) \; \Bigl[\T^*(\sa)  V^{\sa} \T(\sa)  + \T^*(\sa) [D_0 , \T(\sa)]\Bigr] \in I_2(\HH) \end{align}
%for all \emph{static} fields $\sa \in C^\infty_c(\IR^3, \IR^4)$  and
%\begin{align} \hspace{-16mm} 2)\;  \bigl[\,\epsilon\;, \T^*(\sa(t))\,\partial_t \T(\sa(t))\bigr] \in I_2(\HH) \end{align}
%\in \mathfrak{u}_{\rm res}
%for all $ \sa \in \A$ and all $t \in \IR$. If we write $\T(\sa) = e^{Q(\sa)}$ with anti-Hermitian operators $Q(\sa)$ (not the ones defined in \S ??? !), property 2) means that $t \to Q(\sa(t))$ is differentiable for all $\sa \in \A$ with $\dot{Q}(\sa(t))_{\rm odd} \in I_2(\HH)$.  
\end{Proposition}
\begin{proof} Given $\sa \in \A$, we may write $\T_t := \T_t(\sa) = \T(\sa(t))$, since $\T_t$ depends only on $\sa(t)$.\\ 
Suppose $\T$ was in fact a smoothening renormalization. Then we know from Thm. \ref{Thm:Urengenerators} [Generators of the renormalized time evolution] and the differential form of the renormalization \eqref{Zren}, that 
\begin{equation}\label{star} \Bigl[\T^*_t  V^{\sa}(t) \T_t  + \T^*_t [D_0 , \T_t] - i\T^*_t(\partial_t \T_t)\Bigr] \in \mathfrak{u}_{\rm res}\end{equation} 
holds true for all $\sa \in \A$ and all $t \in \IR$. But for any fixed $t$, we may use this identity for a vector potential $\sa' \in \A$ which is constantly equal to $\sa(t)$ in a time interval around $t$, i.e. which satisfies
$\sa'(s) = \sa(t), \; \forall\, s \in (t-\epsilon, t+\epsilon)$. For the field $\sa'$ then, the last term in \eqref{star} vanishes, whereas the first two terms must agree with those for $\sa$, because $\T_t$ depends only on $\sa(t)$.
We conclude that \eqref{star} implies that both 
\begin{equation*} a) \; \T^*(\sa(t))  V^{\sa}(t) \T(\sa(t))  + \T^*(\sa(t)) [D_0 , \T(\sa(t))] \in \mathfrak{u}_{\rm res} \end{equation*} 
\vspace*{-2mm}and
\vspace*{-2mm}\begin{equation*} b) \; - i\, \T^*(\sa(t))\,\partial_t \T(\sa(t)) \in \mathfrak{u}_{\rm res} \end{equation*} 
hold separately for any fixed $t \in \IR$.
Writing $\T(A) = e^{Q(A)}$ for all \emph{static} fields $A \in C^\infty_c(\IR^3, \IR^4)$, with anti-Hermitian operators $Q(A)$ (not the ones defined in \S 6.1!), property $b)$ states that $t \to Q(\sa(t))$ is differentiable for all $\sa \in \A$ with $[\epsilon, \dot{Q}(\sa(t))] \in I_2(\HH)$. But this cannot be the case. If we take, for example, an $\sa \in C^\infty_c(\IR^4, \IR^4)$ with $\sa(t) \equiv 0$ for $t \leq 0$, but $\underline{\sa}(1) \neq 0$, we conclude
\begin{align*} \bigl\lVert [\epsilon, Q(\sa(1))] \bigr\rVert_2 = \bigl\lVert \bigl[ \epsilon , \int\limits_{0}^1 \dot{Q}(\sa(t)) \, \mathrm{d}t\bigr] \bigr\rVert_2 
 \leq \int\limits_{0}^1\, \bigl\lVert\bigl[ \epsilon,  \dot{Q}(\sa(t))\bigr] \bigr\rVert_2 \, \mathrm{d}t < \infty 
 \end{align*}
(assuming integrability on the right-hand side) and thus $Q(\sa(1)) \in \mathfrak{u}_{\rm res}$, hence $\T_1(\sa) = e^{Q(A(1))} \in \Ur(\HH)$. This is in contradiction with $\T$ being a renormalization, because $\underline{\sa}(1) \neq 0$ implies by Thm. \ref{Thm:polarizationclasses} :  $C[\sa(t)] \neq [\HH_-]$.
We conclude that the renormalization $\T$ cannot be smoothening.\\  
\end{proof}

\noindent These considerations help to illustrate how restrictive the ``smoothing'' property for a renormalization actually is. A smoothening renormalization will always require more information on the interaction potential, than is necessary to identify the polarization classes, because it has to control the infinitesimal variation of the field as well. In particular, it follows that the renormalization defined by the operators $e^{Q^{\sa(t)}}$, constructed in \cite{DeDueMeScho} and introduced in Section 6.1, is not smoothening.

\newpage

%\begin{Corollary}[$e^{Q^{\sa(t)}}$ is not smoothening]\label{Cor:expQnotrenormalization}
%\mbox{}\\
%The renormalization $e^{Q^{\sa(t)}}$ constructed in \cite{DeDueMeScho} is not smoothening.
%\end{Corollary}
%\begin{proof}  $\T(\sa(t)):= e^{Q^{\sa(t)}}$ is a renormalization which, by construction, depends only on $\sa(t)$ and not on time-derivatives. Suppose it was even a smoothening renormalization. Then, the previous Lemma implies $[\epsilon\,, \dot{Q}(\sa(t))] \in I_2(\HH)$.
% But this cannot be true. Recall that $Q$ and thus $\dot{Q}$, too, are odd w.r.to the splitting $\HH = \HH_+ \oplus \HH_-$.  Now, in \cite{DeDueMeScho} Lemma III.7 it is proven that $\dot{Q}Q$ is always Hilbert-Schmidt. But $Q^\sa \in I_2(\HH)$ if and only if $\sa=0$. Therefore, $\dot{Q}_{\rm odd} = \dot(Q)$ cannot be of Hilbert-Schmidt type, unless $\sa \equiv 0$.\\
%\end{proof}   

%\newpage

\section{Geometric Second Quantization}
Finally, we define the method of second quantization by parallel transport in the principle bundle
$\Gres(\HH) \rightarrow \Gl_{\rm res}(\HH)$, w.r.to the connection $\Gamma_\Theta$ defined by the one-form \eqref{connection}.\\

\noindent In the following, let $\sa \in C_c^{\infty} (\mathbb{R}^4, \mathbb{R}^4)$ and $U_I^\sa(t,t'), \; t,t', \in \IRinf$ the unitary (interaction picture) time evolution for the external field  $\sa$.
After applying a suitable smoothening renormalization $\T_t$ as definded in the pevious section, the renormalized time evolution $U^{\sa}_{ren}(t,t')$ is a two-parameter semi-group in $\Ur(\HH) \subset \Gl_{\rm res}(\HH)$ and continuously differentiable in $t$ with respect to the differentiable structure on $\Gl_{\rm res}(\HH)$. We construct its second quantization i.e. the lift to the group $\Ures(\HH)$ acting on the Fock space by the following prescription:

%Note that ince $\sa$ has compact support in time, so does $U(\cdot,\cdot)$. In particular there exists  $T \geq 0$ such that 
%$U(t_1,t_0) = U(T , -T)$ whenever $t_1 \geq T, t_0 \leq -T$.\\
%\noindent We construct the lift of the renormalized time evolution to the group $\Ures(\HH)$ acting on the Fock space by the following prescription:
%\definecolor{shadecolor}{rgb}{0.92,0.92,0.92}% Wir definieren im RGB-Farbraum
%\noindent 
%\begin{shaded}
%\begin{framed}{\parbox{15cm}{\onehalfspacing\centering{
%\hspace{-3cm}\vspace{-2mm}For $t_1 \geq t_0 \in \IRinf$ construct the lift\
%\hspace{-3cm}\begin{equation}\hspace{-3cm}\vspace{-3mm}\mathfrak{U}(t_1, t_0) \in \Ures(\HH)\end{equation}
%\hspace{-3cm}as the parallel transport of $\mathds{1} \in \Ures(\HH)$ along the path\\
 %\hspace{-3cm}$s \rightarrow U^\sa(t_0 + s , t_0), \; s \in [0, t_1-t_0]$ in $\Ur(\HH)$.}}}\end{framed}
% \end{shaded}

%\begin{shaded}
\begin{framed}{\parbox{15cm}{\onehalfspacing\centering{
\hspace{-3cm}\vspace{-2mm}For $t_1 \geq t_0 \in \IR$ we define the time evolution
\hspace{-3cm}\begin{equation}\hspace{-3cm}\vspace{-3mm}\mathfrak{U}(t, t_0), \; t \in [t_0 , t_1] \end{equation}
\hspace{-3cm} between  $t_1$ and $t_0$ on the Fock space as the $\Gamma_\Theta$-horizontal lift\\
\hspace{-3cm} of the renormalized one-particle time evolution
$U^\sa_{ren}(t , t_0), \; t \in [t_0 , t_1]$\\
  \hspace{-3cm} to $\Ures(\HH)$ with initial condition $\mathfrak{U}(t_0, t_0) = \mathds{1}_{\Ures}$.\\
 }}}\end{framed}
 
  %\end{shaded}
 %[(\mathds{1}, \mathds{1})] \in \Ures(\HH)$
\vspace*{3mm}
\noindent In other words: for $t_1 \geq t_0 \in \IR$ we define the lift $\mathfrak{U}(t_1, t_0) \in \Ures(\HH)$
of $U^\sa_{ren}(t_1, t_0)$ as the \textit{parallel transport} of $\mathds{1} \in \Ures(\HH)$ along the path
\begin{equation*}s \rightarrow U^\sa_{ren}(t_0 + s , t_0), \; s \in [0, t_1-t_0]\end{equation*} 
in $\Ur(\HH)$ with respect to the connection determined by $\Theta$.\\

\noindent This definition yields a well-defined lift of the renormalized time evolution to the fermionic Fock space. In particular, it gives a smooth prescription for the phase of the implementations. 

 \hspace{-4cm}
\begin{figure}[h]
  \begin{center}
    {\includegraphics[scale=0.50]{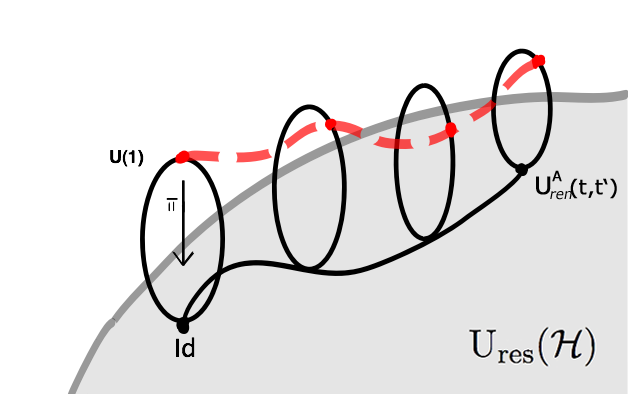}}
    \caption{\label{fig:Parallel Transport} Lifting the time evolution by parallel transport.}
  \end{center}
\end{figure}

\subsubsection{Second Quantization of the S-matrix}

\noindent Since $\sa$ has compact support in time, so do $U_I^\sa(\cdot, \cdot)$ and the renormalization $\T_t(\sa)$. 
Therefore, there exists $ T > 0$ such that $U^\sa_{ren}(t_1,t_0) = U_I^{\sa}(T , -T)$, whenever $t_1 \geq T, t_0 \leq -T$.\\
In particular, for \textit{any} such $T\in \IR$:
\begin{equation*}U^\sa_{ren}(T,-T)= U_I^{\sa}(\infty , -\infty) = S\end{equation*} 
and the S-matrix is unaltered by the renormalization. Therefore,
\begin{equation}\label{gammaforS}\gamma : s \longrightarrow  U^\sa_{ren}(-T +s , -T), \; s \in [0 , 2T]\end{equation} 
is a differentiable path from $\mathds{1}$ to $S$ in $\Ur(\HH) \subset \Gl_{\rm res}(\HH)$. The second quantization $\textbf{S}$ of the scattering operator $S$ is then defined as the parallel transport of  $\mathds{1} \in \Ures(\HH)$ along $\gamma$.\\

\noindent This procedure determines the second quantized scattering matrix in a well-defined manner. However, as we will see, the phase of $\textbf{S}$ will depend on the choice of the renormalization. Therefore, it makes sense to write \begin{equation}\mathbf{S} = \mathbf{S}[\sa,\T]; \, \sa \in \A, \T\; \text{a renormalization.} \end{equation}

\subsubsection{Second Quantization of the generators}
 
\noindent A connection on the principle bundle allows us to lift not only paths, but also vector fields from the base-manifold to the principle bundle in a unique way. In our setting, this means that it distinguishes unique lifts of the \textit{renormalized interaction Hamiltonians} to the universal covering of the Lie algebra. One might call this a \textit{second quantization of the Hamiltonians.}\\

\begin{Theorem}[Second Quantization of the Hamiltonians]
\mbox{}
\begin{enumerate}[i)]
\item Horizontal lifts with respect to the connection $\Gamma_{\Theta}$ yields a continuous section 
\begin{equation} \mathrm{d}\Gamma: \mathfrak{u}_{\rm res} \to \tilde{\mathfrak{u}}_{\rm res} \end{equation}
of the Lie algebra $\mathfrak{u}_{\rm res}$ of $\Ur(\HH)$ into its central extension $\tilde{\mathfrak{u}}_{\rm res}$. This defines a \emph{``second quantization''} of self-adjoint operators satisfying the Shale-Stinespring condition.\\ 
The section gives rise to a \emph{commutator anomaly}
\begin{equation}  [\mathrm{d}\Gamma(X),  \mathrm{d}\Gamma(Y) ] = c(X,Y)\, \mathrm{d}\Gamma([X,Y]) \end{equation}
with $c \in \mathrm{H}^2(\mathfrak{u}_{\rm res}, \IR)$ the Schwinger-cocylce \eqref{Schwingercocycle}.

\item If $V_{\rm ren}^\sa$ is the renormalized interaction \eqref{Zren} and
\begin{equation} h(t) = h[\sa,T](t) := e^{iD_0t} V_{\rm ren}^\sa(t) e^{-iD_0t} \end{equation} 
the renormalized interaction Hamiltonian, the so defined second quantization\\
$\mathfrak{h}(t):= i\, \mathrm{d}\Gamma(-i\, h(t))$ is the generator of the time evolution $\UU(t,t')$ on the Fock space.
\end{enumerate}
\end{Theorem}
\newpage
\noindent Diagrammatically, if $\Gamma$ denotes second quantization of unitary operators and $\mathrm{d}\Gamma$ second quantization of their self-adjoint generators, we have 

\begin{equation*}
\begin{xy}
  \xymatrix{
\tilde{\mathfrak{u}}_{\rm res} & \mathfrak{h}(t) \ar@<5pt>[d]^{\dot{\pi}} \ar[r]^{\exp} &  \UU(t,t')\ar@<5pt>[d]^{\pi}    & \Ures(\HH) \\
    \mathfrak{u}_{\rm res} & h_{\rm}(t)\ar@<5pt>[u]^{\mathrm{d}\Gamma}\ar[r]^{\exp}   & U_{\rm}^\sa(t,t') \ar@<5pt>[u]^{\Gamma}&\Ur(\HH) 
  }
\end{xy}
 \end{equation*}

\noindent What is essential here (and not expressed in the diagram) is that the lift of the Hamiltonians are determined by the geometric structure alone. 

\begin{proof}[Proof of the Theorem] 
\mbox{}
\begin{enumerate}[i)]
\item A connection in general and $\Gamma_{\Theta}$ in particular, determines unique horizontal lifts of vectors on the base manifold $\Ur(\HH)$ to the principle bundle $\Ures(\HH)$.\\
Under the canonical identification 
\begin{equation*} \mathfrak{u}_{\rm res} \cong \T_e\Ur, \hspace{7mm} \tilde{\mathfrak{u}}_{\rm res} \cong \T_e\Ures \end{equation*}
this yields the desired section $\mathrm{d}\Gamma$ in the central extension of Lie algebras:
\begin{equation}
\begin{xy}
\xymatrix{
0 \ar[r] &  \IR\ar[r]_{\dot{\imath}} & \tilde{\mathfrak{u}}_{\rm res} \ar[r]_{\dot{\pi}} & \mathfrak{u}_{\rm res}\ar@/_5mm/[l]_{\mathrm{d}\Gamma}\ar[r] & 0}
\end{xy}.
\end{equation}
We know from the discussion in Section 3.5 that such a section gives rise to a Lie algebra 2-~cocycle and it follows from Prop. \ref{Prop:curvatureSchwinger} that the cocycle corresponding to $\mathrm{d}\Gamma$ equals the Schwinger term.

\item $\mathfrak{h}(t)$ is defined by $-i\,\mathfrak{h}(t) = \mathrm{d}\Gamma(-i h(t))$ for every $t \in \IR$. \\
For fixed $t_0$, the renormalized time evolution $U^\sa_{ren}(t,t')$ satisfies 
\begin{equation*}i\,\partial_t\, U^\sa_{ren}(t,t_0) = h(t) U^\sa_{ren}(t,t_0). \end{equation*}
By construction, $t \mapsto \UU(t,t_0),\, t \in [t_0 , \infty)$ is precisely the integral curve to the
horizontal lift of the vector field 
\begin{equation}X_{\gamma}(t) := -i\, \bigl(R_{U(t,t_0)}\bigr)_* \,h(t) = -i \,h(t)\,U(t,t_0)
\end{equation}
along $\gamma(t) = U^\sa_{ren}(t,t_0),\,  t \in [t_0 , \infty)$.
Since the connection $\Gamma_{\Theta}$ is invariant under the right-action of $\Ures(\HH)$ on itself,
\begin{equation*}\widetilde{X}_{\gamma}(t) := -i\, \bigl(R_{\UU(t,t')}\bigr)_* \, \mathfrak{h}(t) = -i\, \mathfrak{h}(t) \UU(t,t')
\end{equation*}
%-i\,\mathfrak{h}(t) \cdot \UU(t,t')$ 
is horizontal for all $t \in \IR$ and obviously a lift of $X_{\gamma}(t)$ to the tangent bundle $\T\Ures$ of $\Ures(\HH)$.
It follows that $\UU(t,t')$ satisfies
\begin{equation*} i\,\partial_t\, \UU(t,t_0) = \mathfrak{h}(t)\, \UU(t,t_0), \; \forall t,t_0 \in \IR\end{equation*} 
hence $\mathfrak{h}(t)$ generates the second quantized time evolution $\UU(t,t_0)$, as was claimed.
 %Let $t,t' \in \IR$. 
%\begin{equation*} i \partial_t\, \UU(t,t') \, = i \partial_s\bigl\lvert_{s=0} \UU(t+s, t')\UU^*(t,t')\UU(t,t') = \mahfrak{h}(t) \UU(t,t')
\end{enumerate}
\end{proof}

\subsection{Causality}
For the second quantized time evolution $\mathfrak{U}(t,t')$ we can derive the following
important result:

\begin{Theorem}[Semigroup Structure of the Time Evolution]\label{Thm:semigroup}
\mbox{}\\
The lifted time evolution $\mathfrak{U}(t,t')$ in $\Ures(\HH)$, defined by horizontal lifts as above, preserves the semi-group structure of the time evolution i.e. satisfies
\begin{equation}\label{semigroup} 
\left\{
\begin{array}{ll}
\UU(t,t)  = \mathds{1}_{\Ures} & \forall\, t \in \IR \\ \\
 \UU(t_2, t_1)\,\UU(t_1,t_0) = \UU(t_2,t_0)& \forall\, t_0\leq t_1\leq t_2 \in \IR
\end{array}
\right.
%\]
\end{equation}
These properties justify its denomination as a ``time evolution'' on the Fock space.
 %\UU(t_2, t_1)\,\UU(t_1,t_0) = \UU(t_2,t_0) \end{equation}
%for all $t_2 \geq t_1 \geq t_0 \in \IRinf$.
\end{Theorem}

\begin{proof}$\UU(t,t)  = \mathds{1}, \forall\, t \in \IR$ holds by construction. For the composition property, note that $\UU(t_2,t_0)$ is the end-point of the horizontal lift of 
\begin{equation*} s \rightarrow U^\sa_{ren}(t_1+s , t_0),\; s \in [0 , t_2 - t_1]\end{equation*}
with starting point $\UU(t_1 , t_0)$. On the other hand, consider the curve \begin{equation*} s \rightarrow  \UU(s, t_1)\,\UU(t_1,t_0),\; s \in [0 , t_2 - t_1] \end{equation*}
It has the same starting point  $\UU(t_1 , t_0) \in \Ures(\HH)$ and projects down to
\begin{equation*}\pi\bigl( \UU(s, t_1)\,\UU(t_1,t_0) ) = U^\sa_{ren}(s, t_1)\,U^\sa_{ren}(t_1,t_0) = U^\sa_{ren}(s, t_0) \in \Ur(\HH). \end{equation*}
Furthermore, $\frac{d}{ds} \,  \UU(s, t_1)\,\UU(t_1,t_0) = (R_{\UU(t_1,t_0)})_* \, \dot{\UU}(s, t_1)$ is horizontal, because $X_{(s)}:= \dot{\UU}(s, t_1)$ is horizontal by construction of the lift and the connection $\Gamma_{\Theta}$ is invariant under right-action of $\Ures$ on itself. That is, if $\Theta$ denotes the connection one-form, we have:
\begin{equation*}\Theta\bigl((R_{\UU(t_1,t_0)})_* X_{(s)}\bigr) = (R_{\UU(t_1,t_0)})^*\Theta\, (X_{(s)}) =  \Theta(X_{(s)}) = 0 \end{equation*}
By uniqueness of the parallel transport it follows that both curves are actually the same.\\ 
In particular  $\UU(t_2, t_1)\,\UU(t_1,t_0) = \UU(t_2,t_0)$.\\
\end{proof}
\noindent Note that  the same argument wouldn't go through with the left-invariant Langmann-Mickelsson connection. We could however obtain the analogous result by lifting the unitary evolution ``backwards'' in time, i.e. by lifting the path \begin{equation*}s \rightarrow U^\sa_{ren}(t_1, t_1-s), \; s \in [0, t_1-t_0] \end{equation*} 
This is a perfectly valid procedure, but seems more artificial from a physicist's point of view. Therefore we proposed the right-invariant version of the Langmann-Mickelsson connection; this convention fits better with the usual form of the time evolution where subsequent time-steps correspond to unitary operators multiplied from the left.\\

To appreciate the value of this result, let's first make clear what's \textit{not} the point.
The semi-group structure alone wouldn't be worth the trouble, it actually comes fairly cheap:
 \newpage

\begin{Proposition}[Lifting the semi-group Structure]\label{Prop:semigroup}
\mbox{}\\
Let $U(t,t'), t,t' \in \IR$ a two-parameter semi-group in $\Ur(\HH)$ with compact support in time.
\item For every $t \in \IR$ choose \emph{any} lift $\tilde{U}(t, -\infty)$ of $U(t, -\infty)$ to $\Ures(\HH)$ and set
\begin{equation*} \tilde{U}(t_1, t_0) :=  \tilde{U}(t_1, -\infty)\,\tilde{U}(t_0, -\infty)^{-1}\end{equation*}
for $t_1\geq t_0 \in \IRinf$. The so defined lift has the semi-group properties \eqref{semigroup}.\\
%\begin{equation*}\UU(t_2, t_1)\,\UU(t_1,t_0) = \UU(t_2,t_0) \end{equation*}
%for all $t_2 \geq t_1 \geq t_0 \in \IRinf$. 
If  $t \rightarrow \tilde{U}(t, -\infty)$ is continuous/differentiable, then $\tilde{U}(t,t')$ is continuous/differentiable in t. 
\end{Proposition}
\begin{proof}For $t_2\geq t_1\geq t_0 \in \IR$ we find
\begin{align*}\tilde{U}(t_2, t_1)\,\tilde{U}(t_1,t_0)& =  \tilde{U}(t_2, -\infty)\,\tilde{U}(t_1, -\infty)^{-1}\,\tilde{U}(t_1-\infty)\,\tilde{U}(t_0, -\infty)^{-1} \\
& =  \tilde{U}(t_2, -\infty)\,\tilde{U}(t_0, -\infty)^{-1} = \tilde{U}(t_2, t_0) \end{align*}
And of course $\tilde{U}(t,t) =  \tilde{U}(t, -\infty)\,\tilde{U}(t, -\infty)^{-1} = \mathds{1}, \forall t \in \IR$.\\
\end{proof}

%\noindent This construction is just the analogue of [Prop. \ref{Prop:semigrouptimevarying}] on a fixed Fock-space.\\

\noindent To appreciate the virtue of the geometric construction of the time evolution, we need to understand the difference between the results in theorem \ref{Thm:semigroup} and the previous proposition \ref{Prop:semigroup}.
If the semi-group structure is not the point, then what is?

\vspace{1mm}
\noindent The composition property
\begin{equation}\label{compositionproperty} \tilde{U}(t_2, t_1)\,\tilde{U}(t_1,t_0) = \tilde{U}(t_2, t_0), \forall t_0<t_1<t_2 \in \IR \end{equation}
is often related to the physical principle of \textit{causality}. But at this point, this is actually too big a word. The algebraic identities by themselves merely express the basic requirements for a two-parameter family of unitary operators to deserve the title of a ``time evolution'' in the first place. We must arrive at exactly the same state, whether we follow the evolution of a physical system from time $t_0$ to time $t_1>t_0$ and then from time $t_1$ to $t_2 >t_1$, or whether we skip the intermediate time-step and follow the evolution from $t_0$ to $t_2$ directly. If what the theory calls a ``state'' involves a phase, then these phases need to agree as well. So far, this is a question of \textit{consistency} rather than \textit{causality}. To relate the semi-group structure to something like a \textit{causal structure} of the physical theory, we have to take a closer look at the construction of the (second quantized) time evolution and what it involves.\\

The construction in Prop. \ref{Prop:semigroup} is not a serious proposal for a second quantization, but rather an instructive example to demonstrate why the semi-group properties \eqref{semigroup} alone are insufficient. The two-parameter group $\tilde{U}(t_1, t_0)$ in $\Ures(\HH)$ was constructed by choosing lifts of the one-particle time evolution operators $U(t, -\infty)$ first and using those lifts to generate the entire family. This has the following effect: if we want to call $\tilde{U}(t_1, t_0)$ a (second quantized) time evolution, we find that the evolution between times $t_0$ and $t_1>t_0$ depends on the \textit{entire history} of the physical system, i.e. on $\sa(t)$ for $t \in (-\infty, t_1]$. If we altered the electromagnetic fields in the distant past $t \ll t_0$, we would most likely get a different lift for the time evolution between $t_0$ and $t_1$. It's debatable whether this should be called a violation of causality (and the answer will depend on our very definition of that concept). But certainly, such a solution would be at odds with our conventional understanding of a \textit{deterministic} evolution. We would sacrifice what one might call the ``Cauchy  property''\footnote{Or ``Markov property'', if we borrow the language of probability theory.} of the system: that its future evolution is completely determined by the laws of physics, given the current state of the system at any moment in time (or on a spacelike Cauchy-surface in a relativistic setting).\\ 

If in the construction of Prop. \ref{Prop:semigroup} we used a different ``reference time'' than $t = -\infty$ (which we could do), the situation would be even worse. The time evolution of the system before that particular reference time would then depend on the electromagnetic potential \textit{in the future}. Obviously, this constitutes a violation of causality in a very strong sense.\footnote{Of course, we shouldn't forget that in all this, we are merely talking about a phfibre bundlease-factor which is not measurable \textit{per se}. Still, the relative phase between two states or the variation of the phase with the electromagnetic potential can (theoretically) have observable consequences (recall \eqref{eq:Strom von Phase}, for instance). What this means in practice (in the context of QED), however exceeds the author's competence.}\\ 

In our geometric construction, the second quantization of the renormalized time evolution $\UU(t,t')$ and its generators $\mathfrak{h}(t)$ depend only on the corresponding (``first quantized'') objects in the one-particle theory, on the renormalization used to transform them back to $\cu_{\rm res}$ or $\mathfrak{u}_{\rm res}$, respectively, and finally on the bundle-connection introduced in Chapter 7. In particular, the lifts of the renormalized time evolution and the renormalized interaction Hamiltonians are determined solely by the \textit{geometric structure} of the $\Ures(\HH)$ bundle. We have required that the renormalization is causal in the sense that $\T_t(\sa)$ depends only on $\sa(t)$ and its time-derivatives up to finite order $n$. Consequently, the geometric second quantization assures that $\UU(t_1,t_0)$ depends on the external field $\sa(t)$ (and its first $n$ time-derivatives) only at times $t \in [t_0,t_1]$.  Similarly, for fixed $t \in \IR$, the lifted interaction Hamiltonian $\mathfrak{h}(t)$ depends only on $\sa(t)$ and its first $n+1$ time-derivatives at time $t$. Conclusively, we assert that \textit{geometric second quantization preserves the causal structure of the one-particle theory.}

\subsubsection{Application to time-varying Fock spaces}
\noindent We can translate these insights back into the language of time-varying Fock spaces and derive the following result:
\begin{Theorem}[Time Evolution on time-varying Fock Spaces]
\label{Thm:Time evolution on time-varying Fock spaces}
\mbox{}\\
Given a unitary time evolution $U^\sa(t,t')$ and a family of infinite wedge spaces $(\FF_t)_{t \in \IR}$ over $\ocean(C(t))$, there exists right-operations by $R(t,t') \in \cu(\ell)$ such that 
\begin{equation*} \widetilde{U}(t,t'):= \lop{U(t,t')}\rop{R(t,t')}: \FF_{t} \longrightarrow \FF_{t'}, \; \forall t,t' \in \IR \end{equation*} 
defines a time evolution, i.e satisfies:
\begin{equation}
\left\{
\begin{array}{ll}
\UUt(t,t)  = \mathrm{Id}_{\FF_t} & \forall\, t \in \IR \\ \\
 \widetilde{U}(t_2,t_1)\, \widetilde{U}(t_1,t_0) = \widetilde{U}(t_2,t_0), & \forall\, t_0 \leq t_1 \leq t_2 \in \IR
\end{array}
\right.
%\]
 \end{equation}
Moreover, for all $t \geq t' \in \IR$, the operator $R(t.t')$ depends on the A-field (and possibly its time derivates) only in the time-interval $[t',t]$.
\end{Theorem}
\newpage
\begin{proof} Let $(\T_t)_t$ be a smoothening renormalization for the time evolution $U^\sa(t,t')$ and let $\UU(t,t')$ be the second quantized time evolution on $\FF_{\rm geom}$ constructed in Thm. \ref{Thm:semigroup}. We fix a $\Phi_0 \in \seas$ with $\im(\Phi_0) = \HH_-$ so that $\bigwedge \Phi_0$ in $\FF_{\cs(\Phi_0)}$ will play the role of the free vacuum state. Let $C(t)= C[\underline{\sa}(t)]$ be the polarization classes identified in Thm. \ref{Thm:polarizationclasses} [Identification of Polarization classes].  Then we have $\T_t \Phi_0 \in \ocean(C(t)), \; \forall t \in \IR$ and $\mathcal{G}_t := \FF_{\cs(T_t \Phi_0)}$ defines a family of Fock spaces over $C(t),\, t \in \IR.$ 
Under the isomorphism $\wedgeu$ defined in Cor. \ref{Ures on wedge spaces}, $\UU(t,t')$ becomes a time evolution on $\FF_{\cs(\Phi_0)}$. For every $t\geq t' \in \IR$, there exists an operator $S(t,t') \in \cu(\ell)$ such that:
\begin{equation*}\wedgeu\;\UU(t,t') = \lop{(U_{ren}^\sa(t,t'))}\, \rop{S(t,t')} = \lop{(U^0(0,t)\T^*_tU^\sa(t,t')\T_{t'}U^0(t,0))}\, \rop{S(t,t')}\end{equation*} 
and
\begin{equation*}S(t_1,t_0)S(t_2,t_1) = S(t_2,t_0), \; \forall\, t_0 \leq t_1  \leq t_2. \end{equation*}
We may also choose them such that $S(t,t) = \mathds{1}_{\ell}$ for all $t$. Note that the essential part is that those right-operations fit together in the right way to form a two-parameter semi-group. Thus, setting $\widetilde{V}(t,t') := \lop{U^\sa(t,t')}\, \rop{S(t,t')} $
we get a two-parameter family of unitary transformations $ \widetilde{V}(t,t'): \mathcal{G}_{t'} \to \mathcal{G}_t $
with $\widetilde{V}(t_2,t_1)\, \widetilde{V}(t_1,t_0) = \widetilde{V}(t_2,t_0), \; \forall\, t_0 < t_1 <t_2. $
This, however, is not really satisfying, yet, because the Fock spaces themselves are chosen by the smoothening renormalization.
If we have chosen a different family $(\FF_t)_{t \in \IR}$ of infinite wedge spaces over the correct polarization classes, it follows from Cor. \ref{Cor:U(l) acts transitively} [$\cu(\ell)$ acts transitively on oceans] that there exists for every $t$ a unitary transformation $r(t) \in \cu(\ell)$ with $\rop{r(t)}: \mathcal{G}_t \xrightarrow{\;\cong\,} \FF_t$.
For any $t \geq t' \in \IR$, we set
\begin{equation*} R(t,t') := r^*(t') S(t,t') r(t). \end{equation*} 
Then,
\begin{equation*} \widetilde{U}(t,t'):= \lop{U(t,t')}\rop{R(t,t')}  = \rop{r(t)}\widetilde{V}(t,t')\rop{r^*(t')} : \FF_t' \longrightarrow \FF_t \end{equation*}
and for all $t_0 \leq t_1 \leq t_2$ we find
\begin{align*} R(t_1,t_0)R(t_2,t_1) = \,&  r^*(t_0) S(t_1,t_0) r(t_1) \, r^*(t_1) S(t_2,t_1) r(t_2)\\[1.5ex]
= \,& r^*(t_0) S(t_2,t_0) r(t_2) = R(t_2, t_0) \end{align*}
so that 
\begin{equation*}\widetilde{U}(t_2,t_1)\, \widetilde{U}(t_1,t_0) = \widetilde{U}(t_2,t_0).  \end{equation*}
Also, $R(t,t) = r^*(t) S(t,t) r(t) = r^*(t)r(t) = \mathds{1}_{\ell}$, so that $\UUt(t,t) = \mathrm{Id}_{\FF_t}, \; \forall t \in \IR$.\\
%\begin{equation*}
%\begin{xy}
 % \xymatrix{
 %\mathcal{G}_{t_0} \ar[r]^{\widetilde{V}(t_1,t_0)}\ar[d]_{r(t_0)}    &   \mathcal{G}_{t_1} \ar[r]^{\widetilde{V}(t_2,t_1)} \ar[d]_{r(t_1)} & \mathcal{G}_{t_2} \ar[d]^{r(t_2)}  \\
 %     \rr_{t_0}\ar[r]^{\widetilde{U}(t_1,t_0)}          &   \rr_{t_1} \ar[r]^{\widetilde{U}(t_2,t_1)}      & \rr_{t_2} 
  %    }
%\end{xy}
%\end{equation*}
\end{proof}

Note that the parallel transport applied to time-varying Fock spaces in this way does not anymore eliminate the $\cu(1)$-freedom of the lifts. The reason is that the operators $r(t)$ used in the proof to map the family of Fock spaces $(\mathcal{G}_t)_t$ (determined by the renormalization) to the Fock spaces $(\FF_t)_t$ are not unique. If we choose a different family $r'(t)$, they will differ by $r'r^*(t) \in \cu^1(\ell)$, i.e. $\rop {r'(t)} = \det( r'r^*(t)) \cdot \rop {r(t)}: \mathcal{G}_t \xrightarrow{\;\cong\,} \FF_t$. (See also Cor. \ref{Cor:equivalence of infinite wedge spaces} [Equivalence of infinite wedge spaces]). Thus, the phase-freedom now comes from the ambiguity in identifying different, equivalent Fock space representations over the same polarization class. 
\mbox{}\\

%Applying Thm. \ref{Thm:semigroup} in the context of infinite wedge spaces, we find immediately that there exists a family of Fock spaces $(\mathcal{G}_t)_t$ and right-operations $S(t,t')$ such that
%\begin{equation*} \lop{U(t,t')}\rop{S(t,t')}: \mathcal{G}_{t} \longrightarrow \mathcal{G}_{t'}, \; \forall t,t' \in \IR \end{equation*} 
%and
%\begin{equation*}S(t_1,t_0)S(t_2,t_1) = S(t_2,t_0), \; \forall\, t_0 < t_1 <t_2.  \end{equation*}

\subsubsection{The Causal Phase of the S-Matrix}

As mentioned before, it is unclear whether the time evolution is meaningful at all in Quantum Electrodynamics. The predominant position seems to be that only the S-matrix is physically significant. If one adheres to this view, the causal properties of the second quantized (renormalized) time evolution are of secondary interest. A better formulation of causality in terms of the second quantized S-matrix is therefore the following:\\

\noindent Given an external field 
\begin{equation}
\sa = \sa_1 + \sa_2  \in \A\end{equation} 
which splits in two parts with disjoint supports in time, i.e. $\exists \, r \in \IR$ such that
\begin{equation}\label{Asupport}
\mathrm{supp}_t \;\sa_1 \subset (-\infty, r)\; ,\hspace{5mm} \mathrm{supp}_t \;\sa_2 \subset (r, +\infty).
\end{equation}
That is, the field $\sa_1$ vanishes for times $t \geq r$ and the field $\sa_2$ vanishes for $t \leq r$.\\
We say that the phase of the S-matrix is \textit{causal} if
\begin{equation}\label{causalphase} \fingbox{\mathbf{S}[\sa] = \mathbf{S}[\sa_1 + \sa_2] = \mathbf{S}[\sa_2]\,\mathbf{S}[\sa_1].}\end{equation}
The importance of this causality condition was emphasized in particular by G.Scharf in \cite{Scha}.\\
But indeed, the causal properties of the geometric second quantization as discussed above do imply causality in the sense of Scharf:
As the A-field vanishes around time $r$, so does the renormalization. Therefore: 
\begin{align*} &U^\sa_{ren}(r, -\infty) = U_I^\sa(r, -\infty) = S[\sa_1]\\
&U^\sa_{ren}(+\infty,r) = U_I^\sa(+\infty,r) = S[\sa_2]
\end{align*} 
Hence, by construction of the lift and theorem \ref{Thm:semigroup}:
\begin{equation} \mathbf{S}[\sa] = \UU(+\infty , -\infty) = \UU(+\infty , r)\UU(r, -\infty) = \mathbf{S}[\sa_2]\,\mathbf{S}[\sa_1]. \end{equation}

\noindent We note that this finding seems to contradict the results in \cite{Scha}, where it is suggested that the phase of the second quantized scattering operator is completely determined by the causality condition \eqref{causalphase}. The geometric second quantization of Langmann and Mickelsson is causal, still it does not yield a unique phase (unfortunately). In fact, there is pretty much freedom left. We can choose a different $\Gres$-invariant connection and/or different renormalizations of the time evolution. Different connections will inevitable lead to different causal phases, in general, since a bundle connection is uniquely determined by its parallel transport maps. How the choice of the renormalization can affect the phase of $\mathbf{S}$ is discussed in the next section.

%So far, the best causal renormalization we have is the one construction in \cite{LaMi} with $n=1$ i.e. the lifted time evolution depends on the A-field and its first time-derivatives. The lifted Hamiltonian depends even on the  time-derivatives up to order 3. This doesn't seem too nice. However, we know that the right polarization class at time $t$ depends only on $\sa(t)$ - we have the causal renormalization $(e^{Q^{\sa(t)}})_t$ that depends only on $\sa(t)$ and no time-derivatives. Therefore, we conjecture that a causal renormalization exists which doesn't contain any time-derivatives of $\sa$. Then  $\UU(t,t')$ would only depend only on $\sa(s), s \in [t', t]$ and $\mathfrak{h}(t)$ would be a function of $\sa(t)$ and $\partial_t \sa(t)$. One might wish for a causal renormalization of the \textit{Hamiltonian} that doesn't involve time-derivatives of $\sa$, to avoid such terms that don't appear in the first-quantized theory. We conjecture that for the same reasons such a renormalization can \textit{not} exist. But rigorous proof, either way, hasn't been provided yet. 

\newpage
\section{Holonomy of the Bundles}

We have seen that the choice of the renormalization is not unique. Therefore, the renormalized time evolution is -- without additional requirements -- very arbitrary and the urging question arises, how meaningful it can actually be in a physical description. The S-matrix, however, is invariant under renormalizations and one might hope that geometric second quantization yields a well-defined scattering operator on the Fock space, as was suggested in \cite{LaMi}. Unfortunately, the ambiguity in the renormalization is problematic for the second quantization of the S-matrix as well. If we use two different renormalizations to implement the S-matrix on the Fock space via parallel transport, we have two different curves in $\Ur(\HH) \subset \Gl_{\rm res}(\HH)$, both with starting-point $\mathds{1}_{\HH} \in \Ur$ and end-point $S \in \Ur$. Consequently, their horizontal lifts are two different curves in $\Ures(\HH)$. The starting-points of the lifted paths in $\Ures(\HH)$ are, by definition, the same (namely $\mathds{1} \in \Ures$), but their end-points -- the corresponding lift of the scattering matrix -- need not be. (We just know that they both lie in the same fiber above $\Ures(\HH) \ni \pi^{-1}(S) \cong \cu(1)$). Naturally, the question arises how these end-points, i.e. the phase of the scattering matrix, depend on the choice of the renormalization. Geometrically, this freedom is expressed by the \textit{holonomy group} of the bundle $\pi: \Ures(\HH) \rightarrow \Ur(\HH)$.
\vspace{0.4cm}

\begin{figure}[h]
  \begin{center}
    {\includegraphics[scale=0.49]{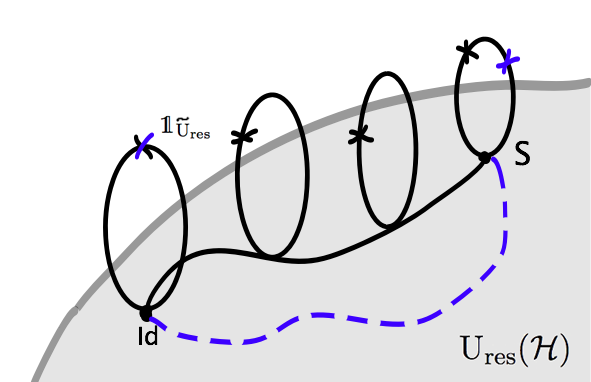}}
    \caption{\label{fig:Holonomy} PT along different paths can end in different points in the fibre over S}
  \end{center}
\end{figure}

%\vspace{0.5cm}
Consider a principle-$G$-bundle $E \xrightarrow{\pi} M$ for a (finite-dimensional) Lie group $G$ over a connected, paracompact manifold $M$, equipped with a connection $\Gamma$.
Let $u_0, u_1$ two points on the base-manifold $M$ and $c, c': [0,1] \rightarrow M$ two piecewise differentiable curves from $u_0$ to $u_1$ . Let $p \in \pi^{-1}(u_0) \subset E$ be an element in the fibre over the starting point $u_0$. 
Now we want to know if parallel transport of $p$ along $c$ yields the same result as parallel transport of $p$ along the other path $c'$. In other words, we want to now if the horizontal lifts of $c$ and $c'$ with starting point $p \in \pi^{-1}(u_0)$ have the same end-point in the fibre $\pi^{-1}(u_1)$ above $u_1$. 

Equivalently, we can look at the closed loop $\gamma := c\circ (c')^{-1}$ in $M$, resulting from moving along the $c$ first and then backwards along $c'$, and ask, whether parallel transport of $p$ along $\gamma$ is the identity or not. As the loop starts and ends in $u_0$, the parallel transport $P_{\gamma}(p)$ is certainly a point in the same fibre $\pi^{-1}(u_0)$ over $u_0$. Since the Lie group $G$ acts transitively on every fibre, there exists a unique element $g := \mathrm{hol}(\gamma)$ in $G$ with $P_{\gamma}(p) = r_g(p) = p \cdot g$.

\noindent This motivates the following definition:
\begin{Definition}[Holonomy Group]
\mbox{}\\
Let $G$ a finite-dimensional Lie group and $E \xrightarrow{\pi} M$ a principle G-bundle with connection $\Gamma$. \\
For every point $p \in P$ we define
%\item $L(\pi(p)) := \lbrace \gamma: [0,1] \rightarrow M \mid$ piecewise differentiable loop around $\pi(p)$ in $M \rbrace$ and 
\begin{equation*}\Hol(\Gamma, p) := \lbrace g \in G \mid \exists\, \text{closed loop}\; \gamma\; \text{around}\; \pi(p)\; \text{in}\; M  \text{s.t.}\; P_{\gamma}(p) = p \cdot g \rbrace \end{equation*}
$\Hol(\Gamma, p)$ is a subgroup of $G$, called the \emph{holonomy group} of $\Gamma$ at $p$.
\item Furthermore, we define the \emph{restricted holonomy group}
\begin{equation*}\Hol_0(\Gamma, p) := \lbrace g \in \Hol(\Gamma, p) \mid \text{the corresponding}\; \gamma\;\text{can be chosen to be null-homotopic} \rbrace \end{equation*}
by restricting to null-homotopic curves in $M$.
\end{Definition}

\noindent In our case, the base-manifold $\Ur(\HH)$ or $\Gl_{\rm res}(\HH)$, respectively, is simply-connected and therefore $\Hol$ and $\Hol_0$ are the same.\\

It is easy to check that the holonomy group is indeed a group. Concatenation of two loops results in multiplication of the corresponding group elements, parallel transport along the constant path yields the identity and the inverse of a group element is obtained by reversing the sense of the corresponding path, i.e. moving along that path backwards. \\

\noindent More formally, if $p \in E$, $I = [0,1] \subset \IR$ and $\gamma: I \to M$ denotes closed loops around $\pi(p) \in M$, we have
\begin{enumerate}[I)]
\item $P_{\gamma_0}(p) = p$, for the constant path $\gamma_0(t) \equiv \pi(p)$.
\item $P_{\gamma}(p) = p\cdot g \Rightarrow P_{\gamma^{-1}}(p) = p \cdot g^{-1}$, for $\gamma^{-1}:I \to M$ defined by $\gamma^{-1}(t) = \gamma (1-t)$.
\item $P_{\gamma_i}(p) = p\cdot g_i \Rightarrow P_{\gamma_1 \cdot \gamma_2}(p) = p\cdot (g_1g_2)$, where $\gamma_1 \cdot \gamma_2 (t) = \left\{
\begin{array}{ll} \gamma_1(2t), & t \in [0,\frac{1}{2}]\\ 
\gamma_2(2t-1), & t \in [\frac{1}{2}, 1] \end{array} 
\right.$.\\
\end{enumerate}

%\left\{
%\begin{array}{ll}
%\UU(t,t)  = \mathds{1}_{\Ures} & \forall\, t \in \IR \\ \\
 %\UU(t_2, t_1)\,\UU(t_1,t_0) = \UU(t_2,t_0)& \forall\, t_0<t_1<t_2 \in \IRinf
%\end{array}
%\right.
%\begin{enumerate}[i)]
%\item  P_{\gamma_0}(p) = p, \;\; \text{for}\; \gamma_0(t) \equiv \pi(p) 
%\item P_{\gamma_i}(p)= p g_i \Rightarrow P_{\gamma_1 \circ \gamma_2} = p(g_1 g_2)
%\item P_{\gamma}(p) = p g \Rightarrow P_{\gamma^{-1}} 

\noindent The holonomy groups at two points which can be connected by parallel transport are conjugated to each other because, if $p, q $ can be connected by a horizontal curve in $E$, we can parallel transport from $q$ to $p$, then along a loop around $p$ and back from $p$ to $q$, which corresponds to parallel transport along a loop around $q$.\\

%\noindent Moreover, the following is true:
%\begin{Lemma}(Holonomy group is a Lie group)
%\item $\Hol_0(\Gamma, p)$ is a connected Lie subgroup of G.
%\item $\Hol(\Gamma , p)$ is a Lie subgroup of G.
%\end{Lemma}

%\noindent The Lemma is in fact a Corollary to the following Theorem
%\begin{Theorem}(Freudenthal, 1941)\\
%Let G a Lie group and H a subgroup of G such that every element of H can be joined to the identity by a piecewise differentiable curve in H. Then H is a connected Lie subgroup of G. 
%\end{Theorem}
%\begin{proof} \cite{KoNo} Appendix 4 \end{proof} 

So much for the basics. Now we are going to show that the holonomy group of the bundle $\Ures(\HH) \xrightarrow{\pi} \Ur(\HH)$ at the identity is homomorphic to $\cu(1)$, i.e. to the entire structure group. We could perform the same calculation for the $\Gres$-bundle, but as we are mainly interested in lifting paths in $\Ur(\HH) \subset \Gl_{\rm res}(\HH)$, we find the unitary case more educative.\\

It will actually suffice to consider loops in a two-dimensional subspace. Thus, we will do the computations for $\cu(2, \IC)$ for simplicity and embed $\cu(2, \IC)$ into $\Ur(\HH)$ in the following way: Let $(e_k)_{k \in \IZ}$ be a basis of $\HH$ such that $(e_k)_{k \geq 0}$ is a basis of $\mathcal{H}_+$ and $(e_k)_{k < 0}$ a basis of $\mathcal{H}_-$. Now can identify $\cu(2)$ with $\cu(\spn(e_0, e_{-1}))$ i.e. 
\begin{align*}
\cu(2, \IC) \hookrightarrow \Ur(\HH);\; 
 \begin{pmatrix}
    a & b\\
    c & d
  \end{pmatrix}
 \longmapsto   \begin{pmatrix}
    a & \; & \vline & b & \; \\
    \; & \mathds{1}&\vline& \; & \textbf{0}\\
    \hline 
    c & \; & \vline & d & \; \\
    \; & \textbf{0} & \vline & \; & \mathds{1}
  \end{pmatrix},
  %\; \in \;
   %\begin{pmatrix}
    %\cu(\HH_+) & I_2(\HH_-, \HH_+)\\
    %I_2(\HH_+,\HH_-) & \cu(\HH_-)
  %\end{pmatrix}  
  \end{align*}
 where the identity matrices are on $(e_0)^\perp \subset \HH_+$ and $(e_{-1})^\perp \subset \HH_-$, respectively.\\

\noindent A general (piecewise differentiable) path in $\cu(2)$ has the form 
\begin{align*}
U(t) =
   \begin{pmatrix}
    a(t) & b(t)\\
    c(t) & d(t)
  \end{pmatrix}
\end{align*}
with 
\begin{align*}
  U^{-1}(t) = U^{*}(t) 
= \begin{pmatrix}
    \alpha(t) & \beta(t) \\
    \gamma(t) & \delta(t) \end{pmatrix} = \begin{pmatrix}
    a^{*}(t) & c^{*}(t)\\
    b^{*}(t) & d^{*}(t)
  \end{pmatrix}.
\end{align*}
Unitarity requires (among other identitites) $\vert a \vert^2 + \vert b \vert^2 = 1$.\\
The formula \eqref{ptformula} for the parallel transport expressed in local coordinates becomes
 \begin{equation}\label{ptformula2D} \exp\Bigl[\int\limits_{-T}^{T} \trace\bigl[\dot{a}(t)(\alpha(t) - a^{-1}(t) )\,+\, \dot{b}(t) \gamma(t) \bigr] dt \Bigr] = \exp\Bigl[\int\limits_{-T}^{T} \bigl[ \dot{a}(t) (a^{*}(t) - a^{-1}(t) ) + \dot{b}(t)b^*(t)\bigr] dt \Bigr]. \end{equation}
We can write
\begin{align*} a(t) & = r(t)\, e^{i\varphi(t)} \\[1.5ex]
b(t) & = \sqrt{1- r^2(t)} \, e^{i \psi(t)}
\end{align*}
with $r(t), \varphi(t)$ and $\psi(t)$ piecewise differentiable, real functions. 
Our path has to stay in the neighborhood of $\mathds{1}$ where $a(t)$ is invertible, i.e. $r(t) \neq 0$ is required. Note that $\vert r(t) \vert\leq 1,\, \forall t$, so $\sqrt{1- r^2(t)}$ is a real, differentiable function.
Actually, we won't even have to exploit the freedom of choosing $\psi(t)$ and can set it to zero. Then we compute: 
\begin{align*}  \dot{a}(t) & = \dot{r}(t) e^{i\varphi(t)} + i \dot{\varphi}(t) r(t) e^{i\varphi(t)} \\[1.2ex]
\dot{a}(t)a^{*}(t) & = \dot{r}(t)r(t) +  i \dot{\varphi}(t) r^2(t) \\[1.2ex]
\dot{a}(t) a^{-1}(t) & = \dot{r}(t) r^{-1}(t) +  i \dot{\varphi}(t)\\[1.2ex]
\dot{b}(t)b^*(t) & = - r(t) \dot{r}(t) \end{align*}
The argument of the exponential in \eqref{ptformula2D} is thus
\begin{align}  \int\Bigl[ \dot{a}(t)a^{*}(t) - \dot{a}(t)a^{-1}(t)& + \dot{b}(t)b^{*}(t)\Bigr] \mathrm{d}t  \notag \\[1.2ex]
= \; & \int\Bigl[  i\, \dot{\varphi}(t) ( r^2(t) - 1) - r^{-1}(t) \dot{r}(t) \Bigr] \mathrm{d}t. 
\end{align}
The second summand is just the derivative of $\log\bigl(r(t)\bigr)$ and gives no contribution when integrated over a closed loop. Thus, we're left with
\begin{align}\label{anyvalue}\exp \Bigl[\oint\bigl[  i\, \dot{\varphi}(t) ( r^2(t) - 1) - r^{-1}(t) \dot{r}(t) \bigr] \mathrm{d}t  \Bigr] = \exp \Bigl[ i\, \oint\bigl[\dot{\varphi}(t) ( r^2(t) - 1) \bigl] \mathrm{d}t  \Bigr],
\end{align}
where the integral is over the parameterization of a closed path around the identity in $\cu(2)$, corresponding to the boundary conditions $r(-T) = r(+T) = 1$ and $\varphi(\pm T) \in 2\pi \IZ$.\\

\noindent But \eqref{anyvalue} can take \textit{any} value in $\cu(1)$: We parameterize over $t \in [0, 1]$ and set $\varphi(t) := 2\pi t  $, such that $\dot{\varphi}(t) \equiv 2\pi$. Now, for any $\delta \in [0,1)$  we can choose a smooth cut-off function $\rho$ with compact support in $[0,1]$ and values in $[0,1)$, such that $\int\limits \rho(t)\, \mathrm{d}t = \delta$. 
Let $r(t) := \sqrt{1-\rho(t)}, \; t \in [0,1]$. $r$ is smooth and strictly positive with $r(0) = r(1) = 1$.  For this, we compute \begin{equation*}\int\limits_{0}^{1}\bigl[\dot{\varphi}(t) ( r^2(t) - 1) \bigl] \mathrm{d}t = 2 \pi \int\limits_{0}^{1} \rho(t) \,\mathrm{d}t = 2\pi \delta \end{equation*}
which can take any value between $0$ and $2\pi$, including 0.\\ 

%$r^2(t)$ there are smooth, positive functions with $r^2(-\pi) = r^2(\pi) = 1$, such that $\int\limits_{-\pi}^{\pi} ( r^2(t) - 1) \mathrm{d}t $ takes any desired value between $-2\pi$ and 0.}
\noindent So, we have found that parallel transport along a closed path (about the identity) can result in multiplication by any complex phase. The holonomy group is the whole $\cu(1)$. An analogous computation reveals that the holonomy group of the $\Gres(\HH)$-bundle, is also homomorphic to its entire structure group $\IC^{\times}$.\\

\noindent We summarize:

\begin{Proposition}[Holonomy Groups]
\mbox{}\\
The holonomy groups of the connection $\Gamma_{\Theta}$ on the principle bundle $\Gres(\HH)$ and of its restriction to  $\Ures(\HH)$ are
\begin{equation*}\Hol(\Gres, \Gamma_{\Theta}) = \Hol_0(\Gres, \Gamma_{\Theta}) =  \IC^{\times}\end{equation*}
and
\begin{equation*}
\Hol(\Ures, \Gamma_{\Theta}) = \Hol_0(\Ures, \Gamma_{\Theta}) =  \cu(1).
\end{equation*}
\end{Proposition}
\mbox{}\\
\noindent For us, this means that the \textit{entire freedom of the geometric phase} that we had eliminated by parallel transport, is actually reintroduced through the ambiguity in the choice of the renormalization. It merely gets a new name: \textit{holonomy}.\\

We remark that on a finite-dimensional bundle we could have immediately derived this result by means of the \textit{Ambrose-Singer Theorem} which states that the Lie algebra of the holonomy group is generated by the curvature two-form of the connection (it would suffice to note that the Schwinger cocycle evaluated on $\mathfrak{u}_{\rm res} \times \mathfrak{u}_{\rm res}$ can take any value in $i \IR = \mathrm{Lie}(\cu(1))$). Unfortunately, generalization to infinite-dimensions is usually problematic, since the proof relies on the Frobenius theorem which has no equipollent infinite-dimensional analogon. Therefore, we prefer to avoid all complications by the direct computation performed above.\\
\newpage
\subsubsection{Holonomy: explicit formula}
\noindent There is an explicit formula to compute how the parallel transport along two curves differ.
In coordinates defined by a local section $\sigma$, the holonomy-group element $\mathrm{hol}(\gamma)$ corresponding to parallel transport along a closed loop $\gamma$ w.r.to the connection one-form $\Theta$ can be computed as
\begin{equation}\label{holformula} 
\mathrm{hol}(\gamma) =  \exp\Bigl[ \oint\limits_{\gamma} (\sigma^* \Theta) \Bigr] = \exp\Bigl[ \int\limits_{S(\gamma)} (\sigma^* \Omega) \Bigr] ,
\end{equation}
where $S(\gamma)$ is the surface enclosed by $\gamma$ and $\Omega = \mathrm{d}\Theta$ is the curvature 2-form. The first equality follows immediately from the local expression for parallel transport, whereas the second equality is an application of Stokes theorem using $\mathrm{d}\Theta = \Omega$.\\
Note that  application of this formula is unproblematic even on infinite-dimensional manifolds, because integration is just along 1-dimensional curves or 2-dimensional surfaces.\\

\noindent If we use two different renormalizations  $\T$ and $\T'$ to lift the S-matrix by parallel transport along a path $\gamma =  \gamma[\T]$ as in \eqref{gammaforS} we will find that
\begin{equation} \mathbf{S}[\sa,\T'] = \mathrm{hol}(\gamma[\T']\cdot\gamma[\T]^{-1}) \, \mathbf{S}[\sa,\T] \end{equation}
and the phase-difference can be computed from \eqref{holformula}.
 
\begin{Remark}[Mickelsson '98]
\mbox{}\\
In \cite{LaMi}, the effects of holonomy are discussed only in connection with gauge transformations, not related to the problem of non-uniqueness of the renormalization. However, J. Mickelsson adresses the issue in a later publication \cite{Mi98}, where he suggests that the second quantization of the S-matrix can be made gauge- and renormalization-independent by introducing suitable counterterms. Those counterterms can be viewed either as an additional modification of the renormalized Hamiltonian, or as choosing a different connection one-form, depending on the A-field. More precisely, it is proven that the phase of the second quantized scattering matrix, defined by parallel transport using the modified Hamiltonian (respectively the modified connection), is invariant under infinitesimal (chirially even) variations of the renormalization \ref{LaMirenormalization} as well as under infinitesimal gauge transformations. 
To me, the proposed solution seems very technically involved and rather ad-hoc. In particular, though, the construction violates causality in the sense discussed above, because the counterterms introduced by Mickelsson depend explicitly on the entire time evolution $U^\sa_I(t, -\infty)$.

In the following section, we propose a different, much simpler approach. We are not going to solve both problems, the gauge-dependence \textit{and} the renormalization-dependence of the parallel transport. Instead, we are going to argue that we can use the second freedom to our advantage and make the geometric second quantization gauge-invariant by choice of suitable renormalizations. 
\end{Remark}

%\subsection{Remark: The Ambrose-Singer Theorem}
%For finite-dimensional bundles, the intuitive relationship between holonomy and curvature is rigorously expressed in a very elegant way through the \textit{holonomy bundle} and the Ambrose-Singer Thm (cf. \cite{KoNo},  \S 8). The Ambrose-Singer Theorem says that the Lie algebra of the holonomy group at some point p (which is always a Lie group in the finite-dimensional case) is generated by the curvature 2-form evaluated at every point that can be connected to p by parallel transport.
%If we could apply the theorem to our case, we would immediately get that the holonomy group of $\Ures$ is $\cu(1)$. It would suffice to note that the Schwinger cocycle (evaluated on $\mathfrak{u}_{\rm res} \times \mathfrak{u}_{\rm res}$) can take any value in $i \IR = \mathrm{Lie}(\cu(1))$. Similarly for $\Gres(\HH)$.\\
%\noindent Unfortunately, generalization of these finite-dimensional results to infinite dimensions is usually very problematic. It might very well be possible in our special case, as we have principle-bundles with finite-dimensional Lie groups over infinite-dimensional Banach-manifolds, which are generally very well-behaved. Also, as the the curvature generates the whole Lie algebras of the structure groups, we would need only one inclusion of the Ambrose-Singer theorem. The other inclusion uses Frobenius-Theorem, of which there's no general infinite-dimensional version. But to circumvent such problems altogether, we have decided to use the explicit computation presented above.\\
\newpage

\section{Gauge Invariance}
\label{Sec:Gauge Invariance}

We have already seen that gauge transformations are not implementable on the Fock space. Actually, we knew that this must be true, because a gauge transformation  
\begin{equation*}\mathcal{G}\ni g: \Psi(\underline x) \rightarrow e^{i\, \Lambda_g(\underline x)} \Psi(\underline x),\;  \Lambda_g \in C^{\infty}_c(\IR^3, \IR)\end{equation*}
changes the spatial component of the A-field and therefore the polarization class.\\
This fact is troubling, but not necessarily a disaster. It just tells us that we have to give up the naive idea about gauge invariance (if we ever had it in the first place) that the symmetry translates directly from the one-particle theory to the second quantized theory. At this level of description, the significance of a gauge transformation in the second quantized theory is rather unclear, anyways. We would really have to understand what quantities in the theoretical description have actual physical significance and are in fact required to be gauge-invariant. However, such a discussion, in its full depth, is certainly beyond the scope of this work. Yet, there is little controversy about the fact that the \textit{scattering matrix} can be taken seriously for the physical description, and therefore, we should require that its gauge invariance carries over from the one-particle theory to the second quantized theory. This is indeed of particular importance, because then:
\begin{equation} \mathbf{S}[\sa_\mu] = \mathbf{S}[\sa_\mu - \epsilon \partial_\mu\Lambda], \; \forall \Lambda \in C^{\infty}_c(\IR^4, \IR) \end{equation}
implies, after taking the derivative with respect to $\epsilon$ at $\epsilon = 0$:
\begin{equation} 0 = - \int \mathrm{d}x\, \frac{\delta}{\delta \sa_\mu(x)} \mathbf{S}[\sa]\,\partial_\mu\Lambda = \int \mathrm{d}x\, \Lambda(x) \partial_\mu \frac{\delta}{\delta \sa_\mu(x)} \mathbf{S}[\sa] \end{equation}
and thus, by definition \eqref{currentdesity} of the current density:
\begin{equation}\label{continuityeq}\addtolength{\fboxsep}{5pt}\boxed{ \partial_\mu \frac{\delta \mathbf{S}}{\delta \sa_\mu(x)} =  \partial_\mu\, j^\mu(x) \equiv 0. }\end{equation}
\vspace{3mm}

\noindent So, the gauge invariance of the second quantized S-matrix \textit{is} physically significant because it implies the continuity equation \eqref{continuityeq} for the current density.
(Actually, equation \eqref{eq:Strom von Phase} shows that a much weaker condition would suffice.
The current density will be gauge-invariant if the first distributional derivative $\frac{\delta \varphi}{\delta A(x)}$ of the phase of $\mathbf{S}$ is gauge-invariant. In that case, the continuity equation follows analogously from $j^\mu [\sa_\mu] = j^\mu[\sa_\mu - \epsilon \partial_\mu\Lambda]$.)\\

\noindent Our method of geometric second quantization is -- a priori -- not gauge-invariant. Although the one-particle S-operator is invariant under compactly supported gauge transformations, the unitary time evolution and the renormalization are \textit{not}. Therefore, if we use parallel transport to lift the S-matrix to $\Ures(\HH)$, once for the external field $\sa = (\sa_\mu)_{\mu =0,1,2,3} \in C^\infty_c(\IR^4, \IR^4)$ and once for the gauge-transformed field $\sa' = \sa_\mu - \partial_\mu \Lambda,\; \Lambda \in C^\infty_c(\IR^4, \IR)$, we will perform the parallel transport along different paths in $\Ur(\HH)$ and again, the lifted S-matrix can differ by any complex phase. (This gauge-anomaly can be explicitly computed using formula \eqref{holformula}. See also the explicit computation in \cite{LaMi}.)\\ 
\newpage
\noindent However, we suggest that it is possible to define the renormalization precisely in such a way that the gauge transformation of the renormalization and the gauge transformation of the time evolution cancel each other out. We will call such a renormalization \textit{gauge-covariant}.\\

\begin{Construction}[Gauge-Covariant Renormalization]
\mbox{}\\
Let $\T: \IR \times \A \rightarrow \cu(\HH)$ be a smoothening renormalization, e.g. the renormalization \eqref{LaMirenormalization}, constructed by Langmann and Mickelsson.
Let $\A$ be the space of vector potentials and $\A \slash \mathcal{G}$ the set of gauge-classes, i.e. 4-vector potentials modula the equivalence relation
\begin{equation} \sa_\mu \sim \sa_\mu - \partial_\mu \Lambda(x), \; \text{for}\; \Lambda \in C^\infty_c(\IR^4, \IR). \end{equation}
More geometrically, we can identify $\A$ with the space $\Omega^1(M , \IR)$ of (compactly supported) one-forms on Minkowski-space (or a general space-time manifold $M$). Then,  $\A \slash \mathcal{G}$ corresponds to the first de-Rham cohomology class $\mathrm{H}^1_{dR}(M, \IR)$.\\

\noindent By axiom of choice, we choose one representative $\sa$ out of each gauge-class $[\sa] \in \A \slash \mathcal{G}$. In particular, we choose $\sa \equiv 0$ as a representative of $\overline{0} \in \A \slash \mathcal{G}$. 
Now, we define a renormalization $\widetilde{\T}: \IR \times \A \rightarrow \cu(\HH)$ by setting $\widetilde{\T}_t\bigl(\sa\bigr) = \T_t\bigl(\sa\bigr)$ and 
\begin{equation}\label{gaugerenormalization} \widetilde{\T}_t \bigl(\sa_\mu - \partial_\mu \Lambda\bigr) := e^{i \Lambda(t, \underline x)}\,\widetilde{\T}_t\bigl(\sa\bigr), \; \forall \Lambda \in C^\infty_c(\IR^4, \IR).\end{equation}  

\noindent By theorem \ref{Thm:Gaugetransformations}, this is indeed a renormalization. Now, with this setting, we find that the renormalized (Schr\"odinger picture) time evolution is in fact gauge-invariant:
%\begin{align*}& e^{i\Lambda(t, x)} \in \cu^0_{\rm res}\bigl(\HH, C[\underline\sa], \HH, C(\underline\sa + \underline{\nabla} \Lambda)\bigr)\\
%\Rightarrow &e^{i\Lambda(t,x)}\, \T(\underline\sa) \in \cu^0_{\rm res}\bigl(\HH, C[\underline\sa], \HH, C(\underline\sa + \underline{\nabla} \Lambda)\bigr)\cdot \cu^0_{\rm res}\bigl(\HH, [\HH_+], \HH, C[\underline\sa]\bigr) = \cu^0_{\rm res}\bigl(\HH, [\HH_+], \HH, C(\underline\sa + \underline{\nabla} \Lambda)\bigr), \; \forall t \end{align*}
%\newpage
\begin{align*} U^{\sa}_{ren}(t,t') =\;  &  \widetilde{\T}^*_t(\sa)\; U^\sa(t, t') \; \widetilde{\T}(\sa(t'))\\[1.5ex]
\vspace{3cm} \xrightarrow{\;\;g\; \in\; \mathcal{G}\;} \; \;& U^{\sa - \partial\Lambda}_{ren}(t,t')= \widetilde{\T}_t^*\bigl(\sa - \partial_\mu\Lambda \bigr) \, U^{\sa - \partial\Lambda}(t, t')\, \widetilde{\T}_t'\bigl(\sa - \partial_\mu\Lambda \bigr) \\[1.5ex]
=\; & \widetilde{\T}^*_t(\sa)\;  e^{-i \Lambda(x)}  \; U^{\sa - \partial\Lambda}(t, t')\; e^{i \Lambda(x)} \; \widetilde{\T}_t({\sa})\\[1.5ex]
=\;  & \widetilde{\T}^*_t({\sa})\; U^{\sa}(t, t') \; \widetilde{\T}_t({\sa}) = \; U^{\sa}_{ren} (t,t')\end{align*}
and because the renormalized time evolution for the distinguished representative is (after switching to the interaction picture) differentiable in $t$ with respect to the differentiable structure on $\Ur(\HH)$, the renormalization $\widetilde{\T}$ is also smoothening. \\

Using such a gauge-covariant renormalization, the second quantization of the S-matrix is \textit{completely invariant} under gauge transformations -- and so is the second quantization of the time evolution for intermediate times. De facto, we do not see gauge transformations in the second quantized theory at all.
In particular, the renormalization acts by ``gauging the field away'' whenever this is possible, which seems like a very elegant solution.

To summarize, we suggest that the problem of the \textit{gauge-anomaly} for the second quantization of the S-matrix by parallel transport, as discussed in \cite{LaMi} and \cite{Mi98}, is not necessarily a problem at all. Since the choice of a renormalization is ambiguous anyways, we can use this freedom to cancel out the effect of gauge transformations entirely. In fact, we can just choose one particular vector potential out of each gauge-class. This choice, however, as well as the smoothing renormalization applied to it, remains ambiguous. 
\end{Construction}

%\begin{Remark}[Hodge Decomposition]
%\mbox{}\\
%For future discussions it might be interesting to know how restrictive the condition of ``gauge-covariance'' is on the renormalization. This is easiest to answer in a setting where $\A$ is realized as the space of smooth one-forms $\Omega^1(M, \IR)$ on a closed manifold $M$.\\ Then, we have the \textit{Hodge decomposition}:
%\begin{equation*} \Omega^1(M) = \im(\mathrm{d}) \oplus \im(\delta) \oplus \mathrm{Harm^1}(M)\end{equation*}
%with \begin{enumerate}[$\cdot$]
%\item $\im(\mathrm{d})$ the space of exact one-forms $\lbrace df \mid f \in C^\infty(M) \rbrace$
%\item $\im(\delta)$ the space of co-exact one-forms $\lbrace \delta\beta \mid \beta \in \Omega^2(M)\rbrace$
%\item $\mathrm{Harm^1}(M)$ the space of harmonic one-forms $\lbrace \alpha \in \Omega^1(M) \mid \Delta \alpha = 0 \rbrace$ \end{enumerate}
%\noindent The condition of gauge-covariance would then determine the renormalization only on exact one-forms, i.e. on $\im(d) \subset \Omega^1(M)$ by $\T(\mathrm{d}\Lambda) = e^{i\Lambda}$. 
% \end{Remark}

\chapter{Closing Arguments}
%\chapter{R\'{e}sum\'{e}}
%\onehalfspace
\vspace*{-1mm}
\subsubsection{What was done.}

We have present several equivalent constructions of the fermionic Fock space. In all cases, we arrive at the Shale-Stinespring criterion \eqref{SS} as a necessary and sufficient condition for a unitary transformation to be implementable on the Fock space. In this case, the second quantization is unique up to a complex phase of modulus one, called the ``geometric phase'' in QED (because it corresponds to the freedom in lifting elements from $\Ur(\HH)$ to the principle-$\cu(1)$-bundle $\Ures(\HH)$).
It turns out that the Dirac time evolution does in general \textit{not} satisfy the Shale-Stinespring condition and thus cannot be second quantized in the usual way. In fact, the time evolution will ``leave'' the Fock space as soon as the spatial component of the external field becomes non-zero (Thm. \ref{Ruij}, Ruijsenaar, Thm. \ref{Thm:polarizationclasses},  Deckert et.al.). Physically, this corresponds to the problem of infinite particle creation in the presence of a magnetic field. In other, more technical terms, it means that the free vacuum and the vacua of an interacting theory correspond to different, non-equivalent representations of the CAR-algebra. Note that those are intrinsic properties of the mathematical structure and not pathologies due to the problem of self-interaction and ultra-violet divergences appearing in the full theory. 

The situation is better in the asymptotic case, if we study the S-matrix only. Before the interaction is turned on and after it is turned off, we have a more or less canonical construction of the Fock space corresponding to the usual Fock representation and we can compare ``in states'' (at $t= -\infty$) and ``out states'' (at $t= +\infty$) with respect to the same vacuum. 

We emphasize that these problems, that were discussed here in the context of external field Quantum Electrodynamics, are certainly not resolved in the ``full-blown'' theory (where the photon-field is quantized and treated as another dynamic variable), nor are they restricted to QED or electromagnetic interactions. In fact, a corresponding result is known as \textit{Haag's theorem} in algebraic quantum field theory ever since the mid 50's (see e.g. \cite{StWigh} for a nice treatment). In the words of A.S. Wightman:
%\begin{quotation}``[T]here is a widespread opinion that the phenomena associated with Haag's Theorem are somehow pathological and irrelevant for real physics. In this section I make one more attempt to explain why that is not the case. [...] \emph{
\begin{quote}\textit{``Not only do strange representations occur ... but different strange representations at very different time. The defining equations of the interaction picture ... are as wrong as they can possibly be. [...] The strange representations associated with Haag's theorem are, in fact, an entirely elementary phenomenon and appear as soon as a theory is euclidean invariant and has a Hamiltonian which does not have the no-particle state as proper vector. This will happen whether or not the theory is relativistically invariant and whether or not there are ultra-violet problems in the theory.''} \cite{Wigh}, S. 255.
%\begin{flushright}
  %  \footnotesize{A.S. Wightman. Introduction to Some Aspects of the Relativistic Dynamics of Quantized Fields. \cite{Wigh}, S. 255.}
 %  \end{flushright}
\end{quote}

After acknowledging those brute and troubling facts, we have presented two different solutions for realizing a time evolution in the external field setting of QED. In a way, the two alternatives seem to be the best that can be done within the existing framework. In Chapter 6, we followed Deckert et.al. and realized the time evolution as unitary transformations between time-varying Fock spaces. In Chapter 8, we introduced the concept of a ``renormalization'' used to render the time evolution implementable on the standard Fock space. The nomenclature, borrowed from \cite{LaMi}, can be misleading, as it doesn't refer to the usual renormalization schemes applied in perturbative field theory but merely to a simple and well-defined prescription for mapping the time evolution back into $\Ur(\HH)$.\footnote{If and how the two concepts are related might be an interesting question to pursue in a future analysis.} In both cases, whether we use renormalizations or time-varying Fock spaces, the implementations of the unitary operators are again unique up to a complex phase. We have proven that the two descriptions are dual to each other and showed how the renormalization can be used to translate between them. One advantage of the renormalized theory might be that it allows a second quantization of the Hamiltonians as well, after a suitable (``smoothening'') renormalization is applied. However, the renormalization introduces a bunch of artificial terms into the renormalized Hamiltonian \eqref{Zren}, and the physical relevance of these objects remains unclear.\\

In Chapter  7, we introduced the ``Langmann-Mickelsson connection'' on the principle bundle $\Ures(\HH)$, or rather its complexification $\Gres(\HH) \rightarrow \Gl_{\rm res}(\HH)$. Parallel transport with respect to this connection defines a unique second quantization of a continuous family of unitary operators, including a smooth prescription for the phase. It also defines a unique second quantization of the renormalized Hamiltonians, generating the second quantized time evolution. However, the hope that the additional geometric structure of a bundle connection can eliminate the $\cu(1)$-freedom of the geometric phase completely, at least for the second quantization of the scattering matrix, did not stand up to scrutiny. The reason is that if we want to apply the method of parallel transport to the unitary time evolution, we will always (except for the most trivial cases) need a smoothening renormalization in the sense of Def. \ref{Def:Renormalization} to make the time evolution differentiable with respect to the differentiable structure on $\GL_{\rm res}(\HH)$. The choice of such a renormalization, however, is not unique and different choices will lead to different renormalized time evolutions. Although the scattering matrix remains invariant under renormalization, different choices will correspond to parallel transport along different paths in $\Ur(\HH)$ which can lead to different phases for the second quantized scattering operator. The path-dependence of the phase is expressed by the holonomy group of the $\Ures(\HH)$ bundle which we computed to be isomorphic the the whole structure group $\cu(1)$. In this sense, the entire freedom of the geometric phase, supposedly eliminated by the parallel transport, reappears in different disguise. We have to conclude that the result stated in \cite{LaMi}, saying that \textit{``the phase [of the second quantized scattering operator] is uniquely determined ... by the geometric structure of the central extension of the group of one-particle (renormalized) time evolution operators''} is too optimistic. We emphasize that the authors themselves have revised their initial statement and addressed the problem in \cite{Mi98}.

However, we have shown that this method of \textit{geometric second quantization} has other important benefits. In Chapter 8 we proved that second quantization by parallel transport preserves the semi-group structure of the time evolution and -- more importantly -- the causal structure of the one-particle theory. In particular, the phase of the second quantized S-matrix is causal in the sense of Scharf (\cite{Scha}, \S 2.8).\footnote{This feature was already mentioned in \cite{Mi98}, however without explanation or proof.} The results can be applied in the context of time-varying Fock spaces, showing that the implementations can be chosen in such a way to yield an actual time evolution. Furthermore, we have proposed a simple solution to the problem of gauge invariance, showing that the geometric second quantization prescription for the S-matrix can be made gauge-invariant by using suitable, ``gauge-covariant'' renormalizations. Those were the ``optimistic'' results of our analysis.\\ 

Still, we have to realize that at this point, the ambiguities in the construction of the time evolution seem too vast to provide meaningful physical content. In the framework of time-varying Fock spaces, this fact might be somewhat concealed because the language might suggest the intuition that the physical state is what it is (modulo phase) and only the mathematical space that inhibits it somehow requires additional specification. But in the renormalized theory, this ambiguity translates into the freedom to lift basically \textit{any} unitary family of operators to the (fixed) Fock space, as long as it stays in $\Ur(\HH)$ and agrees with the Dirac time evolution whenever the interaction is turned off. Therefore, it seems that -- without additional structure -- the time evolution doesn't actually tell us anything. In other words: we do not even know what physical quantities should characterize the states represented on different Fock spaces. Without additional ingredients in the theory, specifying (instantaneous) vacuum states, we cannot even say how many particles and anti-particles exist at any given time.\footnote{This is one of the big problems of quantum field theory in general. QFT on curved space-time is a good conscious raiser for this issue. Rather spectacular phenomena like the Unruh effect or Hawking radiation can be traced back to it.} So, we have to ask: what actually is the \textit{physics} behind our mathematical formalism?

\subsubsection{What it means.}

We concede that -- up to this point -- our rigorous approaches to the time evolution for the external field problem contain too many unspecified degrees of freedom to allow an unambiguously defined time evolution and a straight forward physical interpretation in the language of many-body Quantum Mechanics. It is an interesting realization on its own -- apparent mainly in the context of time varying Fock spaces -- that this problem is closely related to the question of a consistent \textit{particle interpretation}. Such a particle interpretation requires a well-defined specification of (instantaneous) vacuum states that would also define a unique time evolution on time-varying Fock spaces, up to the geometric phase.\\ 

Now, there are basically three general attitudes that  can be taken towards this and I'd like to conclude this work by taking the liberty to discuss them from a rather abstract point of view. First, we could hope that we can reduce the degrees of freedom -- and thereby the ambiguities -- by introducing additional structure, or that we can impose \textit{physical} principles to render the choices involved in our constructions essentially unique. Our analysis shows that the principles of \textit{causality} and \textit{gauge invariance} are not sufficient to do the job. The question of \textit{Lorentz invariance}, however, was addressed only briefly and requires further discussion. The situation might also be better, if we restrict the class of permissible interaction potentials, e.g. to actual solutions of Maxwell's equations. Most likely, though, the crucial insights would have to come from a fully-interacting theory. We will return to this speculative area later on.\\

Second, we could conclude that our theory of QED -- despite its great success within its limits -- is deeply and fundamentally flawed and that a fundamentally different approach will be necessary to resolve the various issues. This would certainly be a premature judgement if based only on our narrow treatment of the external field problem. But given the fact that \textit{all} known version of QED stubbornly resist a consistent mathematical formulation, I personally think that this alternative has great appeal.\\ 

Lastly, one may regard the freedoms contained in the construction not as ambiguities, but rather as a kind of ``internal symmetry'' of the theory and look for a physical description in terms of physical quantities that are invariant under these symmetries, i.e. independent of the particular choice of Fock spaces, right-operations etc.. This seems to be more or less the way in which quantum field theories are usually handled, with indisputable practical success. However, I find such an approach very unsatisfying, not because the endeavor is practically unpromising, but because it must inevitably lead to the ontological vagueness that I criticized in the established formulations of quantum field theory. I am convinced that a good physical theory must provide a clear physical interpretation on its most fundamental level, preferably in terms of a clearly specified local ontology. 

For example, the results of this work might very well be understood as sustaining the predominant understanding that \textit{not the time evolution but only the S-matrix should be taken seriously in relativistic quantum field theory}. But the right question is, whether our failure to do better is a statement about nature, or merely a statement about the capability of our existing theories. We can very well conclude with Scharf and Fierz  that ``the notion of particles has only asymptotic meaning" (\cite{FS}, p. 453). But then, we have to ask, in what terms we are supposed to think about the world as it is \textit{right now}, i.e. between $t= -\infty$ and $t= +\infty$. If there are no particles, \textit{what else IS there}?\\

It is often suggested to be a very deep feature of relativistic quantum theory that physical relations can be defined only operationally, i.e. in terms of transition amplitudes, scattering cross-sections, etc.. (This point of view probably goes back to Heisenberg \cite{HeisA}, \cite{HeisB}, who was convinced that the S-matrix should take over the role of the Hamiltonian in a ``future'' theory and that all ``observable quantities'' are encoded in it.) However, one should be sceptical about the arguments commonly put forward to sustain this claim; quite often, they are either very vague or claim a generality which is unjustified, because they are actually trapped within the narrow limitations of the current, deficient, framework. In fact, there are some remarkable examples, demonstrating that a relativistic quantum theory is very well compatible with a local ontology. There is Tumulkas relativistic version of GRW with flash ontology, for instance, which has caused some furor in recent years (although the theory is not interacting, yet, see \cite{Tum}). It is also possible -- quite easily, in fact -- to write down a ``Bell-type`` quantum field theory with particle trajectories that reproduces the predictions of regularized QED, including particle creation and annihilation (although the formulation is not fully Lorentz-invariant, yet, see \cite{DGTZ}). Of course, neither of these examples are ultimately satisfying (or even meant to be), but they are striking counterexamples to widespread claims, that an understand of relativistic quantum theory in the proper sense cannot be possible.\\

\newpage

The analysis presented in this work reveals that the scattering matrix is in general a well-defined object in the second quantized theory, whereas the time evolution is not. But it's important to note what was actually proven. The one and only reason why things work better in the asymptotic case is that one assumes \textit{the interaction to be turned off} in the distant future and past (or at least falling off quickly, if the regularity conditions on the fields are weakened). In my opinion, this points to the conclusion that the difficulties are not a matter of Lorentz invariance or some deep feature of relativistic space-time, but that we need a better understanding of the fully interacting theory in terms of the fundamental interactions of its elementary physical entities -- presumably the particles.\\

It was suggested by Dirk Deckert and Detlef D\"urr (in talks and private communications) that the Dirac sea theory with infinitely many particles should be understood as an \textit{effective description} (a thermodynamic limit) of a more fundamental theory, involving a large, but finite number of particles with a well-defined pair interaction -- either transmitted by the photon field, or without any photon field at all, like in the Wheeler-Feynman formulation of classical electrodynamics \cite{WF45}, \cite{WF49}. An analysis of the particle dynamics would then tell us what the appropriate ``vacuum state'' is that we should use in the effective description for specific physical interactions. It would correspond to an equilibrium state of the ``Dirac sea''  in which the pair interactions within the sea effectively cancel, justifying the description in terms of Dirac's hole theory (see also the quote of Dirac cited in Chapter 1). At the moment, these are just speculations, of course. However, I think that this is a path worth pursuing, since it would provide the clearest and most ``down to earth'' account for the phenomena of Quantum Electrodynamics.

%This would be without a doubt a beautiful solution, if only such invariant quantities where directly available or, in some sense, obvious. But this is not the case; a physical picture in such terms would have to emerge from a higher-level description on the formal (mathematical) side. heory. 

%So, the usual formulations of Quantum Electrodynamics resist stubbornly a consistent mathematical formulation. 

%\noindent This would constitute remarkable progress. Still, I must conclusively say that I would personally find the solution too constructed and not completely satisfying.  
\singlespace

\newpage 
\thispagestyle{empty}
\quad 
\newpage

\appendix
\chapter*{}

\vspace{-2cm}
\begin{bfseries} \huge Appendix A \end{bfseries}\\
\setcounter{chapter}{1}

\addcontentsline{toc}{chapter}{Appendix}
\section{Commensurability and Polarization Classes}
\begin{nndef}[Commensurability]
\mbox{}\\
Two polarizations $V,W \in \pol(\mathcal{H})$ are called \emph{commensurable}, if $V\cap W$ has finite codimension in both $V$ and $W$.
\end{nndef}

\begin{nnprop}[Comsurable Polarizations and Polarization Classes]
\mbox{}\\
Let $ V \in \pol(\mathcal{H})$. We denote by $\mathrm{Gr}(\HH, V)$ the restricted Grassmannian of $V$ i.e. its polarization class endowed with the structure of a complex Hilbert-manifold modeled on $I_2(V, V^{\perp})$.
Then, the set of all polarizations $ W \in \pol(\mathcal{H})$ commensurable with $V$ is dense in $\mathrm{Gr}(\HH, V)$ and for any such $W$, $\charge(V,W)$ coincides with \eqref{relcharge}, i.e.
\begin{equation*} \charge(V,W) = \ind(P_V|_{W\rightarrow V}) = \dim(V/(V\cap W)) - \dim(W/(V\cap W)) \end{equation*}
\end{nnprop}

\begin{proof}
Let $\mathrm{Comns(\HH,V)} := \lbrace W \in \pol(\mathcal{H}) | W\;\text{commensurable with}\; V \rbrace$.

\begin{enumerate}[i)]
\item Claim: $\mathrm{Comns(\HH,V)} \subset  \mathrm{Gr}(\HH,V)$.\\
Let $W \in \mathrm{Comns(\HH,V)}$. We have to show $P_W - P_V \in I_2(\HH)$. Obviously, $P_W - P_V$ is zero on $V\cap W \subseteq \HH$, as well as on $(V + W)^{\bot} \subseteq \HH$. Thus, since $V$ and $W$ are commensurable, i.e. $V\cap W$ has finite codimension in $V$ and $W$, it has also finite codimension in $V + W$ and $P_W - P_V$ is non-zero on a finite-dimensional subspace only and therefore of Hilbert-Schmidt type.

\item Claim: $\charge(V,W) = \dim(V/(V\cap W)) - \dim(W/(V\cap W))$.\\ 
We can write $W = (W\cap V)\oplus \ker(P_V|_W)\oplus R$ with some subspace $R \subset \HH$.\\ 
Then, $\dim(\ker(P_V|_W)\oplus R) < \infty$ and: 
\begin{align*}\ind(P_V|_W) =& \dim \ker(P_V|_W) - \dim \coker(P_V|_W)\\
=& \dim\bigl(W/(V\cap W \oplus R)\bigr) - \dim\bigl(V/(V\cap W\oplus P_V \, R)\bigr)\\
=& \dim\bigl(W/(V\cap W)\bigr) - \dim\bigl(V/(V\cap W)\bigr) \end{align*}
since $P_V\lvert_R$ is a finite-dimensional isomorphism, hence $\dim(R) = \dim(P_V R) < \infty.$ 

\item Claim: $\mathrm{Comns(\HH,V)} \subset  \mathrm{Gr}(\HH,V)$ is dense (in the topology introduced in \S \ref{Sec:Grassmannian}).\\
$\mathrm{Gr}(\HH,V)$ is a homogeneous space for $\Ur(\HH; V) := \lbrace U \in \cu(\HH) \mid [P_V - P_{V^\perp}] \in I_2(\HH) \rbrace$.  $\mathrm{Comns(\HH,V)}$, on the other hand, is the orbit of $V$ under $\lbrace  \cu(\HH) \mid [P_V - P_{V^\perp}]\; \text{has finite rank} \; \rbrace$. Since the finite-rank operators are dense in the Hilbert-Schmidt class (and the Hilbert-Schmidt norms are smaller than the $\Ur$-norm) it follows that $\mathrm{Comns(\HH,V)}$ is dense in $\mathrm{Gr}(\HH,V)$. 

%Let $X \in \mathrm{Gr}(\HH,V)$. Let  $W := \im(P_X|_{V\to X}) \oplus \ker(P_X|_{V\to X})$. Then, $W / (V \cap W) \cong %\ker(P_X\lvert_V)$ and $V / (V \cap W) \cong \ker(P_X\lvert_V)$

%$W$ isa polarization, commensurable with V, since $P_V|_W$ is a Fredholm operator and this implies $\dim(V/(V\cap %W')) = \dim \ker(P_V|_W) < \infty$ as well as $\dim(V/(V\cap W')) = \dim \coker(P_V|_W) < \infty$. Now, there is an %operator $T: W' \rightarrow (W')^{\bot}$ with $W = \mathrm{Graph}(T)\subset \mathcal{H}$. This operator is given by: 
%\begin{equation*}T\lvert_{\ker(P_V|_W)} = 0\; \text{and}\; T\lvert_{\ker(P_V|_W)}^{\bot} = (P_W - P_V)|_{\ker(P_V|_W)}^%{\bot}
%\end{equation*}
%Thus, T is a Hilbert-Schmidt operator (i.e. W lies in the domain of the chart around W' as defined in \eqref{U_W} ). Now, if %we note that the graph of any finite-rank operator $G: W' \rightarrow (W')^{\bot}$ is also commensurable with V, the claim %follows directly from the fact that any Hilbert-Schmidt operator can be approximated by finite-rank operators (in the %Hilbert-Schmidt norm, which is the relevant norm here).
\end{enumerate}
\end{proof}
\newpage
\section{On Banach Lie groups}
\begin{Lemma}[Continuous One-parameter Groups are smooth]
\mbox{}\\
Let $\mathcal{A}$ a Banach algebra and $G \subseteq \mathcal{A}^{\times}$ a linear Lie group.\\
Then, every continuous one-parameter group $(\gamma(t))_{t \in \IR}$ in $G$ is smooth.
\end{Lemma}
\begin{proof}
Let $\gamma(t), \, t \in \IR$ a continuous one-parameter subgroup in $G$.
We write $j: G \hookrightarrow \mathcal{A}^{\times}$ for the natural inclusion and define $\gamma_j(t) := j \circ \gamma(t)$. Since $\gamma$ is continuous, there exists $\epsilon >0$ with $\lVert \gamma(t) - \mathds{1} \rVert < \frac{1}{2}, \; \forall t \in \IR, \lvert t \rvert < \epsilon$. Let $\rho: \IR \to \IR_+$ be a smooth cut-off function with compact support in $[-\epsilon, \epsilon]$ and $\int\limits_{\IR} \rho(s) \mathrm{d}s = 1$.\\
In a Banach space (even in infinite dimensions) we can do path-integration and thus define
\begin{align*} \tilde{\gamma}(t)  := \rho * \gamma_j (t) & =  \int\limits_{\IR} \rho(s) \gamma_j(t-s) \mathrm{d} s =  \gamma_x(t) \int\limits_{\IR} \rho(s) \gamma_j(-s) \mathrm{d} s\\ 
& =  \gamma_j(t) \int\limits_{\epsilon}^{\epsilon} \rho(s) \gamma_j(-s) \mathrm{d} s
 =:  \gamma_j(t) \cdot g \end{align*}
with $g :=  \int\limits_{\epsilon}^{\epsilon} \rho(s) \gamma_j(-s) \mathrm{d} s \in \mathcal{A}$. We can estimate: 
\begin{align*} \lVert g - \mathds{1} \rVert  =  \rVert \int\limits_{\epsilon}^{\epsilon} \rho(s) \gamma_x(-s) \mathrm{d} s - \mathds{1} \rVert  \leq \int\limits_{\epsilon}^{\epsilon} \rho(s) \lvert \gamma_x(-s) - \mathds{1} \rVert \mathrm{d} s
 \leq \frac{1}{2} \int\limits_{\epsilon}^{\epsilon} \rho(s) \mathrm{d} s = \frac{1}{2}
\end{align*}
so that $g \in \mathcal{A}^{\times}$ i.e. $g$ is invertible. Therefore: $ \gamma_j(t) := \tilde{\gamma}(t) \cdot g^{-1}$.
But $\tilde{\gamma}(t)$ is smooth, since
\begin{equation*}\frac{d^k}{dt^k} \tilde{\gamma}(t) = \frac{d^k}{dt^k} \int\limits_{\IR} \rho(s) \gamma_x(t-s) \mathrm{d} s =   \frac{d^k}{dt^k} \int\limits_{\IR} \rho(t-s) \gamma_x(s) \mathrm{d} s =   \int\limits_{\IR} \rho^{(k)}(t-s) \gamma_x(s) \mathrm{d} s. \end{equation*}
Therefore, $\gamma_j = \tilde{\gamma} \cdot g^{-1}$ is smooth and so is $\gamma(t)$.\\ 
\end{proof}

\begin{Proposition}[Functorial property of the Lie Algebras]\label{propfunctorialproperty}
\mbox{}\\
Let $G_1, G_2$ be linear Banach Lie groups with Lie algebras $\mathrm{Lie}(G_1)$  and  $\mathrm{Lie}(G_2)$ , respectively. Let $\varphi: G_1 \to G_2$ a \emph{continuous} group homomorphism. Then, the derivative 
\begin{equation*} \frac{d}{dt}\bigl\lvert_{t=0} \varphi(\exp(t x) ) =: \mathrm{Lie}(\varphi) x \end{equation*}
exists for every $x \in \mathrm{Lie}(G_1)$. 
This defines the unique (continuous) Lie algebra homomorphism $\mathrm{Lie}(\varphi)$ such that the following diagram commutes:
\begin{equation*}
\begin{xy}
  \xymatrix{
G_1  \ar[r]^ \varphi & G_2    \\
   \mathrm{Lie}(G_1)\ar[u]^{\exp} \ar[r]^ {\mathrm{Lie}(\varphi)}         &  \mathrm{Lie}(G_2) \ar[u]_{\exp} 
  }
\end{xy}
 \end{equation*}
\end{Proposition}

\begin{proof} Let $x \in \mathrm{Lie}(G_1)$. Then, $\varphi(\exp(tx))$ is a continuous one-parameter group in $G_2$ and therefore, according to the previous Lemma, even differentiable, i.e. $\frac{d}{dt}\bigl\lvert_{t=0} \varphi(\exp(t x) )$ exists and we define this to be $\mathrm{Lie}(\varphi) x$. But the unique one-parameter group in $G_2$ whose differential at the identity equals $\mathrm{Lie}(\varphi) x$ is $t \mapsto \exp (t\,  \mathrm{Lie}(\varphi) x )$. We conclude:
\begin{equation}\label{Lievarphi} \varphi \circ  \exp_{G_1} (x)  = \exp_{G_2}  \circ\, \mathrm{Lie}(\varphi) (x) , \; \forall x \in \mathrm{Lie}(G_1). \end{equation}
\noindent Obviously, $\mathrm{Lie}(\varphi)$ is $\IR$-, respectively $\IC$- homogeneous.\\ 
To show that it defines a Lie algebra homomorphism, we use the Trotter product formula:
\begin{equation}\lim\limits_{k \to \infty} \bigl( \exp(x/k) \exp(y/k)\bigl) ^k = \exp(x+y) \end{equation}
and the commutator formula
\begin{equation}\label{commutatorformula}
\lim\limits_{k \to \infty} \bigl( \exp(x/k) \exp(y/k) - \exp(y/k) \exp(x/k)    \bigl) ^k = \exp(xy - yx). \end{equation}
Then:
\begin{align*} \varphi(\exp(x+y)) =  & \lim\limits_{k \to \infty} \varphi \bigl( \exp(x/k) \exp(y/k)\bigl) ^k =  \lim\limits_{k \to \infty} \varphi \bigl( \exp(x/k) \bigr)^k  \varphi \bigl( \exp(y/k)\bigl) ^k \\
= & \lim\limits_{k \to \infty}  \exp(\frac{1}{k}\mathrm{Lie}(\varphi)x ) \bigr)^k \exp(\frac{1}{k}\mathrm{Lie} (\varphi) y)\bigl) ^k  = \exp\bigl( \mathrm{Lie}(\varphi) x + \mathrm{Lie}(\varphi) y\bigr).
\end{align*}
which means \begin{equation*} \mathrm{Lie}(\varphi)(x+y) = \mathrm{Lie}(\varphi)x +  \mathrm{Lie}(\varphi)y. \end{equation*}
Similarly, with \eqref{commutatorformula}, 
$ \varphi( \exp([x,y] ) = \exp\bigl( [ \mathrm{Lie}(\varphi)x, \mathrm{Lie}(\varphi) y ] \bigr)$ and thus   \begin{equation*}\mathrm{Lie}(\varphi)([x,y]) = \bigl[\mathrm{Lie}(\varphi)x, \mathrm{Lie}(\varphi)y\bigr]. \end{equation*}

\noindent Hence $\mathrm{Lie}(\varphi)$ defines a Lie algebra homomorphism. From \eqref{Lievarphi} and the fact that the exponential map is a local diffeomorphism around $0$, it follows that $\mathrm{Lie}(\varphi)$ is continuous. 
 
\end{proof}
\newpage

\section{Derivation of the Charge Conjugation}
\noindent We present a general derivation of the charge conjugation operator $\CC$ mapping negative\\ 
energy solutions of the Dirac equation to positiv energy solutions with opposite charge.\\

\noindent More precisely, we want a transformation $\CC$ satisfying
\begin{enumerate}[i)]
\item $\CC H(e) \CC^{-1} = - H(-e)$\\
\item $ i \hbar \frac{\partial}{\partial t} \Psi = H(e) \Psi \iff  i \hbar \frac{\partial}{\partial t} \CC \Psi = H(-e) \CC\Psi$
\end{enumerate}
for H(e) the Dirac-Hamiltonian with charge e and any eigenstate $\Psi$.\\
i) and ii) imply
\begin{equation*} i \hbar \frac{\partial}{\partial t} \CC \Psi = - \CC H(e) \Psi \end{equation*}
For an eigenstate $\Psi \neq 0$ this is only possible, if $\CC$ is \emph{anti-linear}. We can thus write
\begin{equation*} \CC \Psi = \CC \Psi^{cc.}\end{equation*}
(cc. denotes complex conjugation).\\

\noindent Now we take a look at the Hamiltonian:
\begin{equation*} H(e) = -i \underline{\alpha} \cdot \underline{\nabla} - e \underline{\alpha} \cdot \underline{A} + m \beta + e \Phi \end{equation*}
and observe that in order to satisfy i), we need
\begin{equation}\label{conjtransform} \beta C = - C \beta^{cc.} \,; \; \alpha_k\,C = C\, \alpha^{cc.}_{k} \end{equation}
This, together with the (anti-) commutation relations for the $\alpha$ matrices (or $\gamma$ matrices, respectively) are enough to determine the form of the charge conjugation operator in any given representation. In particular, note that 
\begin{align*}
&\lbrace \gamma^i , \alpha^j \rbrace = 0, \; \text{for}\;  i=0\; \text{or} \; i=j  \\
&\left[ \gamma^i , \alpha^j \right]  = 0, \; else
\end{align*}
We look at the two most common examples:\\

\noindent \textbf{Standard Representation}: Only $\alpha^2$ (and thus $\gamma^2$) is imaginary, all the other matrices are real. Thus $C = const. \, \gamma^2$. Conventionally: $C = i \gamma^2 = i \beta \alpha^2$\\

\noindent \textbf{Standard Representation}: All the $\gamma$'s are imaginary. In particular, $\beta$ is purely imaginary and all the $\alpha-$matrices are real. Hence, looking at \eqref{conjtransform}, we see that we can take $C = \mathds{1}$, i.e. charge-conjugation is just complex conjugation.

\newpage
\section{Miscellaneous}
\vspace*{4mm}

\begin{Lemma}[For use in \eqref{L-M}]
\mbox{}\\
For all $U \in \Ur(\HH)$ it is true that 
\begin{align*} \dim\ker (U_{++}) = \dim\ker(U^*_{--})\\
\dim\ker(U_{--}) = \dim\ker(U^*_{++}) \end{align*}
\end{Lemma}
\begin{proof} With respect to the splitting $\HH = \HH_+ \oplus \HH_-$, we write $U = U_{\rm even} + U_{\rm odd}$, with $U_{\rm even}$ the diagonal parts and $U_{\rm odd}$ the off-diagonal parts. $UU^* = U^*U = \mathds{1}$ then implies
\begin{enumerate}[i)] 
\item $(U^*U)_{\rm even} = U^*_{\rm even}U_{\rm even} + U^*_{\rm odd}U_{\rm odd} = \mathds{1}$
\item $(U^*U)_{\rm odd} =  U^*_{\rm even}U_{\rm odd} + U^*_{\rm odd}U_{\rm even} = 0 $
\item $(UU^*)_{\rm even} = U_{\rm even}U^*_{\rm even} + U_{\rm odd}U^*_{\rm odd} = \mathds{1}$
\item $(UU^*)_{\rm odd} =  U_{\rm even}U^*_{\rm odd} + U_{\rm odd}U^*_{\rm even} = 0 $
\end{enumerate}
\mbox{}\\
\noindent 
ii) implies \hspace{3cm} $U_{\rm odd} \ker(U_{\rm even}) \subset \ker(U^*_{\rm even})$
\vspace{2mm}
\item since \hspace{2cm} $U^*_{\rm even} U_{\rm odd}\, x = - U^*_{\rm odd}U_{\rm even}\, x = 0, \; \forall x \in \ker(U_{\rm even})$.
\vspace{2mm}
\item Similarly, iv) implies\hspace{1cm} $U^*_{\rm odd} \ker(U^*_{\rm even}) \subset \ker(U_{\rm even})$.
\item We conclude:
\begin{equation*} U_{\rm odd}\, \ker(U_{\rm even}) \subset \ker(U^*_{\rm even}) \stackrel{iii)}{=} U_{\rm odd}U^*_{\rm odd} \ker(U^*_{\rm even}) \subset U_{\rm odd} \ker(U_{\rm even}) \end{equation*}
And thus
\begin{equation*} U \ker(U_{\rm even}) = U_{\rm odd} \ker(U_{\rm even}) = \ker(U^*_{\rm even}) \end{equation*}
Separating w.r.to the $\pm-$splitting yields
\begin{equation}U \ker(U_{++}) = \ker (U^*_{--})\; \text{and}\; U \ker(U_{--}) = \ker (U^*_{++}) \end{equation}
which implies the claimed identities.\\ 
\end{proof}
\newpage
\noindent\textbf{Norm Estimates for the Dyson Series, Lemma \ref{Lem:Estimates}}\\

\noindent For the terms in the Dyson series \eqref{Dyson} we proove the norm-estimates  [Lem. \ref{Lem:Estimates}]:
\begin{equation*}\lVert U_n(t,t') \rVert \leq \frac{1}{n!}\Biggl( \int\limits_{t'}^t \lVert V(s) \rVert \mathrm{d}s \Biggr)^n, \; \forall n \geq 0 \end{equation*}
and 
\begin{align*}
& \lVert [\epsilon, U_1(t,t')] \rVert_2  \leq \int\limits_{t'}^{t} \lVert[\epsilon, V(s)]\rVert_2 \; \mathrm{d}s\\
\lVert [\epsilon ,  U_n(t,t')] \rVert_2 \leq & \frac{1}{(n-2)!} \int\limits_{t'}^t \lVert [\epsilon, V(s)] \rVert_2\; \mathrm{d}s \Biggl( \int\limits_{t'}^t \lVert V(r) \rVert \mathrm{d}r \Biggr)^{n-1}, \; \forall n \geq 1
\end{align*}

\begin{proof}
Obviously, $\lVert U_0(t,t') \rVert = \lVert \mathds{1}\rVert = 1,\; \forall t,t'$. And for $n \geq 1$: 
\begin{align*}U_n (t,t') &= (-i)^{n} \int\limits^{t}_{t'} V_I(s_1)  \int\limits^{s_1}_{t'}V_I(s_2) \ldots \int\limits^{s_{n-1}}_{t'}V_I(s_n)\; \mathrm{d}s_1 \ldots  \mathrm{d}s_{n}\\
\Rightarrow\; \lVert U_n(t,t') \rVert & \leq  \int\limits^{t}_{t'} \lVert V_I(s_1)\rVert  \int\limits^{s_1}_{t'}\lVert V_I(s_2)\rVert \ldots \int\limits^{s_{n-1}}_{t'}\lVert V_I(s_n)\rVert \; \mathrm{d}s_1 \ldots  \mathrm{d}s_{n}\\
& =   \int\limits^{t}_{t'} \lVert V(s_1)\rVert  \int\limits^{s_1}_{t'}\lVert V(s_2)\rVert \ldots \int\limits^{s_{n-1}}_{t'}\lVert V(s_n)\rVert \; \mathrm{d}s_1 \ldots  \mathrm{d}s_{n}\\
& = \frac{1}{n!}\Bigl(\int\limits^t_{t'} \lVert V(s) \rVert \mathrm{d}s \Bigl)^n
\end{align*}
For the $I_2$-estimates we first note that conjugation with $e^{iD_0t}$ doesn't change the Hilbert-Schmidt norm of the odd parts. Thus:
\begin{equation*}\lVert [\epsilon, U_1(t,t')] \rVert_2 = \lVert \int\limits_{t'}^{t} [\epsilon, V_I(s)] \mathrm{d}s\;\rVert_2 \leq \int\limits_{t'}^{t} \lVert[\epsilon, V_I(s)]\rVert_2 \; \mathrm{d}s = \int\limits_{t'}^{t} \lVert[\epsilon, V(s)]\rVert_2 \; \mathrm{d}s\end{equation*}
Furthermore,
\begin{align*}\lVert [\epsilon, U_{n+1} (t,t')] \rVert_2 & \leq \int\limits_{t'}^{t} \lVert[\epsilon, V_I(s)U_n(s,t')]\rVert_2 \; \mathrm{d}s\\
& \leq \int\limits_{t'}^{t} \lVert[\epsilon, V_I(s)]\, U_n(s,t')\rVert_2 \; \mathrm{d}s + \int\limits_{t'}^{t} \lVert V_I(s)[\epsilon, U_n(s,t')] \rVert_2 \; \mathrm{d}s\\
&\leq  \int\limits_{t'}^{t} \lVert[\epsilon, V_I(s)]\rVert_2\, \lVert U_n(s,t')\rVert \; \mathrm{d}s + \int\limits_{t'}^{t} \lVert V_I(s)\rVert \; \lVert[\epsilon, U_n(s,t')] \rVert_2 \; \mathrm{d}s\\
& = \int\limits_{t'}^{t} \lVert[\epsilon, V(s)]\rVert_2\, \lVert U_n(s,t')\rVert \; \mathrm{d}s + \int\limits_{t'}^{t} \lVert V(s)\rVert \; \lVert[\epsilon, U_n(s,t')] \rVert_2 \; \mathrm{d}s
\end{align*}
For n=1 this yields
\begin{align*}\lVert [\epsilon, U_2 (t,t')] \rVert_2 &\leq  \int\limits_{t'}^t\int\limits_{t'}^s \bigl( \lVert[\epsilon, V(s)]\rVert_2\, \lVert V(r)\rVert + \lVert V(s)\rVert \,  \lVert[\epsilon, V(r)]\rVert_2 \bigl)\, \mathrm{d}r\, \mathrm{d}s\\
& = \frac{1}{2} \int\limits_{t'}^t\int\limits_{t'}^t \bigl( \lVert[\epsilon, V(s)]\rVert_2\, \lVert V(r)\rVert + \lVert V(s)\rVert \,  \lVert[\epsilon, V(r)]\rVert_2 \bigl)\, \mathrm{d}r\, \mathrm{d}s\\
& = \int\limits_{t'}^t  \lVert[\epsilon, V(s)]\rVert_2 \; \mathrm{d}s \int\limits_{t'}^t  \lVert V(s)\rVert\, \mathrm{d}s
\end{align*}
And thus, inductively, for $n \geq2$
\begin{align*}
&\lVert [\epsilon, U_{n+1} (t,t')] \rVert_2  \leq \int\limits_{t'}^{t} \lVert[\epsilon, V(s)]\rVert_2\, \lVert U_n(s,t')\rVert \; \mathrm{d}s + \int\limits_{t'}^{t} \lVert V(s)\rVert \; \lVert[\epsilon, U_n(s,t')] \rVert_2 \; \mathrm{d}s\\
\leq & \frac{1}{n!} \int\limits_{t'}^{t} \lVert[\epsilon, V(s)]\rVert_2 \Bigl(\int\limits^s_{t'} \lVert V(r) \rVert \mathrm{d}r \Bigl)^n \mathrm{d}s + \frac{1}{(n-2)!} \int\limits_{t'}^t \lVert V(s) \rVert \int\limits_{t'}^s \lVert [\epsilon, V(s')] \rVert_2\, \mathrm{d}s' \Bigl(\int\limits_{t'}^s \lVert V(r) \rVert \mathrm{d}r \Bigr)^{n-1}\mathrm{d}s\\
 \leq & \frac{1}{n!} \int\limits_{t'}^{t} \lVert[\epsilon, V(s)]\rVert_2\; \mathrm{d}s\, \Bigl(\int\limits^t_{t'} \lVert V(r) \rVert \mathrm{d}r \Bigl)^n + \int\limits_{t'}^t \lVert [\epsilon, V(s)] \rVert_2\,\mathrm{d}s\; \frac{(n-1)}{(n-1)!} \int\limits_{t'}^t \lVert V(s) \rVert  \Bigl( \int\limits_{t'}^s \lVert V(r) \rVert \mathrm{d}r \Bigr)^{n-1}\mathrm{d}s\\
=\;&  \frac{1}{n!} \int\limits_{t'}^{t} \lVert[\epsilon, V(s)]\rVert_2\, \mathrm{d}s\, \Bigl(\int\limits^t_{t'} \lVert V(r) \rVert \mathrm{d}r \Bigl)^n + \; \frac{(n-1)}{n!}\,\int\limits_{t'}^{t} \lVert[\epsilon, V(s)]\rVert_2\; \mathrm{d}s\, \Bigl(\int\limits^t_{t'} \lVert V(r) \rVert \mathrm{d}r \Bigl)^n\\
=\;&  \frac{1}{(n-1)!} \int\limits_{t'}^{t} \lVert[\epsilon, V(s)]\rVert_2\; \mathrm{d}s\, \Bigl(\int\limits^t_{t'} \lVert V(r) \rVert \mathrm{d}r \Bigl)^n
\end{align*}
%&\leq \frac{1}{(n-1)!} \int\limits_{t'}^{t}\int\limits^s_{t'}\bigl(\lVert[\epsilon, V(s)]\rVert_2 \lVert V(s') \rVert + \lVert[\epsilon, V(s')]\rVert_2 \lVert V(s) \rVert \bigr) \Bigl(\int\limits^s_{t'} \lVert V(r) \rVert \mathrm{d}r \Bigl)^{n-1} \mathrm{d}s' \mathrm{d}s
%&= \frac{1}{(n-1)!} \int\limits_{t'}^{t}\int\limits^s_{t'}\bigl(\lVert[\epsilon, V(s)]\rVert_2 \lVert V(s') \rVert + \lVert[\epsilon, V(s')]\rVert_2 \lVert V(s) \rVert \bigr) \Bigl(\int\limits^{max\lbrace s,s'\rbrace}_{t'} \lVert V(r) \rVert \mathrm{d}r \Bigl)^{n-1} \mathrm{d}s' \mathrm{d}s\\
%&\leq \frac{1}{2}\frac{1}{(n-1)!} \int\limits_{t'}^{t}\int\limits^t_{t'}\bigl(\lVert[\epsilon, V(s)]\rVert_2 \lVert V(s') \rVert + \lVert[\epsilon, V(s')]\rVert_2 \lVert V(s) \rVert \bigr) \Bigl(\int\limits^s_{t'} \lVert V(r) \rVert \mathrm{d}r \Bigl)^{n-1} \mathrm{d}s' \mathrm{d}s\\
%\end{align*}
\end{proof}

\newpage 
\thispagestyle{empty}
\quad

\newpage
\thispagestyle{empty}

\vspace*{2cm}
\noindent {\huge\textbf{Selbstst\"andigkeitserkl\"arung}}

\begin{flushleft}
    \large Hiermit erkl\"are ich, dass ich die vorliegende Arbeit selbstst\"andig
        und ohne Benutzung anderer als der von mir angegebenen Quellen und
        Hilfsmittel verfasst habe.
   
 \vspace{3cm}

\noindent {\Large\textbf{Academic Honor Principle}}\\
\mbox{}\\
\noindent \large I hereby certify that I have completed the present thesis independently and
without the use of sources other than those stated in the thesis.
     \newline
\vspace{5cm}

Dustin Lazarovici\\
M\"unchen, 30. September 2011
\end{flushleft}

\newpage 
\thispagestyle{empty}
\quad 
\newpage
%\lhead{Danksagung}
\thispagestyle{empty}

\noindent {\huge\textbf{Danksagung}} \hfill {\LARGE{Acknowledgement}}

\vspace*{1.9cm}
\large{
Ich danke meinen Eltern, Doriana und Dan Lazarovici, die mir mein Studium erm\"oglicht haben, f\"ur ihre Liebe und ihre Unterst\"utzung und f\"ur ihre Geduld.\\

Ich danke meinem Bruder Remy daf\"ur, dass ich immer auf ihn z\"ahlen konnte.\\

Ich danke meinem Lehrer, Detlef D\"urr, von dem ich mehr gelernt habe, als man von irgendeinem Professor erwarten kann. Er hat mir meinen Glauben bewahrt, dass ein physikalisches Verst\"andins der Welt, im eigentlich Sinne, tats\"achlich m\"oglich ist.\\

Ich danke Peter Pickl f\"ur die Betreuung der Arbeit und daf\"ur, dass er immer die richtigen, das hei\ss t, physikalischen, Fragen stellt.\\

Ich danke Dirk Deckert und Franz Merkl f\"ur wertvolle Gespr\"ache und kompetenten Rat. Ohne sie h\"atte ich oft nicht gewusst, was ich eigentlich mache. \\

Ich danke meinen guten Freundinnen Janine Adomeit und Sarina Balkhausen f\"urs Korrekturlesen. Die verbliebenen Tippfehler -- es werden zahlreiche sein -- sind allein meiner Schludrigkeit und dem Zeitdruck geschuldet. \\

Ich danke Ir\`ene Lassmann f\"ur die K\"atzchen -- sie wei\ss\, schon, was gemeint ist.\\

Ich danke Ingrid Scherer und Verena Heinemann vom mathematischen Institut der LMU M\"unchen, ohne die ich im B\"urokratiedschungel verhungert w\"are.\\

}

\vspace{3cm}
\flushright Dustin Lazarovici\\

%\noindent Ich danke meinen Eltern, Doriana und Dan Lazarovici, die mir mein Studium erm\"oglicht haben, f\"ur ihre Liebe und ihre Unterst\"utzung und f\"ur ihre Geduld. Ich danke meinem Bruder, Remy Lazarovici, daf\"ur, dass ich immer auf ihn z\"ahlen konnte. Ich danke meinem Lehrer, Detlef D\"urr, von dem ich mehr gelernt habe, als man von irgendeinem Dozenten oder Betreuer erwarten kann. Er hat mir meinen Glauben bewahrt, dass ein physikalisches Verst\"andins der Welt, im eigentlich Sinne, tats\"achlich m\"oglich ist. Ich danke Peter Pickl f\"ur die Betreuung der Arbeit und daf\"ur, dass er immer die richtigen, d.h. physikalischen, Fragen stellt. Ich danke Dirk Deckert und Franz Merkl f\"ur wertvolle Gespr\"ache und kompetenten Rat. Ohne sie h\"atte ich oft gar nicht gewusst, was ich eigentlich mache. Ich danke meinen guten Freundinnen Janine Adomeit und Sarina Balkhausen f\"urs Korrekturlesen. Die verbliebenen Tippfehler -- es werden zahlreiche sein -- sind allein meiner Schludrigkeit und dem Zeitdruck geschuldet. Ich danke Ir\`ene Lassmann f\"ur die Katzen -- sie wei\ss, was gemeint ist. Ich danke Ingrid Scherer und Verena Heinemann vom mathematischen Institut der LMU M\"unchen, ohne die ich im B\"urokratiedschungel verhungert w\"are. Ich danke allen meinen Lehrern, Kommilitonen und Kollegen f\"ur alles, was sie mir beigebracht haben. 
 
\end{document}